%% file: QPA.tex
\documentclass[aps,pra,nofootinbib,reprint,superscriptaddress]{revtex4-2}
\usepackage{amsmath,amssymb,mathtools,amsthm,mathrsfs,physics,amsfonts,dsfont}
\usepackage[pass]{geometry}
\usepackage{pdfrender}
\usepackage{bibunits}
\usepackage{graphicx}
\usepackage{tikz}
\usetikzlibrary{quantikz2}
\usepackage{algorithm,algpseudocode}
\usepackage{tabularx,siunitx,booktabs}

\usepackage{xcolor}
\usepackage{colortbl}
\definecolor{lavender}{rgb}{0.64,0.47,0.87}
\definecolor{darklavender}{rgb}{0.48,0.35,0.65}
\definecolor{revisionblue}{RGB}{0,0,0}
\newcommand{\rev}[1]{#1}
\makeatletter
\newcommand{\normalsizemaketag}[1]{\hbox{\m@th\normalfont\normalsize#1}}
\newcommand{\normalsizeequationtag}{\let\maketag@@@\normalsizemaketag}
\makeatother

\usepackage{ytableau}
\usepackage[enableskew]{youngtab}
\usepackage{xparse}
\usepackage{caption,subcaption}
\usepackage{hyperref}
\hypersetup{
    colorlinks=true,
    linkcolor=darklavender,
    citecolor=darklavender,
    urlcolor=darklavender,
}
\usepackage{aliascnt}
\usepackage[nameinlink]{cleveref}

\usepackage{enumitem,comment,adjustbox,extarrows}
\usepackage{float}
\usepackage[normalem]{ulem}

\usepackage{etoolbox}
\usepackage{footmisc}

\makeatletter
\renewcommand\listoffigures{\@starttoc{lof}}
\renewcommand\listoftables{\@starttoc{lot}}

\makeatother

\newtheorem{theorem}{Theorem}
\theoremstyle{definition}
\newaliascnt{definition}{theorem}
\newtheorem{definition}[definition]{Definition}
\aliascntresetthe{definition}
\newaliascnt{lemma}{theorem}
\newtheorem{lemma}[lemma]{Lemma}
\aliascntresetthe{lemma}
\newaliascnt{corollary}{theorem}
\newtheorem{corollary}[corollary]{Corollary}
\aliascntresetthe{corollary}
\newaliascnt{notation}{theorem}

\aliascntresetthe{notation}
\newaliascnt{proposition}{theorem}

\aliascntresetthe{proposition}
\newaliascnt{conjecture}{theorem}

\aliascntresetthe{conjecture}
\newaliascnt{assumption}{theorem}

\aliascntresetthe{assumption}
\newaliascnt{observation}{theorem}

\aliascntresetthe{observation}
\newaliascnt{fact}{theorem}

\aliascntresetthe{fact}
\newaliascnt{remark}{theorem}
\newtheorem{remark}[remark]{Remark}
\aliascntresetthe{remark}
\newaliascnt{claim}{theorem}

\aliascntresetthe{claim}
\newaliascnt{example}{theorem}
\newtheorem{example}[example]{Example}
\aliascntresetthe{example}
\newaliascnt{problem}{theorem}

\aliascntresetthe{problem}
\newaliascnt{open}{theorem}

\aliascntresetthe{open}
\newaliascnt{hypothesis}{theorem}

\aliascntresetthe{hypothesis}
\newaliascnt{question}{theorem}

\aliascntresetthe{question}

\newcommand{\suppresetaliascounter}[1]{%
  \aliascntresetthe{#1}%
  \expandafter\let\csname theH#1\endcsname\theHtheorem
}

\newcommand{\supplementnextanchor}[1]{%
  \hypersetup{next-anchor=#1.S.\number\value{#1}}%
}

\newcommand{\beginsupplement}{
  \clearpage
  \onecolumngrid
  \newgeometry{margin=1in}
  \hypersetup{pageanchor=false}
  \setcounter{page}{1}
  \setcounter{section}{0}
\renewcommand{\thesection}{S\arabic{section}}
\gdef\theHsection{S.\arabic{section}}
\counterwithin{theorem}{section}
\setcounter{theorem}{0}
\renewcommand{\thetheorem}{\thesection.\arabic{theorem}}
\gdef\theHtheorem{\thesection.\arabic{theorem}}
\forcsvlist{\suppresetaliascounter}{definition,lemma,corollary,notation,proposition,conjecture,assumption,observation,fact,remark,claim,example,problem,open,hypothesis,question}
  \setcounter{equation}{0}
  \renewcommand{\theequation}{S\arabic{equation}}
  \gdef\theHequation{S.\arabic{equation}}
  \AtBeginEnvironment{equation}{%
    \gdef\theHequation{S.\arabic{equation}}%
    \supplementnextanchor{equation}%
  }
  \setcounter{table}{0}
  \renewcommand{\thetable}{S\arabic{table}}
  \gdef\theHtable{S.\arabic{table}}
  \AtBeginEnvironment{table}{\gdef\theHtable{S.\arabic{table}}}
  \AtBeginEnvironment{table*}{\gdef\theHtable{S.\arabic{table}}}
  \setcounter{figure}{0}
  \renewcommand{\thefigure}{S\arabic{figure}}
  \gdef\theHfigure{S.\arabic{figure}}
  \AtBeginEnvironment{figure}{\gdef\theHfigure{S.\arabic{figure}}}
  \AtBeginEnvironment{figure*}{\gdef\theHfigure{S.\arabic{figure}}}
 \renewcommand{\bibnumfmt}[1]{[S##1]}
\renewcommand{\citenumfont}[1]{S##1}
}

\definecolor{celeste}{rgb}{0.09, 0.81, 0.9}
\definecolor{emerald}{rgb}{0, 0.677778, 0.555556}

\crefname{lemma}{Lemma}{Lemmas}
\crefname{proposition}{Proposition}{Propositions}
\crefname{definition}{Definition}{Definitions}
\crefname{theorem}{Theorem}{Theorems}
\crefname{conjecture}{Conjecture}{Conjectures}
\crefname{corollary}{Corollary}{Corollaries}
\crefname{claim}{Claim}{Claims}
\crefname{notation}{Notation}{Notations}
\crefname{assumption}{Assumption}{Assumptions}
\crefname{observation}{Observation}{Observations}
\crefname{fact}{Fact}{Facts}
\crefname{remark}{Remark}{Remarks}
\crefname{example}{Example}{Examples}
\crefname{problem}{Problem}{Problems}
\crefname{open}{Open Problem}{Open Problems}
\crefname{hypothesis}{Hypothesis}{Hypotheses}
\crefname{question}{Question}{Questions}
\crefname{section}{Section}{Sections}
\crefname{paragraph}{Section}{Sections}
\crefname{appendix}{Appendix}{Appendices}
\crefname{figure}{Fig.}{Figs.}
\crefname{equation}{Eq.}{Eqs.}
\crefname{table}{Table}{Tables}
\makeatletter
\renewcommand{\p@paragraph}{\p@subsection\thesubsection\,}

\makeatother
\theoremstyle{definition}

\input{notations.tex}
\begin{document}

\title{Quantum Purity Amplification for Arbitrary Eigenstates and Multiple Outputs}

\author{Zhaoyi Li}
\thanks{These authors contributed equally.\\Emails: \texttt{ladmon@mit.edu}, \texttt{edmt@math.ku.dk}.}
\affiliation{Department of Physics, Massachusetts Institute of Technology, Cambridge MA 02139, USA}
\affiliation{The NSF AI Institute for Artificial Intelligence and Fundamental Interactions}

\author{Elias Theil}
\thanks{These authors contributed equally.\\Emails: \texttt{ladmon@mit.edu}, \texttt{edmt@math.ku.dk}.}
\affiliation{Department of Mathematical Sciences, University of Copenhagen, 2100 Copenhagen, Denmark}

\author{Aram W. Harrow}
\affiliation{Department of Physics, Massachusetts Institute of Technology, Cambridge MA 02139, USA}

\author{Isaac Chuang}
\affiliation{Department of Physics, Massachusetts Institute of Technology, Cambridge MA 02139, USA}

\date{\today}

\begin{abstract}
Quantum purity amplification (QPA) is the task of coherently transforming $n$ copies of a mixed state into high-fidelity copies of a chosen eigenstate. We develop a general theory of QPA for $n$ input copies, $m$ output copies, and arbitrary choices of target eigenstate, local dimension $d$, and input spectrum. We introduce the overhang-removal protocol, derive its all-site and one-site fidelities across output regimes, and show that it is optimal in several regimes, both asymptotically and nonasymptotically. For the asymptotic analysis, we use a path-graph parametrization to show that, when the target eigenvalue has a constant spectral gap $D_{k,\mathrm{min}}$, achieving worst-case all-site error $\varepsilon$ scales as $O(m/(\varepsilon D_{k,\mathrm{min}}^2))$. When $m/n$ approaches a constant, the performance exhibits phase-like regimes, which we characterize explicitly. For the nonasymptotic analysis, we develop a theory of generalized Young diagrams that yields tight sample complexity bounds and provides the first tight dimension-uniform guarantee for optimal QPA. We also provide asymptotically efficient implementations of the overhang-removal protocol. Together, these results establish QPA as a rigorous example of coherent quantum information processing with dimension-uniform sample complexity, supplying the technical foundation for the coherent-incoherent separation developed in the companion work.
\end{abstract}

\maketitle
\begin{bibunit}[apsrev4-2]
\section{Introduction}

Physical quantum information processing is inevitably noisy, since any quantum system can couple to its environment, and repeated preparations generally produce a mixed state $\rho$ rather than a pure eigenstate.
Over the past three decades, an approach to noise suppression has emerged in which coherent processing of \(n\) independent copies exploits the collective structure of \(\rho^{\otimes n}\), complementing existing quantum error correction frameworks~\cite{G98,CDT09,LKH23}.
Early proposals along these lines include error symmetrization~\cite{P99,B97}, and soon after, representation theoretic tools were introduced to clarify the symmetry structure underlying multi-copy processing~\cite{CEM99,KW99}. Subsequent developments evolved into a broader toolbox, including SWAP-test-based constructions and density matrix exponentiation type strategies~\cite{F16,CFLL+23,YCHC+24,DGHM+25,BHPM25,KM25}.
The first optimality results were obtained in the qubit setting, but several generalizations have since been explored, including extensions to multiple output qubits~\cite{KW99} and, more recently, optimal qudit (\(d\)-dimensional quantum state) strategies brought under the quantum purity amplification (QPA) framework~\cite{LFIC24} via symmetry reduction and asymptotic analysis.

Despite this progress, three central gaps remain, including the absence of a fully general QPA theory and a sharp dimension-uniform sample-complexity separation between coherent QPA and its best incoherent (measurement-mediated) counterparts, as developed in the companion paper \emph{An Exponential Sample-Complexity Advantage for Coherent Quantum Inference}~\cite{CI26}.

First, the QPA formulation of Ref.~\cite{LFIC24} is intrinsically a single-output primitive, whereas many resource-state preparation tasks instead call for an \(n\!\to\! m\) coherent channel that produces multiple purified outputs. This is especially relevant for states prepared by digital or adiabatic simulation for later reuse, problem-specific states, and non-Clifford resource states~\cite{SS06,BGBW18,HHL09,LMR14,R25,DGHM+25}, where one would often expect to need multiple such copies rather than just one. We also prove that a naive strategy that partitions the \(n\) inputs into \(m\) batches and applies a single-output protocol to each batch is in general suboptimal, incurring a factor of \(\bigO{m^2}\) in sample complexity, whereas the optimal scaling is only \(\bigO{m}\). This makes it essential to characterize the tradeoff between yield and accuracy.

Second, a central advantage of optimal QPA over both full-scale quantum error correction and suboptimal QPA protocols lies in its reduced sample complexity~\cite{LFIC24}. To make this advantage precise, it is essential to determine the achievable fidelities sharply. Existing analyses either yield asymptotics in big-$O$ form without dimension-uniform control~\cite{LFIC24}, or provide dimension-uniform guarantees that are too weak to match the leading constants of the asymptotic regime~\cite{DGHM+25,GLLP+25}.

Finally, most available guarantees for general system sizes focus on amplifying the principal eigenstate, while many algorithmic settings such as spectral feature extraction~\cite{HB26} and excited-state preparation~\cite{MIY24} require coherent extraction of a specified eigencomponent that need not be dominant. Overall, the landscape remains only partially mapped, largely due to the mathematical complexity introduced by mixed symmetry structure.

We develop a general theory of QPA that closes all three gaps above.

\begin{enumerate}[leftmargin=*, itemindent=0pt, labelsep=0.5em]
\item \emph{Asymptotic optimality of general QPA.}
We extend the QPA task to the general setting of extracting \(m\) copies of the \(k\)-th input eigenstate~(\cref{sec:general_QPA}). We introduce the overhang-removal protocol for the task and derive its asymptotic performance for nondegenerate spectra~(\cref{equ:intensive_asymptote,equ:extensive_asymptote,equ:one_site_asymptote}).
We show that the overhang-removal protocol is optimal across a broad range of settings: it is asymptotically all-site optimal pointwise in \(m\), sector-wise all-site optimal against extensive deviations, and asymptotically one-site optimal uniformly in \(m\) among protocols with totally symmetric outputs~(\cref{par:asymptotically_optimal_scaling}), with these guarantees applying to various loss functions~(\cref{equ:infidelity_loss,equ:purified_loss,equ:bures_loss,equ:trace_loss,equ:cross_entropy_loss}).

\item \emph{Sharp nonasymptotic bounds.} We provide sharp, dimension-independent bounds on the performance of the overhang-removal protocol for arbitrary input spectra. These bounds are tight with the optimal leading-order constant $1$ and capture the correct asymptotic scaling~(\cref{equ:utility_bound_non_asymptotic,equ:utility_remainder_bound_non_asymptotic,equ:one_site_non_asymptotic}). Under explicit row-gap conditions, we further prove exact sector-wise optimality of overhang removal. For depolarized inputs, overhang removal is exactly optimal for the one-site utility for all parameters~(\cref{par:strict_optimality}).

\item \emph{Exponential separation between coherent and incoherent protocols (Together with companion paper~\cite{CI26}).} In the companion paper, we use the one-site utility bounds obtained here to derive a sample-complexity lower bound for incoherent QPA protocols based on measurement-mediated strategies that scales with $d$, and is therefore exponential in the local system size. Taken together, these results exhibit an exponential versus constant separation between incoherent and coherent protocols in their sample-complexity dependence on the local system size.

\item \emph{Efficient implementations.} We give circuit-level implementations of the overhang-removal protocol based on generalized quantum phase estimation and on Clebsch--Gordan transforms. The GQPE-OQPA implementation uses \(\bigO{n^4\ln^p(dn/\varepsilon)}\) gates with exponentially small failure probability when \(p_{k} - m/n\geq c>0\), while the CG-OQPA implementation uses \(\bigO{dr^3(n+m)\ln^p(d(n+m)/\varepsilon)}\) gates, where \(r\) is the input rank, and admits a streaming-memory version using \(\bigO{dr\ln^p(d(n+m)/\varepsilon)}\) qubits~(\cref{subsec:implementations}).

\item \emph{Mathematical advances.} To derive the above results, we make three main mathematical advances. Due to the inherent symmetry of the QPA optimization problem, we decompose it into independent optimization of each symmetry sector and use concentration bounds for the Schur--Weyl (SW) distribution~\cite{OW16} to obtain overall performance guarantees from sector-wise results~(\cref{equ:SW_decomposition_input_state,thm:main_part_concentration_YD,equ:main_part_proof_outline_f_all_as_average}). At the sector-wise level, we: (i) introduce a novel parametrization of Weyl tableaux (WTs) via path graphs, which lets us derive the asymptotics of sector-wise fidelities~(\cref{thm:main_part_limit_Weyl_average_geometric_average,fig:main_part_gt_path_param}); (ii) introduce generalized Young diagrams (gYDs) with constraints, which give sharp lower bounds on the fidelity on each sector and provide an alternative proof of Schur log-concavity~(\cref{thm:main_part_splitting_diagram_average_weight,thm:main_part_monotonicity_constrained_generalized_averages,thm:inequality_Schur_polynomials_products}); and (iii) employ $F$-symbols to relate one-site figures of merit to their all-site counterparts and characterize the \rev{domination relations} needed for optimality~(\cref{equ:main_part_proof_outline_one_fidelity_decomposition,equ:main_part_proof_outline_fusion_coefficient_1,subsec:props_F_symbols}).
\end{enumerate}

This paper complements Ref.~\cite{CI26}, which develops the broader theory of coherent quantum inference, with QPA as a central special case. Here we provide full derivations and extended discussion of the QPA-specific results stated there, together with additional developments tailored to general QPA. After the posting of our preprint, subsequent independent work proved a bound related to~\cref{equ:utility_bound_non_asymptotic} using different techniques~\cite{SSW26}.

\section{The task of general QPA}
\label{sec:general_QPA}
General QPA protocols coherently process multiple copies of an unknown quantum state to amplify its $k$-th eigenstate. Informally, given $n$ copies of the input state
$\rho = \sum_{i=1}^{d} p_{i} \psi_{i}$
where $\psi_{i} = \ketbra{\psi_{i}}{\psi_{i}}$ with spectrum $\boldsymbol{p}\dot{=}\mleft(p_1,\ldots,p_d\mright)$, a QPA protocol is a quantum channel $\mathcal{T}$ such that the output approximates the pure states
\begin{equation}
    \mathcal{T}\mleft(\rho^{\otimes n}\mright) \approx \psi_k^{\otimes m},
\end{equation}
as illustrated in~\cref{fig:qpa_illustration}. We quantify the performance of a QPA protocol using two risk functionals, each corresponding to a different operational use of the output. The all-site risk $^k\mathcal{L}_{\mathrm{all}}$ assesses the protocol when all \(m\) output copies are used jointly as inputs to a later quantum subroutine. The second is the marginal one-site risk $^k\mathcal{L}_{\mathrm{one}}$, which is relevant when only a single output copy is used and any correlations with the remaining outputs are disregarded. We also note that the case with $k=1, m=1$ reduces to one-output principal-eigenstate QPA, recovering the setting studied in Ref.~\cite{LFIC24}.
\begin{figure}
    \centering
    \includegraphics[width=\columnwidth]{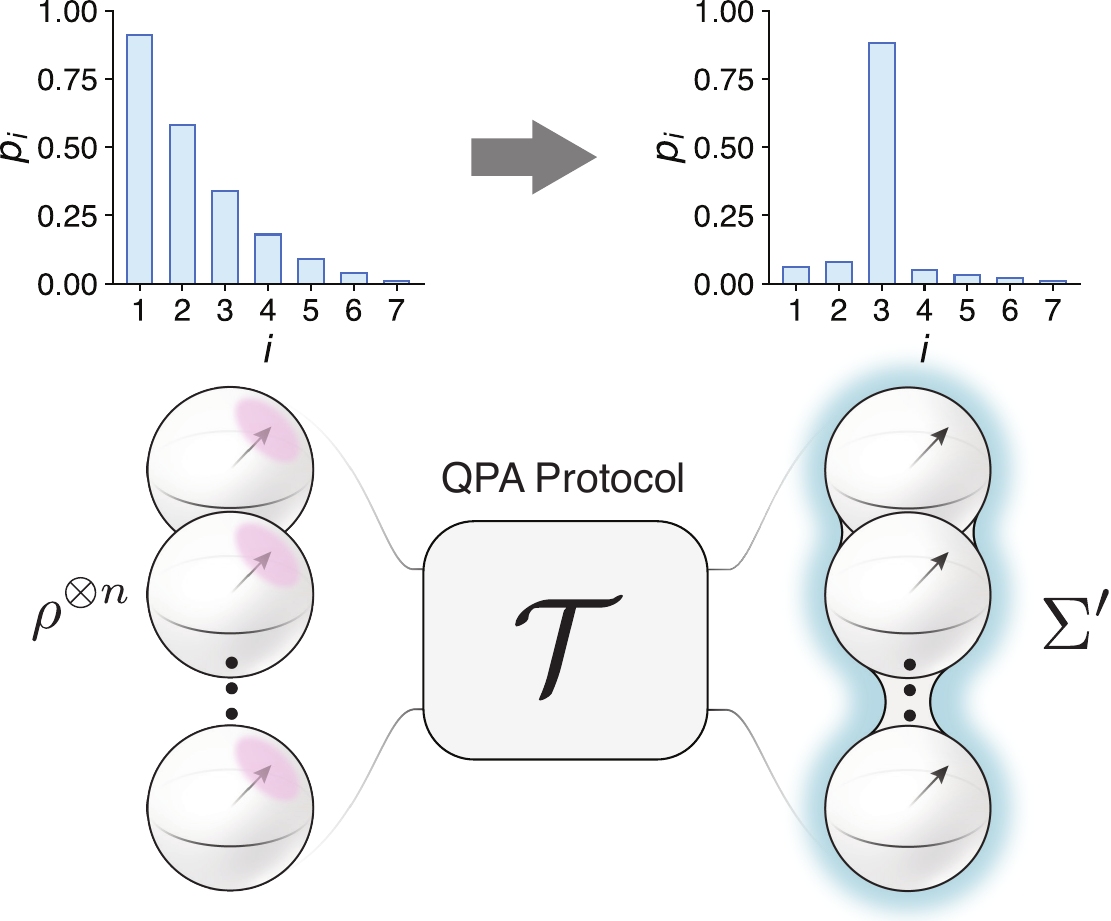}
    \caption{Illustration of the general QPA task for the $k$-th eigenstate, $m$ outputs, and $n$ inputs. Given $n$ copies of an unknown state $\rho$ with sorted spectrum $\boldsymbol{p}\dot{=}(p_1,\ldots,p_d)$, the QPA protocol $\mathcal{T}$ produces an $m$-partite output $\Sigma^\prime$ that approximates $m$ copies of the $k$-th eigenstate $\psi_k^{\otimes m}$.}
    \label{fig:qpa_illustration}
\end{figure}

We formalize QPA as a coherent quantum inference task~\cite{CI26} as follows.
\begin{definition}[General QPA]
Choose $1\leq k \leq d$, positive integers $n$ and $m$, and a spectrum $\boldsymbol{p}\dot{=}\mleft(p_1,\ldots,p_d\mright)$ whose entries we take, without loss of generality, to be sorted and whose $k$-th eigenvalue is nondegenerate. Let $\mathcal{X}\subseteq\mathcal{B}(\mathbb{H})$ be the unitary orbit of states with this spectrum, where $\mathbb{H}=\mathbb{C}^d$:
    \begin{equation}
        \mathcal{X}\coloneqq \left\{U\mathrm{diag}(\boldsymbol{p})U^\dagger \,\middle|\, U\in SU(d)\right\} \, .
    \end{equation}
For each $\rho\in \mathcal{X}$, let $\psi_k$ denote the projection onto the eigenspace corresponding to $p_{k}$, and let
    \begin{equation}
        ^k\Gamma(\rho) \coloneqq \psi_k^{\otimes m}, \quad {^k}\gamma(\rho) \coloneqq \psi_k
    \end{equation}
    be the all-site and one-site target states, respectively. For a quantum channel $\mathcal{T}:\mathbb{H}^{\otimes n}\to \mathbb{H}^{\otimes m}$, its performance is defined by the worst-case all-site risk
    \begin{equation}
        \label{equ:all_site_risk}
        ^k\mathcal{L}_{\mathrm{all}}\mleft(\mathcal{T}\mright)
        \coloneqq \sup_{\rho\in\mathcal X}L\mleft(^k\Gamma(\rho),\Sigma^\prime\mright) .
    \end{equation}
\noindent Here, $\Sigma^\prime=\mathcal{T}(\rho^{\otimes n})$ denotes the output corresponding to the input $\rho^{\otimes n}$, and \(L\) denotes the loss function. The worst-case one-site risk is defined analogously, but we restrict the optimization domain to channels \(\mathcal{T}:\mathbb{H}^{\otimes n}\to \mathbb{W}^{\yd{m}}\), where \(\mathbb{W}^{\yd{m}}\) is the totally symmetric subspace of \(\mathbb{H}^{\otimes m}\). We then set
    \begin{equation}
        \label{equ:one_site_risk}
        ^k\mathcal{L}_{\mathrm{one}}\mleft(\mathcal{T}\mright)\coloneqq \sup_{\rho\in\mathcal X}L\mleft(^k\gamma(\rho),\Tr_{2,\ldots,m}\mleft(\Sigma^{\prime}\mright)\mright) \, .
    \end{equation}
    In both the all-site and one-site cases, we equivalently quantify performance via the utility $\mathcal{U}_{(\cdot)} \coloneqq 1-\mathcal{L}_{(\cdot)}$, with $(\cdot)\in\{\mathrm{all},\mathrm{one}\}$.
\end{definition}

We first remark that one could define \(^k\mathcal{L}_{\mathrm{one}}\) over the same optimization domain as \(^k\mathcal{L}_{\mathrm{all}}\).
Here, however, we restrict the optimization to \(\yd{m}\)-output channels, thereby requiring all outputs to be ``the same.'' Beyond this interpretational motivation, this restriction is sufficient for our purposes, including taking the entanglement-breaking limit in the companion paper~\cite{CI26}. At the same time, we demonstrate that relaxing this restriction can lead to different optimal channels; see, e.g., the discussion in~\cref{subsec:optimality_and_robustness,fig:edge_cases}.

A separate point concerns the choice of the admissible family $\mathcal X$. We choose the target eigenstate index \(k\) and specialize \(\mathcal X\) to a single-spectrum unitary orbit. This is motivated by the fully agnostic setting, in which no information about the input eigenbasis is assumed and the performance is evaluated over all orientations. Importantly, however, the overhang-removal protocol introduced later depends only on \(k\), and not on the spectrum itself; the spectrum enters only through the achieved risk. We therefore specialize to a single unitary orbit to make the spectral dependence of the performance explicit. The same results can consequently be extended to settings in which \(\mathcal X\) contains multiple unitary orbits.

Finally, throughout the paper, we take the loss, without loss of generality, to be the infidelity~\cite{BJ23},
\begin{equation}
    \label{equ:infidelity_loss}
    L\mleft(\sigma,\rho\mright)
    = 1 - F(\sigma,\rho)
    = 1 - \Tr\mleft(\sqrt{\sigma}\sqrt{\rho}\mright)^{2} \, .
\end{equation}
The corresponding utility is denoted by $\mathcal{F}_{(\cdot)}$.
The reason is that many other losses can be expressed as monotone functions of the infidelity loss, including those based on the purified and Bures distances:
\begin{subequations}
\begin{align}
L_\mathrm{Purified}
&= P = \sqrt{L} \, ,\label{equ:purified_loss}\\
L_\mathrm{Bures}
&=\frac{D_B}{\sqrt{2}} = \sqrt{1-\sqrt{1-L}} \, .\label{equ:bures_loss}
\end{align}
\end{subequations}
By a symmetry argument given later in the paper, the QPA output is diagonal in the eigenbasis of the pure target state. Consequently, trace-distance and cross-entropy losses reduce to
\begin{subequations}
\begin{align}
L_\mathrm{Trace}
&=D_{\mathrm{tr}} = L \, ,\label{equ:trace_loss}\\
L_\mathrm{Cross\,entropy}
&=H_\times=-\ln(1-L)\,.\label{equ:cross_entropy_loss}
\end{align}
\end{subequations}

\subsection{Overhang removal rule and fidelity laws}
\label{subsec:overhang_removal}
The search for an optimal QPA protocol is guided by the symmetry structure of the problem. As we show below, the same protocol emerges for both the one-site and all-site risks. We first motivate and present this protocol and its performance guarantees before turning to the optimality proofs.

An optimal QPA protocol can be taken to have a universal three-step structure: (1) Schur sampling, which resolves the input into symmetry sectors; (2) isometric embedding, which reshapes the symmetry structure of the inputs; and (3) tracing out, which removes the excess degrees of freedom. See~\cref{fig:qpa_channel} for an illustration.

\begin{figure}[H]
    \centering
    \includegraphics[width=\columnwidth]{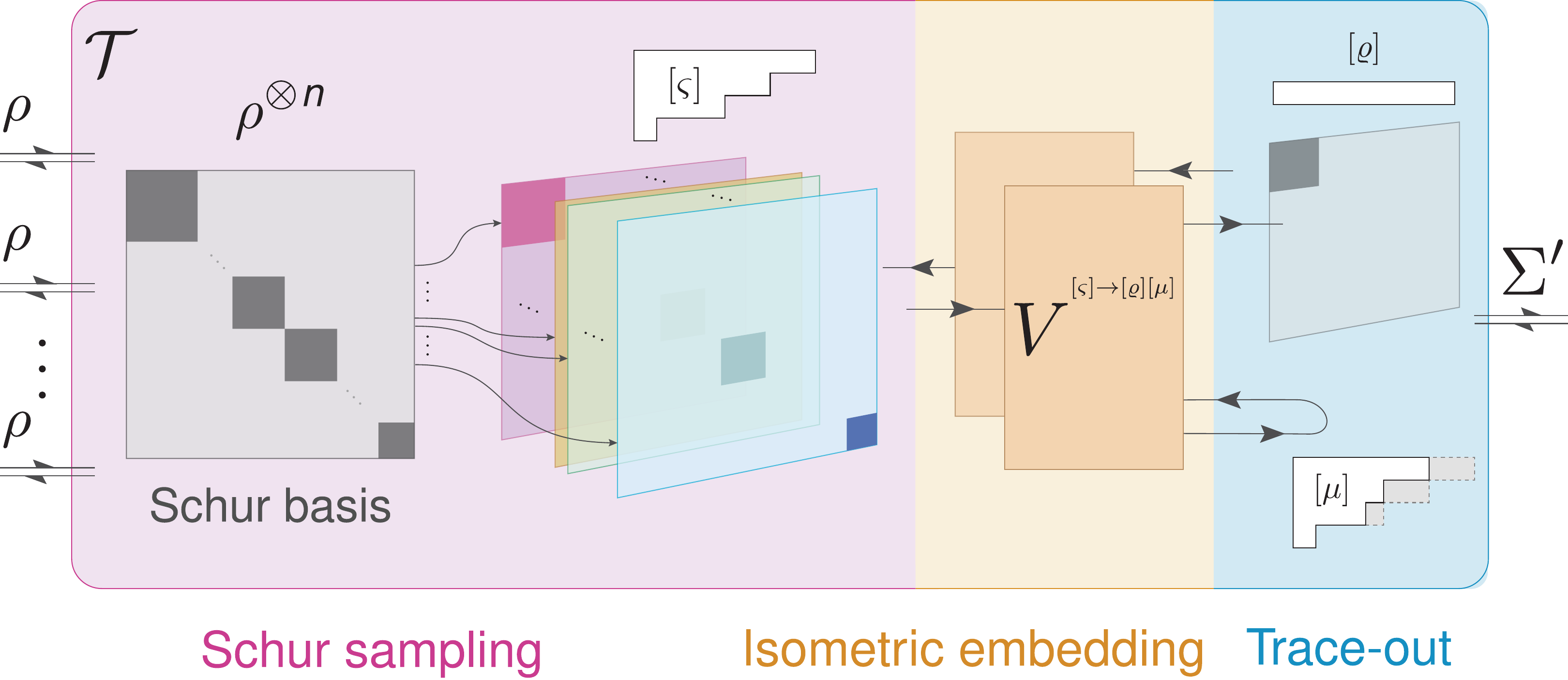}
    \caption{Schematic diagram of an optimal QPA protocol. On each symmetry sector $\yd{\varsigma}$, an extremal channel maps the input irreducible representation (irrep) $\mathbb{W}^{\yd{\varsigma}}$ to the output irrep $\mathbb{W}^{\yd{m}}$ via an intertwining Stinespring isometry, tracing out the remainder irrep $\mathbb{W}^{\yd{\mu}}$.}
    \label{fig:qpa_channel}
\end{figure}

This universal structure is dictated by symmetry; see Ref.~\cite[Section~S1 C]{CI26}: QPA is exchange-invariant under permutations of the input and output copies, i.e., under $S_n\times S_m$, and covariant under joint $U(d)$ rotations of the inputs and outputs. SW duality then decomposes the input and output spaces into symmetry sectors labeled by Young diagrams (YDs), respectively $\yd{\varsigma}\dot{=}\yd{\varsigma_{1},\ldots,\varsigma_{d}}$ and $\yd{\varrho}\dot{=}\yd{\varrho_{1},\ldots,\varrho_{d}}$. Each sector factors into a Specht module $\mathbb{V}^{\yd{\varsigma}}$, carrying the permutation-group action, and a Weyl module $\mathbb{W}^{\yd{\varsigma}}$, carrying the unitary-group action.

By linearity of the utility, twirling cannot decrease the utility, so we may restrict to symmetric channels.
Permutation symmetry then allows the optimal channel to be chosen to act trivially on the Specht modules. Thus, without loss of generality, step (1) may be taken to be Schur sampling. On each input symmetry sector, unitary covariance leaves only the choice of the isometric embedding in step (2). Specifically, each sector-wise channel is specified by
\label{equ:stinespring_dilation}
\begin{equation}
\begin{aligned}
\label{equ:main_part_proof_outline_irrep_channel_decomposition}
\extChannel{\mathcal{T}}{\yd{\varsigma}}{\yd{\varrho}}{\yd{\mu}}(\cdot)
= \Tr_{\yd{\mu}}\mleft(
\intertwiner{W}{\yd{\varsigma}}{\yd{\mu}}{\yd{\varrho}}{}
\,\cdot\,
\intertwineradj{W}{\yd{\varsigma}}{\yd{\mu}}{\yd{\varrho}}{}
\mright),
\end{aligned}
\end{equation}
where
$\intertwiner{W}{\yd{\varsigma}}{\yd{\mu}}{\yd{\varrho}}{}:
\mathbb{W}^{\yd{\varsigma}}\to
\mathbb{W}^{\yd{\mu}}\otimes\mathbb{W}^{\yd{\varrho}}$
is the intertwining Stinespring isometry, given by the inverse Clebsch--Gordan transform restricted to \(\mathbb{W}^{\yd{\varsigma}}\). The adjoint is denoted by $\intertwineradj{W}{\yd{\varsigma}}{\yd{\mu}}{\yd{\varrho}}{}$. Here $\yd{\varsigma}$, $\yd{\varrho}$, and $\yd{\mu}$ label the input, output, and environment irreps, respectively; for brevity, when labeling registers, we use $\yd{\lambda}$ as shorthand for the corresponding Weyl module $\mathbb{W}^{\yd{\lambda}}$. See Ref.~\cite[Theorem~19]{MT25}.
We refer the reader to the companion paper~\cite{CI26} for a more general treatment of symmetries in coherent-inference protocols.

For the all-site risk, only the totally symmetric output irrep can have risk less than one, while for the one-site risk the output is restricted to \(\yd{m}\)-output channels by definition.
Hence, in both cases, we may fix \(\yd{\varrho}=\yd{m}\). Therefore, the only remaining factor that determines the channel is the environment irrep \(\yd{\mu}\). Its admissibility is governed by the Littlewood--Richardson (LR) rules, under which \(\yd{\mu}\) may be viewed as the YD obtained from \(\yd{\varsigma}\) by removing \(m\) boxes. Here and throughout, we use the term YDs to include mixed YDs, whose entries may be negative; see Supplemental \cref{sec:math_prelim_supp}. Adopting the convention \(\varsigma_{d+1}=-\infty\), let \(\Delta_{i,j}\coloneqq\varsigma_i-\varsigma_j\), for \(1\leq i<j\leq d+1\), denote the row gaps. For \(1\leq k\leq d\), we refer to the boxes in row \(k\) of \(\yd{\varsigma}\) that extend beyond row \(k+1\), i.e.,\ its \(\Delta_{k,k+1}\) rightmost boxes, as the \emph{overhang} of row \(k\). We find that the \emph{overhang removal rule} presented below is optimal across a broad range of asymptotic and nonasymptotic regimes.

\begin{definition}[Overhang Removal Rule]
\label{def:overhang_removal_rule}
    Define the terminal row index as \(i^\ast \coloneqq \min\{i\mid \Delta_{k,i+1}>m\}\).
    Then the environment YD $\yd{\mu}$ is chosen to be:
    \begin{equation}
    \mu_i \coloneqq
    \mleft\{\begin{aligned}
    &\varsigma_i, & 1 \leq i < k,\\[4pt]
    &\varsigma_{i+1}, & k \leq i < i^\ast,\\[4pt]
    &\varsigma_k - m, & i = i^\ast,\\[4pt]
    &\varsigma_i, & i^\ast < i \leq d.
    \end{aligned}\mright.
    \end{equation}
\end{definition}

The overhang removal rule corresponds to removing boxes in the overhang on each row, starting with the $k$-th row and moving downwards, as illustrated in \cref{fig:overhang_removal}. This is a generalization of the optimal $\yd{\mu}$ described in Ref.~\cite{LFIC24}, where only one box is removed from the first overhang.

\begin{figure}[H]
    \centering
    \includegraphics[width=0.95\columnwidth]{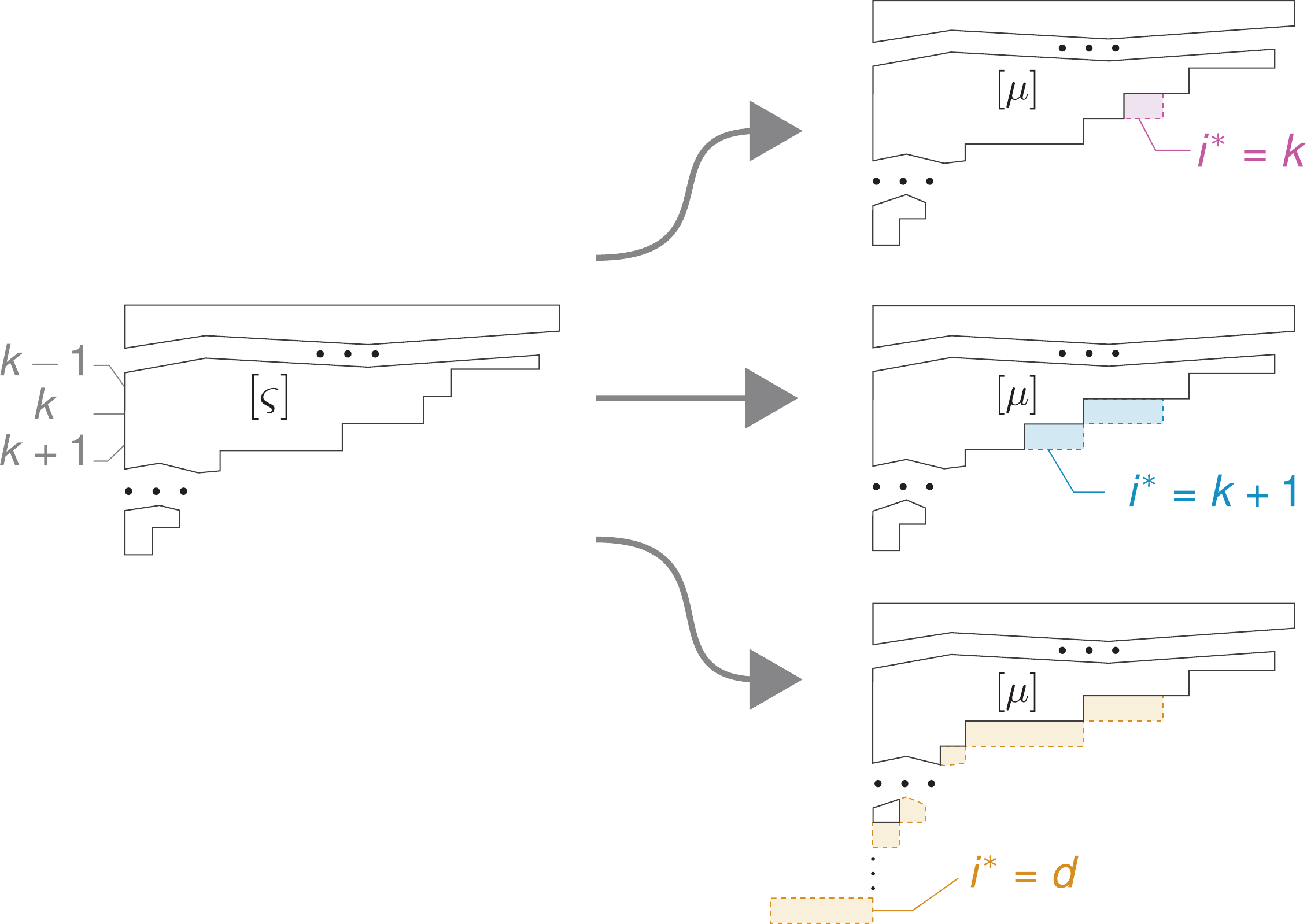}
    \caption{Illustration of the overhang removal rule. Three typical cases are shown: pink with the first overhang partially removed; blue with one full overhang removed, continuing into the next row; yellow with all overhangs removed and further into the $d$-th row, creating virtual boxes and turning the YD into a mixed YD. Starting from the input YD $\yd{\varsigma}$, boxes are removed sequentially beginning at row $k$ and proceeding downward until $m$ boxes have been removed; the remaining diagram is the environment YD $\yd{\mu}$.}
    \label{fig:overhang_removal}
\end{figure}

\subsection{Performance guarantee}
We analyze the performance of overhang removal first asymptotically, revealing its large-\(n\) scaling, and then nonasymptotically to obtain quantitative finite-\(n\) guarantees. In each analysis, we consider both the all-site and one-site risks.

\paragraph{Asymptotic performance}

For the all-site risk, we study two output regimes. In the \emph{intensive regime}, \(m\) is treated as a parameter independent of \(n\), and the asymptotic bounds are pointwise in \(m\), with the constants in the remainder allowed to depend on \(m\). In the \emph{extensive regime}, \(m\) may scale with \(n\), and the bounds are uniform in \(m\). The performance of overhang removal for the all-site risk with infidelity loss, defined in~\cref{equ:infidelity_loss}, differs markedly between the two regimes. In the intensive regime, full amplification is possible, and we obtain an expansion up to \(\bigO{n^{-1}}\). In the extensive regime, the relevant expansion is up to \(\bigO{1}\), and full purification becomes impossible once \(m\) grows linearly or faster with \(n\).

For intensive \(m\), the risk admits a simple asymptotic expansion in the large-$n$ limit when evaluated on the overhang-removal protocol:
\begin{equation}
\label{equ:intensive_asymptote}
^k\mathcal{L}_{\mathrm{all}}=\frac{m}{n}\sum_{i\neq k}\frac{p_i}{D^2_{k,i}}
+\lilo[d,\boldsymbol{p},m]{n^{-1}},
\end{equation}
\noindent where $D_{k,i}=p_k-p_i$ are the gap sizes for $1\leq i\leq d$, and the $o$-term can potentially depend on $d$, $\boldsymbol{p}$, and \(m\). Intuitively, in the intensive regime QPA can amplify the state arbitrarily close to the target as \(n\) increases, provided the target state remains spectrally distinguishable from its neighbors.

In the extensive regime, we write \(R\coloneqq m/n\). When \(R>0\), it is not possible to achieve a perfect asymptotic all-site utility. Nevertheless, one can still obtain asymptotic characterizations, proved in~\cref{thm:extensive_overhang_all_site_fidelity}:
\begin{equation}
\begin{aligned}
\label{equ:extensive_asymptote}
^k\mathcal{F}_\mathrm{all} = \prod_{i=k+1}^{I^\ast}\frac{D_{k,i}^2}{p_{k}R}
    \prod_{i\notin \{k,\ldots,I^\ast\}} \frac{D_{k,i}^2}{D_{k,i}^2 + p_{i}R} +\lilo[d,\boldsymbol{p}]{1}.
\end{aligned}
\end{equation}
\noindent With the convention \(D_{k,d+1}\coloneqq\infty\), we define the macroscopic terminal index
$I^\ast \coloneqq \min\{i\mid\, D_{k,i+1}>R\}$, and the above expression is uniform in \(m\). Thus, for any \(m\) with arbitrary dependence on \(n\), the performance approaches the corresponding asymptotic expression at the stated rate as \(n\to\infty\). The utility has cusps when the removal reaches different overhangs, displaying phase-like behavior that will be exemplified later.

For the one-site risk, \(^k\mathcal{L}_{\mathrm{one}}\) converges to $0$ uniformly in \(m\), and we obtain the stronger \(\bigO{n^{-1}}\) characterization even in the extensive regime. Asymptotically, as proved in~\cref{thm:asymptotics_one_fidelity},
\begin{equation}\begin{aligned}
\label{equ:one_site_asymptote}
&^k\mathcal{L}_{\mathrm{one}} \!=\!\frac{1}{n}\mleft(\sum\limits_{i\neq k}\frac{p_{i}}{D_{k,i}^2} \!+\!\! \sum\limits_{i=k+1}^{I^\ast}\!\! \mleft(D_{k,i}^{-1} - R^{-1}\mright)\mright) + \lilo[d,\boldsymbol{p}]{n^{-1}}.
\end{aligned}\end{equation}

\paragraph{Nonasymptotic performance}

Although these asymptotic characterizations are precise for large \(n\), they do not by themselves provide the finite-\(n\) convergence guarantees needed in practice. Indeed, the \(o\)-terms may still hide a substantial dependence on \(d\): for example, they could remain of order one for \(n<d^2\) and only enter the \(n^{-1}\) regime once \(n\geq d^2\). Separately, our asymptotic analysis assumes a nondegenerate spectrum as a consequence of the proof technique.
To obtain explicit, dimension-uniform nonasymptotic control that also applies to arbitrary spectra, we establish the following finite-\(n\) bounds. For all $k$, $d$, $\boldsymbol{p}$, and $m$, and all $n\geq N$, let $c=12$ and $N_0\coloneqq8c^2D_{k,\mathrm{min}}^{-2}\ln\mleft(4c^2D_{k,\mathrm{min}}^{-2}\mright)$. For $k<d$, set $N\coloneqq\max\mleft(N_0,2m/D_{k,k+1}\mright)$; for $k=d$, set $N\coloneqq N_0$. Then \cref{thm:nonasymptotic_all_site_fidelity_bound} gives:
\begin{equation}
\label{equ:utility_bound_non_asymptotic}
^k\mathcal{L}_{\mathrm{all}} \leq \frac{m}{n}\sum_{i\neq k}\frac{p_i}{|D_{k,i}|D_{k,\mathrm{min}}} + R_{\mathrm{all}} \, ,
\end{equation}
\noindent where $D_{k,\mathrm{min}}=\min_{i\neq k} |D_{k,i}|$, and
\begin{equation}
\label{equ:utility_remainder_bound_non_asymptotic}
R_{\mathrm{all}}\le\frac{1}{n}\frac{4\sqrt{\ln n}}{n^{1/2}}
\mleft(
1+\sum_{i\neq k}\frac{6mp_i}{|D_{k,i}|D_{k,\mathrm{min}}^2}
\mright) \, .
\end{equation}

\noindent Under the same assumptions as for the all-site bound, the one-site loss obeys the nonasymptotic bound:
\begin{equation}\begin{aligned}
\label{equ:one_site_non_asymptotic}
&^k\mathcal{L}_{\mathrm{one}} \leq\frac{1}{n}\sum\limits_{i\neq k}\frac{p_i}{|D_{k,i}|D_{k,\mathrm{min}}} + R_{\mathrm{one}},
\end{aligned}\end{equation}
\noindent where $R_{\mathrm{one}}$ denotes the specialization of \(R_{\mathrm{all}}\) obtained by setting \(m=1\); see~\cref{thm:nonasymptotic_one_site_fidelity_bound}.

To illustrate the generic behavior of the utility, we choose an arbitrary spectrum parameterized by a single parameter $\lambda$, interpolating between a pure state and a maximally mixed state of $d=4$. We plot the corresponding asymptotic expressions for the all-site and one-site utilities for target eigenstates $k=1,2$ in \cref{fig:qpa_fidelity}(a)--(c), (e). To place the one-site utility on the same footing as the \(m\)-copy utility, we consider \(m\) independent copies of the marginal output, namely \(\mathcal{F}^m_\mathrm{one}\). The spectrum of the inputs is shown in \cref{fig:qpa_fidelity}(d).
As in the intensive regime, increasing the dimension ($d\to\infty$) does not affect the qualitative behavior of QPA. This manifests as the utility converging stably to a nonzero constant when removals are restricted to the first overhang, as shown for depolarized inputs in~\cref{fig:qpa_fidelity}(f).

We compare our general setting with previous guarantees in~\cref{tab:qpa_existing_bounds}.

\begin{table}[H]
\caption{Comparison of existing performance guarantees.}
\label{tab:qpa_existing_bounds}
\scriptsize
\renewcommand{\arraystretch}{1}
\makebox[\columnwidth][c]{\begin{tabular}{@{}>{\raggedright\arraybackslash}m{0.15\columnwidth}>{\raggedright\arraybackslash}m{0.46\columnwidth}>{\raggedright\arraybackslash}m{0.34\columnwidth}@{}}
\toprule
Source & Performance guarantee & Applicability domain \\
\midrule
\multicolumn{3}{@{}l}{\emph{Intensive all-site risk}, $\mathcal{L}_{\mathrm{all}}$ (pointwise in \(m\))} \\
This work
& $\frac{m}{n}\sum_{i\neq k}\frac{p_i}{D_{k,i}^2}+\lilo[d,\boldsymbol{p},m]{n^{-1}}$
& arbitrary $k,m,d,\boldsymbol{p}$\textsuperscript{a} \\
\addlinespace[2pt]
Ref.~\cite{LFIC24}
& $\frac{1}{n}\sum_{i=2}^{d}\frac{p_i}{D_{1,i}^2}+\lilo[d,\boldsymbol{p}]{n^{-1}}$
& $k=m=1$ \\
\addlinespace[3pt]
Ref.~\cite{DGHM+25}
& $\bigO{\frac{1-p_1}{D_{1,2}^2}n^{-1}}$
& $k=m=1$ \\
\addlinespace[3pt]
Ref.~\cite{GLLP+25}
& $\bigO{4^{f(p_1,D_{1,2})}n^{-1}}$
& $k=m=1$, $D_{1,2}<2/3$ \\
\addlinespace[3pt]
Ref.~\cite{GLLP+25}
& $\bigO{(1-p_1)n^{-1}}$
& $k=m=1$, $D_{1,2}\geq2/3$ \\
\midrule
\multicolumn{3}{@{}l}{\emph{Extensive all-site utility}, $\mathcal{F}_{\mathrm{all}}$ (uniform in \(m\))} \\
This work
&\cref{equ:extensive_asymptote}
& arbitrary $k,m,d,\boldsymbol{p}$\textsuperscript{a} \\
\addlinespace[3pt]
\raisebox{3pt}{Ref.~\cite{KW01}}
& \raisebox{3pt}{$\dfrac{(2p_1-1)^2}{(2p_1-1)^2+R(1-p_1)}+\lilo[\boldsymbol{p}]{1}$}
& $k=1,d=2$, $R\leq2p_1-1$ \\
\addlinespace[3pt]
\raisebox{3pt}{Ref.~\cite{KW01}}
& \raisebox{3pt}{$\dfrac{(2p_1-1)^2}{Rp_1}+\lilo[\boldsymbol{p}]{1}$}
& $k=1,d=2$, $R\geq2p_1-1$ \\
\addlinespace[3pt]
Ref.~\cite{W98}
& $\multiset{d}{n}/\multiset{d}{m}$
& $k=1$, $\boldsymbol{p}\dot{=}(1,0,\ldots,0)$ \\
\midrule
\multicolumn{3}{@{}l}{\emph{One-site utility}, $\mathcal{F}_{\mathrm{one}}$ (uniform in \(m\))} \\
This work
&\cref{equ:one_site_asymptote}
& arbitrary $k,m,d,\boldsymbol{p}$\textsuperscript{a} \\
\addlinespace[3pt]
Ref.~\cite{KW01}
& $1-\dfrac{m(2p_1-1)}{8n(1-p_1)^2}+\lilo[\boldsymbol{p}]{n^{-1}}$
& $k=1,d=2$ \\
\addlinespace[8pt]
Ref.~\cite{KW99}
& $\dfrac{n(m+d)+m-n}{m(n+d)}$
& $k=1$, $\boldsymbol{p}\dot{=}(1,0,\ldots,0)$ \\
\midrule
\multicolumn{3}{@{}l}{\emph{Minimax risk upper bound}, $\mathcal{L}_{\mathrm{all}}^*\leq$} \\
This work
& $\frac{m}{n}\sum_{i\neq k}\frac{p_i}{|D_{k,i}|D_{k,\mathrm{min}}} + R_{\mathrm{all}}$
& \raisebox{-6pt}{\shortstack[l]{arbitrary $k,m,d,\boldsymbol{p}$\\and \(n\geq N\)}} \\
\addlinespace[2pt]
Ref.~\cite{SSW26} & \(\frac{2032+4m}{n}\frac{1-p_{1}}{p_{1}-p_{2}}\)
& arbitrary \(m,d,\boldsymbol{p}\) and \(k=1\) \\
\bottomrule
\end{tabular}}
\par\smallskip
\noindent\textsuperscript{a}The spectrum \(\boldsymbol p\) is assumed to be nondegenerate.
\end{table}
\noindent
For $k=m=1$, the asymptote in~\cref{equ:intensive_asymptote} reduces to the principal-eigenstate expression of Ref.~\cite{LFIC24}. For $k=1$ and $d=2$, the all-site and one-site formulas in~\cref{equ:extensive_asymptote,equ:one_site_asymptote} reduce to the qubit expressions of Ref.~\cite{KW01}. In the pure-state limit $k=1$ and $\boldsymbol{p}\dot{=}(1,0,\ldots,0)$, the same formulas recover the optimal pure-state cloning benchmarks of Refs.~\cite[Eq.~(3.5)]{W98} and~\cite[Theorem~1]{KW99}.

For nonasymptotic upper bounds, all previous work has been about the case \(k=1\). Our estimate removes the potentially large leading-order constant appearing in Ref.~\cite[Prop.~10]{DGHM+25}. Compared with Ref.~\cite[Thm.~15]{GLLP+25}, it attains the optimal constant \(1\) and avoids the exponential factor $4^{f(p_1,D_{1,2})}$ in the small-gap regime $D_{1,2}<2/3$, where $f$ is bounded by $O(D_{1,2}^{-1})$ in general and can be logarithmic in $D_{1,2}^{-1}$ when $p_2$ is bounded away from zero. Compared with subsequent work~\cite{SSW26}, our bound similarly avoids the large constant, although it does not capture the refined dependence on \(1-p_1\).

\begin{figure*}
  \centering
  \includegraphics[width=\textwidth]{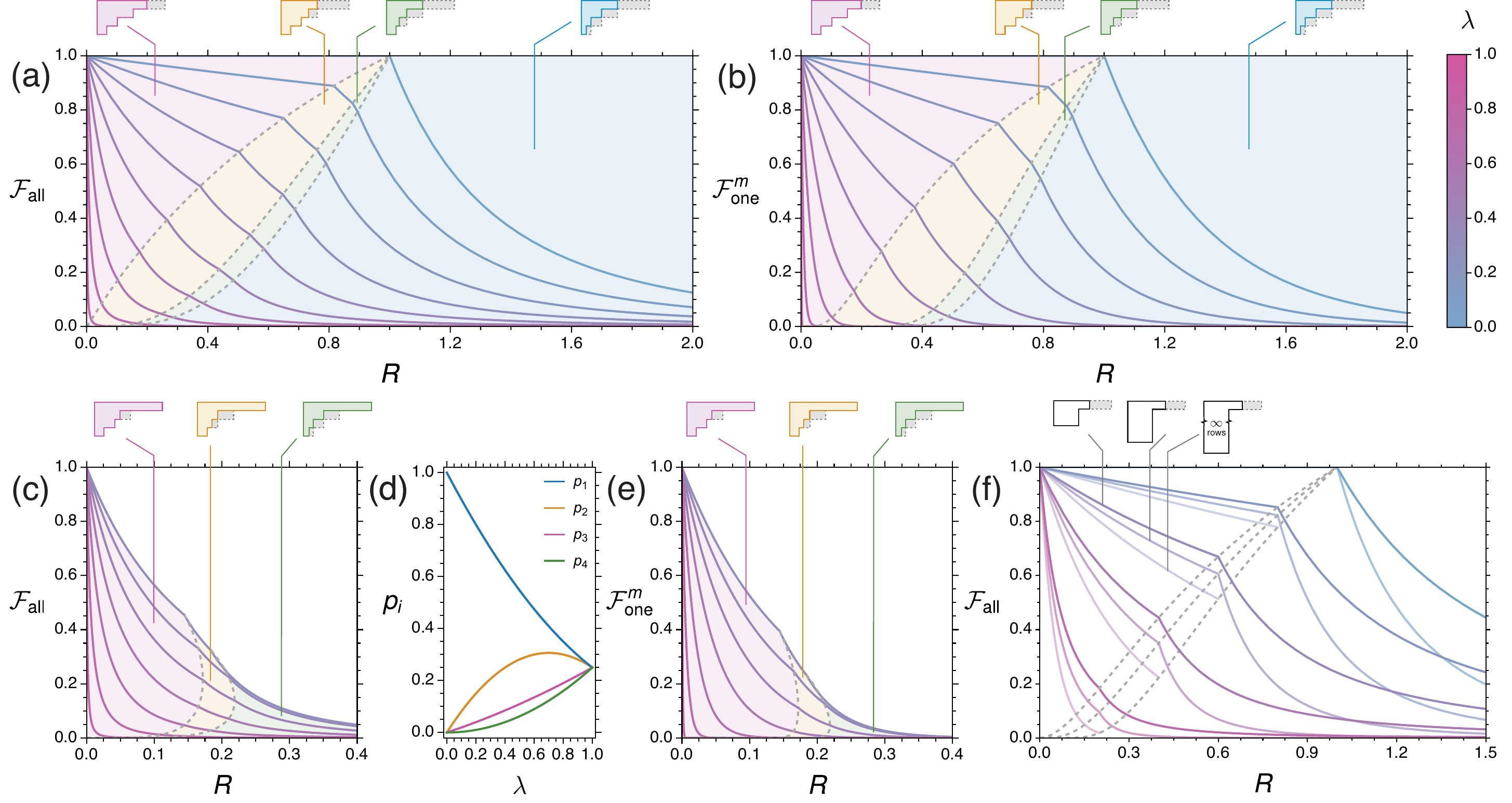}
    \caption{{(a)} All-site QPA fidelities (\cref{equ:extensive_asymptote}) for $k=1$; {(b)} One-site QPA fidelities (\cref{equ:one_site_asymptote}) for $k=1$; {(c)} All-site QPA phase diagrams (\cref{equ:extensive_asymptote}) for $k=2$; {(d)} input spectrum $\boldsymbol{p}\dot{=}(p_1,\ldots,p_4)$ used throughout panels (a)--(e); {(e)} one-site QPA phase diagram (\cref{equ:one_site_asymptote}) for $k=2$. Different colors correspond to different noise parameters $\lambda$.
For (a), (b), (c), and (e), dashed separators indicate QPA phase boundaries; the callouts illustrate representative removal rules, not to scale.
For (c) and (e), curves are plotted within the $\lambda$ range $[0.3,1]$ for clarity.
{(f)} QPA phase diagrams for depolarized inputs in dimensions $d=3,5,\infty$, respectively; the color legend in {(b)} maps $\lambda$ values to curve colors and is used across panels (a)--(c), (e), (f). In panel {(f)}, Phase II (where $\mathcal{F}=0$) is not shown for $d\to\infty$. All-site curves are evaluated from \cref{equ:extensive_asymptote} and one-site curves from \cref{equ:one_site_asymptote}.}
\label{fig:qpa_fidelity}
\end{figure*}

\subsection{Optimality and robustness}
\label{subsec:optimality_and_robustness}
The optimality of the overhang removal rule can be distinguished at two levels. First, the overhang removal rule attains the same asymptotic risk scaling as the minimal risk, and it is therefore sample-optimal in the small-error regime. Second, in a stronger sector-wise sense, the optimal channels can be characterized exactly by the overhang removal rule, and on most symmetry sectors any deviation from this rule leads to a strictly larger risk. These optimality guarantees are summarized in~\cref{tab:qpa_optimality_guarantees}.

\begin{table}[H]
\caption{Optimality guarantees for overhang removal.}
\label{tab:qpa_optimality_guarantees}
\scriptsize
\renewcommand{\arraystretch}{1.12}
\begin{tabularx}{\columnwidth}{@{}>{\raggedright\arraybackslash}p{0.30\columnwidth}>{\raggedright\arraybackslash}p{0.22\columnwidth}>{\raggedright\arraybackslash}X@{}}
\toprule
Optimality guarantee & Utility & Applicability domain \\
\midrule
Asymptotic  & all-site & pointwise in $m$; nondegenerate spectrum \\
\addlinespace[3pt]
Asymptotic  & one-site, $\yd{m}$ & uniform in $m$; nondegenerate spectrum \\
\addlinespace[3pt]
Sector-wise against extensive deviations & all-site & $\boldsymbol m_n=n\boldsymbol R^\infty+\lilo{n}$\\
\addlinespace[3pt]
Sector-wise exact & all-site & row-gap condition in~\cref*{equ:all_site_nonasymptotic_sector_gap} \\
\addlinespace[3pt]
Sector-wise exact & one-site, $\yd{m}$ & row-gap condition in~\cref*{equ:one_site_nonasymptotic_sector_gap} \\
\addlinespace[3pt]
Exact             & one-site, $\yd{m}$ & depolarized inputs; all $n,m,d$\\
\bottomrule
\end{tabularx}
\end{table}

\paragraph{Asymptotic optimality}\label{par:asymptotically_optimal_scaling}
In this part, we assume the spectrum is nondegenerate.

For the all-site risk, the overhang-removal risk~\cref{equ:intensive_asymptote} matches the asymptotic expansion of the minimal risk ${}^k\mathcal{L}_{\mathrm{all}}^{\ast}$ pointwise in $m$, as shown explicitly in~\cref{thm:asymptotic_optimality_QPA}. In addition, we establish sector-wise optimality against protocols whose environment irreps satisfy \(\yd{\mu_n}=\yd{\varsigma_n}-\boldsymbol m_n\), where \(\boldsymbol m_n=n\boldsymbol R^\infty+\lilo{n}\) componentwise, and whose removal ray \(\boldsymbol R^\infty\) differs from that of the overhang removal rule. Every such linearly deviating removal rule has exponentially decaying sector-wise utility as \(n\rightarrow \infty\); see~\cref{cor:extensive_deviation_overhang_suppressed}.

 For the one-site utility, the overhang-removal risk~\cref{equ:one_site_asymptote} matches the asymptotic expansion of the minimal risk ${}^k\mathcal{L}_{\mathrm{one}}^{\ast}$ uniformly in $m$, as shown explicitly in~\cref{thm:asymptotic_optimality_QPA}.

\paragraph{Nonasymptotic optimality}\label{par:strict_optimality} The above statements hold only for sufficiently large $n$ with unspecified thresholds. To sharpen the optimality
statement, we provide two further characterizations. First, as long as $\Delta_{k,\mathrm{min}}=\min_{i\neq k}|\Delta_{k,i}|$ satisfies
\begin{equation}
\label{equ:all_site_nonasymptotic_sector_gap}
\Delta_{k,\mathrm{min}}>2m^2\mleft(2m-1+\frac{1}{D_{k,\mathrm{min}}}\mright),
\end{equation}
for the all-site utility and \begin{equation}
\label{equ:one_site_nonasymptotic_sector_gap}
\Delta_{k,\mathrm{min}}>\max\mleft(m,\frac{2}{D_{k,\mathrm{min}}}\mright),
\end{equation}
for the one-site utility, the overhang-removal protocol is uniquely optimal within that symmetry sector; see~\cref{cor:non_asymptotic_optimality_irrep_level}.

Second, we show that overhang removal is exactly one-site optimal for depolarized inputs. That is, if \(0<\eta<1\), \(p_{1}=1-\eta(d-1)/d\), and \(p_i = \eta/d\) for \(1<i\leq d\), then the overhang-removal protocol is exactly optimal for $^1\mathcal{F}_{\mathrm{one}}$. This holds for all $n$, $m$, and $d$ in every symmetry sector (\cref{cor:optimality_QPA_depolarizing_F_one}).

\paragraph{Edge cases}
Finally, we note that although overhang removal is guaranteed to be optimal in most practically relevant cases, there remain a few edge cases in which optimality need not hold.

First, the rule need not be exactly optimal at small \(n\), although we conjecture that it is optimal in the principal eigenvalue case. For example, for the spectrum $p_1=1-\lambda$, $p_2=4\lambda/5$, $p_3=\lambda/5$ with spectral parameter $\lambda\in(0,1)$, one finds, as shown in~\cref{fig:edge_cases}(a), that whenever \( \frac{5}{42}(7-\sqrt{7})<\lambda<\frac{5}{9} \), the optimal sector-wise choice is to remove a \(k=1\) box, whereas the overhang removal rule would instead prescribe removing a \(k=2\) box.

Second, our optimality results under the one-site utility rely on restricting the output to the totally symmetric subspace. Without this restriction, protocols can potentially achieve a higher one-site utility than the overhang-removal protocol.
A counterexample appears already for $d=3$, $n=2$ and $m=2$. For $\yd{\varsigma}=[1,1,0], \yd{\varrho}=[1,1,0]$, and $\yd{\mu}=[0,0,0]$, the one-site utility is $\res{f}{[1,1,0]}{[1,1,0]}{[0,0,0]}_{\mathrm{one}}=\frac{1}{27}\bigl(3\lambda-2\lambda^2\bigr)$, which is strictly larger than the utility for $\yd{\varsigma}=[1,1,0]$, $\yd{\varrho}=[2,0,0]$, and $\yd{\mu}=[1,-1,0]$, namely $\res{f}{[1,1,0]}{[2,0,0]}{[1,-1,0]}_{\mathrm{one}}=\frac{1}{216}\bigl(21\lambda-13\lambda^2\bigr)$, as shown in \cref{fig:edge_cases}(b). This contrasts sharply with the qubit case~\cite{KW01}, where the all-site and one-site utilities define the same optimal protocol, since for $d\!=\!2$ there is no nontrivial antisymmetric sector to exploit. For most practical purposes, however, such as establishing the entanglement-breaking limit~\cite{CI26}, it is sufficient to restrict attention to totally symmetric outputs.

\begin{figure}
    \centering
    \includegraphics[width=\columnwidth]{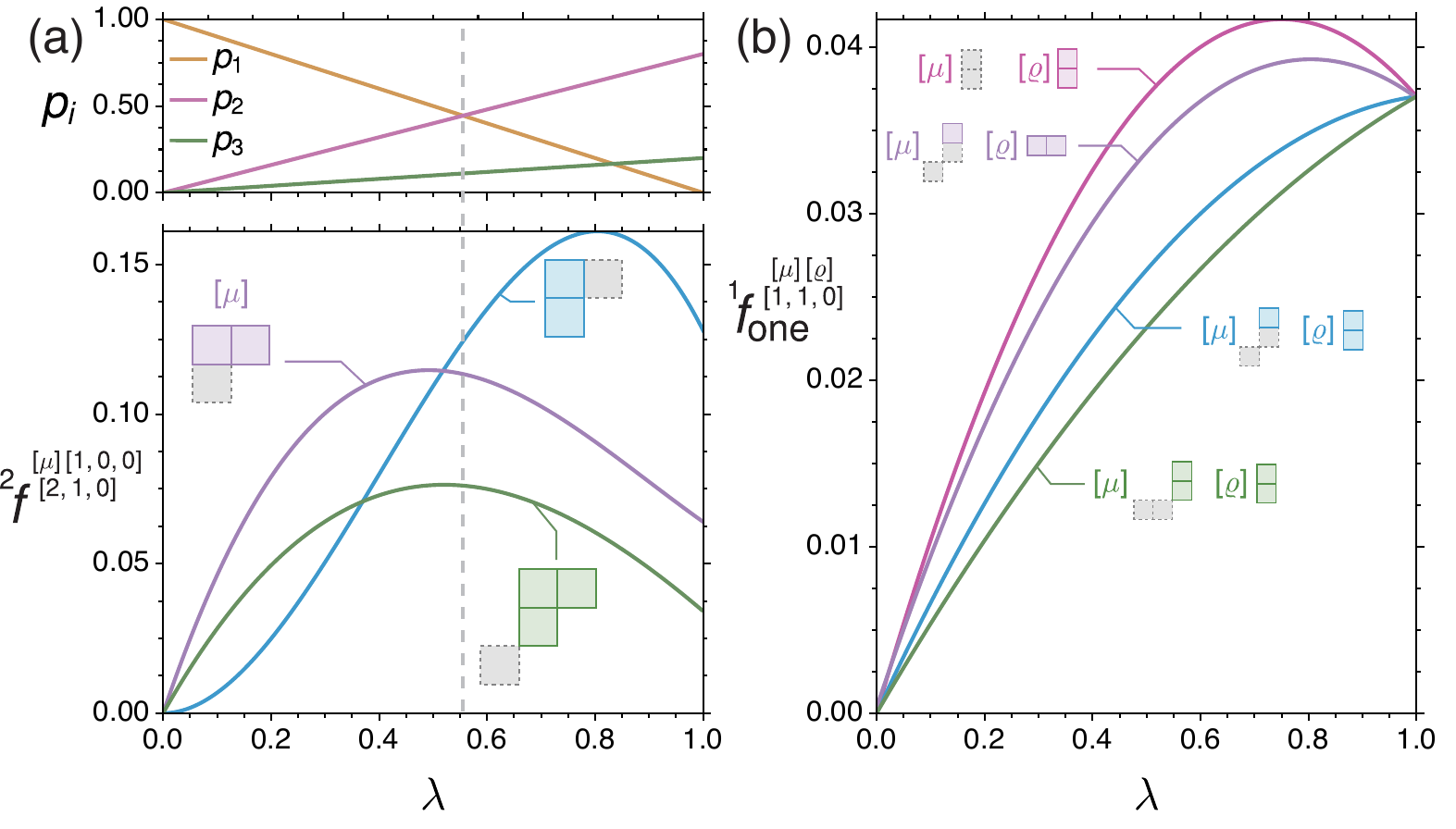}
    \caption{Edge cases for the overhang removal rule. (a) Suboptimality at small $n=3$ for $m=1$, $d=3$, and a specific spectrum $\boldsymbol{p}$: the overhang-removal protocol (purple) is strictly below the optimal protocol (blue) in all-site utility within a $\lambda$ interval, showing the large-$n$ optimality guarantee does not extend to finite $n$ in general. $\lambda$ parametrizes the input spectrum. (b) One-site utility for output $\yd{\varrho}=[1,1,0]$ outside the totally symmetric subspace, with $n=m=2$ and $d=3$: the antisymmetric-output protocol (environment $\yd{\mu}=[0,0,0]$) achieves a strictly higher one-site utility than the overhang-removal protocol with totally symmetric output ($\yd{\mu}=[1,-1,0]$), demonstrating that overhang removal is not one-site optimal beyond the totally symmetric subspace for $d\geq 3$.}
\label{fig:edge_cases}
\end{figure}

\subsection{Implementations}
\label{subsec:implementations}

\begin{figure}
    \centering
    \includegraphics[width=\columnwidth]{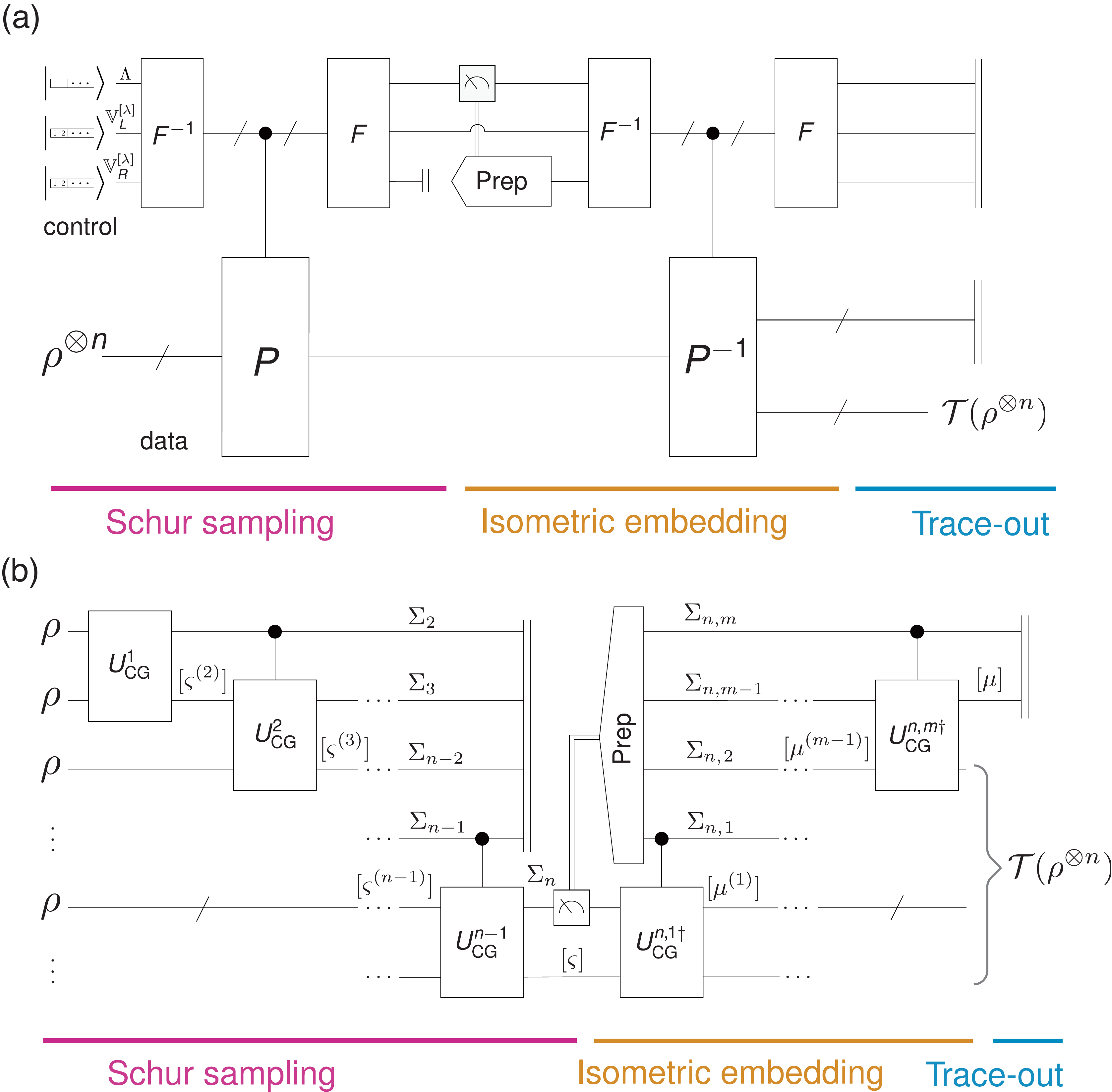}
    \caption{Circuit diagrams for the general QPA protocol. (a) The generalized quantum phase estimation (GQPE)-based overhang removal QPA (GQPE-OQPA) algorithm: 1. Schur sampling is performed via GQPE, followed by measurement of the YD. 2. The isometric embedding is implemented by preparing the Specht state and reinserting it into the system through inverse GQPE. 3. Extra registers are traced out. (b) The Clebsch--Gordan (CG)-based overhang removal QPA (CG-OQPA) algorithm: 1. Schur sampling is performed through a sequence of CG transforms, followed by measurement of the $\Sigma_n$ register storing the YD. 2. The isometric embedding is implemented by preparing the target state and applying the inverse sequence of CG transforms. 3. Extra registers are traced out.}
    \label{fig:qpa_circuit}
\end{figure}

We provide two circuit-level implementations of the overhang-removal protocol: generalized quantum phase estimation-based optimal QPA (GQPE-OQPA) and Clebsch--Gordan-based optimal QPA (CG-OQPA). Their correctness follows from the more detailed constructions and explanations given in Refs.~\cite{LFIC24,MT25}.

GQPE-OQPA is closely related to the original algorithm of Ref.~\cite{LFIC24} and the high-dimensional Schur transform of Ref.~\cite{BFGL+25}. The algorithm acts on the data register through a control register with generalized quantum phase estimation (GQPE). In Step 1, Schur sampling is implemented by indirectly measuring the YD from the control register. If $n$ is sufficiently large and $m$ sufficiently small (cf.~\cref{lem:upper_bound_probability_row_lengths}), then, with high probability, we measure a YD \(\yd{\varsigma}\) with $\varsigma_{k} \geq m.$
This allows us to reset the control register to a state that corresponds to the first \(n-m\) inputs being in the irrep labeled by \(\yd{\mu}\) (see~\cref{def:overhang_removal_rule}), and the last \(m\) inputs being in the totally symmetric irrep labeled by \(\yd{m}\). We then uncompute the GQPE, thereby injecting the irrep decomposition back into the data register, and finally we trace out the first \(n-m\) registers.
This algorithm implements QPA based on postselection, and the probability of success is \(1-\eta\) as long as \(np_{k} \geq m\), where
\begin{equation}
    \eta \leq 2\exp\left(-\,\frac{\,(\sqrt{n}\left(p_{k} - \frac{m}{n}\right) - 2)^2}{8}\right) \, .
\end{equation}
On failure, we simply output the maximally mixed state, which decreases the final fidelity by at most \(\eta\). We can implement this algorithm up to error \(\delta\) in diamond-norm with a gate complexity of \(T\), where
\begin{equation}
    T = \bigO{n^4\ln^p(dn/\delta)} \, ,
\end{equation}
and \(p\approx 1.44\). The diamond-norm error is due to a Solovay--Kitaev construction. For details, see~\cref{thm:implementation_GPE}.

CG-OQPA operates on the data register directly with an algorithm based on Ref.~\cite{MT25}. In Step 1, Schur sampling is implemented through a sequence of CG transforms, followed by measurement of the \(\Sigma_n\) register encoding the YD. In Step 2, the isometric embedding is realized by preparing the corresponding resource state, similar to GQPE-OQPA, and applying the inverse sequence of CG transforms. In Step 3, the auxiliary registers are traced out. Preparing the resource state can be done in a way similar to implementing the inverse isometry in Ref.~\cite{BFGL+25}, with gate complexity
\begin{equation}
    T = O\mleft(r^3d(n+m)\ln^p \mleft((n+m)d/\delta\mright)\mright) \, ,
\end{equation}
where \(r\) is the rank of the input state, \(\delta\) is the error in diamond-norm with which we want to implement the circuit, again due to a Solovay--Kitaev construction, and \(p\approx 1.44\). See~\cref{thm:implementation_CG} for details.

\section{Proof outline}
We sketch the main ideas behind the optimality proofs and risk scaling here, with full derivations in the SM.
\paragraph{Decomposition into symmetry sectors}
To start, we formalize the symmetry reduction introduced in~\cref{subsec:overhang_removal}. By symmetry, the input space decomposes as
\begin{equation}
\begin{aligned}
\mathbb{H}^{\otimes n}
\simeq
\bigoplus_{\yd{\varsigma}\vdash_{d}n}
\mathbb{V}^{\yd{\varsigma}}
\otimes
\mathbb{W}^{\yd{\varsigma}} .
\end{aligned}
\end{equation}
where \(\mathbb{V}^{\yd{\varsigma}}\) is the Specht module. Equivalently, the input \(\rho^{\otimes n}\) will be block-diagonal, i.e.,
\begin{equation}
    \label{equ:SW_decomposition_input_state}
    \rho^{\otimes n} \simeq \bigoplus_{\yd{\varsigma}\vdash_{d}n}g^{\yd{\varsigma}}\schur{\yd{\varsigma}} \frac{I_{\mathbb{V}^{\yd{\varsigma}}}}{g^{\yd{\varsigma}}} \otimes \frac{\rho^{\yd{\varsigma}}}{\schur{\yd{\varsigma}}} \, ,
\end{equation}
where \(g^{\yd{\varsigma}}\) is the dimension of \(\mathbb{V}^{\yd{\varsigma}}\) and \(\schur{\yd{\varsigma}}\) is the Schur polynomial with arguments \(\boldsymbol{p}\dot{=}(p_{1},\ldots,p_{d})\), which is also equal to the trace of \(\rho^{\yd{\varsigma}}\).
Accordingly, the optimal channel $\mathcal{T}$ can be chosen to respect the same block structure, and the utility decomposes accordingly as
\begin{equation}
\begin{aligned}
    ^k\mathcal{F}_\mathrm{(\cdot)}\mleft(\mathcal{T}\mright) = \sum_{\yd{\varsigma}\vdash_{d}n} g^{\yd{\varsigma}}\schur{\yd{\varsigma}} \, {^kf}_{(\cdot)}^{\yd{\varsigma}}(\mathcal{T}) =\mleft\langle {^kf}_{(\cdot)}^{\yd{\varsigma}}(\mathcal{T})\mright\rangle_{\yd{\varsigma}\sim\mathrm{SW}(\boldsymbol{p})}\,, \\
\end{aligned}
\end{equation}
Here, $g^{\yd{\varsigma}}\schur{\yd{\varsigma}}$ is the probability mass assigned to $\yd{\varsigma}$ by the SW distribution, equivalently the probability of observing the input YD $\yd{\varsigma}$ under Schur sampling, and ${^kf}_{(\cdot)}^{\yd{\varsigma}}(\mathcal{T})$ denotes the sector-wise utility. Thus, the analysis of optimal QPA reduces to understanding the behavior of the SW distribution and the sector-wise fidelities.\\

The SW distribution has been thoroughly analyzed, and in particular its connection to the input spectrum has proved useful for spectrum estimation and related tasks \cite{OW15,OW16,OW17}. The normalized row gaps of \(\yd{\varsigma}\) converge to the spectral gaps. More precisely:
\begin{theorem}[\cref{lem:upper_bound_probability_row_differences}, Informal]
    \label{thm:main_part_concentration_YD}
    For \(1\leq i \leq d-1\) and \(\alpha\geq 0\), we have
    \begin{equation}
    \begin{aligned}
        \Pr\big( \, |\overline{\Delta}_{i,i+1} - D_{i,i+1}| \ge \alpha \, \big) \leq 2e^{-\,\frac{\,(\sqrt{n}\alpha - 4)^2}{32}} \,,
    \end{aligned}
    \end{equation}

where the normalized row gaps are given by \(\overline{\Delta}_{i,j} = \Delta_{i,j}/n\).
\end{theorem}
Thus, at large $n$, the irreps obtained from Schur sampling have the property

\begin{equation}
\label{equ:main_part_proof_outline_Delta_approximately_nD}
    \Delta_{k-1,k} \approx nD_{k-1,k}\, , \quad \Delta_{k,k+1} \approx nD_{k,k+1} \,
\end{equation}
with probability approaching unity.
As we will see in the analysis of \({^kf}_{(\cdot)}^{\yd{\varsigma}}(\mathcal{T})\), the row gaps \(\Delta_{i,j}\), in particular the target row gaps \(\Delta_{k-1,k}\) and \(\Delta_{k,k+1}\), will be central for calculating the fidelities. \newline

\paragraph{The sector-wise all-site utility}
As discussed, the optimal sector-wise channel can be chosen as~\cref{equ:stinespring_dilation}. Therefore, from~\cref{equ:SW_decomposition_input_state} we get that the sector-wise utility \({^kf}_{\rm all}^{\yd{\varsigma}}(\mathcal{T})\) is given by
\begin{equation}
\begin{aligned}
    \label{equ:main_part_proof_outline_irrep_level_fidelity}
    {^kf}_{\rm all}^{\yd{\varsigma}}(\mathcal{T}) &= \Tr\mleft(\psi_k^{\otimes m} \extChannel{\mathcal{T}}{\yd{\varsigma}}{\yd{m}}{\yd{\mu}} \mleft(\frac{\rho^{\yd{\varsigma}}}{\Tr(\rho^{\yd{\varsigma}})}\mright) \mright) \,,
\end{aligned}
\end{equation}
By unitary covariance, it suffices to consider a diagonal input state with eigenstates $\ket{\psi_i}$ and eigenvalues $\boldsymbol{p}$. For the proof, we relabel the eigenbasis so that the target eigenstate $\ket{\psi_k}$ is mapped to $\ket{d}$. Define the sorting permutation
\begin{equation}
\sigma(i) \coloneqq
\mleft\{\begin{aligned}
    &i, & i < k, \\
    &d, & i = k, \\
    &i-1, & i > k.
\end{aligned}\mright.
\end{equation}
We then set $
\ket{\psi_i} \coloneqq \ket{\sigma(i)},$
and define the relabeled probability vector $\boldsymbol{q}\dot{=}(q_1,\ldots,q_d)$ by $
    q_{\sigma(i)} \coloneqq p_i.$
We expand \(\rho^{\yd{\varsigma}}\) in the Gel'fand--Tsetlin (GT) basis, equivalently labeled by WTs, as
\begin{equation}
\begin{aligned}
    \rho^{\yd{\varsigma}} = \sum_{\wt{w}\vdash_{d}\yd{\varsigma}}\ketbra{\wt{w}}{\wt{w}} {\boldsymbol{q}}^{\#\mleft(\wt{w}\mright)} \, ,
\end{aligned}
\end{equation}
where the monomial weight is
\begin{equation}
\label{equ:main_part_proof_outline_monomial_weights}
{\boldsymbol{q}}^{\#\mleft(\wt{w}\mright)} = \prod_{i=1}^{d}q_{i}^{\#_{i}\mleft(\wt{w}\mright)},
\end{equation}
where \(\#_{i}\mleft(\wt{w}\mright)\) denotes the occupancy of \(i\), i.e., the number of \(i\)'s in the WT \(\wt{w}\).
Combining~\cref{equ:main_part_proof_outline_irrep_channel_decomposition} with~\cref{equ:main_part_proof_outline_irrep_level_fidelity}, and using that \(\intertwiner{W}{\yd{\varsigma}}{\yd{\mu}}{\yd{m}}{}\) is given by CG transforms, the sector-wise utility evaluates to
\begin{equation}
\begin{aligned}
\label{equ:main_part_proof_outline_irrep_fidelity}
    {^kf}_{\rm all}^{\yd{\varsigma}}(\mathcal{T}) = \sum_{\substack{\wt{v}\vdash_{d}\yd{\mu}, \\ \wt{w}\vdash_{d}\yd{\varsigma}}} \mleft|\bra{\wt{v}}_{\yd{\mu}} \bra{d}^{\otimes m}_{\yd{m}}\ket{\wt{w}}_{\yd{\varsigma}}\mright|^2 \frac{{\boldsymbol{q}}^{\#\mleft(\wt{w}\mright)}}{\Tr(\rho^{\yd{\varsigma}})} \, .
\end{aligned}
\end{equation}
This can be interpreted as the average of the utility component
\begin{equation}
\label{equ:main_part_proof_outline_f_all_as_average}
        {^kf}_{\rm all}^{\yd{\varsigma}}(\mathcal{T}) = \weylavg{\res{f}{\yd{\mu}}{\yd{m}}{\yd{\varsigma}}_\mathrm{all}\mleft(\wt{w}\mright)}{\boldsymbol{q}}{\yd{\varsigma}} \, ,
\end{equation}
with respect to the Weyl distribution \(\mathrm{W}(\boldsymbol{p},\yd{\varsigma})\), whose probability mass is defined by normalizing the monomial weights in~\cref{equ:main_part_proof_outline_monomial_weights}. Using the \(SU(d)\) Clebsch--Gordan coefficients (CGCs) of Ref.~\cite{KV95}, the higher-dimensional analogue of Wigner \(3j\)-symbols, we find:
\begin{theorem}[\cref{thm:utility_component}, informal]
    \begin{equation}
    \begin{aligned}
\label{equ:main_part_proof_outline_CG_coefficient_formula}
        &\res{f}{\yd{\mu}}{\yd{m}}{\yd{\varsigma}}_\mathrm{all}\mleft(\wt{w}\mright) \\
        =&\binom{m}{\mathbf m}\;
        \frac{
        \displaystyle
        \prod_{\substack{1\le i\le d\\ 1\le j\le d-1}}
        \rpoch{\Delta_{j,i} - \#_{d,j}\mleft(\wt{w}\mright) + (i - j)}{m_i}
        }{
        \displaystyle
        \prod_{\substack{1\le i,j\le d\\ i\neq j}}
        \rpoch{\Delta_{j,i} + (i - j)+1}{m_i}
        } \, .
    \end{aligned}
    \end{equation}
\end{theorem}
Here, \(\boldsymbol{m}\dot{=}(m_{1},\ldots,m_{d})\) with \(m_{i} = \varsigma_{i} - \mu_{i}\), and \(\#_{d,i}\mleft(\wt{w}\mright)\) denotes the occupancy of \(d\)'s on the $i$-th row. The main challenge now is to derive the asymptotic expansion and nonasymptotic estimates for the sector-wise utility.\newline

\paragraph{Asymptotic behavior of \({^kf}_{\rm all}^{\yd{\varsigma}}(\mathcal{T})\)}
For the asymptotic analysis, we generalize the approach taken in Ref.~\cite{LFIC24}. The concentration of row gaps from~\cref{equ:main_part_proof_outline_Delta_approximately_nD} leads us to introduce the \emph{typical set}
\begin{equation}
    \mathsf{R}_{\epsilon} \coloneqq \left\{\overline{\yd{\varsigma}} \, \middle| \, \yd{\varsigma}\vdash_d n,\ \overline{\Delta}_{i,j}\geq\epsilon \right\} \, .
\end{equation}
Equivalently, these are diagrams for which the row gaps grow sufficiently quickly. For a nondegenerate spectrum \(p_{1}>\ldots>p_{d}\), one may choose
\begin{equation}
  0<  \epsilon < \min_{i<j}D_{i,j} \, ,
\end{equation}
such that the typical set contains the point corresponding to \(D_{i,j}\) and remains separated from the degeneracies \(\overline{\Delta}_{i,j}=0\).

The role of the row gaps appears directly through the unnormalized monomial weights defining the Weyl distribution, given in~\cref{equ:main_part_proof_outline_monomial_weights}. The largest monomial weight is achieved by the unique WT \(\wt{\lw_{\sigma}}\) with
\begin{equation}
{\boldsymbol{q}}^{\#\mleft(\wt{\lw_{\sigma}}\mright)} = \prod_{i=1}^{d}p_{i}^{\varsigma_{i}} = \prod_{i=1}^{d}q_{\sigma(i)}^{\varsigma_{i}} \, .
\end{equation}

This WT is obtained by filling the rows \(1\leq i \leq k-1\) with \(\varsigma_{i}\) letters \(i\), filling the rows \(k\leq i \leq d-1\) with \(\varsigma_{i+1}\) letters \(i\), and filling the remaining boxes with \(d\)'s. For example, for \(d=5\), \(k=2\), and \(\yd{\varsigma}\dot{=}\yd{7,5,3,1,0}\), \cref{fig:wt_deviation_example} shows \(\wt{\lw_\sigma}\) in panel (a) and an arbitrary \(\wt{w}\) with the corresponding deviations in panel (b).
\begin{figure}[H]
    \centering
    \includegraphics[width=\columnwidth]{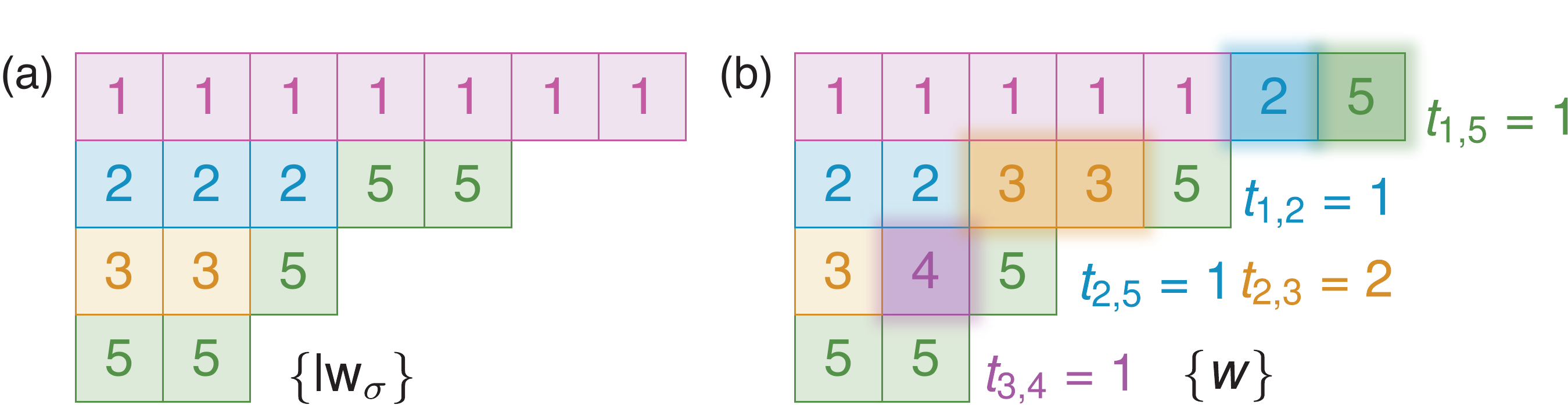}
    \caption{Example WT parametrization for \(d=5\), \(k=2\), and \(\yd{\varsigma}\dot{=}\yd{7,5,3,1,0}\). (a) The lowest-weight WT \(\wt{\lw_\sigma}\). (b) An arbitrary WT \(\wt{w}\) together with its deviation variables ${t_{i,j}}$.}
    \label{fig:wt_deviation_example}
\end{figure}

For an arbitrary WT \(\wt{w}\), we parametrize its deviation from \(\wt{\lw_{\sigma}}\). Changing an entry \(i\) to an entry \(j\) contributes the relative factor \(q_{j}/q_{i}\). For \(1\leq i<j \leq d-1\) or \(1\leq i<k, j=d\) we set
\begin{equation}
    t_{i,j} \coloneqq \#_{j,i}(\wt{w}) \, , \quad  r_{i,j}\coloneqq \frac{q_{j}}{q_{i}} \, .
\end{equation}
The deviation variables \(t_{i,j}\) count how many \(i\)'s in row \(i\) are replaced by \(j\)'s, contributing the factor \(r_{i,j}^{t_{i,j}}\). For \(k\leq i \leq d-1\), we similarly set
\begin{equation}
    t_{i,d} \coloneqq \Delta_{i,i+1} - \#_{d,i}(\wt{w}) \, , \quad  r_{i,d}\coloneqq \frac{q_{i}}{q_{d}} \, .
\end{equation}
Here \(t_{i,d}\) counts how many \(d\)'s in row \(i\) are replaced by \(i\)'s, again contributing the factor \(r_{i,d}^{t_{i,d}}\).
Thus \(\boldsymbol{t}\dot{=}\{t_{1,2},\ldots,t_{d-1,d}\}\) uniquely parametrizes \(\wt{w}\).
With the parametrization above, the utility component in~\cref{equ:main_part_proof_outline_CG_coefficient_formula} takes a factorized form after substituting
\begin{equation}
    \#_{d,i}(\wt{w}) =
    \left\{\begin{aligned}
        &t_{i,d}, & i<k, \\[4pt]
        &\Delta_{i,i+1}-t_{i,d}, & k\leq i < d, \\[4pt]
        &\varsigma_{d}, & i=d.
    \end{aligned}\right.
\end{equation}

The monomial weights are proportional to
\begin{equation}
\frac{{\boldsymbol{q}}^{\#\mleft(\wt{w}\mright)}}{{\boldsymbol{q}}^{\#\mleft(\wt{\lw_{\sigma}}\mright)}} = \prod_{i<j}r_{i,j}^{t_{i,j}}.
\end{equation}
Thus, before imposing admissibility, the probability mass factors as a product of geometric weights in the variables \(\boldsymbol{t}\), with parameters \(\boldsymbol{r}\dot{=}\{r_{1,2},\ldots,r_{d-1,d}\}\). The actual Weyl distribution is obtained by restricting this geometric distribution to the allowed values of \(\boldsymbol{t}\) by the WT admissibility constraints.
For large row gaps, these constraints become inactive: if \(\yd{\varsigma}\in \mathsf{R}_{\epsilon}\), then we choose $\epsilon$ small enough such that every \(\boldsymbol{t}\) satisfying $0\leq t_{i,j} \leq n\epsilon$
is admissible, and the remaining geometric tail is exponentially small. This gives the following theorem.
\begin{theorem}[\cref{thm:asymptotics_weyl_distribution}, informal]
\label{thm:main_part_limit_Weyl_average_geometric_average}
    For the all-site utility component,
    \begin{equation}
        \weylavg{\res{f}{\yd{\mu}}{\yd{m}}{\yd{\varsigma}}_\mathrm{all}\mleft(\wt{w}\mright)}{\boldsymbol{q}}{\yd{\varsigma}} \!\!\!\!= \mleft\langle \res{f}{\yd{\mu}}{\yd{m}}{\yd{\varsigma}}_\mathrm{all}\mleft(\boldsymbol{t}\mright) \mright\rangle_{\mathrm{G}(\boldsymbol{r})} \!\!\!\!+\,\, \bigO[d,\mathbf{r},\epsilon][n]{\exp(-n)} \, .
\end{equation}
\end{theorem}

In the intensive case, \cref{thm:main_part_limit_Weyl_average_geometric_average} allows us to evaluate the \(\boldsymbol{t}\)-average explicitly and expand the sector-wise utility as
\begin{equation}
    ^k\res{f_\mathrm{all}}{\yd{\varsigma-m\mathbf{e}_k}}{\yd{m}}{\yd{\varsigma}} = 1 - \sum_{i\neq k}\frac{m}{\Delta_{k,i}}\frac{p_{i}}{D_{k,i}}
    + \bigO[d,m,\mathbf{r},\epsilon][n]{n^{-2}} \, .
\end{equation}
where the remainder is uniformly controlled for $\overline{\yd{\varsigma}}\in\mathsf{R}_\epsilon$. Applying~\cref{lem:sw_averaging_uniform_sector_expansion} recovers the intensive all-site expansion in~\cref{equ:intensive_asymptote}. In the extensive regime where \(m\) scales with \(n\), the same framework gives the leading \(\bigO{1}\) term in~\cref{equ:extensive_asymptote}.

In~\cref{subsec:gt_param_supp}, we develop the general path-graph parametrization of WTs for arbitrary orderings of \(\boldsymbol{q}\), with the parametrization above appearing as a special case.

\paragraph{Nonasymptotic bounds for \({^kf}_{\rm all}^{\yd{\varsigma}}(\mathcal{T})\)}
We outline the lower bounds on the fidelity when \(\yd{\mu}\) is given by the removal rule and \(m<\Delta_{k,k+1}\). This restriction is justified by~\cref{equ:main_part_proof_outline_Delta_approximately_nD}: after choosing \(m<nD_{k,k+1}\), the complementary case gives only trivial bounds. The upper bounds for the fidelity when \(\yd{\mu}\) is not given by the removal rule follow from very similar calculations. We first consider the simplest cases: either \(k=1\) and \(\yd{\varsigma'}\dot{=}\yd{\varsigma_{1},0,\ldots,0}\), or \(k=d\) and \(\yd{\varsigma''}\dot{=}\yd{\varsigma_{1},\ldots,\varsigma_{1},0}\). In both cases, our assumption gives \(m<\varsigma_{1}\). For the first case, with \(k=1\) and \(\yd{\varsigma'}\), the removal rule gives \(\yd{\mu'}\dot{=}\yd{\varsigma' - m\boldsymbol{e}_{1}}\), and~\cref{equ:main_part_proof_outline_CG_coefficient_formula} takes the simple form

\begin{equation}
\begin{aligned}
    \label{equ:main_part_proof_outline_f_all_symmetric_irrep}
    \res{f}{\yd{\mu'}}{\yd{m}}{\yd{\varsigma'}}_\mathrm{all}\mleft(\wt{w}\mright) = \frac{
    \fpoch{\#_{d}\mleft(\wt{w}\mright)}{m}
    }{
    \fpoch{\varsigma_{1}}{m}
    } \, .
\end{aligned}
\end{equation}
Since \(\yd{\varsigma'}\) labels the totally symmetric subspace, the WTs are ordered lists with entries \(1,\ldots,d\), so averages over \(\mathrm{W}(\boldsymbol{q},\yd{\varsigma'})\) are more tractable. A direct calculation gives
\begin{equation}
    \label{equ:main_part_proof_outline_symmetric_average}
        \weylavg{\fpoch{\#_{d}}{m}}{\boldsymbol{q}}{\yd{\varsigma'}} \geq \fpoch{\varsigma_{1} - \sum_{i=2}^{d}\frac{p_{i}}{p_{1} - p_{i}}}{m} \, .
\end{equation}
Similarly, for \(k=d\) and \(\yd{\varsigma''}\dot{=}\yd{\varsigma_{1},\ldots,\varsigma_{1},0}\), the removal rule gives \(\yd{\mu''}\dot{=}\yd{\varsigma'' - m\boldsymbol{e}_{d}}\), and~\cref{equ:main_part_proof_outline_CG_coefficient_formula} takes the form
\begin{equation}
\begin{aligned}
    \label{equ:main_part_proof_outline_f_all_dual_symmetric_irrep}
    \res{f}{\yd{\mu''}}{\yd{m}}{\yd{\varsigma''}}_\mathrm{all}\mleft(\wt{w}\mright) = \frac{
    \rpoch{\varsigma_{1} - \#_{d}\mleft(\wt{w}\mright) + 1}{m}
    }{
    \rpoch{\varsigma_{1}+d}{m}
    } \, .
\end{aligned}
\end{equation}
The irrep labeled by \(\yd{\varsigma''}\) is dual to the irrep denoted by \(\yd{\varsigma'}\). This again allows us to use the simple structure of totally symmetric subspaces to obtain
\begin{equation}
\begin{aligned}
\label{equ:main_part_proof_outline_dual_symmetric_average}
    \weylavg{\rpoch{\varsigma_{1} - \#_{d} + 1}{m}}{\boldsymbol{q}}{\yd{\varsigma''}}
    \geq \rpoch{\varsigma_{1} + d - \sum_{i=1}^{d-1}\frac{p_{i}}{p_{i} - p_{d}}}{m} \, .
\end{aligned}
\end{equation}
Putting~\cref{equ:main_part_proof_outline_f_all_as_average,equ:main_part_proof_outline_f_all_symmetric_irrep,equ:main_part_proof_outline_symmetric_average,equ:main_part_proof_outline_f_all_dual_symmetric_irrep,equ:main_part_proof_outline_dual_symmetric_average} together, we get
\begin{equation}
\begin{aligned}
    {^1f}_{\rm all}^{\mathcal{T}}(\yd{\varsigma'}) \geq & \frac{
    \fpoch{\varsigma_{1} - \sum_{i=2}^{d}\frac{p_{i}}{p_{1} - p_{i}}}{m}
    }{
    \fpoch{\varsigma_{1}}{m}
    } \, , \\
    {^df}_{\rm all}^{\mathcal{T}}(\yd{\varsigma''}) \geq &\frac{
    \rpoch{\varsigma_{1} + d - \sum_{i=1}^{d-1}\frac{p_{i}}{p_{i} - p_{d}}}{m}
    }{
    \rpoch{\varsigma_{1} + d}{m}
    } \, .
\end{aligned}
\end{equation}

We now reduce the more complicated expression \({^kf}_{\rm all}^{\yd{\varsigma}}(\mathcal{T})\), for general \(k\) and \(\yd{\varsigma}\), to a product of these two expressions. The assumption \(m<\Delta_{k,k+1}\) implies that the removal rule gives \(\yd{\mu} = \yd{\varsigma - m\boldsymbol{e}_{k}}\). Straightforward inequalities yield
\begin{equation}
\begin{aligned}
    \label{equ:main_part_proof_outline_f_all_WT_level_lower_bound}
    \res{f}{\yd{\mu}}{\yd{m}}{\yd{\varsigma}}_\mathrm{all}\mleft(\wt{w}\mright) &\geq \frac{\fpoch{\#_{d}|_{k}^{d}\mleft(\wt{w}\mright) -\varsigma_{k+1})}{m}}{\fpoch{\Delta_{k,k+1}}{m}} \\
    &\quad \frac{\rpoch{\Delta_{k-1,k} + 2 - (\#_{d}|_{1}^{k-1}\mleft(\wt{w}\mright) + k - 1)}{m}}{\rpoch{\Delta_{k-1,k} + 2}{m}} \, .
\end{aligned}
\end{equation}
Here
\begin{equation}
    \#_{d}|_{i}^{j}\mleft(\wt{w}\mright) \coloneqq \sum_{\ell = i}^{j} \#_{d,\ell}\mleft(\wt{w}\mright)
\end{equation}
denotes the number of \(d\)'s in rows \(i\) through \(j\). There are two obstructions to applying lower bounds analogous to~\cref{equ:main_part_proof_outline_symmetric_average,equ:main_part_proof_outline_dual_symmetric_average}. First, the two factors are not independent, so their joint average cannot be split into a product of averages. Second, the average is over an arbitrary irrep \(\yd{\varsigma}\), rather than one of the special cases \(\yd{\varsigma'},\yd{\varsigma''}\) discussed above. To overcome the first obstruction, we develop the theory of \emph{gYDs with constraints}. This theory yields the following result.
\begin{theorem}[\cref{cor:splitting_diagram_average_weight}, informal]
    \label{thm:main_part_splitting_diagram_average_weight}
    Let
    \begin{equation}
        D|_{1}^{k-1}:\mathbb{Z}_{\geq0} \rightarrow \mathbb{R}_{\geq 0} \, , \quad I|_{k}^{d}:\mathbb{Z}_{\geq0} \rightarrow \mathbb{R}_{\geq 0}
    \end{equation}
    be two functions such that \(D|_{1}^{k-1}\) is nonincreasing and \(I|_{k}^{d}\) is nondecreasing. Then
    \begin{equation}
    \begin{aligned}
    \label{equ:main_part_proof_outline_split_average_YD}
        &\weylavg{D|_{1}^{k-1}\mleft(\#_{d}|_{1}^{k-1}\mleft(\wt{w}\mright)\mright) I|_{k}^{d}\mleft(\#_{d}|_{k}^{d}\mleft(\wt{w}\mright)\mright)}{\boldsymbol{q}}{\yd{\varsigma}} \\
        \geq \, &\weylavg{D|_{1}^{k-1}\mleft(\#_{k,1},\ldots,\#_{k,k-1}\mright)}{(\boldsymbol{q}|_{1}^{k-1},q_{d})}{\yd{\Delta|_{1}^{k}}} \\
        &\quad \weylavg{I|_{k}^{d}\mleft(\#_{d-k+1,1},\ldots,\#_{d-k+1,d-k+1}\mright)}{\boldsymbol{q}|_{k}^{d}}{\yd{\varsigma|_{k}^{d}}} \, .
    \end{aligned}
\end{equation}
Here \((\boldsymbol{q}|_{1}^{k-1},q_{d})\dot{=}(q_{1},\ldots,q_{k-1},q_{d})\), \(\boldsymbol{q}|_{k}^{d}\dot{=}(q_{k},\ldots,q_{d})\), \(\yd{\Delta|_{1}^{k}}\dot{=}\yd{\Delta_{1,k},\ldots,\Delta_{k-1,k},0}\), and \(\yd{\varsigma|_{k}^{d}}\dot{=}\yd{\varsigma_{k},\ldots,\varsigma_{d}}\).
\end{theorem}
The statement says that an average over a product of a decreasing function \(D|_{1}^{k-1}\) on the upper part of the YD and an increasing function \(I|_{k}^{d}\) on the lower part can be split into the product of the respective averages; see~\cref{fig:gYT}(g).
We now apply this theorem to the functions needed here. Set
\begin{equation}
\begin{aligned}
    &D|_{1}^{k-1}(a) \coloneqq \rpoch{\Delta_{k-1,k} + 2 - (a + k - 1)}{m} \, , \\
    &I|_{k}^{d}(a) \coloneqq \fpoch{a -\varsigma_{k+1}}{m} \, ,
\end{aligned}
\end{equation}
Using~\cref{equ:main_part_proof_outline_f_all_as_average,equ:main_part_proof_outline_f_all_WT_level_lower_bound,equ:main_part_proof_outline_split_average_YD}, we get
\begin{equation}
\begin{aligned}
    {^kf}_{\rm all}^{\yd{\varsigma}}(\mathcal{T}) \geq& \frac{\weylavg{D|_{1}^{k-1}\mleft(\#_{k}\mleft(\wt{w}\mright)\mright)}{(\boldsymbol{q}|_{1}^{k-1},q_{d})}{\yd{\Delta|_{1}^{k}}}}{\fpoch{\Delta_{k,k+1}}{m}} \\
    &\quad \frac{\weylavg{I|_{k}^{d}\mleft(\#_{d-k+1}\mleft(\wt{w}\mright)\mright)}{\boldsymbol{q}|_{k}^{d}}{\yd{\varsigma|_{k}^{d}}}}{\rpoch{\Delta_{k-1,k} + 2}{m}} \, .
\end{aligned}
\end{equation}
For the second obstruction, we replace \(\yd{\Delta|_{1}^{k}}\) by \(\yd{\Delta_{1,k},\ldots,\Delta_{1,k},0}\) and \(\yd{\varsigma|_{k}^{d}}\) by \(\yd{\varsigma_{k},0,\ldots,0}\), reducing to simple structures for which inequalities analogous to~\cref{equ:main_part_proof_outline_symmetric_average,equ:main_part_proof_outline_dual_symmetric_average} can be obtained. To this end, we prove the following result.
\begin{theorem}[\cref{cor:monotonicity_functions_young_diagrams}, informal]
\label{thm:main_part_monotonicity_functions_young_diagrams}
Let $F:\mathbb{Z}_{\geq0}\mapsto \mathbb{R}$ be a nondecreasing function. Then, for any \(\yd{\lambda} \leq \yd{\varrho}\) entrywise,
\begin{equation}
        \weylavg{F(\#_d\mleft(\wt{w}\mright))}{\boldsymbol{q}}{\yd{\lambda}} \le \weylavg{F(\#_d\mleft(\wt{w}\mright))}{\boldsymbol{q}}{\yd{\varrho}} \, .
\end{equation}
\end{theorem}
We have \(\yd{\Delta|_{1}^{k}} \leq \yd{\Delta_{1,k},\ldots,\Delta_{1,k},0}\) and \(\yd{\varsigma|_{k}^{d}} \geq \yd{\varsigma_{k},0,\ldots,0}\) entrywise, so we can apply this theorem to \(-D|_{1}^{k-1}\) and \(I|_{k}^{d}\) separately to get
\begin{equation}
\begin{aligned}
    {^kf}_{\rm all}^{\yd{\varsigma}}(\mathcal{T}) &\geq \frac{\weylavg{D|_{1}^{k-1}\mleft(\#_{k}\mleft(\wt{w}\mright)\mright)}{(\boldsymbol{q}|_{1}^{k-1},q_{d})}{\yd{\Delta_{1,k},\ldots,\Delta_{1,k},0}}}{\fpoch{\Delta_{k,k+1}}{m}} \\
    &\quad \frac{\weylavg{I|_{k}^{d}\mleft(\#_{d-k+1}\mleft(\wt{w}\mright)\mright)}{\boldsymbol{q}|_{k}^{d}}{\yd{\varsigma_{k},0,\ldots,0}}}{\rpoch{\Delta_{k-1,k} + 2}{m}} \, .
\end{aligned}
\end{equation}
These expressions are in a form amenable to explicit calculation on the (dual) totally symmetric subspace. The final lower bound follows by the same argument as in~\cref{equ:main_part_proof_outline_symmetric_average,equ:main_part_proof_outline_dual_symmetric_average}. As discussed above, the upper bounds for general \(\yd{\mu}\) not given by the removal rule follow from a very similar argument, with nondecreasing and nonincreasing interchanged throughout and the inequality signs reversed. \newline

\paragraph{Reducing the one-site utility}
As in~\cref{equ:main_part_proof_outline_irrep_level_fidelity}, the sector-wise one-site utility \({^kf}_{\rm one}^{\yd{\varsigma}}(\mathcal{T})\) can be written as
\begin{equation}
\begin{aligned}
    {^kf}_{\rm one}^{\yd{\varsigma}}(\mathcal{T}) &= \Tr\mleft(\psi_k\Tr_{2,\ldots,m}\mleft(\mathcal{T}^{\yd{\varsigma}}\mleft(\frac{\rho^{\yd{\varsigma}}}{\Tr(\rho^{\yd{\varsigma}})}\mright)\mright) \mright) \, .
\end{aligned}
\end{equation}
As before, it suffices to consider the extremal cases for \(\mathcal{T}^{\yd{\varsigma}}\), so we set
\begin{equation}\mathcal{T}^{\yd{\varsigma}} = \extChannel{\mathcal{T}}{\yd{\varsigma}}{\yt{1\cdots m}}{\yd{\mu}} \, ,
\end{equation}
for some \(\mu\). The marginal channel \(\Tr_{2,\ldots,m}\mleft(\extChannel{\mathcal{T}}{\yd{\varsigma}}{\yt{1\cdots m}}{\yd{\mu}} (.) \mright)\) from \(\mathbb{W}^{\yd{\varsigma}}\) to \(\mathbb{W}\) is unitary-equivariant, but not necessarily extremal. Thus it can be written as a convex combination of extremal channels
\begin{equation}
    \Tr_{2,\ldots,m}\mleft(\extChannel{\mathcal{T}}{\yd{\varsigma}}{\yt{1\cdots m}}{\yd{\mu}}(\cdot)\mright) = \sum_{i=1}^{d}c_{i} \extChannel{\mathcal{T}}{\yd{\varsigma}}{\yd{1}}{\yd{\varsigma-\mathbf{e}_i}}(\cdot) \, ,
\end{equation}
where \(c_{i}>0\) only when \(\yd{\varsigma-\mathbf{e}_i}\) is a valid YD. This gives
\begin{equation}
\label{equ:main_part_proof_outline_one_fidelity_decomposition}
{^kf}_{\rm one}^{\yd{\varsigma}}(\mathcal{T}) = \sum_{i=1}^{d}c_{i} \, \weylavg{\res{f}{\yd{\varsigma - \boldsymbol{e}_{i}}}{\yd{1}}{\yd{\varsigma}}_\mathrm{all}\mleft(\wt{w}\mright)}{\boldsymbol{q}}{\yd{\varsigma}} \, .
\end{equation}
Therefore, we can use the asymptotics and bounds already obtained for \({^kf}_{\rm all}^{\yd{\varsigma}}(\mathcal{T})\) with \(m=1\).
To calculate the \(c_{i}\), we use the two equivariant chains of embeddings
\begin{equation}
\begin{aligned}
    \label{equ:embedding_chains}
    \mathbb{W}^{\yd{\varsigma}} \hookrightarrow \mathbb{W}^{\yd{\mu}} \otimes \mathbb{W}^{\yd{m}} \hookrightarrow \mathbb{W}^{\yd{\mu}} \otimes \mathbb{W}^{\yd{m-1}} \otimes \mathbb{W} \, , \\
    \mathbb{W}^{\yd{\varsigma}} \hookrightarrow \mathbb{W}^{\yd{\varsigma - \boldsymbol{e}_{i}}} \otimes \mathbb{W} \hookrightarrow \mathbb{W}^{\yd{\mu}} \otimes \mathbb{W}^{\yd{m-1}} \otimes \mathbb{W} \, .
\end{aligned}
\end{equation}
These chains are connected by the fusion symbols detailed, for example, in Ref.~\cite{BFGL+25}:
\begin{equation}
    \label{equ:main_part_proof_outline_fusion_coefficient_1}
    \resSix{F}{\yd{\mu}}{\yd{m-1}}{\yd{1}}{\yd{\varsigma-\mathbf{e}_i}}{\yd{m}}{\yd{\varsigma}} = \sqrt{\dfrac{1}{m}\dfrac{\prod_{j=1}^{d}(\varsigma_{i} - \mu_{j} + (j-i))}{\prod_{j\neq i}(\varsigma_{i} - \varsigma_{j} + (j-i))}}
\end{equation}
if \(\yd{\varsigma-\mathbf{e}_i}\) is a valid YD and $\resSix{F}{\yd{\mu}}{\yd{m-1}}{\yd{1}}{\yd{\varsigma-\mathbf{e}_i}}{\yd{m}}{\yd{\varsigma}} = 0$ otherwise.

\begin{figure}[H]
    \centering
    \includegraphics[width=0.6\columnwidth]{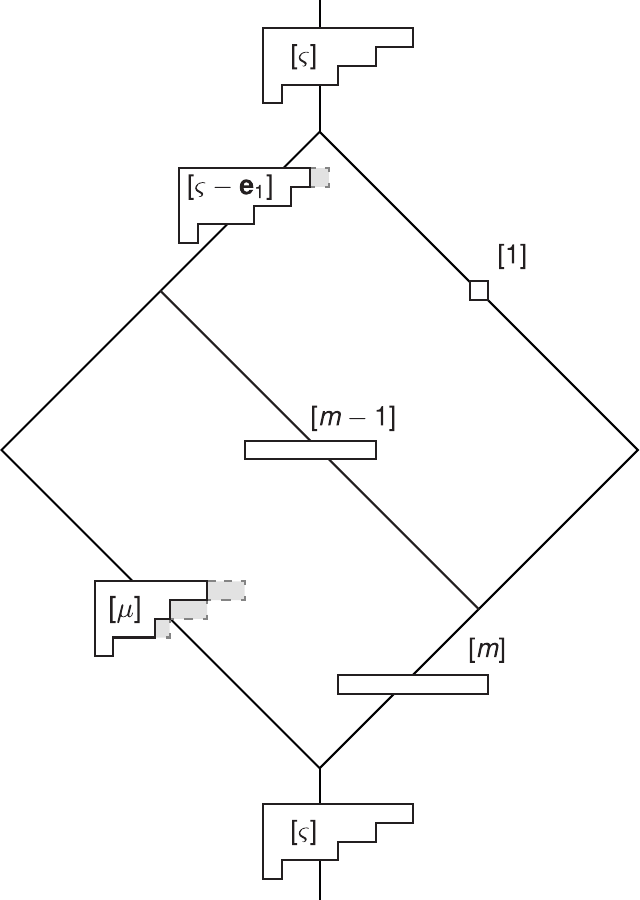}
    \caption{Tensor network diagram illustrating the intertwiner connecting two sides of~\cref{equ:embedding_chains}, $\resSix{F}{\yd{\mu}}{\yd{m-1}}{\yd{1}}{\yd{\varsigma-\mathbf{e}_i}}{\yd{m}}{\yd{\varsigma}} I^{\yd{\varsigma}}$, where each split corresponds to an intertwining isometry.}
    \label{fig:fusion_diagram}
\end{figure}

By~\cref{equ:main_part_proof_outline_irrep_channel_decomposition}, the marginal channel is obtained by first applying the intertwiners and then tracing out the factor \(\mathbb{W}^{\yd{\mu}} \otimes \mathbb{W}^{\yd{m-1}}\). In this decomposition, that trace is equivalent to tracing out each \(\mathbb{W}^{\yd{\varsigma-\boldsymbol{e}_{i}}}\) component appearing in the convex combination.
Comparing with~\cref{equ:main_part_proof_outline_one_fidelity_decomposition}, we find $c_{i} = \resSix{F}{\yd{\mu}}{\yd{m-1}}{\yd{1}}{\yd{\varsigma-\mathbf{e}_i}}{\yd{m}}{\yd{\varsigma}}^2$.
For the asymptotic and nonasymptotic bounds, we combine the \(m=1\) all-site results with corresponding bounds on the fusion coefficients, which are explicitly calculated using~\cref{equ:main_part_proof_outline_fusion_coefficient_1}.\\

\paragraph{Path-graph parametrization}
Here we give a short summary of the path-graph parametrization for WTs.
We define an ordering of the entries \(\{1,\ldots,d\}\) by choosing a permutation \(\pi\in S_d\) such that
\begin{equation}
    \pi(1) \succ \cdots \succ \pi(d) \, .
\end{equation}
Next, we take the unique WT that maximizes the number of \(\pi(1)\), then the number of \(\pi(2)\), and so on. We label this WT by \(\wt{\lw_{\pi}}\), and it has the property that
\begin{equation}
    \#_{j}(\wt{\lw_{\pi}}) = \varsigma_{\pi^{-1}(j)} \, .
\end{equation}
The simplest such case is when we choose
\begin{equation}
    1\succ 2 \succ \ldots \succ d \, ,
\end{equation}
and the resulting WT is the lowest weight tableau with \(1\)'s in the first row, \(2\)'s in the second row, and so on. We use \(\wt{\lw_{\pi}}\) as the starting point of the parametrization and label every WT by its deviation from \(\wt{\lw_{\pi}}\). To this end, we first reframe the WTs as Gel'fand--Tsetlin (GT) patterns. A GT pattern consists of \(d\) rows of increasing length, stacked from bottom to top, and the \(i\)-th row records the row-lengths of the YD formed by all entries of \(\wt{w}\) that are \(\leq i\). For example, the WT
\begin{equation}
\wt{w}=\begin{tikzpicture}[baseline={([yshift=-0.3em]current bounding box.center)}, x=1.15em, y=0em, every node/.style={draw, minimum size=1.15em, inner sep=0pt, outer sep=0pt, font=\sffamily\fontsize{8.4}{9}\selectfont}]
    \node at (0,0) {2};
    \node at (1,0) {3};
    \node at (2,0) {3};
\end{tikzpicture}
\qquad\longleftrightarrow\qquad
\begin{tikzpicture}[baseline={(current bounding box.center)}, x=1.15em, y=1.15em, every node/.style={draw, minimum size=1.15em, inner sep=0pt, outer sep=0pt, font=\sffamily\fontsize{8.4}{9}\selectfont}]
    \node at (0,0) {3};
    \node at (1,0) {0};
    \node at (2,0) {0};
    \node at (0.5,-1) {1};
    \node at (1.5,-1) {0};
    \node at (1,-2) {0};
\end{tikzpicture}.
\end{equation}
We record two facts. First, the top row is always fixed to be \(\yd{\varsigma}\). Second, if we know \(\yd{\varsigma}\) and we know for each \((a,b)\) the difference
\begin{equation}
    w_{a^\prime(a,b),b+1} - w_{a,b}
\end{equation}
for some \(1\leq a^\prime(a,b) \leq b+1\), we can still reconstruct the whole diagram. In fact, it is possible to pick each \(a^\prime(a,b)\) such that
\begin{equation}
    a^\prime(x,b) \neq a^\prime(y,b) \quad \text{for } x\neq y
\end{equation}
and
\begin{equation}
    a^\prime(a,b) = a \quad \text{or} \quad a^\prime(a,b) = a + 1 \, .
\end{equation}
Following the path of differences from the top row to the bottom, we reach a terminal level \(b^\prime(i)\), and call this path \(\boldsymbol{\gamma}_{b^\prime(i)}\). A \emph{path graph} is a disjoint collection of these paths \(\boldsymbol{\gamma}_{j}\) for \(1\leq j \leq d\), and each path consists of the vertices \(\gamma_{j,i}\) for \(j\leq i \leq d\).
\begin{figure}[H]
    \centering
    \includegraphics[width=0.55\columnwidth]{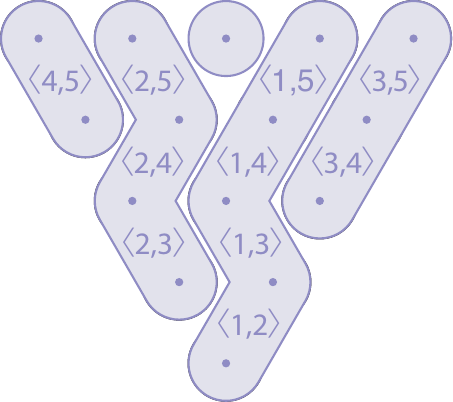}
    \caption{Path-graph illustration for the GT-pattern parametrization. Each dot is a GT-pattern vertex, and the circled dots indicate one path through the triangular lattice.}
\label{fig:main_part_gt_path_param}
\end{figure}
\noindent For a path graph, the map \(b^\prime\in S_{d}\) is a permutation, and conversely, there exists a corresponding path graph for each permutation \(\pi\in S_{d}\).

\begin{lemma}[\cref{lem:ordering_path_graph}, informal]
For every permutation $\pi\in S_d$ there exists a unique path graph \(\boldsymbol{\gamma}\) whose vertices satisfy
\begin{equation}
    \boldsymbol{\gamma}_{a,d} = (\pi^{-1}(a),d) \, .
\end{equation}
\end{lemma}

For a given path graph \(\boldsymbol{\gamma}\) coming from a permutation \(\pi\), we can define the \emph{edge variables}
\begin{equation}
    t_{a,b}^\prime \coloneqq w_{\gamma_{a,b}} - w_{\gamma_{a,b-1}} \, ,
\end{equation}
which corresponds precisely to the parametrization above. In addition, \(t_{a,b}^\prime = 0\) for all \(1\leq a<b\leq d\) precisely parametrizes the WT \(\wt{\lw_{\pi}}\). Thus, the \(t_{a,b}^\prime \) again parametrize deviations from the lowest weight \(\wt{\lw_{\pi}}\). It follows that for a given pair \((a,b)\), either
\begin{equation}
    t_{a,b}^\prime  \geq 0 \quad \text{or} \quad t_{a,b}^\prime  \leq 0 \, ,
\end{equation}
and to simplify the parameter ranges, we choose \(t_{a,b}\coloneqq |t^\prime _{a,b}|\). For some spectrum ordered by \(\pi\), i.e.,
\begin{equation}
    q_{\pi(1)} > \ldots > q_{\pi(d)} \, ,
\end{equation}
this parametrization again yields the simple relative monomial weights
\begin{equation}
    \frac{{\boldsymbol{q}}^{\#\mleft(\wt{w}\mright)}}{{\boldsymbol{q}}^{\#\mleft(\wt{\lw_{\pi}}\mright)}} = \prod_{i<j}r_{i,j}^{t_{i,j}} \, ,
\end{equation}
which allows us to generalize~\cref{thm:main_part_limit_Weyl_average_geometric_average} to arbitrarily ordered spectra.
\newline

\paragraph{gYDs with constraints}
We give a brief overview of gYDs with constraints, which provide a new approach to Schur polynomials and averages over the Weyl distribution. While YDs are usually depicted as rows of boxes with decreasing lengths from top to bottom, gYDs allow any finite collection of such boxes, potentially with gaps between them. Formally, \(\yd{\lambda}_{\mathrm{g}} \subseteq \mathbb{Z}^{2}\). Generalized Weyl tableaux (gWTs) \(\wt{w}_{\mathrm{g}}\vdash_{d}\yd{\lambda}_{\mathrm{g}}\) are fillings with letters \(1,\ldots,d\), obeying the same rules as WTs: they are nondecreasing to the right and increasing downward. Formally,
\begin{equation}
\begin{aligned}
    &\wt{w}_{\mathrm{g}} : \yd{\lambda}_{\mathrm{g}} \rightarrow \{1,\ldots,d\} \, , \\
    &\wt{w}_{\mathrm{g}}(i,j) \leq \wt{w}_{\mathrm{g}}(i,j+1) \quad \text{for} \quad (i,j),(i,j+1) \in \yd{\lambda}_{\mathrm{g}} \, , \\
    &\wt{w}_{\mathrm{g}}(i,j) < \wt{w}_{\mathrm{g}}(i+1,j) \quad \text{for} \quad (i,j),(i+1,j) \in \yd{\lambda}_{\mathrm{g}} \, .
\end{aligned}
\end{equation}
The final ingredient is a set of constraints on \(\yd{\lambda}_{\mathrm{g}}\). A constraint \(\boldsymbol{x}\) restricts the possible gWTs by providing a minimum and maximum value for each box. Formally,
\begin{equation}
\begin{aligned}
    &\boldsymbol{x} : \yd{\lambda}_{\mathrm{g}} \rightarrow \{1,\ldots,d\}^{2} \, , \\
    &\boldsymbol{x}(i,j) = (x_{l}(i,j),x_{u}(i,j)) \, ,
\end{aligned}
\end{equation}
and we write \(\wt{w}_{\mathrm{g}}\vdash_{d}\yd{\lambda}_{\mathrm{g}}|\boldsymbol{x}\) if
\begin{equation}
    x_{l}(i,j) \leq \wt{w}_{\mathrm{g}}(i,j) \leq x_{u}(i,j) \quad \text{for all} \quad (i,j)\in \yd{\lambda}_{\mathrm{g}} \, .
\end{equation}
\begin{figure*}[t]
    \centering
    \includegraphics[width=\textwidth]{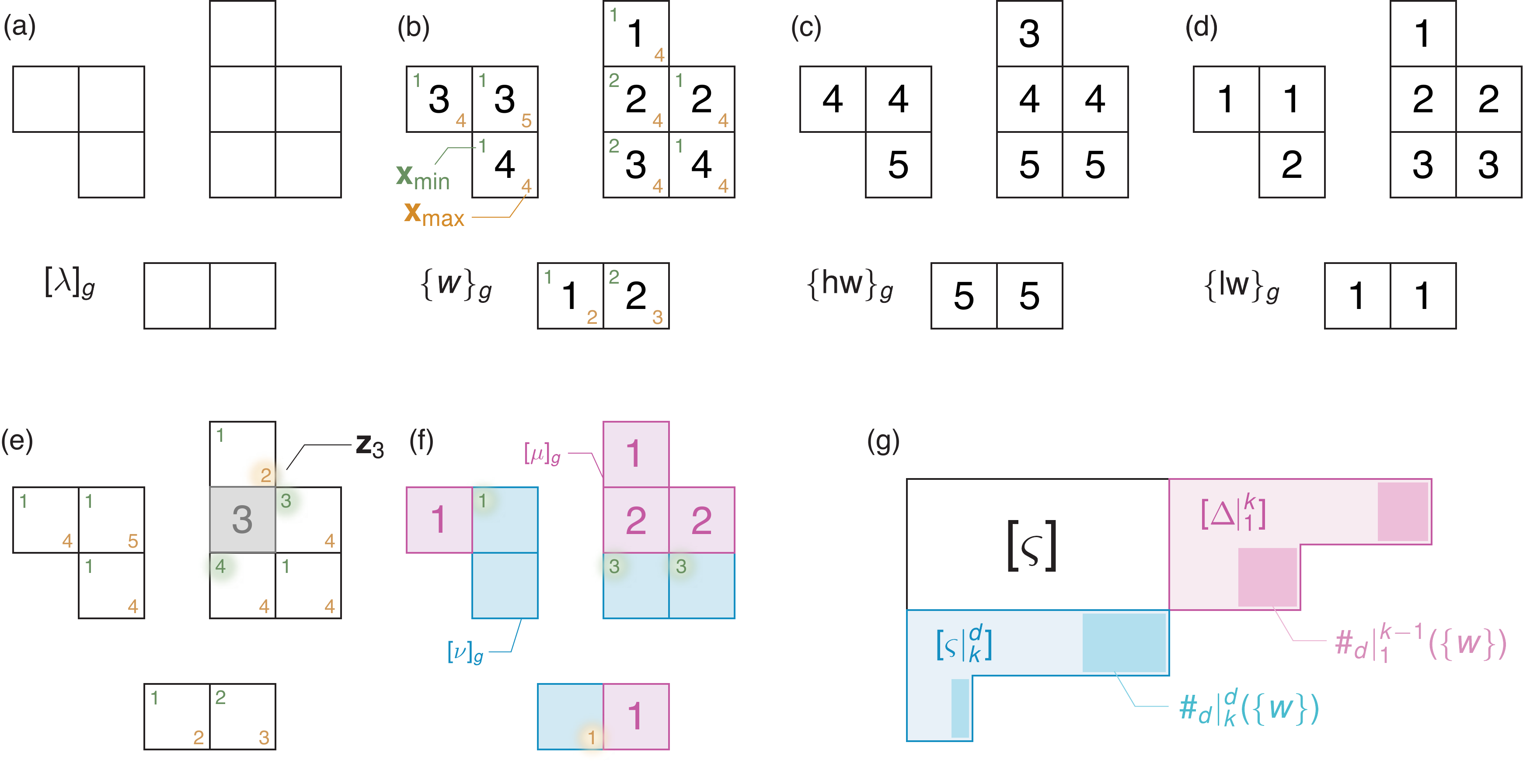}
    \caption{gYDs and gWTs with constraints in the proof outline. (a) A gYD \(\yd{\lambda}_{\mathrm{g}}\) with disjoint components. (b) A gWT with constraints, with lower and upper filling bounds. (c,d) Highest and lowest weight gWTs. (e) Decomposition of an average with constraints after fixing one box. (f) Restriction to complementary subdiagrams \(\yd{\mu}_{\mathrm{g}}\) and \(\yd{\nu}_{\mathrm{g}}\). (g) The gYD partition used to split Weyl averages of products of monotone functions.}
    \label{fig:gYT}
\end{figure*}
The Weyl distribution with constraints \(\mathrm{W}(\boldsymbol{p},\yd{\lambda}_{\mathrm{g}}|\boldsymbol{x})\) is defined analogously: the unnormalized probability mass for \(\wt{w}_{\mathrm{g}}\vdash_{d}\yd{\lambda}_{\mathrm{g}}|\boldsymbol{x}\) is
\begin{equation}
    P\mleft(\wt{w}_{\mathrm{g}}\mright) = {\boldsymbol{p}}^{\#\mleft(\wt{w}_{\mathrm{g}}\mright)} = \prod_{i=1}^{d}p_{i}^{\#_{i}\mleft(\wt{w}_{\mathrm{g}}\mright)} \, .
\end{equation}
Here, the \(\#_{i}\mleft(\wt{w}_{\mathrm{g}}\mright)\) denote the number of \(i\)'s in \(\wt{w}_{\mathrm{g}}\). With these definitions in place, the first result is
\begin{theorem}[\cref{thm:monotonicity_constrained_generalized_averages}, informal]
    \label{thm:main_part_monotonicity_constrained_generalized_averages}
    Let $\boldsymbol{x},\boldsymbol{x}'$ be two constraints such that
    \begin{equation}\begin{aligned}
        x_{l}(i,j) \leq x'_{l}(i,j) \, , \quad x_{u}(i,j) \leq x'_{u}(i,j) \, ,
    \end{aligned}\end{equation}
    and suppose there exist \(\wt{w}_{\mathrm{g}}\vdash_{d}\yd{\lambda}_{\mathrm{g}}|\boldsymbol{x}\) and \(\wt{w'}_{\mathrm{g}}\vdash_{d}\yd{\lambda}_{\mathrm{g}}|\boldsymbol{x'}\). Let $F:\mathcal{W}_{\yd{\lambda}_{\mathrm{g}}}^{d} \rightarrow \mathbb{R}$ be nondecreasing in every entry of the WTs, meaning that \(F\mleft(\wt{w_{1}}_{\mathrm{g}}\mright) \leq F\mleft(\wt{w_{2}}_{\mathrm{g}}\mright)\) whenever \(\wt{w_{1}}_{\mathrm{g}}\leq \wt{w_{2}}_{\mathrm{g}}\) pointwise. Then
    \begin{equation}
            \weylavg{F}{\boldsymbol{p}}{\yd{\lambda}_{\mathrm{g}} | \boldsymbol{x}} \leq \weylavg{F}{\boldsymbol{p}}{\yd{\lambda}_{\mathrm{g}} | \boldsymbol{x}'} \, .
\end{equation}
\end{theorem}
Although the proof is nontrivial, the theorem has a simple interpretation. Increasing the lower and upper bounds for the fillings of individual boxes makes their entries larger on average. The result is also independent of \(\boldsymbol{p}\). It is therefore true in particular for \(\boldsymbol{p}\dot{=}(1/d,\ldots,1/d)\), showing that it is combinatorial in nature. To prove~\cref{thm:main_part_monotonicity_constrained_generalized_averages}, consider the simple example \(\yd{\lambda}\dot{=}\yd{n}\), for which we decompose the corresponding Schur polynomial
\begin{equation}
    \schur{\yd{\lambda}}[\boldsymbol{p}] = \sum_{\wt{w}\vdash_{d}\yd{\lambda}} \boldsymbol{p}^{\#\mleft(\wt{w}\mright)}
\end{equation}
in the following way
\begin{equation}
    \label{equ:main_part_proof_outline_decomposition_schur_polynomial}
    \schur{\yd{n}}[\boldsymbol{p}] = \sum_{i=1}^{d}p_{i}\schur{\yd{n-1}}[\boldsymbol{p}|_{i}^{d}] \, .
\end{equation}

Here, one fixes the entry in the first box, which constrains the following boxes to be equal or larger, and then sums over all possible entries in the first box. Importantly, this reduces the relevant YD by one. We can apply a similar decomposition to averages with constraints over gYDs and proceed by induction to obtain the theorem; see~\cref{fig:gYT}(e).
Despite its simple form,~\cref{thm:main_part_monotonicity_constrained_generalized_averages} is powerful enough to rederive the following major theorem, which we use to prove~\cref{thm:main_part_monotonicity_functions_young_diagrams} and which illustrates the strength of the formalism of gYDs with constraints.
\begin{theorem}[Theorem~5 of Ref.~\cite{LPP05}, reproved as~\cref{thm:inequality_Schur_polynomials_products}]
    Let \(\yd{\lambda_{1}},\yd{\lambda_{2}},\yd{\nu_{1}},\yd{\nu_{2}}\) be YDs with \(\yd{\nu_{1}} \leq \yd{\lambda_{1}}\) and \(\yd{\nu_{2}} \leq \yd{\lambda_{2}}\) entrywise. Let \(\yd{\lambda_{\min}},\yd{\lambda_{\max}}\) be given by
    \begin{equation}
    \begin{aligned}
        (\lambda_{\min})_{i} \coloneqq \min\mleft((\lambda_{1})_{i},(\lambda_{2})_{i}\mright) \, , \\
        (\lambda_{\max})_{i} \coloneqq \max\mleft((\lambda_{1})_{i},(\lambda_{2})_{i}\mright) \, ,
    \end{aligned}
    \end{equation}
    and let \(\yd{\nu_{\min}},\yd{\nu_{\max}}\) be defined similarly. Then we have
    \begin{equation}
        \schur{\yd{\lambda_{\max}}\setminus \yd{\nu_{\max}}}[\boldsymbol{p}] \, \schur{\yd{\lambda_{\min}}\setminus \yd{\nu_{\min}}}[\boldsymbol{p}] \geq \schur{\yd{\lambda_{1}}\setminus \yd{\nu_{1}}}[\boldsymbol{p}] \, \schur{\yd{\lambda_{2}}\setminus \yd{\nu_{2}}}[\boldsymbol{p}] \, .
    \end{equation}
    Here we interpret \(\yd{\lambda}\setminus\yd{\nu}\) as the skew YD given by removing the smaller diagram from the larger one.
\end{theorem}
We prove it by choosing a gYD \(\yd{\lambda}_{\mathrm{g}}\), a function \(F\), and constraints \(\boldsymbol{x}\geq \boldsymbol{x}'\) so that~\cref{thm:main_part_monotonicity_constrained_generalized_averages} gives
\begin{equation}
        \weylavg{F\mleft(\wt{w}_{\mathrm{g}}\mright)}{\boldsymbol{p}}{\yd{\lambda}_{\mathrm{g}} | \boldsymbol{x}} \geq \weylavg{F}{\boldsymbol{p}}{\yd{\lambda}_{\mathrm{g}} | \boldsymbol{x}'} \, ,
\end{equation}
but we also have
\begin{equation}
\begin{aligned}
\weylavg{F\mleft(\wt{w}_{\mathrm{g}}\mright)}{\boldsymbol{p}}{\yd{\lambda}_{\mathrm{g}} | \boldsymbol{x}} &= \frac{\schur[F]{\yd{\lambda}_{\mathrm{g}}| \boldsymbol{x}}[\boldsymbol{p}]}{\schur{\yd{\lambda}_{\mathrm{g}}| \boldsymbol{x}}[\boldsymbol{p}]} = \frac{\schur{\yd{\lambda_{\min}}\setminus \yd{\nu_{\min}}}[\boldsymbol{p}]}{\schur{\yd{\lambda_{2}}\setminus \yd{\nu_{2}}}[\boldsymbol{p}]} \, ,
    \\
    \weylavg{F\mleft(\wt{w}_{\mathrm{g}}\mright)}{\boldsymbol{p}}{\yd{\lambda}_{\mathrm{g}} | \boldsymbol{x}'}& = \frac{\schur[F]{\yd{\lambda}_{\mathrm{g}}| \boldsymbol{x}'}[\boldsymbol{p}]}{\schur{\yd{\lambda}_{\mathrm{g}}| \boldsymbol{x}'}[\boldsymbol{p}]} = \frac{\schur{\yd{\lambda_{1}}\setminus \yd{\nu_{1}}}[\boldsymbol{p}]}{\schur{\yd{\lambda_{\max}}\setminus \yd{\nu_{\max}}}[\boldsymbol{p}]} \, .
\end{aligned}
\end{equation}
The proof of~\cref{thm:main_part_monotonicity_functions_young_diagrams} then follows the argument of Ref.~\cite[Proposition~8.1]{KT21}. For the final result of the theory of gYDs, we introduce the following notions. The lowest weight \(\wt{\lw}_{\mathrm{g}}\vdash_{d}\yd{\lambda}_{\mathrm{g}}\) is the unique filling of \(\yd{\lambda}_{\mathrm{g}}\) in which every entry is as small as possible, and similarly for the highest weight \(\wt{\hw}_{\mathrm{g}}\vdash_{d}\yd{\lambda}_{\mathrm{g}}\). Formally, for every \(\wt{w}_{\mathrm{g}}\vdash_{d}\yd{\lambda}_{\mathrm{g}}\),
\begin{equation}
    \wt{\lw}_{\mathrm{g}}(i,j) \leq \wt{w}_{\mathrm{g}}(i,j) \leq \wt{\hw}_{\mathrm{g}}(i,j) \, .
\end{equation}
When \(\yd{\lambda}_{\mathrm{g}}\) is a regular YD, they correspond to the usual lowest / highest weight WTs, illustrated in~\cref{fig:gYT}(c,d).
Further, let \(\yd{\lambda}_{\mathrm{g}}\) comprise two disjoint gYDs \(\yd{\mu}_{\mathrm{g}}\) and \(\yd{\nu}_{\mathrm{g}}\). Formally,
\begin{equation}
    \yd{\mu}_{\mathrm{g}} \subseteq \yd{\lambda}_{\mathrm{g}} \, ,\quad \yd{\nu}_{\mathrm{g}} = \yd{\lambda}_{\mathrm{g}} \setminus \yd{\mu}_{\mathrm{g}} \, .
\end{equation}
For some \(\wt{w'}_{\mathrm{g}}\vdash_{d} \yd{\nu}_{\mathrm{g}}\), we define \(\boldsymbol{x}\mleft(\wt{w'}_{\mathrm{g}}\mright)\) as the constraint that forces
\begin{equation}
    \wt{w''}_{\mathrm{g}}\vdash_{d} \yd{\mu}_{\mathrm{g}}|\boldsymbol{x}\mleft(\wt{w}_{\mathrm{g}}\mright)
\end{equation}
to be a valid extension of \(\wt{w'}_{\mathrm{g}}\) onto \(\yd{\lambda}_{\mathrm{g}}\). More formally, we have
\begin{equation}
    \wt{w''}_{\mathrm{g}}\vdash_{d} \yd{\mu}_{\mathrm{g}}|\boldsymbol{x}\mleft(\wt{w'}_{\mathrm{g}}\mright) \quad \Leftrightarrow \quad \mleft(\wt{w'}_{\mathrm{g}} \cup \wt{w''}_{\mathrm{g}}\mright) \vdash_{d} \yd{\lambda}_{\mathrm{g}} \, .
\end{equation}
Finally, we introduce restrictions of gWTs. For \(\wt{w}_{\mathrm{g}} \vdash_{d} \yd{\lambda}_{\mathrm{g}}\), we define
\begin{equation}
    \wt{w|_{\yd{\nu}_{\mathrm{g}}}}_{\mathrm{g}} \vdash_{d} \yd{\nu}_{\mathrm{g}}
\end{equation}
to be the gWT restricted to the entries on \(\yd{\nu}_{\mathrm{g}} \subseteq \yd{\lambda}_{\mathrm{g}}\). For the special case of lowest and highest weight gWTs, we introduce

\begin{equation}
\begin{aligned}
    \boldsymbol{x}_{\min} \coloneqq \boldsymbol{x}\mleft(\wt{\lw|_{\yd{\nu}_{\mathrm{g}}}}_{\mathrm{g}}\mright) \, , \quad \boldsymbol{x}_{\max} \coloneqq \boldsymbol{x}\mleft(\wt{\hw|_{\yd{\nu}_{\mathrm{g}}}}_{\mathrm{g}}\mright) \, .
\end{aligned}
\end{equation}
The final main theorem is the following.
\begin{theorem}[\cref{thm:restriction_averages_generalized_YD_monotone_functions}, informal]
    \label{thm:main_part_restriction_averages_generalized_YD_monotone}
    Let
    \begin{equation}
        F\mleft(\wt{w}_{\mathrm{g}}\mright) = F|_{\yd{\mu}_{\mathrm{g}}}\mleft(\wt{w|_{\yd{\mu}_{\mathrm{g}}}}_{\mathrm{g}}\mright) F|_{\yd{\nu}_{\mathrm{g}}}\mleft(\wt{w|_{\yd{\nu}_{\mathrm{g}}}}_{\mathrm{g}}\mright) \, ,
    \end{equation}
    where all functions involved are nonnegative. If \(F|_{\yd{\mu}_{\mathrm{g}}}\) is nondecreasing in every entry, then we have
    \begin{equation}
    \begin{aligned}       \weylavg{F\mleft(\wt{w}_{\mathrm{g}}\mright)}{\boldsymbol{p}}{\yd{\lambda}_{\mathrm{g}}} \geq \weylavg{F|_{\yd{\mu}_{\mathrm{g}}}\mleft(\wt{w'}_{\mathrm{g}}\mright)}{\boldsymbol{p}}{\yd{\mu}_{\mathrm{g}}|\boldsymbol{x}_{\min}} \\
        \cdot \, \weylavg{F|_{\yd{\nu}_{\mathrm{g}}}\mleft(\wt{w|_{\yd{\nu}_{\mathrm{g}}}}_{\mathrm{g}}\mright)}{\boldsymbol{p}}{\yd{\lambda}_{\mathrm{g}}}
    \end{aligned}
\end{equation}
and
    \begin{equation}
    \begin{aligned}   \weylavg{F\mleft(\wt{w}_{\mathrm{g}}\mright)}{\boldsymbol{p}}{\yd{\lambda}_{\mathrm{g}}} \leq \weylavg{F|_{\yd{\mu}_{\mathrm{g}}}\mleft(\wt{w'}_{\mathrm{g}}\mright)}{\boldsymbol{p}}{\yd{\mu}_{\mathrm{g}}|\boldsymbol{x}_{\max}} \\
        \cdot \, \weylavg{F|_{\yd{\nu}_{\mathrm{g}}}\mleft(\wt{w|_{\yd{\nu}_{\mathrm{g}}}}_{\mathrm{g}}\mright)}{\boldsymbol{p}}{\yd{\lambda}_{\mathrm{g}}} \, .
    \end{aligned}
\end{equation}
On the other hand, if \(F|_{\yd{\mu}_{\mathrm{g}}}\) is nonincreasing in every entry, then the inequality signs are reversed.
\end{theorem}
Taking \(\yd{\lambda}_{\mathrm{g}}\) to be a YD and restricting to products of functions of the top and bottom rows gives~\cref{thm:main_part_splitting_diagram_average_weight}. A schematic partition of the gYD into the subdiagrams supporting the two averages with constraints is shown in~\cref{fig:gYT}(f,g).
To illustrate the proof idea of~\cref{thm:main_part_restriction_averages_generalized_YD_monotone}, consider again \(\yd{\lambda} = \yd{n}\). Let \(\yd{\mu} = \yd{n-1}\), and take a function \(F\) on \(\mathcal{W}_{\yd{\lambda}_{\mathrm{g}}}^{d}\) which does not depend on the first box, that is,
\begin{equation}
    F\mleft(\wt{w}_{\mathrm{g}}\mright) = F|_{\yd{\mu}_{\mathrm{g}}}\mleft(\wt{w|_{\yd{\mu}_{\mathrm{g}}}}_{\mathrm{g}}\mright) \, .
\end{equation}
Restating~\cref{equ:main_part_proof_outline_decomposition_schur_polynomial} with this notation gives
\begin{equation}
    \label{equ:main_part_proof_outline_decomposition_schur_polynomial_constraints}
    \schur{\yd{\lambda}_{\mathrm{g}}}[\boldsymbol{p}] = \sum_{i=1}^{d}p_{i}\schur{\yd{\mu}_{\mathrm{g}}|\boldsymbol{x}(i)}[\boldsymbol{p}] \, ,
\end{equation}
where \(\boldsymbol{x}(i)\) restricts all the entries to \(\{i,\ldots,d\}\). Similarly, we obtain
\begin{equation}
    \label{equ:main_part_proof_outline_decomposition_weighted_schur_polynomial_constraints}
    \schur[F]{\yd{\lambda}_{\mathrm{g}}}[\boldsymbol{p}] = \sum_{i=1}^{d}p_{i}\schur[F|_{\yd{\mu}_{\mathrm{g}}}]{\yd{\mu}_{\mathrm{g}}|\boldsymbol{x}(i)}[\boldsymbol{p}] \, .
\end{equation}
Here, one fixes the first box to have filling \(i\) and then sums over the remaining possibilities. One checks that \(\boldsymbol{x}(i) \leq \boldsymbol{x}(j)\) entrywise for \(i\leq j\). If \(F|_{\yd{\mu}_{\mathrm{g}}}\) is additionally entrywise nondecreasing, then~\cref{thm:main_part_monotonicity_functions_young_diagrams} gives
\begin{equation}
\begin{aligned}
\label{equ:main_part_proof_outline_decomposition_weighted_schur_polynomial_constraints_inequality}
    \schur[F|_{\yd{\mu}_{\mathrm{g}}}]{\yd{\mu}_{\mathrm{g}}|\boldsymbol{x}(i)}[\boldsymbol{p}] =& \schur{\yd{\mu}_{\mathrm{g}}|\boldsymbol{x}(i)}[\boldsymbol{p}] \weylavg{F|_{\yd{\mu}_{\mathrm{g}}}\mleft(\wt{w'}_{\mathrm{g}}\mright)}{\boldsymbol{p}}{\yd{\mu}_{\mathrm{g}}|\boldsymbol{x}(i)} \\
    \geq& \schur{\yd{\mu}_{\mathrm{g}}|\boldsymbol{x}(i)}[\boldsymbol{p}] \weylavg{F|_{\yd{\mu}_{\mathrm{g}}}\mleft(\wt{w'}_{\mathrm{g}}\mright)}{\boldsymbol{p}}{\yd{\mu}_{\mathrm{g}}|\boldsymbol{x}(1)} \, .
\end{aligned}
\end{equation}
In addition, it follows that \(\wt{\lw}_{\mathrm{g}}\vdash_{d}\yd{\lambda}_{\mathrm{g}}\) is given precisely by the gWT with all \(1\)'s, and therefore
\begin{equation}
    \boldsymbol{x}_{\min} = \boldsymbol{x}(1) \, .
\end{equation}
In particular, \(\boldsymbol{x}(1)\) is equivalent to no constraint at all in this case, although this observation is not needed below. Combining~\cref{equ:main_part_proof_outline_decomposition_schur_polynomial_constraints,equ:main_part_proof_outline_decomposition_weighted_schur_polynomial_constraints,equ:main_part_proof_outline_decomposition_weighted_schur_polynomial_constraints_inequality} gives
\begin{equation}
\begin{aligned}
    &\weylavg{F\mleft(\wt{w}_{\mathrm{g}}\mright)}{\boldsymbol{p}}{\yd{\lambda}_{\mathrm{g}}} = \frac{\schur[F]{\yd{\lambda}_{\mathrm{g}}}[\boldsymbol{p}]}{\schur{\yd{\lambda}_{\mathrm{g}}}[\boldsymbol{p}]} \\
    \geq& \weylavg{F|_{\yd{\mu}_{\mathrm{g}}}\mleft(\wt{w'}_{\mathrm{g}}\mright)}{\boldsymbol{p}}{\yd{\mu}_{\mathrm{g}}|\boldsymbol{x}(1)} \, .
\end{aligned}
\end{equation}
The full proof of~\cref{thm:main_part_restriction_averages_generalized_YD_monotone} extends this idea to other gYDs and \(F|_{\yd{\nu}_{\mathrm{g}}} \neq 1\).

\section{Outlook}
This paper leaves several natural directions open. The asymptotic results assume a nondegenerate input spectrum, while the nonasymptotic optimality is currently established only under large row-gap assumptions. We also expect that the overhang removal rule may remain optimal beyond the regime proven here, possibly for all $n$ in the principal-eigenstate case $k\!=\!1$, and that tighter nonasymptotic bounds may be available. More broadly, the QPA task studied here should be viewed as one instance of a wider theory of coherent spectrum transformation; characterizing which such transformations are coherently achievable, and with what sample complexity, would extend the present framework. Finally, finding more efficient and practically relevant implementations for the overhang-removal protocol without relying on primitives such as GQPE or CG transforms remains a concrete algorithmic challenge.

\section{Acknowledgments}
We thank A. Gilyén and D. Grinko for helpful discussions. After the first version of this work was posted on arXiv, we became aware of independent work~\cite{SSW26}. We thank T. Scharnhorst, J. Spilecki and J. Wright for subsequently sharing their manuscript with us and for helpful discussions. Z.L., A.W.H. and I.L.C. acknowledge support from the U.S. Department of Energy, Office of Science, National Quantum Information Science Research Centers, Co-design Center for Quantum Advantage (C\(^2\)QA), under contract No. DE-SC0012704.
This work is supported by the National Science Foundation under Cooperative Agreement PHY-2019786 (The NSF AI Institute for Artificial Intelligence and Fundamental Interactions).
ET is supported by the European Union via an ERC grant (QInteract, Grant No 101078107) and VILLUM FONDEN (Grant No 10059 and 37532).
The authors further acknowledge that they used generative AI tools for writing and editing, literature review, and assistance with proofs of some lemmas and theorems. The main conceptual and technical contributions originated with the authors. In particular, the main concepts underlying general QPA, the path-graph parametrization used to analyze the asymptotics of Weyl distributions, and the framework of generalized YDs and its proofs were developed without AI assistance.

\putbib[qpa_main.bib]
\end{bibunit}

\beginsupplement

\begin{bibunit}[apsrev4-2]
\begin{center}
  {\Large \textbf{Supplemental Material}}\\[0.25cm]
  {\normalsize for}\\[0.2cm]
  {\large \textbf{Quantum Purity Amplification}}\\[0.08cm]
  {\large \textbf{for Arbitrary Eigenstates and Multiple Outputs}}\\[0.35cm]

  {\normalsize Zhaoyi Li,\textsuperscript{1,2,*} Elias Theil,\textsuperscript{3,*} Aram W. Harrow,\textsuperscript{1} and Isaac Chuang\textsuperscript{1}}\\[0.25cm]

  {\footnotesize\itshape
  \textsuperscript{1}Department of Physics, Massachusetts Institute of Technology, Cambridge MA 02139, USA \\
  \textsuperscript{2}The NSF AI Institute for Artificial Intelligence and Fundamental Interactions \\
  \textsuperscript{3}Department of Mathematical Sciences, University of Copenhagen, 2100 Copenhagen, Denmark \\
  \textsuperscript{*}These authors contributed equally. Emails: \texttt{ladmon@mit.edu}, \texttt{edmt@math.ku.dk}.}
\end{center}

\vspace{0.5cm}
\noindent In this Supplemental Material (SM), we provide definitions, complete proofs, and additional technical details supporting the quantum purity amplification (QPA) results in the main text.
We first introduce the mathematical background required for the technical arguments in \cref{sec:math_prelim_supp}: asymptotic notation conventions (\cref{subsec:math_conventions_supp}), representation-theoretic preliminaries (\cref{subsec:math_preliminaries}), extremal channels and the optimal Choi matrix (\cref{subsec:extremal_supp}), and isometric intertwiners realizing extremal channels (\cref{subsec:isometric_supp}).
The construction and sector-wise reduction of QPA are presented in \cref{sec:general_qpa_supp}: we set up the general task and notation (\cref{subsec:setup_supp}), characterize the protocols and their utilities (\cref{subsec:protocols_supp}), reduce the overall utility to the sector-wise utility via SW concentration (\cref{subsec:concentration_supp}), and give the explicit sector-wise utility expression (\cref{subsec:expression_supp}).
Implementation details are collected in \cref{sec:implementation}, where we analyze the correctness and complexity of the GQPE-OQPA and CG-OQPA implementations (\cref{subsec:implementation_via_GPE,subsec:implementation_via_CG_transform}).
The asymptotic behavior of sector-wise utilities is analyzed in \cref{sec:asymptotic_supp}: after introducing the GT parametrization with path graphs (\cref{subsec:gt_param_supp}), we treat the all-site utility for intensive outputs (\cref{subsubsec:intensive_supp}) and extensive removal (\cref{subsubsec:extensive_supp}), followed by the one-site utility (\cref{subsec:asymptotic_one_site_supp}).
Nonasymptotic, dimension-uniform bounds on the sector-wise utility are developed in \cref{sec:qpa_irrep_level_all_copy}, beginning with gYDs with constraints (\cref{subsec:constrained_yd_supp}), followed by the all-site analysis (\cref{subsec:nonasymptotic_allcopy_supp}), the one-site analysis (\cref{subsec:nonasymptotic_onecopy_supp}), and overall bounds (\cref{subsec:nonasymptotic_overall_bounds_supp}).
Finally, optimality of overhang removal is treated in \cref{sec:optimality_supp}, beginning with structural properties of the \(F\)-symbols (\cref{subsec:props_F_symbols}) and followed by asymptotic and nonasymptotic optimality statements (\cref{subsec:asymptotic_optimality_supp,subsec:nonasymptotic_optimality_supp}).
\newpage
\section*{Contents}
\vspace{0.3cm}
\begin{enumerate}[label=S\arabic*.]
    \item \textit{From coherent inference to QPA}\dotfill \pageref{sec:math_prelim_supp}
        \begin{enumerate}[label=\Alph*.]
            \item \textit{Math conventions}\dotfill \pageref{subsec:math_conventions_supp}

            \item \textit{Representation-theoretic preliminaries}\dotfill \pageref{subsec:math_preliminaries}
            \item \textit{Extremal channels and the optimal Choi matrix}\dotfill \pageref{subsec:extremal_supp}
            \item \textit{Isometric intertwiners realizing extremal channels}\dotfill \pageref{subsec:isometric_supp}
        \end{enumerate}
    \item \textit{General QPA}\dotfill \pageref{sec:general_qpa_supp}
        \begin{enumerate}[label=\Alph*.]
            \item \textit{Setup and notation}\dotfill \pageref{subsec:setup_supp}
            \item \textit{QPA protocols and their utilities}\dotfill \pageref{subsec:protocols_supp}
            \item \textit{From overall utility to sector-wise utility via SW concentration}\dotfill \pageref{subsec:concentration_supp}
            \item \textit{Expression of sector-wise utilities}\dotfill \pageref{subsec:expression_supp}
        \end{enumerate}
    \item \textit{Implementation}\dotfill \pageref{sec:implementation}
        \begin{enumerate}[label=\Alph*.]
            \item \textit{GQPE-OQPA correctness and complexity}\dotfill \pageref{subsec:implementation_via_GPE}
            \item \textit{CG-OQPA correctness and complexity}\dotfill \pageref{subsec:implementation_via_CG_transform}
        \end{enumerate}
    \item \textit{Asymptotic analysis of sector-wise utilities}\dotfill \pageref{sec:asymptotic_supp}
        \begin{enumerate}[label=\Alph*.]
            \item \textit{GT parametrization with path graphs}\dotfill \pageref{subsec:gt_param_supp}
            \item \textit{Intensive outputs}\dotfill \pageref{subsubsec:intensive_supp}
            \item \textit{Asymptotic analysis of the all-site utility}\dotfill \pageref{subsubsec:extensive_supp}
                \begin{enumerate}[label=\roman*.]
                    \item \textit{Overhang Removal Ray}\dotfill \pageref{subsubsec:extensive_overhang_ray_supp}
                    \item \textit{Distinct-ray competitors}\dotfill \pageref{subsubsec:extensive_distinct_ray_supp}
                \end{enumerate}
            \item \textit{Properties of the \(F\)-symbols}\dotfill \pageref{subsec:props_F_symbols}
            \item \textit{Asymptotic analysis of the one-site utility}\dotfill \pageref{subsec:asymptotic_one_site_supp}
        \end{enumerate}
    \item \textit{Nonasymptotic analysis of the sector-wise utilities}\dotfill \pageref{sec:qpa_irrep_level_all_copy}
        \begin{enumerate}[label=\Alph*.]
            \item \textit{gYDs with constraints}\dotfill \pageref{subsec:constrained_yd_supp}
                \begin{enumerate}[label=\roman*.]
                    \item \textit{Definitions and monotonicity of averages under constraints}\dotfill \pageref{subsubsec:definitions_monotonicity_averages}
                    \item \textit{Restrictions and Schur monotonicity}\dotfill \pageref{subsubsec:log_convexity_Schur_polynomials}
                    \item \textit{Splitting averages}\dotfill \pageref{subsubsec:splitting_averages}
                \end{enumerate}
            \item \textit{Nonasymptotic analysis of the all-site utility}\dotfill \pageref{subsec:nonasymptotic_allcopy_supp}
                \begin{enumerate}[label=\roman*.]
                    \item \textit{Miscellaneous lemmas}\dotfill \pageref{subsec:misc_lemmas_supp}
                    \item \textit{Sector-wise utility bounds}\dotfill \pageref{subsubsec:sector_utility_bounds_supp}
                \end{enumerate}
            \item \textit{Nonasymptotic analysis of the one-site utility}\dotfill \pageref{subsec:nonasymptotic_onecopy_supp}
            \item \textit{Nonasymptotic overall bounds}\dotfill \pageref{subsec:nonasymptotic_overall_bounds_supp}
        \end{enumerate}
    \item \textit{Optimality of overhang removal}\dotfill \pageref{sec:optimality_supp}
        \begin{enumerate}[label=\Alph*.]
            \item \textit{Asymptotic optimality}\dotfill \pageref{subsec:asymptotic_optimality_supp}
            \item \textit{Nonasymptotic optimality}\dotfill \pageref{subsec:nonasymptotic_optimality_supp}
        \end{enumerate}
\end{enumerate}
\clearpage

\section{From coherent inference to QPA}
\label{sec:math_prelim_supp}
\subsection{Math conventions}
\label{subsec:math_conventions_supp}
\paragraph{Asymptotic notation.}
We use the standard asymptotic notations
\(\bigO{\cdot}\), \(\lilo{\cdot}\), \(\bigOm{\cdot}\),
\(\lilom{\cdot}\), and \(\bigTheta{\cdot}\).
A subscript specifies exactly the parameters on which the implicit constants or comparison functions may depend.
A superscript specifies the asymptotic variable(s), with the corresponding asymptotic regime determined by the surrounding statement, such as \(n\to\infty\) or \(\epsilon\to0\).
The estimate is understood to be uniform in all displayed variables not appearing in the subscript.
The domain of each variable is likewise understood from the surrounding statement, for example \(n\in\mathbb N\) or \(\epsilon\in(0,1)\).
When the asymptotic variable is itself the displayed argument, we omit the superscript; for example,
\(\bigO[][n]{\ln n}\) is written simply as \(\bigO{\ln n}\).

The statement
\begin{equation*}
f(n,\boldsymbol a,\boldsymbol u)
=
\bigO[\boldsymbol u][n]{g(n,\boldsymbol a,\boldsymbol u)}
\tag{S0a}\label{eq:asymptotic_bigO_statement}
\end{equation*}
means that there exist a constant \(C_{\boldsymbol u}>0\) and a cutoff \(N_{\boldsymbol u}\in\mathbb N\), depending only on the parameters appearing in the subscript, such that
\begin{equation*}
\lvert f(n,\boldsymbol a,\boldsymbol u)\rvert
\le
C_{\boldsymbol u}\lvert g(n,\boldsymbol a,\boldsymbol u)\rvert
\tag{S0b}\label{eq:asymptotic_bigO_bound}
\end{equation*}
for all \(n\geq N_{\boldsymbol u}\), uniformly over all displayed variables not appearing in the subscript. If a variable is explicitly declared fixed in the surrounding statement, it is excluded from this uniformity requirement, and we may suppress it from the asymptotic-notation subscript.

The notation
\(
\lilo[\boldsymbol u][n]{g}
\)
is defined analogously, except that the constant \(C_{\boldsymbol u}\) is replaced by a nonnegative comparison function
\(\eta_{\boldsymbol u}(n)\), depending only on the parameters appearing in the subscript, satisfying
\(\eta_{\boldsymbol u}(n)\to0\)
along the specified asymptotic regime.
The notation
\(
\bigOm[\boldsymbol u][n]{g}
\)
is obtained by replacing the upper bound in the definition of
\(
\bigO[\boldsymbol u][n]{g}
\)
with the corresponding lower bound.
Likewise,
\(
\lilom[\boldsymbol u][n]{g}
\)
is obtained by replacing
\(
\eta_{\boldsymbol u}(n)\to0
\)
with
\(
\eta_{\boldsymbol u}(n)\to\infty
\),
and
\(
\bigTheta[\boldsymbol u][n]{g}
\)
means that both the corresponding
\(
\bigO[\boldsymbol u][n]{g}
\)
and
\(
\bigOm[\boldsymbol u][n]{g}
\)
bounds hold.

We also write
\(
\bigO[\boldsymbol u][n]{\exp(-n)}
\)
to denote exponential decay as \(n\to\infty\), namely
\begin{equation*}
\lvert f(n,\boldsymbol a,\boldsymbol u)\rvert
\le
C_{\boldsymbol u}e^{-c_{\boldsymbol u}n}
\tag{S0c}\label{eq:asymptotic_exponential_bound}
\end{equation*}
for some constants \(C_{\boldsymbol u},c_{\boldsymbol u}>0\), depending only on the parameters appearing in the subscript, for all sufficiently large \(n\), uniformly over all displayed variables not appearing in the subscript.

For example, in a statement taking \(n\to\infty\),
\(f(n,\boldsymbol a,\boldsymbol u)=\bigO[\boldsymbol u][n]{1}\) means that \(|f(n,\boldsymbol a,\boldsymbol u)|\le C_{\boldsymbol u}\) for all sufficiently large \(n\), uniformly in \(\boldsymbol a\), with \(C_{\boldsymbol u}\) allowed to depend only on \(\boldsymbol u\).
By contrast, \(f(n,\boldsymbol a,\boldsymbol u)=\bigO[\boldsymbol u]{1}\) means that the same bound holds uniformly in both \(n\) and \(\boldsymbol a\).

\paragraph{Component expansion.}
We use \(\dot{=}\) for component expansion. For example, \(\boldsymbol{x}\dot{=}(x_1,\ldots,x_d)\) records the components of \(\boldsymbol{x}\). The same convention applies to Young diagrams (YDs), tuples, and other indexed objects.

\paragraph{Empty sums and products.}
When the allowed range of a sum or product is empty, we set the sum equal to $0$ and the product equal to $1$.

\subsection{Representation-theoretic preliminaries}
\label{subsec:math_preliminaries}
YDs, Young tableaux (YTs), and Weyl tableaux (WTs) form the combinatorial backbone of the representation theory used here (symmetric groups and unitary groups via Schur--Weyl (SW) duality, Littlewood--Richardson (LR) coefficients, and multiplicity spaces).
For succinct definitions and properties we follow Ref.~\cite{LFIC24S} (and its Supplement), to which we refer for details. For more details on the topic of rational representations and equivalence of representations, we refer to the introductions of Refs.~\cite{N23S,grinko2023gelfand}. For intertwiners and unitary-equivariant channels, we refer to Ref.~\cite{MT25S}. We emphasize the following concepts.

\paragraph{Rational representations.} We label rational irreps of \(U(d)\) by mixed YDs, that is, generalizations of YDs whose entries may be negative. We write \(\yd{\varsigma}\vdash_{d}(n,m)\) if
\begin{equation}
    \yd{\varsigma} \dot{=} [\varsigma_{1},\ldots,\varsigma_{d}] \quad , \quad \sum_{i:\varsigma_{i}\geq 0} \varsigma_{i} \leq n \quad , \quad \sum_{i:\varsigma_{i}\leq 0} |\varsigma_{i}| \leq m \quad , \quad \sum_{i=1}^{d} \varsigma_{i} = n-m \, .
\end{equation}
This \(U(d)\) labeling is particularly convenient because it directly reflects the mixed-tensor structure.
In the case $\yd{\varsigma}\vdash_{d}(n,0)$, these reduce to ordinary YDs, and we simply write $\yd{\varsigma}\vdash_{d} n$.

Throughout this paper, the symmetry actions relevant to our results are the restrictions of these \(U(d)\) actions to \(SU(d)\). Under this restriction, uniformly shifting every row of a YD produces an equivalent \(SU(d)\) irrep: \(\mathbb{W}^{\yd{\varsigma}}\simeq\mathbb{W}^{\yd{\varrho}}\) whenever \(\yd{\varsigma}+\yd{z,\ldots,z}=\yd{\varrho}\) entrywise for \(z\in\mathbb Z\). Such uniform shifts leave the corresponding LR multiplicities, CG decompositions and intertwiners, \(F\)-symbols, and relevant dimensions unchanged. As a nomenclature convention, throughout this paper, YDs include mixed YDs; similarly, YTs and WTs include mixed YTs and WTs. Importantly, this identification applies only to the \(SU(d)\) irreps and not to the corresponding symmetric-group representations.

\paragraph{Dual representations.} For \(\yd{\varsigma} \vdash_{d} (n,m)\) we denote the dual irrep label by \(\yd{\varsigma}^{\star} \vdash_{d} (m,n)\), and we have
\begin{equation}
    \yd{\varsigma}^{\star} \dot{=} [-\varsigma_{d},\ldots,-\varsigma_{1}] \, .
\end{equation}

\paragraph{LR multiplicity space for a triple \(\triple{\yd{\mu}}{\yd{\varrho}}{\yd{\varsigma}}\).}
Given YDs \(\yd{\varrho}\), \(\yd{\mu}\), and \(\yd{\varsigma}\), the LR coefficient
\(\res{c}{\yd{\mu}}{\yd{\varrho}}{\yd{\varsigma}}\) gives the multiplicity of the irrep \(\yd{\varsigma}\) appearing in
\(\mathbb{W}^{\yd{\mu}}\otimes \mathbb{W}^{\yd{\varrho}}\) under the Clebsch--Gordan decomposition. Equivalently,
\begin{equation}
\mathbb{W}^{\yd{\mu}}\otimes \mathbb{W}^{\yd{\varrho}}
\simeq
\bigoplus_{\yd{\varsigma}}
\mathbb{W}^{\yd{\varsigma}}
\otimes
\res{\mathbb{D}}{\yd{\mu}}{\yd{\varrho}}{\yd{\varsigma}},
\end{equation}
where the degeneracy space has dimension
\begin{equation}
\dim
\res{\mathbb{D}}{\yd{\mu}}{\yd{\varrho}}{\yd{\varsigma}}
=
\res{c}{\yd{\mu}}{\yd{\varrho}}{\yd{\varsigma}}.
\end{equation}
Equivalently, by tensor--hom adjunction, the irrep \(\yd{\mu}\) appears with the same multiplicity as a subrepresentation of
\(\mathbb{W}^{\yd{\varsigma}}\otimes \mathbb{W}^{\yd{\varrho}^{\star}}\).
\paragraph{Intertwiner notation.}
For each unit vector
\(\ket{\phi}\in\res{\mathbb{D}}{\yd{\mu}}{\yd{\varrho}}{\yd{\varsigma}}\),
there is a unitary-equivariant isometric embedding \begin{equation}
\intertwiner{W}{\yd{\varsigma}}{\yd{\mu}}{\yd{\varrho}}{\ket{\phi}}:\; \mathbb{W}^{\yd{\varsigma}}\;\longrightarrow\; \mathbb{W}^{\yd{\mu}}\otimes \mathbb{W}^{\yd{\varrho}},
\end{equation}
that satisfies ${\intertwineradj{W}{\yd{\varsigma}}{\yd{\mu}}{\yd{\varrho}}{\ket{\phi}}}\,\intertwiner{W}{\yd{\varsigma}}{\yd{\mu}}{\yd{\varrho}}{\ket{\phi}}=I_{\yd{\varsigma}}$.
Here $\intertwineradj{W}{\yd{\varsigma}}{\yd{\mu}}{\yd{\varrho}}{\ket{\phi}}={\intertwiner{W}{\yd{\varsigma}}{\yd{\mu}}{\yd{\varrho}}{\ket{\phi}}}^\dagger$ and we adopt the Condon-Shortley phase convention, so that this embedding is equivalent to the inverse CG transform restricted to the subspace $\mathbb{W}^{\yd{\varsigma}}\otimes \ketbra{\phi}{\phi}$, that is, $\intertwiner{W}{\yd{\varsigma}}{\yd{\mu}}{\yd{\varrho}}{\ket{\phi}}=\mleft.U^\dag_\mathrm{CG}\mright|_{\mathbb{W}^{\yd{\varsigma}}\otimes \ketbra{\phi}{\phi}}$.

\subsection{Extremal channels and the optimal Choi matrix}
\label{subsec:extremal_supp}
We follow the approach of Ref.~\cite{MT25S}, which characterizes the extreme points of the set of symmetric channels and their Choi matrices.

\begin{theorem}[Optimal Choi matrix]
\label{thm:SM_QPA_choi}
A generic element of a symmetric channel decomposes into irrep blocks associated with the joint symmetry $S_m \otimes S_n \otimes SU(d)$.
Its Choi matrix admits a block-diagonal decomposition, where each block acts on the input Specht module $\mathbb{V}^{\yd{\varsigma}}$ with irrep $\yd{\varsigma}$, the output Specht module $\mathbb{V}^{\yd{\varrho}}$ with irrep $\yd{\varrho}$, and the Weyl module $\mathbb{W}^{\yd{\lambda}^{\star}}$ with irrep $\yd{\lambda}^{\star}$. For brevity, when labeling registers, we use $\yd{\lambda}$ as shorthand for the corresponding Weyl module $\mathbb{W}^{\yd{\lambda}}$.
The matrix $\res{M}{\yd{\varsigma}^{\star}}{\yd{\varrho}}{\yd{\lambda}^{\star}}$ acts on the degeneracy register labeled $D$, namely $\res{\mathbb{D}}{\yd{\varsigma}^{\star}}{\yd{\varrho}}{\yd{\lambda}^{\star}}$, which is associated with the LR coefficient $\res{c}{\yd{\varsigma}^{\star}}{\yd{\varrho}}{\yd{\lambda}^{\star}}$.
\begin{equation}
T
\;\simeq\;
\bigoplus_{\yd{\varsigma}\vdash_d n,\yd{\varrho}\vdash_d m,\yd{\lambda}\vdash_{d}(n,m)}
\;
I_{\mathbb{V}^{\yd{\varsigma}}}
\;\otimes\;
\frac{1}{g^{\yd{\varrho}}}I_{\mathbb{V}^{\yd{\varrho}}}
\;\otimes\;
\frac{d^{\yd{\varsigma}}}{d^{\yd{\lambda}}}\,
I_{\yd{\lambda}^{\star}}
\;\otimes\;
\res{M}{\yd{\varsigma}^{\star}}{\yd{\varrho}}{\yd{\lambda}^{\star}}_D.
\end{equation}
Here, $g^{\yd{\nu}}$ denotes the dimension of the Specht module $\mathbb{V}^{\yd{\nu}}$, and $d^{\yd{\nu}}$ denotes the dimension of the Weyl module $\mathbb{W}^\yd{\nu}$, and we remark that \(d^{\yd{\nu}} = d^{\yd{\nu}^{\star}}\).
Moreover, the matrix $\res{M}{\yd{\varsigma}^{\star}}{\yd{\varrho}}{\yd{\lambda}^{\star}}$ satisfies
\begin{subequations}
\begin{align}
    \res{M}{\yd{\varsigma}^{\star}}{\yd{\varrho}}{\yd{\lambda}^{\star}} &\succeq 0 \, ,\\
    \sum_{\yd{\varrho},\yd{\lambda}} \Tr\mleft(\res{M}{\yd{\varsigma}^{\star}}{\yd{\varrho}}{\yd{\lambda}^{\star}}\mright) &= 1\, .
\end{align}
\end{subequations}
\end{theorem}

Since the objective function is linear and the set of symmetric channels is compact and convex, the maximum is attained at an extreme point of this set. The extreme points of the set of symmetric channels have been characterized in Ref.~\cite{MT25S}, and are presented as follows:

\begin{corollary}[Choi Matrices of Extremal Channels]
\label{cor:extremal_channels}
For an extremal channel, the matrix $\res{M}{\yd{\varsigma}^{\star}}{\yd{\varrho}}{\yd{\lambda}^{\star}}$ reduces to a rank-1 projector supported on a particular choice of $\yd{\varrho}$ and $\yd{\lambda}$ for each $\yd{\varsigma}$, such that the corresponding degeneracy register is nonzero. This means in particular that
\begin{equation}
    \res{M}{\yd{\varsigma}^{\star}}{\yd{\varrho}}{\yd{\lambda}^{\star}}=\ketbra{\phi}{\phi}
\end{equation}
for one particular triple of \(\triple{\yd{\varsigma}^{\star}}{\yd{\varrho}}{\yd{\lambda}^{\star}}\) for each \(\yd{\varsigma}\vdash_{d} n\), while it is \(0\) otherwise.
\end{corollary}

To illustrate this, we define the following sector-wise Choi matrix for each branch, living in the space $\mathbb{W}^{\yd{\varsigma}^{\star}}\otimes\mathbb{W}^{\yd{\varrho}}$:
\begin{equation}
\res{T}{\yd{\varsigma}^{\star}}{\yd{\varrho}}{\yd{\lambda}^{\star}}_{\phi}=U^\dag_\mathrm{CG}\mleft(O\oplus\cdots\oplus O\oplus\frac{d^\yd{\varsigma}}{d^\yd{\lambda}}I_{\yd{\lambda}^{\star}}
\,\otimes\,
\ketbra{\phi}{\phi}_D\oplus O\oplus\cdots\oplus O\mright)U_{\mathrm{CG}},
\end{equation}
where we denote the zero matrix as $O$.
The overall Choi matrix is given by the direct sum over all sectors, with the unused output irreps (those not equal to $\yd{\varrho}$) padded by the zero matrix.
\begin{equation}
 T
\;=\; U_{\mathrm{Sch}_n}^{\dagger} \left(\bigoplus_{\yd{\varsigma}\vdash_d n}I_{\mathbb{V}^{\yd{\varsigma}}}\otimes U_{\mathrm{Sch}_m}^{\dagger}\mleft( O\oplus\cdots\oplus O\oplus
I_{\mathbb{V}^{\yd{\varrho}}}\otimes \res{T}{\yd{\varsigma}^{\star}}{\yd{\varrho}}{\yd{\lambda}^{\star}}_{\phi}\oplus O\oplus\cdots\oplus O\mright) U_{\mathrm{Sch}_m} \right)U_{\mathrm{Sch}_n}\,,
\end{equation}
where \(U_{\mathrm{Sch}_k}\) denotes a choice of Schur transform. In the next section, we discuss the implementation of the channel corresponding to this Choi matrix.

\subsection{Isometric intertwiners realizing extremal channels}
\label{subsec:isometric_supp}
Block-projector Choi matrices of the form appearing in~\cref{cor:extremal_channels} admit a natural operational interpretation as quantum channels, which we characterize here using the intertwiner formalism introduced in~\cref{subsec:math_preliminaries}. We first remark that we can canonically identify \(\res{\mathbb{D}}{\yd{\varsigma}^{\star}}{\yd{\varrho}}{\yd{\lambda}^{\star}} \simeq \res{\mathbb{D}}{\yd{\lambda}}{\yd{\varrho}}{\yd{\varsigma}}\); see Ref.~\cite{MT25S}. With this in hand, we can state the following theorem, which is proven in~\cite{MT25S}.

\begin{theorem}[Extremal Channel Form]
    \label{thm:extremal_channel_form}
  Given such an isometry, an extremal covariant channel on the irrep block is obtained by tracing out the irrep remainder,
\begin{equation}
\extChannel{\mathcal{T}}{\yd{\varsigma}}{\yd{\varrho}}{\yd{\lambda}}[\ket{\phi}](\cdot)
\;=\;
\Tr_{\yd{\lambda}}\mleft(\intertwiner{W}{\yd{\varsigma}}{\yd{\lambda}}{\yd{\varrho}}{\ket{\phi}}\,\cdot\,\intertwineradj{W}{\yd{\varsigma}}{\yd{\lambda}}{\yd{\varrho}}{\ket{\phi}}\mright).
\end{equation}

The Choi matrix of this channel is precisely \(\res{T}{\yd{\varsigma}^{\star}}{\yd{\varrho}}{\yd{\lambda}^{\star}}_{\phi}\) in~\cref{thm:SM_QPA_choi}.
\end{theorem}

In the case where there is no multiplicity, i.e., \(\res{c}{\yd{\lambda}}{\yd{\varrho}}{\yd{\varsigma}}=1\), the isometry is uniquely determined by $\yd{\varsigma}$, $\yd{\varrho}$, and $\yd{\lambda}$.

In the following discussion, we will restrict our consideration to totally symmetric outputs on the irrep $\yd{m}$ (a single row with $m$ boxes), and we denote the corresponding intertwiner by
\begin{equation}
\intertwiner{W}{\yd{\varsigma}}{\yd{\lambda}}{\yd{m}}{}:\; \mathbb{W}^{\yd{\varsigma}}\;\longrightarrow\; \mathbb{W}^{\yd{\lambda}}\otimes \mathbb{W}^{\yd{m}}.
\end{equation}
On each sector corresponding to $\yd{\varsigma}$, the optimal channel amounts to selecting a single output sector $\yd{\varrho}$ and a single feasible $\yd{\lambda}$.

		The environment irreps \(\yd{\lambda}\) are obtained by removing a total of \(m\) boxes across the rows: \(\lambda_i=\varsigma_i-m_i\), where \(m_i\geq0\) are integers satisfying \(\sum_i m_i=m\) and \(m_i\leq\varsigma_i-\varsigma_{i+1}\) for all \(i\leq d\), with the convention \(\varsigma_{d+1}=-\infty\).
For an optimization problem with a linear objective, the optimum is achieved by an extremal channel as characterized earlier in~\cref{subsec:isometric_supp}.
\section{General QPA}
\label{sec:general_qpa_supp}
\subsection{Setup and notation}
\label{subsec:setup_supp}
By the symmetry reduction of Theorem~S1.6 in Ref.~\cite{CI26S}, we restrict to symmetric channels, for which both the worst-case utilities
\begin{subequations}
\begin{align}
^k\mathcal{F}_{\mathrm{w,all}}\mleft(\mathcal{T}\mright)
&\coloneqq \inf_{\rho\in\mathcal X}\,F\mleft(^k\gamma(\rho)^{\otimes m},\ \mathcal T(\rho^{\otimes n})\mright),\\
^k\mathcal{F}_{\mathrm{w,one}}\mleft(\mathcal{T}\mright)
&\coloneqq \inf_{\rho\in\mathcal X}\,F\mleft(^k\gamma(\rho),\ \Tr_{2,\ldots,m}\mathcal T(\rho^{\otimes n})\mright)
\end{align}
\end{subequations}
and the average utilities
\begin{subequations}
\begin{align}
^k\mathcal{F}_{\mathrm{a,all}}\mleft(\mathcal{T}\mright)
&\coloneqq \int_{\mathcal X} F\mleft(^k\gamma(\rho)^{\otimes m},\ \mathcal T(\rho^{\otimes n})\mright)\,\mathrm{d}\mu(\rho),\\
^k\mathcal{F}_{\mathrm{a,one}}\mleft(\mathcal{T}\mright)
&\coloneqq \int_{\mathcal X} F\mleft(^k\gamma(\rho),\Tr_{2,\ldots,m}\mathcal T(\rho^{\otimes n})\mright)\,\mathrm{d}\mu(\rho)
\end{align}
\end{subequations}
with infidelity loss $L\mleft(\cdot,\cdot\mright)=1-F(\cdot,\cdot)$ reduce to the fidelity evaluated on a single input state, here taken to be diagonal with sorted spectrum in the basis $\ket{\psi_i}$.
The minimax utilities, obtained by optimizing the worst-case utility over all valid QPA protocols \(\mathcal T\), and the Bayesian utilities, obtained by optimizing the average utility over \(\mathcal T\), coincide in this setting.
This follows from the single-orbit equivalence of Lemma~S1.8 in Ref.~\cite{CI26S}.
We refer to these common values as the optimal utilities \({}^k\mathcal F_{\mathrm{all}}^\ast\) and \({}^k\mathcal F_{\mathrm{one}}^\ast\).

\subsection{QPA protocols and their utilities}
\label{subsec:protocols_supp}
To determine the optimal channel, i.e., the symmetric channel that achieves the optimal utility, we analyze the all-site utility directly using representation-theoretic methods, then relate it to one-site quantities via recoupling theory. We begin with the all-site case.

Through a sequence of reductions, we express the resulting all-site utility as expectations over two probability distributions related to the spectrum: the SW distribution and the Weyl distribution.
\begin{equation}
    \label{equ:the_fidelity_is}
\begin{aligned}
    ^k\mathcal{F}_\mathrm{all}\mleft(\mathcal{T}\mright)
=&F\mleft(\,\psi_k^{\otimes m},\ \mathcal{T}\mleft(\rho^{\otimes n}\mright)\mright)\\
=&F\mleft(\psi_k^{\otimes m},\sum_{\yd{\varsigma}}g^{\yd{\varsigma}}s^{\yd{\varsigma}}\extChannel{\mathcal{T}}{\yd{\varsigma}}{\yt{1,\ldots,m}}{\yd{\mu}}\mleft(\frac{\rho^{\yd{\varsigma}}}{s^{\yd{\varsigma}}}\mright)\mright)\\
=&\swavg{\sum_{\wt{w},\wt{v}}\mleft|\bra{\wt{v}}_{\yd{\mu}} \bra{\psi_k}_{\mathrm{out}}^{\otimes m}\intertwiner{W}{\yd{\varsigma}}{\yd{\mu}}{\yt{1,\ldots,m}}{}\ket{\wt{w}}_{\yd{\varsigma}}\mright|^2\mleft\langle \wt{w} \mleft| \frac{\rho^{\yd{\varsigma}}}{s^{\yd{\varsigma}}} \mright| \wt{w} \mright\rangle}{\boldsymbol{p}}{n}\\
\end{aligned}
\end{equation}
We use the definition of $\mathcal{F}_\mathrm{all}$ and decomposition of $\mathcal{T}$ in the first two equalities. For the second equality, we normalize the sampling outcome $\rho^{\yd{\varsigma}}$ and factor out $s^{\yd{\varsigma}}$, the Schur polynomial with argument $\boldsymbol{p}$, which corresponds to the post-selected density matrix on the specific Schur sampling outcome. The resulting weight on the YDs $g^{\yd{\varsigma}}s^{\yd{\varsigma}}$ makes up the SW distribution~\cite{OW17S}. The SW distribution governs the probability of observing a particular symmetry sector $\yd{\varsigma}$ when performing Schur sampling on the input state $\rho^{\otimes n}$. The average over the SW distribution defined by the spectrum $\boldsymbol{p}$ is denoted by $\swavg{\cdot}{\boldsymbol{p}}{n}$.

In the third step, we insert two resolutions of the identity, namely, $\wt{w}$ (in $\mathbb{W}^{\yd{\varsigma}}$) and $\wt{v}$ (in $\mathbb{W}^{\yd{\mu}}$), corresponding to the GT bases of the respective registers. For the fourth equality, we express the coupling of $\mathbb{W}^{\yd{\mu}}\otimes\mathbb{W}^{\yd{m}}$ to $\mathbb{W}^{\yd{\varsigma}}$ in terms of the Clebsch--Gordan coefficient (CGC) $\mathrm{CGC}(\wt{v},\,\wt{k\cdots k}\,|\,\wt{w})$. Moreover, in the GT basis, the density matrix is diagonal and its eigenvalues simplify to the weights of the GT basis vectors (i.e., the number of occurrences of each symbol in the GT basis vector).

Finally, we write the expression in~\cref{equ:the_fidelity_is} more compactly as the expectation of the sector-wise utility \begin{equation}
     ^k\mathcal{F}_\mathrm{all}\mleft(\mathcal{T}\mright)=\swavg{\res{{^kf}}{\yd{\mu}}{\yd{m}}{\yd{\varsigma}}_\mathrm{all}}{\boldsymbol{p}}{n},\, \res{{^kf}}{\yd{\mu}}{\yd{m}}{\yd{\varsigma}}_\mathrm{all}=\sum_{\wt{w},\wt{v}}\mleft|\bra{\wt{v}}_{\yd{\mu}} \bra{\psi_k}_{\mathrm{out}}^{\otimes m}\intertwiner{W}{\yd{\varsigma}}{\yd{\mu}}{\yt{1,\ldots,m}}{}\ket{\wt{w}}_{\yd{\varsigma}}\mright|^2\mleft\langle \wt{w} \mleft| \frac{\rho^{\yd{\varsigma}}}{s^{\yd{\varsigma}}} \mright| \wt{w} \mright\rangle
\end{equation}

We turn to the one-site utility:
\begin{equation}
\label{equ:the_one_copy_fidelity_is}
\begin{aligned}
^k\mathcal{F}_{\mathrm{one}}\mleft(\mathcal{T}\mright)
=&\;
F\mleft(\psi_k,\ \Tr_{2,\ldots,m}\,\mathcal{T}\mleft(\rho^{\otimes n}\mright)\mright)\\
=&\;
F\mleft(\psi_k,\ \Tr_{2,\ldots,m}\sum_{\yd{\varsigma}} g^{\yd{\varsigma}}\,s^{\yd{\varsigma}}\,
\extChannel{\mathcal{T}}{\yd{\varsigma}}{\yt{1\cdots m}}{\yd{\mu}}\mleft(\frac{\rho^{\yd{\varsigma}}}{s^{\yd{\varsigma}}}\mright)\mright)\\
=&\;
\swavg{
F\mleft(\psi_k,\ \Tr_{2,\ldots,m}\,
\extChannel{\mathcal{T}}{\yd{\varsigma}}{\yt{1\cdots m}}{\yd{\mu}}\mleft(\frac{\rho^{\yd{\varsigma}}}{s^{\yd{\varsigma}}}\mright)\mright)
}{\boldsymbol{p}}{n}\\
\end{aligned}
\end{equation}
Consider the partially traced channel
\(\Tr_{2,\ldots,m}\mleft(\extChannel{\mathcal{T}}{\yd{\varsigma}}{\yt{1\cdots m}}{\yd{\varrho}}(\cdot)\mright)\).
It is itself a symmetric channel, and by~\cref{subsec:extremal_supp} it admits a decomposition into extremal single-output channels. This decomposition follows from an explicit analysis of the \(SU(d)\) recoupling \(F\)-symbols, which in turn allows us to relate the induced one-site channel to the fidelity of the $m=1$ case.

\begin{lemma}
    \label{lem:one_site_channel_decomposition}
        Let $\yd{\varsigma}\dot{=}[\varsigma_{1},\ldots,\varsigma_{d}]$. Further, let $m\geq 1$ and $\yd{\lambda}$ such that $\res{c}{\yt{1\cdots m}}{\yd{\lambda}}{\yd{\varsigma}} = 1$.
        Then the channel decomposes as
        \begin{equation}\begin{aligned} \Tr_{2,\ldots,m}\mleft(\extChannel{\mathcal{T}}{\yd{\varsigma}}{\yt{1\cdots m}}{\yd{\lambda}}(\cdot)\mright) = \sum_{i=1}^{d}\resSix{F}{\yd{\lambda}}{\yd{m-1}}{\yd{1}}{\yd{\varsigma-\mathbf{e}_i}}{\yt{1\cdots m}}{\yd{\varsigma}}^2 \extChannel{\mathcal{T}}{\yd{\varsigma}}{\yd{1}}{\yd{\varsigma-\mathbf{e}_i}}(\cdot) \, ,
        \end{aligned}\end{equation}
        where the $F$-symbols of Ref.~\cite[Lemma~5]{BFGL+25S} are given by
        \begin{equation}
        \label{equ:definition_F_symbols}
        \begin{aligned}
            \resSix{F}{\yd{\lambda}}{\yd{m-1}}{\yd{1}}{\yd{\varsigma-\mathbf{e}_i}}{\yt{1\cdots m}}{\yd{\varsigma}}
            &= \mleft\{\begin{aligned}
            &\sqrt{\dfrac{1}{m}\dfrac{\prod_{j=1}^{d}(\varsigma_{i} - \lambda_{j} + (j-i))}{\prod_{j\neq i}(\varsigma_{i} - \varsigma_{j} + (j-i))}}, & \yd{\varsigma-\mathbf{e}_i} \text{ is a valid YD} \, , \\[10pt]
            &0, & \text{otherwise}.
            \end{aligned}\mright.
        \end{aligned}
        \end{equation}
\end{lemma}

\begin{remark}
    \label{rem:F_symbols_isometry}
    Since the $F$-symbols are the matrix elements of an isometry with indices $\yd{\varsigma-\mathbf{e}_i}$ and $\yt{1\cdots m}$, we have
    \begin{equation}\begin{aligned}
        \sum_{i=1}^{d} \resSix{F}{\yd{\lambda}}{\yd{m-1}}{\yd{1}}{\yd{\varsigma-\mathbf{e}_i}}{\yt{1\cdots m}}{\yd{\varsigma}}^2 = 1 \, .
    \end{aligned}\end{equation}
\end{remark}

\begin{proof}
    We know from~\cref{thm:extremal_channel_form} that
    \begin{equation}\begin{aligned}
        \label{equ:proof_one_site_channel_decomposition_isometry_form}
        \extChannel{\mathcal{T}}{\yd{\varsigma}}{\yt{1\cdots m}}{\yd{\lambda}}(\cdot)
        \;=\;
        \Tr_{\yd{\lambda}}\mleft(\intertwiner{W}{\yd{\varsigma}}{\yd{\lambda}}{\yt{1\cdots m}}{}\,\cdot\,\intertwineradj{W}{\yd{\varsigma}}{\yd{\lambda}}{\yt{1\cdots m}}{}\mright).
    \end{aligned}\end{equation}
    On the other hand, from Ref.~\cite[Section~2.8]{BFGL+25S}, we know that
    \begin{equation}\begin{aligned}
        \label{equ:proof_one_site_channel_decomposition_F_decomposition}
        \mleft(I_{\yd{\lambda}} \otimes \intertwiner{W}{\yt{1\cdots m}}{\yd{m-1}}{\yd{1}}{}\mright) \intertwiner{W}{\yd{\varsigma}}{\yd{\lambda}}{\yt{1\cdots m}}{} = \sum_{i:\yd{\varsigma-\mathbf{e}_i} \vdash_{d}(n,1)} \resSix{F}{\yd{\lambda}}{\yd{m-1}}{\yd{1}}{\yd{\varsigma-\mathbf{e}_i}}{\yt{1\cdots m}}{\yd{\varsigma}} W_{i} \, ,
    \end{aligned}\end{equation}
    where
    \begin{equation}\begin{aligned}
        W_{i} \coloneqq \mleft(\intertwiner{W}{\yd{\varsigma-\mathbf{e}_i}}{\yd{\lambda}}{\yd{m-1}}{} \otimes I_{d}\mright) \intertwiner{W}{\yd{\varsigma}}{\yd{\varsigma-\mathbf{e}_i}}{\yd{1}}{} \, .
    \end{aligned}\end{equation}
    Finally, we remark that
    \begin{equation}\begin{aligned}
        \Tr_{2,\ldots,m}\mleft(\rho^{\yt{1\cdots m}}\mright) = \Tr_{\yd{m-1}}\mleft(\intertwiner{W}{\yt{1\cdots m}}{\yd{m-1}}{\yd{1}}{}\,\rho^{\yt{1\cdots m}}\,\intertwineradj{W}{\yt{1\cdots m}}{\yd{m-1}}{\yd{1}}{}\mright) \, .
    \end{aligned}\end{equation}
    Using this fact, we can take $\Tr_{2,\ldots,m}$ on both sides of~\cref{equ:proof_one_site_channel_decomposition_isometry_form} and apply~\cref{equ:proof_one_site_channel_decomposition_F_decomposition} to get
    \begin{equation}\begin{aligned}
        \label{equ:proof_one_site_channel_decomposition_trace_out}
        \Tr_{2,\ldots,m}\mleft(\extChannel{\mathcal{T}}{\yd{\varsigma}}{\yt{1\cdots m}}{\yd{\lambda}}(\cdot)\mright) = \sum_{i,j} \resSix{F}{\yd{\lambda}}{\yd{m-1}}{\yd{1}}{\yd{\varsigma-\mathbf{e}_i}}{\yt{1\cdots m}}{\yd{\varsigma}} \overline{\resSix{F}{\yd{\lambda}}{\yd{m-1}}{\yd{1}}{\yd{\varsigma-\mathbf{e}_j}}{\yt{1\cdots m}}{\yd{\varsigma}}} \Tr_{\yd{\lambda},\yd{m-1}}\mleft(W_{i}\cdot{W_{j}}^{\dagger}\mright) \, .
    \end{aligned}\end{equation}
    By Schur's lemma, it follows that
    \begin{equation}\begin{aligned}
        \intertwineradj{W}{\yd{\varsigma-\mathbf{e}_j}}{\yd{\lambda}}{\yd{m-1}}{} \intertwiner{W}{\yd{\varsigma-\mathbf{e}_i}}{\yd{\lambda}}{\yd{m-1}}{} = 0
    \end{aligned}\end{equation}
    for any $i\neq j$. In particular, ${W_{j}}^{\dagger} W_{i} = 0$ for any $i\neq j$. Inserting into~\cref{equ:proof_one_site_channel_decomposition_trace_out}, together with the fact that
    \begin{equation}\begin{aligned}
        \extChannel{\mathcal{T}}{\yd{\varsigma}}{\yd{1}}{\yd{\varsigma-\mathbf{e}_i}}(\cdot)
        \;=\;
        \Tr_{\yd{\varsigma-\mathbf{e}_i}}\mleft(\intertwiner{W}{\yd{\varsigma}}{\yd{\varsigma-\mathbf{e}_i}}{\yd{1}}{}\,\cdot\,\intertwineradj{W}{\yd{\varsigma}}{\yd{\varsigma-\mathbf{e}_i}}{\yd{1}}{}\mright)
    \end{aligned}\end{equation}
    then completes the proof.
\end{proof}

With~\cref{lem:one_site_channel_decomposition}, we find
\begin{equation}
\begin{aligned}
^k\mathcal{F}_{\mathrm{one}}\mleft(\mathcal{T}\mright)=&\swavg{
\sum_i \resSix{F}{\yd{\mu}}{\yd{m-1}}{\yd{1}}{\yd{\varsigma-\mathbf{e}_i}}{\yt{1\cdots m}}{\yd{\varsigma}}^2\;
F\mleft(\psi_k,\
\extChannel{\mathcal{T}}{\yd{\varsigma}}{\yd{1}}{\yd{\varsigma-\mathbf{e}_i}}\mleft(\frac{\rho^{\yd{\varsigma}}}{s^{\yd{\varsigma}}}\mright)\mright)
}{\boldsymbol{p}}{n}\\
=&
\swavg{
\res{{^kf}}{\yd{\mu}}{\yd{m}}{\yd{\varsigma}}_\mathrm{one}
}{\boldsymbol{p}}{n},
\end{aligned}
\end{equation}
where we have similarly defined the sector-wise utility and related it to the all-site sector-wise utility as
\begin{equation}
\label{equ:one_site_decomposition}
\res{{^kf}}{\yd{\mu}}{\yd{m}}{\yd{\varsigma}}_\mathrm{one} =
\sum_i \resSix{F}{\yd{\mu}}{\yd{m-1}}{\yd{1}}{\yd{\varsigma-\mathbf{e}_i}}{\yt{1\cdots m}}{\yd{\varsigma}}^2\;
\res{{^kf}}{\yd{\varsigma-\mathbf{e}_i}}{\yd{1}}{\yd{\varsigma}}_\mathrm{all}.
\end{equation}

\subsection{From overall utility to sector-wise utility via SW concentration}
\label{subsec:concentration_supp}

We denote the normalized Weyl chamber by
\begin{equation}
    \mathcal{C}
    \coloneqq
    \mleft\{
    \boldsymbol{x}\in\mathbb{R}^{d}:
    x_1\geq\cdots\geq x_d\geq0,\quad
    \sum_{i=1}^{d}x_i=1
    \mright\}.
\end{equation}
For a fixed nondegenerate spectrum $\boldsymbol{p}\in\mathcal{C}$, choose $\epsilon>0$ sufficiently small such that $\boldsymbol{p}\in\mathsf{R}_{\epsilon}$ and define the typical region
\begin{equation}
\label{eq:gap_separated_typical_region}
    \mathsf{R}_{\epsilon}
    \coloneqq
    \mleft\{
    \overline{\yd{\varsigma}}\in\mathcal{C}:
    \min_{i<j}\overline{\Delta}_{i,j}\geq\epsilon
    \mright\}.
\end{equation}
\rev{For the case $p_d=0$, we further require $\varsigma_d=0$.}
\begin{lemma}[SW averaging of uniform sector-wise expansions]
\label{lem:sw_averaging_uniform_sector_expansion}
Let $\overline{\yd{\varsigma}}=\yd{\varsigma}/n$ be sampled from the SW distribution $\mathrm{SW}(n,\boldsymbol{p})$. By the SW large-deviation bound~\cite[Theorem~3.1]{OW15S}, the atypical event satisfies
\begin{equation}
    \swprob{\overline{\yd{\varsigma}}\notin\mathsf{R}_{\epsilon}}{\boldsymbol{p}}{n}
    =
    \bigO[d,\boldsymbol{p},\epsilon][n]{\exp(-n)} \, .
\end{equation}
Let $0\leq f_n\leq 1$ be a sector-wise utility so that uniformly on $\mathsf{R}_{\epsilon}$ we have
\begin{equation}
    \label{eq:uniform_sector_expansion}
    f_n(\overline{\yd{\varsigma}})
    =
    1-\frac{c(\overline{\yd{\varsigma}})}{n}
    +
    \lilo[d,\boldsymbol{p},\epsilon]{n^{-1}} \, .
\end{equation}
Here \(c(\overline{\yd{\varsigma}})\) is Lipschitz on \(\mathsf{R}_{\epsilon}\), uniformly over all other parameters such as \(R\), with constant \(L_{d,\boldsymbol p,\epsilon}\). Then
\begin{equation}
    \swavg{f_n(\overline{\yd{\varsigma}})}{\boldsymbol{p}}{n}
    =
    1-\frac{c(\boldsymbol{p})}{n}
    +
    \lilo[d,\boldsymbol{p},\epsilon]{n^{-1}} \, .
\end{equation}
Likewise, in the extensive case, if $f_n(\overline{\yd{\varsigma}})=C(\overline{\yd{\varsigma}})+\lilo[d,\boldsymbol{p},\epsilon]{1}$ uniformly on $\mathsf{R}_{\epsilon}$ with $C$ Lipschitz, then
\begin{equation}
    \swavg{f_n(\overline{\yd{\varsigma}})}{\boldsymbol{p}}{n}=C(\boldsymbol{p})+\lilo[d,\boldsymbol{p},\epsilon]{1} \, .
\end{equation}
\end{lemma}

\begin{proof}
We first split the SW average into typical and atypical sectors.
Since $0\leq f_n\leq 1$, the atypical contribution is bounded by the atypical probability:
\begin{equation}
    \swavg{f_n}{\boldsymbol{p}}{n}
    =
    \swavg{f_n\mathbf{1}_{\mathsf{R}_{\epsilon}}}{\boldsymbol{p}}{n}
    +
    \bigO[d,\boldsymbol{p},\epsilon][n]{\exp(-n)} \, .
\end{equation}
On $\mathsf{R}_{\epsilon}$, we average each term of the expansion separately:
\begin{equation}
\begin{aligned}
    \swavg{f_n\mathbf{1}_{\mathsf{R}_{\epsilon}}}{\boldsymbol{p}}{n}
    &=
    \swprob{\mathsf{R}_{\epsilon}}{\boldsymbol{p}}{n}
    -
    \frac{1}{n}
    \swavg{c(\overline{\yd{\varsigma}})\mathbf{1}_{\mathsf{R}_{\epsilon}}}{\boldsymbol{p}}{n} +
    \swavg{R_n(\overline{\yd{\varsigma}})\mathbf{1}_{\mathsf{R}_{\epsilon}}}{\boldsymbol{p}}{n} \, .
\end{aligned}
\end{equation}
where \(R_n(\overline{\yd{\varsigma}})\) denotes the uniform little-\(o\) remainder in~\cref{eq:uniform_sector_expansion}.
Here $\swprob{\mathsf{R}_{\epsilon}}{\boldsymbol{p}}{n}=1 + \bigO[d,\boldsymbol{p},\epsilon][n]{\exp(-n)}$, and the uniform remainder gives
\begin{equation}
    \swavg{R_n(\overline{\yd{\varsigma}})\mathbf{1}_{\mathsf{R}_{\epsilon}}}{\boldsymbol{p}}{n}
    =
    \lilo[d,\boldsymbol{p},\epsilon]{n^{-1}} \, .
\end{equation}
It remains to control the coefficient term.
We write
\begin{equation}
\begin{aligned}
    \swavg{c(\overline{\yd{\varsigma}})\mathbf{1}_{\mathsf{R}_{\epsilon}}}{\boldsymbol{p}}{n}
    &=
    c(\boldsymbol{p})\swprob{\mathsf{R}_{\epsilon}}{\boldsymbol{p}}{n} +
    \swavg{\mleft(c(\overline{\yd{\varsigma}})-c(\boldsymbol{p})\mright)\mathbf{1}_{\mathsf{R}_{\epsilon}}}{\boldsymbol{p}}{n} .
\end{aligned}
\end{equation}
The SW fluctuation bound of Ref.~\cite[Theorem~1.1]{OW16S} gives \(\swavg{\|\overline{\yd{\varsigma}}-\boldsymbol{p}\|_2^2}{\boldsymbol{p}}{n}\leq d/n\). By Cauchy--Schwarz, \(\swavg{\|\overline{\yd{\varsigma}}-\boldsymbol{p}\|_2}{\boldsymbol{p}}{n}=\bigO[d][n]{n^{-1/2}}\). Since \(c(\overline{\yd{\varsigma}})\) is Lipschitz on \(\mathsf{R}_{\epsilon}\), uniformly over all other parameters, the difference term is bounded by the unrestricted fluctuation:
\begin{equation}
\begin{aligned}
    \mleft| \swavg{\mleft(c(\overline{\yd{\varsigma}})-c(\boldsymbol{p})\mright)\mathbf{1}_{\mathsf{R}_{\epsilon}}}{\boldsymbol{p}}{n}
    \mright|
    &\leq
    L_{d,\boldsymbol p,\epsilon}\swavg{\|\overline{\yd{\varsigma}}-\boldsymbol{p}\|_2\mathbf{1}_{\mathsf{R}_{\epsilon}}}{\boldsymbol{p}}{n} \\
    &\leq
    L_{d,\boldsymbol p,\epsilon}
    \swavg{\|\overline{\yd{\varsigma}}-\boldsymbol{p}\|_2}{\boldsymbol{p}}{n}
    =
    \bigO[d,\boldsymbol{p},\epsilon][n]{n^{-1/2}} \, .
\end{aligned}
\end{equation}
Therefore
\begin{equation}
    \swavg{c(\overline{\yd{\varsigma}})\mathbf{1}_{\mathsf{R}_{\epsilon}}}{\boldsymbol{p}}{n}
    =
    c(\boldsymbol{p})
    +
    \bigO[d,\boldsymbol{p},\epsilon][n]{n^{-1/2}} \, .
\end{equation}
Multiplying this error by \(1/n\) gives \(\bigO[d,\boldsymbol{p},\epsilon][n]{n^{-3/2}}\), which is absorbed into \(\lilo[d,\boldsymbol{p},\epsilon]{n^{-1}}\).
The extensive case follows by the same argument without the prefactor $1/n$.
\end{proof}

For the asymptotic analyses, we apply~\cref{lem:sw_averaging_uniform_sector_expansion} to the intensive-output expansions to obtain the all-site utility~\cref{equ:intensive_asymptote} and the one-site utility~\cref{equ:one_site_asymptote}.
For extensive outputs, applying the same lemma to the expansion $\res{f}{\yd{\varsigma}}{\yd{m}}{\yd{\mu}}_{(\cdot)}=C(\overline{\yd{\varsigma}})+\lilo[d,\boldsymbol{p},\epsilon]{1}$ gives~\cref{equ:extensive_asymptote}.

For the nonasymptotic analysis, in order to remove the $d$-dependence, we first establish a concentration bound for the row differences $\Delta_{k,k+1} = \varsigma_k - \varsigma_{k+1}$.
\begin{lemma}[Azuma's inequality for row differences]
\label{lem:azuma_row_difference}
Let each component of $w\dot{=}(w_1,\ldots,w_n)$ be drawn i.i.d.\ from a distribution with probabilities $\boldsymbol{p}$, and let $\varsigma_i=\varsigma_i(w)$ denote the $i$-th row-length produced by the sampling procedure under consideration (e.g., $\yd{\varsigma}=\mathrm{shRSK}(w)$).
Let $\Delta_{k,k+1} = \varsigma_{k}-\varsigma_{k+1}$ denote the row difference. Then for any $t > 0$,
\begin{equation}
\Pr(|\langle\overline\Delta_{k,k+1}\rangle - \overline\Delta_{k,k+1}| \ge t)
\le
2e^{-\,\frac{\,n t^2}{32}}.
\end{equation}
\end{lemma}

\begin{proof}
Following the proof of Ref.~\cite[Proposition~4.8]{OW17S}, consider the Doob martingale
\begin{equation}
X^{(i)} = \mathbb{E}\mleft(\Delta_{k,k+1} \mid w_1,\ldots,w_i\mright),
\quad i=0,1,\ldots,n,
\end{equation}
so that $X^{(0)} = \langle\Delta_{k,k+1}\rangle$ and $X^{(n)} = \Delta_{k,k+1}$. By the bounded difference property for single-symbol changes (each $\varsigma_\ell$ changes by at most $2$ when only one $w_i$ is modified), each martingale increment satisfies
\begin{equation}
|X^{(i)}-X^{(i-1)}|
\le|\varsigma_{k+1}-\varsigma_{k+1}'|+|\varsigma_k-\varsigma_k'|
\le 4,
\end{equation}
where the primed variables refer to changing only $w_i$.
Hence Azuma's inequality applies with $c_i=4$ for all $i$, yielding for any $\epsilon>0$,
\begin{equation}
\Pr(|\langle\Delta_{k,k+1}\rangle - \Delta_{k,k+1}| \ge T)
\le
2e^{-\,\frac{T^2}{2\sum_{i=1}^n c_i^2}}
=
2e^{-\,\frac{\,T^2}{32n}}.
\end{equation}
Normalizing by $n$ gives the stated result.
\end{proof}

\begin{lemma}[Upper bound on probability of row-difference deviations]
    \label{lem:upper_bound_probability_row_differences}
    For \(\alpha\geq 0\) and \(1\leq k \leq d-1\), we have
    \begin{equation}\begin{aligned}
        \Pr\mleft(\mleft|\overline{\Delta}_{k,k+1} - D_{k,k+1}\mright| \ge \alpha \mright) \leq 2e^{-\,\frac{\,(\sqrt{n}\alpha - 4)^2}{32}} \, .
    \end{aligned}\end{equation}
    In particular, for \(2\leq k \leq d-1\),
    \begin{equation}\begin{aligned}
        \Pr\mleft(\,|\overline{\Delta}_{k-1,k}-D_{k-1,k}|\ge \alpha\ \lor\ |\overline{\Delta}_{k,k+1}-D_{k,k+1}|\ge \alpha\,\mright)
        \ \le\ 4e^{-\,\frac{(\sqrt{n}\alpha-4)^2}{32}} \, .
    \end{aligned}\end{equation}
\end{lemma}

\begin{proof}

    Applying Ref.~\cite[Theorem~1.4]{OW17S} to rows \(k\) and \(k+1\), we have
    \begin{equation}\begin{aligned}
        \mleft|\langle\overline{\Delta}_{k,k+1}\rangle - D_{k,k+1}\mright| \leq \frac{4}{\sqrt{n}} \, .
    \end{aligned}\end{equation}
    Combining with~\cref{lem:azuma_row_difference}, we get
    \begin{equation}\begin{aligned}
        \Pr\mleft(\mleft|\overline{\Delta}_{k,k+1} - D_{k,k+1}\mright| \ge t + \frac{4}{\sqrt{n}}\mright) \leq 2e^{-\,\frac{\,nt^2}{32}} \, .
    \end{aligned}\end{equation}
    Assuming \(\alpha>4/\sqrt{n}\), we set \(t=\alpha-4/\sqrt{n}>0\) to obtain the desired inequality. We further see that \(2\exp\left(-\,\frac{\,(\sqrt{n}\alpha - 4)^2}{32}\right) \geq 1\) for \(0\leq \alpha \leq \frac{4}{\sqrt{n}}\), so we can extend the validity of the inequality to the regime \(0\leq \alpha\).
\end{proof}

\begin{lemma}[Upper bound on probability of two row-difference deviations]
    \label{lem:upper_bound_probability_two_row_differences}
    For \(1<k<d\) and \(\alpha\geq 0\), we have
    \begin{equation}\begin{aligned}
        \Pr\mleft(\,|\overline{\Delta}_{k-1,k}-D_{k-1,k}|\ge \alpha\ \lor\ |\overline{\Delta}_{k,k+1}-D_{k,k+1}|\ge \alpha\,\mright)
        \ \le\ 4e^{-\,\frac{(\sqrt{n}\alpha-4)^2}{32}} \, .
    \end{aligned}\end{equation}
\end{lemma}

\begin{proof}
    By the union bound,
    \begin{equation}\begin{aligned}
        &\Pr\mleft(\,|\overline{\Delta}_{k-1,k}-D_{k-1,k}|\ge \alpha\ \lor \ |\overline{\Delta}_{k,k+1}-D_{k,k+1}|\ge \alpha\,\mright)\\
        \le&
        \Pr\mleft(|\overline{\Delta}_{k-1,k}-D_{k-1,k}|\ge \alpha\mright)
        +
        \Pr\mleft(|\overline{\Delta}_{k,k+1}-D_{k,k+1}|\ge \alpha\mright).
    \end{aligned}\end{equation}
    Applying~\cref{lem:upper_bound_probability_row_differences} to each term yields
    \begin{equation}\begin{aligned}
        \Pr\mleft(|\overline{\Delta}_{k-1,k}-D_{k-1,k}|\ge \alpha\mright),
        \Pr\mleft(|\overline{\Delta}_{k,k+1}-D_{k,k+1}|\ge \alpha\mright)
        \le 2e^{-\,\frac{(\sqrt{n}\alpha-4)^2}{32}},
    \end{aligned}\end{equation}
    and summing the two bounds gives the claim.
\end{proof}

As in~\cref{lem:upper_bound_probability_row_differences}, we can derive a bound for the probability of obtaining a given row length.

\begin{lemma}[Upper bound on probability of row-length deviations]
    \label{lem:upper_bound_probability_row_lengths}
    For $\alpha\geq 0$ and \(1\leq k \leq d\), we have
    \begin{equation}\begin{aligned}
        \Pr\mleft(\mleft|\varsigma_{k} - np_{k}\mright| \ge n\alpha \mright) \leq 2e^{-\,\frac{\,(\sqrt{n}\alpha - 2)^2}{8}} \, .
    \end{aligned}\end{equation}
\end{lemma}

\begin{proof}
    We proceed as in~\cref{lem:azuma_row_difference,lem:upper_bound_probability_row_differences}, with the difference that we consider the same Doob martingale as in the proof of Ref.~\cite[Proposition~4.8]{OW17S}, which means we have
    \begin{equation}
        |X^{(i)}-X^{(i-1)}|
        \le|\varsigma_{k}-\varsigma_{k}'|
        \le 2 \, .
    \end{equation}
    In addition, Ref.~\cite[Theorem~1.4]{OW17S} gives
    \begin{equation}
        \mleft|\langle\varsigma_{k}\rangle - np_{k}\mright| \leq 2\sqrt{n} \, .
    \end{equation}
    All other steps are the same, and lead to the stated result.
\end{proof}

\subsection{Expression of sector-wise utilities}
\label{subsec:expression_supp}

To evaluate the sector-wise fidelities, it is convenient to reindex the target eigenstate as the $d$-th state using the permutation $\sigma$:
\begin{equation}
\label{equ:permutation}
\sigma(i) \coloneqq
\mleft\{\begin{aligned}
    &i, & i < k, \\
    &d, & i = k, \\
    &i-1, & i > k.
\end{aligned}\mright.
\end{equation}
We then set $\boldsymbol{q}\dot{=}(q_{1},\ldots,q_{d})$ by $q_{\sigma(i)}\coloneqq p_{i}$, and the corresponding eigenstates $\ket{\sigma(i)}=\ket{\psi_{i}}$. This reindexing allows us to express the utility components in a more compact form, as shown in~\cref{thm:utility_component}.

On each symmetry sector $\yd{\varsigma}$, the sector-wise all-site utility evaluates to the expectation over the Weyl distribution of the utility component, a function of the GT basis $\wt{w}$:
 \begin{equation}
\begin{aligned}
^k\res{f}{\yd{\mu}}{\yd{m}}{\yd{\varsigma}}_\mathrm{all}=&\frac{1}{s^{\yd{\varsigma}}}\sum_{\wt{v}\vdash\yd{\mu},\wt{w}\vdash\yd{\varsigma}}\mleft|\mathrm{CGC}\mleft(
\wt{v}\,\wt{d\cdots d}|\wt{w}\mright)\mright|^2\prod_{i=1}^{\rev{d}} {q_i}^{\#_i\mleft(\wt{w}\mright)}\\
=&\weylavg{\res{f}{\yd{\mu}}{\yd{m}}{\yd{\varsigma}}_\mathrm{all}\mleft(\wt{w}\mright)}{\boldsymbol{q}}{\yd{\varsigma}}\,,\quad
\end{aligned}
\end{equation}
Here, the Weyl distribution $W(\boldsymbol{q},\yd{\varsigma})$ is the distribution over Weyl tableaux $\wt{w}$ weighted by $\prod_{i=1}^{\rev{d}} {q_i}^{\#_i\mleft(\wt{w}\mright)}$. $\#_i\mleft(\wt{w}\mright)$ counts the symbol $i$ in $\wt{w}$. More explicitly, the average over the Weyl distribution is given by
\begin{equation}
\langle \cdot\rangle_{\wt{w}\sim\mathrm{W}(\boldsymbol{q},\yd{\varsigma})} = \frac{1}{s^{\yd{\varsigma}}}\sum_{\wt{w}\vdash\yd{\varsigma}}(\cdot) \prod_{i=1}^{\rev{d}} {q_i}^{\#_i\mleft(\wt{w}\mright)}.
\end{equation}

In the second line, we also define the sector-wise utility
\begin{equation}
    \label{equ:sector_wise_fidelity}
\res{f}{\yd{\mu}}{\yd{m}}{\yd{\varsigma}}_\mathrm{all}\mleft(\wt{w}\mright) =\sum_{\wt{v}\vdash\yd{\mu}}\mleft|\mathrm{CGC}\mleft(
\wt{v}\,\wt{d\cdots d}|\wt{w}\mright)\mright|^2
\end{equation}
Note that the utility components are independent of the choice of $k$, since the dependence on $k$ has been absorbed into the reindexed spectrum $\boldsymbol{q}$.

\begin{theorem}
The utility component can be written compactly as:
\label{thm:utility_component}
\begin{equation}
\res{f}{\yd{\mu}}{\yd{m}}{\yd{\varsigma}}_\mathrm{all}\mleft(\wt{w}\mright)=\binom{m}{\mathbf m}\;
\frac{
\displaystyle
\prod_{\substack{1\le i\le d\\ 1\le j\le d-1}}
\rpoch{\widetilde w_{j,d-1}-\widetilde w_{i,d}}{m_i}
}{
\displaystyle
\prod_{\substack{1\le i,j\le d\\ i\neq j}}
\rpoch{\widetilde w_{j,d}-\widetilde w_{i,d}+1}{m_i}
},
\end{equation}
where $\yd{\mu}$ is given by the removal rule
$\yd{\mu}=\yd{\varsigma-\mathbf{m}}\dot{=}\yd{\varsigma_1-m_1,\ldots,\varsigma_d-m_d}$.
\end{theorem}
\begin{proof}
We first recall that the CGCs are given by:
{\scriptsize\normalsizeequationtag
\begin{equation}
\label{equ:CGC_pouts}
\begin{aligned}
&\mathrm{CGC}\mleft(\wt{v}\,\wt{d,\ldots,d}|\wt{v'}\mright)\\
=&\mathrm{CGC}\mleft(
\begin{array}{cc|c}
\{v_{i,d}\}_{i=1}^d & (m,0) & \{v_{i,d}^\prime\}_{i=1}^d \\
\{v_{i,d-1}\}_{i=1}^{d-1} & (0,0) & \{v_{i,d-1}\}_{i=1}^{d-1}
\end{array}
\mright)
\\=&
\mleft(
\dfrac{m\,!
\displaystyle
\prod_{1\le s\le k\le d-1}\!\bigl(\widetilde v_{s,d}^{\prime}-\widetilde v_{k,d-1}\bigr)!\;
\prod_{1\le s<k\le d}\!\bigl(\widetilde v_{s,d-1}-\widetilde v_{k,d}-1\bigr)!\;
\prod_{1\le s\le k\le d}\!\bigl(\widetilde v_{s,d}^{\prime}-\widetilde v_{k,d}^{\prime}\bigr)!\;
\prod_{1\le s<k\le d}\!\bigl(\widetilde v_{s,d}-\widetilde v_{k,d}^{\prime}-1\bigr)!
}{
\displaystyle
\prod_{1\le s<k\le d}\!\bigl(\widetilde v_{s,d-1}-\widetilde v_{k,d}^{\prime}-1\bigr)!\;
\prod_{1\le s\le k\le d-1}\!\bigl(\widetilde v_{s,d}-\widetilde v_{k,d-1}\bigr)!\;
\prod_{1\le s<k\le d}\!\bigl(\widetilde v_{s,d}^{\prime}-\widetilde v_{k,d}^{\prime}-1\bigr)!\;
\prod_{1\le s\le k\le d}\!\bigl(\widetilde v_{s,d}^{\prime}-\widetilde v_{k,d}\bigr)!
}
\mright)^{\frac{1}{2}},
\end{aligned}
\end{equation}
}
where
\(
m=\sum_{i=1}^d v_{i,d}^{\prime}-\sum_{i=1}^d v_{i,d},\) \(
\widetilde v_{i,b}=v_{i,b}-i,\) and \(
\widetilde v_{i,b}^{\prime}=v_{i,b}^{\prime}-i.\)
{\scriptsize\normalsizeequationtag
\begin{equation}
\begin{aligned}
&\mathrm{CGC}\mleft(
\begin{array}{cc|c}
\{w_{i,d}\}_{i=1}^d & (0,-m) & \{w_{i,d}^\prime\}_{i=1}^d \\
\{w_{i,d-1}\}_{i=1}^{d-1} & (0,0) & \{w_{i,d-1}\}_{i=1}^{d-1}
\end{array}
\mright)
\\=&
\mleft(
\dfrac{m\,!\displaystyle
\prod_{1\le s<k\le d}\!\bigl(\widetilde w_{s,d-1}-\widetilde w_{k,d}^{\prime}-1\bigr)!\;
\prod_{1\le s\le k\le d-1}\!\bigl(\widetilde w_{s,d}-\widetilde w_{k,d-1}\bigr)!\;
\prod_{1\le s\le k\le d}\!\bigl(\widetilde w_{s,d}^{\prime}-\widetilde w_{k,d}^{\prime}\bigr)!\;
\prod_{1\le s<k\le d}\!\bigl(\widetilde w_{s,d}^{\prime}-\widetilde w_{k,d}-1\bigr)!
}
{\displaystyle
\prod_{1\le s\le k\le d-1}\!\bigl(\widetilde w_{s,d}^{\prime}-\widetilde w_{k,d-1}\bigr)!\;
\prod_{1\le s<k\le d}\!\bigl(\widetilde w_{s,d-1}-\widetilde w_{k,d}-1\bigr)!\;
\prod_{1\le s < k\le d}\!\bigl(\widetilde w_{s,d}^{\prime}-\widetilde w_{k,d}^{\prime}-1\bigr)!\;
\prod_{1\le s\le k\le d}\!\bigl(\widetilde w_{s,d}-\widetilde w_{k,d}^{\prime}\bigr)!
}
\mright)^{\frac{1}{2}},
\end{aligned}
\end{equation}
}
where
\(
m=\sum_{k=1}^d w_{k,d}-\sum_{k=1}^d w_{k,d}^{\prime}.\)
We note that the formula in Ref.~\cite{KV95S} contains several inconsistencies, in particular in the last factor of Eq.~(18.2.6)(9) and in Eq.~(18.2.8)(6).
For clarity, we therefore present below the corrected (dual) $G_1$ factor governing the CG rules, derived directly from Ref.~\cite{KV95S} by starting from the right-hand side of the first equality in Eq.~(18.2.6)(9) together with Eq.~(18.2.8)(5).

To apply these formulae to calculate the fidelity, we will focus on the first case starting from here. Consider the removal rule $\yd{\mu}=\{v_{i,d}^\prime\}_{i=1}^d=\{v_{i,d}\}_{i=1}^d-\mathbf{m}=\yd{\varsigma-\mathbf{m}}$, where $\mathbf{m}\dot{=}(m_1,\ldots,m_d)$.
\cref{equ:CGC_pouts} squared reduces to
\begin{equation}
  \label{equ:parametrized_CGC}
\begin{aligned}
\mleft|\mathrm{CGC}\mleft(\wt{v},\wt{d\cdots d}|\wt{v'}\mright)\mright|^2=&
\frac{m!
\displaystyle
\prod_{1\le i\le j\le d-1}
\rpoch{\widetilde v_{i,d}-\widetilde v_{j,d-1}+1}{m_i}
\;
\prod_{1\le i<j\le d}
\fpoch{\widetilde v_{i,d-1}-\widetilde v_{j,d}-1}{m_j}
}{
\displaystyle
\prod_{1\le i\le j\le d}
\fpoch{\widetilde v'_{i,d}-\widetilde v_{j,d}}{m_j}
\;
\prod_{1\le i<j\le d}
\rpoch{\widetilde v_{i,d}-\widetilde v'_{j,d}}{m_i}
}\\
=&
\binom{m}{\mathbf m}\;
\frac{
\displaystyle
\prod_{1\le i\le j\le d-1}
\fpoch{\widetilde v_{j,d-1}-\widetilde v_{i,d}-1}{m_i}
\;\;
\prod_{1\le j< i\le d}
\fpoch{\widetilde v_{j d-1}-\widetilde v_{i,d}-1}{m_i}
}{
\displaystyle
\prod_{1\le i< j\le d}
\fpoch{\widetilde v'_{j,d}-\widetilde v_{i,d}}{m_i}
\;\;
\prod_{1\le j< i\le d}
\fpoch{\widetilde v'_{j,d}-\widetilde v_{i,d}}{m_i}
}
\\[8pt]
=&
\binom{m}{\mathbf m}\;
\frac{
\displaystyle
\prod_{\substack{1\le i\le d\\ 1\le j\le d-1}}
\fpoch{\widetilde v_{j,d-1}-\widetilde v_{i,d}-1}{m_i}
}{
\displaystyle
\prod_{\substack{1\le i,j\le d\\ i\neq j}}
\fpoch{\widetilde v'_{j,d}-\widetilde v_{i,d}}{m_i}
},
\\[10pt]
\end{aligned}
\end{equation}
where $\binom{m}{\mathbf m}$ denotes the multinomial coefficient.
With the totally symmetric representation $\yd{m}$ on the output, for each $\wt{w}$, the CGCs are either all zero, or there exists a single GT basis vector $\wt{v}$ that contributes. Therefore, we can reparameterize the CGC in terms of $\wt{w}$ by a change of variables. First consider those $\wt{w}$ that have a corresponding $\wt{v}$. In this case, only one $\wt{v}$ contributes to the sum, yielding:
\begin{equation}
\label{equ:general_CGC_parametrized}
\begin{aligned}
\mleft|\mathrm{CGC}\mleft(\wt{v},\wt{d\cdots d}|\wt{w}\mright)\mright|^2=&\binom{m}{\mathbf m}\;
\frac{
\displaystyle
\prod_{\substack{1\le i\le d\\ 1\le j\le d-1}}
\rpoch{\widetilde w_{j,d-1}-\widetilde w_{i,d}}{m_i}
}{
\displaystyle
\prod_{\substack{1\le i,j\le d\\ i\neq j}}
\rpoch{\widetilde w_{j,d}-\widetilde w_{i,d}+1}{m_i}
}.
\end{aligned}
\end{equation}
If a given $\wt{w}$ does not correspond to any valid $\wt{v}$, then the removal is excessive in the sense that $w_{\ell,d-1}-w_{\ell,d}<m_\ell$ for some $\ell$.
In this case, the right-hand side of~\cref{equ:general_CGC_parametrized} evaluates to $0$. This agrees with the utility component because no $\wt{v}$ contributes to the sum. This completes the proof.
\end{proof}

\begin{corollary}
\label{cor:utility_component_single_row}
    When the removal is restricted to a single row, i.e., $\mathbf{m}=\sum_\ell m_\ell \boldsymbol{e}_\ell$ with $\boldsymbol{e}_\ell$ being the unit vector, the utility component simplifies to:
    \begin{equation}
        \res{f}{\yd{\mu}}{\yd{m}}{\yd{\varsigma}}_\mathrm{all}\mleft(\wt{w}\mright)=
\frac{
\displaystyle
\prod_{1\le j\le d-1}
\rpoch{\widetilde w_{j,d-1}-\widetilde w_{\ell,d}}{m_\ell}
}{
\displaystyle
\prod_{\substack{1\le j\le d\\ j\neq \ell}}
\rpoch{\widetilde w_{j,d}-\widetilde w_{\ell,d}+1}{m_\ell}
},
    \end{equation}
\end{corollary}

\section{Implementation}
\label{sec:implementation}

For the implementation of QPA, we give two algorithms with different gate-complexity scalings in the regimes \(d\ll n\) and \(n\ll d\): one based on generalized quantum phase estimation (GQPE), following Ref.~\cite{LFIC24S}, and one based on CG transforms, following Ref.~\cite{MT25S}. Both use the recent insight of Ref.~\cite{BFGL+25S} that we can perform the embedding $\intertwiner{W}{\yd{\varsigma}}{\yd{\mu}}{\yd{m}}{}$
efficiently on a quantum computer via a coherent calculation of the \(F\)-symbols, which translates to similar embeddings for the Specht modules via SW duality. In particular, we have the following result.

\begin{lemma}[Efficient initialization of resource state]
    \label{lem:implementation_resource_state}
    Let \(\yd{\varsigma}\vdash_{d}n\) be such that \(r\) is the number of rows in \(\yd{\varsigma}\). Further, let \(m\) and \(\yd{\mu}\vdash_{d}(n,m)\) be such that \(\res{c}{\yd{m}}{\yd{\mu}}{\yd{\varsigma}} = 1\), and let \(\ket{p_{\yd{\varsigma}\to \,\yd{\mu}}^{0,m,d}}\) be the resource state for the embedding \(\intertwiner{W}{\yd{\varsigma}}{\yd{\mu}}{\yd{m}}{}\), as described in~\cite[Proposition 11]{MT25S}, that is, it consists of a superposition of paths in the Bratteli diagram from \(\yd{\varsigma} \rightarrow \yd{\mu}\) that correspond to the irrep \(\yd{m}\) of the subalgebra of the walled Brauer algebra generated by \(S_{m}\) acting on the \(m\) dual registers; see~\cite{MT25S}. Then we can prepare the state in a streaming manner in the reverse encoding of Ref.~\cite{grinko2023gelfand} up to trace-norm error \(\delta\) by applying \(T\) one-qubit or two-qubit gates from a universal gate set and using \(M\) qubits of memory, where
    \begin{equation}
    \begin{aligned}
        T = \bigO{r^2m\ln^p(mr/\delta)} \, , \\
        M = \bigO{r^2\ln^p((n+m)r/\delta)} \, .
    \end{aligned}
    \end{equation}
    Here, \(p\approx 1.44\), and the states are stored as a sequence of mixed YDs \(\yd{\nu^{(0)}},\ldots,\yd{\nu^{(m)}}\) with \(\yd{\nu^{(0)}} = \yd{\mu}\) and \(\yd{\nu^{(m)}} = \yd{\varsigma}\). Streaming means that we output the registers \(\yd{\nu^{(k)}}\) one after another in reverse order, from \(k=m\) to \(k=0\).

    If \(\yd{\mu}\vdash_{d} n-m\), that is, if \(\yd{\mu}\) is actually a YD, then this means that \(\yd{\nu^{(k)}}\vdash_{d}n-m+k\), and the encoding reduces to the regular order of the encoding for the last \(m\) YDs in Ref.~\cite{BFGL+25S}. For any Young--Yamanouchi (YY) basis vector \(\ket{T_{\mathrm{YY}}^{\yd{\mu}}}\in \mathbb{V}^{\yd{\mu}}\), there exists a \(\ket{T^{\yd{\varsigma}}}\in \mathbb{V}^{\yd{\varsigma}}\) so that in the aforementioned encoding
    \begin{equation}
        \ket{T_{\mathrm{YY}}^{\yd{\mu}}} \otimes \ket{p_{\yd{\varsigma}\to \,\yd{\mu}}^{0,m,d}} = \ket{T^{\yd{\varsigma}}} \, .
    \end{equation}
    In addition, under the action of the subgroup \(S_{n-m}\times S_{m} \subseteq S_{n}\), we have
    \begin{equation}
        \ket{T^{\yd{\varsigma}}}_{\mathbb{V}^{\yd{\varsigma}}} \simeq \ket{T_{\mathrm{YY}}^{\yd{\mu}}}_{\mathbb{V}^{\yd{\mu}}} \otimes \ket{\yt{m}}_{\mathbb{V}^{\yd{m}}} \, ,
    \end{equation}
    where \(\yt{m}\) is the only possible standard YT of \(\yd{m}\). The above results imply that we can prepare the state \(\ket{T^{\yd{\varsigma}}}_{\mathbb{V}^{\yd{\varsigma}}}\) in the standard encoding of Ref.~\cite{BFGL+25S} up to trace-norm error \(\delta\) by applying \(T\) one-qubit or two-qubit gates from a universal gate set, where
    \begin{equation}
        T = \bigO{r(n-m)\ln(n-m)+r^2m\ln^p(mr/\delta)} \, .
    \end{equation}
    Here, \(p\approx 1.44\).
\end{lemma}

\begin{proof}
    Finding \(\ket{p_{\yd{\varsigma}\to \,\yd{\mu}}^{0,m,d}}\) reduces to finding the invariant (totally symmetric) subspace under the permutation action of \(S_{m}\) on all the Bratteli paths from \(\yd{\varsigma}\) to \(\yd{\mu}\) when removing \(m\) boxes. The permutation action on this space is the same as the permutation action on the usual standard YTs, modulo the reverse ordering of the mixed YDs from larger (\(\yd{\varsigma}\)) to smaller (\(\yd{\mu}\)), as was discussed in Ref.~\cite{grinko2023gelfand}. This can be done via ``splitting off'' the \(m\) boxes into the totally symmetric subspace, which we can do with the \(F\)-symbols, cf. the arguments in Ref.~\cite{BFGL+25S}. We initialize the YD \(\yd{\varsigma}\) in time \(\bigO{r\ln(n)}\). As in Ref.~\cite{BFGL+25S}, for \(1\leq k \leq m-1\) we can now recursively and coherently calculate the \(F\)-symbols
    \begin{equation}
        \resSix{F}{\yd{\mu}}{\yd{m-k-1}}{\yd{1}}{\yd{\nu^{(m-k+1)}-\mathbf{e}_i}}{\yd{m-k}}{\yd{\nu^{(m-k+1)}}}
    \end{equation}
    for \(1\leq i \leq d\) up to error \(\bigO{\delta/(mr^2)}\) in time \(\bigO{r^2\ln(mr/\delta)}\), and for \(k\in\{0,m\}\) we simply set \(\yd{\nu^{(0)}}=\yd{\mu}\) and \(\yd{\nu^{(m)}}=\yd{\varsigma}\). The idea here is that the expression in~\cref{lem:one_site_channel_decomposition} has at most \(r+1\) nontrivial terms in the product, and there are at most \(r+1\) possible positions for \(i\), since we can remove a box only from rows \(1,\ldots,r\) or from the last row. The last point is due to the fact that we remove \(\yd{m}\), so we cannot remove boxes that lie above each other. For a more detailed treatment of this idea, compare Ref.~\cite[Appendix B.3]{CMT24S}. After calculating the relevant \(F\)-symbols, we use Givens rotations conditioned on these values to initialize the register for \(\yd{\nu^{(m-k)}}\) as
    \begin{equation}
        \ket{\yd{\nu^{(m-k)}}} = \sum_{i=1}^{\rev{d}} \resSix{F}{\yd{\mu}}{\yd{m-k-1}}{\yd{1}}{\yd{\nu^{(m-k+1)}-\mathbf{e}_i}}{\yd{m-k}}{\yd{\nu^{(m-k+1)}}} \ket{\yd{\nu^{(m-k+1)}-\boldsymbol{e}_{i}}} \, .
    \end{equation}
    \rev{Here, the \(F\)-symbols are nonzero only if \(\yd{\nu^{(m-k+1)}-\boldsymbol{e}_{i}}\) is a valid YD, cf.~\cref{equ:definition_F_symbols}, which happens at most at \(r+1\) positions, and otherwise we take \(\ket{\yd{\nu^{(m-k+1)}-\boldsymbol{e}_{i}}}\) to be the zero vector.} We can implement the Givens rotations via Solovay--Kitaev (cf. Ref.~\cite{CMT2024} for the analysis of the implementation and Ref.~\cite{K2025} for the exponent \(p\)) up to error \(\bigO{\delta/(mr^2)}\), which means the \(k\)-th step takes \(\bigO{r^2\ln^p(mr/\delta)}\) gates, and we need \(\bigO{r\ln^p(mr/\delta)}\) qubits to calculate the \(F\)-symbols and \(r\ln(n)\) qubits to store \(\yd{\nu^{(m-k)}}\). By additivity of error, the total error in trace-norm is then \(\delta\).

    In the case where \(\yd{\mu}\vdash_{d} n-m\) and we also want to initialize the first \(n-m\) registers corresponding to \(\ket{T_{\mathrm{YY}}^{\yd{\mu}}}\), in each step we initialize up to \(r\) integers up to size \(n-m\), which means each step takes \(\bigO{r\ln(n-m)}\) gates, and needs \(r\ln(n-m)\) qubits to store them. We then simply add this gate complexity to the previous complexity of initializing \(\ket{p_{\yd{\varsigma}\to \,\yd{\mu}}^{0,m,d}}\).
\end{proof}

\subsection{GQPE-OQPA correctness and complexity}
\label{subsec:implementation_via_GPE}

This implementation follows the GQPE approach of Ref.~\cite{LFIC24S}.

\begin{theorem}[GQPE-OQPA correctness and complexity]
    \label{thm:implementation_GPE}
    Take the input \(\rho^{\otimes n}\), where \(\rho\) has \(k\)-th eigenvalue \(p_{k}\), and let \(np_{k}\geq m\). Then, conditioned on success, the GQPE-OQPA algorithm performs \(k\)-th eigenstate QPA; the success probability is \(1-\eta\), where
    \begin{equation}
        \eta \leq 2\exp\left(-\,\frac{\,(\sqrt{n}\left(p_{k} - \frac{m}{n}\right) - 2)^2}{8}\right) \, .
    \end{equation}
    For the deterministic protocol where we output the maximally mixed state conditioned on failure, the all-site and one-site fidelities decrease by at most \(\eta\), and we can implement the algorithm up to diamond-norm error \(\delta\) by applying \(T\) one-qubit or two-qubit gates from a universal gate set, where
    \begin{equation}
        T = \bigO{n^4\ln^p(dn/\delta)} \, .
    \end{equation}    
    Here, \(p\approx 1.44\).
\end{theorem}

\begin{proof}
    By~\cref{lem:upper_bound_probability_row_lengths}, the probability of measuring \(\yd{\varsigma}\) so that \(\varsigma_{k}< m\) can be upper bounded by
    \begin{equation}
        \Pr[\varsigma_{k}<m] \leq 2\exp\left(-\,\frac{\,(\sqrt{n}\left(p_{k} - \frac{m}{n}\right) - 2)^2}{8}\right) \, .
    \end{equation}
    Let \(\mathcal{T}\) be the overhang-removal channel, and let \(\mathcal{T}^{\prime}\) be the modified channel that outputs the maximally mixed state on all \(\yd{\varsigma}\) with \(\varsigma_{k}<m\). Then we have
    \begin{equation}
        ^{k}\mathcal{F}_{\mathrm{all}}(\mathcal{T}^{\prime}) \geq ^{k}\mathcal{F}_{\mathrm{all}}(\mathcal{T}) - \Pr[\varsigma_{k}<m] = ^{k}\mathcal{F}_{\mathrm{all}}(\mathcal{T}) - \eta \, .
    \end{equation}
    Here we set the fidelity loss on the irreps \(\varsigma_{k}<m\) to \(1\), which provides a lower bound, and the same inequality also holds for \(^{k}\mathcal{F}_{\mathrm{one}}\).

    For the implementation, the algorithm first applies GQPE for the permutation group and then measures the irrep label \(\yd{\varsigma}\), as in Ref.~\cite{LFIC24S}. If we measure \(\varsigma_{k}<m\), then we simply output the maximally mixed state, which can be done with gate complexity \(O(m\log(d))\). If \(m \leq \varsigma_{k}\), the irrep label \(\yd{\mu}\) given by the overhang removal rule satisfies \(\yd{\mu}\vdash_{d}n-m\). Therefore we can use~\cref{lem:implementation_resource_state} to initialize the state \(\ket{T^{\yd{\varsigma}}}\) described there. After performing inverse GQPE and tracing out qudits \(1,\ldots,n-m\), we have implemented the channel. The argument for this is the same as in Ref.~\cite{LFIC24S}. The gate complexity of GQPE is \(\bigO{n^4\ln^p(dn/\delta)}\). For the resource state, \(\yd{\varsigma}\) has at most \(n\) rows, so we can set \(r=n\). Combining this with~\cref{lem:implementation_resource_state} gives the stated complexity.

    For the overall diamond-norm error, we can view the state preparation from~\cref{lem:implementation_resource_state} as an entanglement-breaking quantum channel with diamond-norm \(\leq \delta\). Since the diamond-norm errors are subadditive when composing channels, we can implement GQPE and the resource state generation up to error \(\delta/2\), which does not change the complexity and gives a total error of \(\delta\).
\end{proof}

\subsection{CG-OQPA correctness and complexity}
\label{subsec:implementation_via_CG_transform}

This implementation follows the CG-transform approach of Ref.~\cite{MT25S}.

\begin{theorem}[CG-OQPA correctness and complexity]
    \label{thm:implementation_CG}
    For input \(\rho^{\otimes n}\), where \(\rho\) has rank \(r\), we can perform \(k\)-th eigenstate QPA up to diamond-norm error \(\delta\) by applying \(T\) one-qubit or two-qubit gates from a universal gate set, where
    \begin{equation}
        T = \bigO{dr^3(m+n)\ln^p(d(m+n)/\delta)} \, .
    \end{equation}
    Here, \(p\approx 1.44\). In addition, we can do this in a streaming manner, using only \(M\) qubits of memory, where
    \begin{equation}
        M = \bigO{dr\ln^p(d(m+n)/\delta)} \, .
    \end{equation}
\end{theorem}

\begin{proof}
    We use the resource state from~\cref{lem:implementation_resource_state} with the implementation results of Ref.~\cite[Proposition~11, Theorem~14]{MT25S}. The translation from the notation of Ref.~\cite{MT25S} to ours is
    \begin{equation}
    \begin{alignedat}{2}
        \lambda &\rightarrow \varsigma \, , &\qquad m &\rightarrow n \, , \\
        n &\rightarrow m \, , & r &\rightarrow r \, , \\
        r^\prime &\rightarrow 1 \, , & k &\rightarrow m \, , \\
        l &\rightarrow 0 \, , & \iota_{p_{\mu\rightarrow\lambda}^{k,l,d}} &\rightarrow \intertwiner{W}{\yd{\varsigma}}{\yd{\mu}}{\yd{m}}{} \, , \\
        \epsilon &\rightarrow \delta & & \, .
    \end{alignedat}
    \end{equation}
    Combining Ref.~\cite[Proposition~11]{MT25S} with~\cref{lem:implementation_resource_state}, we can implement \(\intertwiner{W}{\yd{\varsigma}}{\yd{\mu}}{\yd{m}}{}\) with \(\bigO{mdr^3\ln^p(dm/\delta)}\) gates. Ref.~\cite[Theorem~14]{MT25S} then implies that we can implement the whole QPA channel with \(\bigO{(n+m)dr^3\ln^p(d(n+m)/\delta)}\) gates. Since the resource-state initialization of~\cref{lem:implementation_resource_state} can be performed in a streaming manner, the streaming memory complexity follows by combining it with Ref.~\cite{MT25S}.

    For the overall diamond-norm error, we use an argument similar to that in the proof of~\cref{thm:implementation_GPE}.
\end{proof}

\section{Asymptotic analysis of sector-wise utilities}
\label{sec:asymptotic_supp}

It is generally hard to directly evaluate the average over CGCs in~\cref{equ:sector_wise_fidelity}. However, most qualitative behaviors can be extracted by studying it in the asymptotic regime of large $n$.
To do so, we build on a generalized version of the argument in Ref.~\cite{LFIC24S} to evaluate the sector-wise utility as a quotient of weight-normalized terms
\begin{equation}
    ^k\res{f_\mathrm{all}}{\yd{\varsigma}}{\yd{m}}{\yd{\mu}}=
  \frac{\widetilde{\res{f_\mathrm{all}}{\yd{\varsigma}}{\yd{m}}{\yd{\mu}}}}{\widetilde{s^{\yd{\varsigma}}}}.
\end{equation}
For some unsorted spectrum \(\boldsymbol{q}\dot{=}(q_{1},\ldots,q_{d})\), let us define the sorting permutation \(\sigma\in S_{d}\) such that $\sigma^{-1}(\boldsymbol{q})\dot{=}(q_{\sigma(1)},\ldots,q_{\sigma(d)})$ is sorted in nonincreasing order, i.e.,
\begin{equation}
     q_{\sigma(1)} \geq \ldots \geq q_{\sigma(d)}.
\end{equation}
The weight-normalized utility and Schur polynomial are given by
\begin{equation}
    \widetilde{\res{f_\mathrm{all}}{\yd{\varsigma}}{\yd{m}}{\yd{\mu}}}\coloneqq \frac{1}{\sigma^{-1}(\boldsymbol{q})^{\#\mleft(\wt{\lw}\mright)}}\sum_{\wt{w}\vdash\yd{\varsigma}}\res{f_\mathrm{all}}{\yd{\varsigma}}{\yd{m}}{\yd{\mu}}\mleft(\wt{w}\mright)\prod_{i=1}^{d} {q_i}^{\#_{i}\mleft(\wt{w}\mright)} \quad \text{ and} \quad \widetilde{s^{\yd{\varsigma}}}\coloneqq \frac{s^{\yd{\varsigma}}}{\sigma^{-1}(\boldsymbol{q})^{\#\mleft(\wt{\lw}\mright)}} \, .
\end{equation}
The normalization is the largest term in the Schur polynomial, corresponding to the probability weight of the lowest weight WT with the sorted probability vector \(\sigma^{-1}(\boldsymbol{q})\). We have
\begin{equation}
    \sigma^{-1}(\boldsymbol{q})^{\#\mleft(\wt{\lw}\mright)} = \prod_{i=1}^{d}q_{\sigma(i)}^{\varsigma_{i}} \, .
\end{equation}
We will now reparametrize the GT patterns $\wt{w}$ in a way that is compatible with the ordering \(\sigma\) of the spectrum \(\boldsymbol{q}\). For this purpose, we introduce path graphs.

\subsection{GT parametrization with path graphs}
\label{subsec:gt_param_supp}
We introduce the following notion of GT lattices. Let \(V = \{(i,j) \mid 1 \le i \le j \le d\}\) denote the vertex set corresponding to the index set of the GT pattern. For each \((i,j) \in V\) with \(j < d\), we define two edges: \((i,j)\ue(i,j+1)\) and \((i,j) \ue (i+1,j+1)\), which we label by $\setminus$ and $/$ respectively.

\begin{definition}
    A \emph{path graph} is a collection of \(d\) disjoint directed paths \(\boldsymbol{\gamma}=\{\boldsymbol{\gamma}_a\}\). For each \(a\in\{1,\ldots,d-1\}\), the path $\boldsymbol{\gamma}_a$ is an ordered sequence of vertices from level $a$ to level $d$:
    \begin{equation}
    \begin{aligned}
        \boldsymbol{\gamma}_a\,\dot{=}\,&\boldsymbol{\gamma}_{a,a}\ue\boldsymbol{\gamma}_{a,a+1}\ue\cdots\ue\boldsymbol{\gamma}_{a,d}\\
        \,=\,&        ((\gamma_{a,a})_{1},a)\ue((\gamma_{a,a+1})_{1},a+1)\ue\cdots\ue((\gamma_{a,d})_{1},d) \, ,
    \end{aligned}
    \end{equation}
    and we denote the edge \(\gamma_{a,b-1}\ue\gamma_{a,b}\) as \(\graphpath{a}{b}\). We denote the set of all path graphs as $\Gamma$.
\end{definition}

Since the paths of a path graph are disjoint, they collectively tile the vertex set $V$, together with $\boldsymbol{\gamma}_{d,d}$ defined as the remaining terminal vertex. This way, we can express $V$ alternatively as $V=\{\boldsymbol\gamma_{a,b}\mid 1\le a\le b\le d\}$.
\begin{figure}[H]
    \centering
    \includegraphics[width=0.3\textwidth]{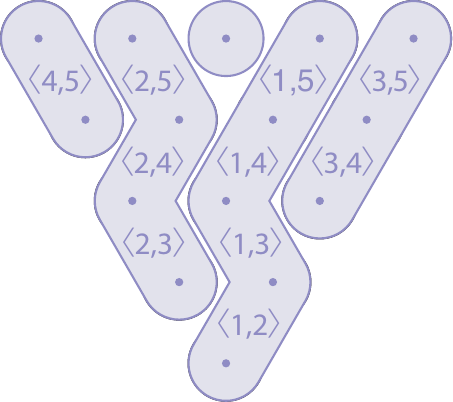}
    \caption{
        A path graph with \(d=5\).
        Each dot represents a vertex \(w_{i,j}\) arranged on a triangular lattice.
        Circled dots indicate a path, with each edge labeled by an index pair \(\graphpath{a}{b}\).
        }
    \label{fig:gt_path_param}
\end{figure}
We introduce an alternative parametrization of GT patterns via these path graphs.

\begin{definition}
    \label{def:path_graph_parametrization}
    For a given path graph \(\boldsymbol{\gamma}\) and GT pattern \(\wt{w}\) we define for each edge \(\graphpath{a}{b}\) of \(\boldsymbol{\gamma}\) the quantity
    \begin{equation}
        t^\prime_{a,b} \coloneqq w_{\gamma_{a,b-1}}- w_{\gamma_{a,b}} \, ,
    \end{equation}
    where we write \(w_{\gamma_{a,b}} = w_{(\gamma_{a,b})_{1},(\gamma_{a,b})_{2}}\). The signed tuple \(\boldsymbol{t}^\prime\dot{=}\{t^\prime_{a,b}\}_{a,b}\) is the \emph{path-graph parametrization} of \(\wt{w}\) corresponding to the path graph \(\boldsymbol{\gamma}\). Since the signs are fixed by the path graph, we write \(t_{a,b}\coloneqq |t^\prime_{a,b}|\) for the corresponding nonnegative edge variables. The \emph{admissible domain} is the set of all such nonnegative tuples for the fixed path graph under consideration and YD $\yd{\varsigma}$, denoted by $\mathcal{W}^{\yd{\varsigma}} \subseteq \mathbb{Z}_{\geq0}^{d(d-1)/2}$.
\end{definition}

\begin{remark}
    We can easily see that \(\boldsymbol{t}\) uniquely determines a given \(\wt{w}\) and vice versa: Given the top row \(\yd{\varsigma}=\yd{w_{1,d},\ldots,w_{d,d}}\), the parametrization determines all remaining entries \(w_{i,j}\) by propagating downward along the paths. Namely,
    \begin{equation}
    \label{equ:path_graph_w_value}
        w_{\gamma_{a,b}}=w_{\gamma_{a,b+1}}+t^\prime_{a,b+1}=\cdots=\varsigma_{(\gamma_{a,d})_{1}}+\sum_{s=b+1}^dt^\prime_{a,s}
    \end{equation}
Therefore, we view \(\boldsymbol{t}\) as a reparametrization of \(\wt{w}\), and write functions of GT patterns as functions of \(\boldsymbol{t}\). Similarly, sums over GT patterns can be written as sums over all parametrizations in the admissible domain:
    \begin{equation}
        \sum_{\wt{w}\vdash_{d} \yd{\varsigma}} f\mleft(\wt{w}\mright) = \sum_{\boldsymbol{t} \in \mathcal{W}^{\yd{\varsigma}}} f(\boldsymbol{t}) \, .
    \end{equation}
\end{remark}

We seek to link the path graphs to the unsorted spectrum \(\boldsymbol{q}\). In particular, we would like the path starting in the first bottom row to end up at the place \(\sigma(1)\) in the last row. As an example,~\cref{fig:gt_path_param} would correspond to the spectrum
\begin{equation}
    q_{4} \geq q_{2} \geq q_{5} \geq q_{1} \geq q_{3} \, .
\end{equation}
To this end, we go derive the following lemma.

\begin{lemma}[Path-graphs of permutations]
\label{lem:ordering_path_graph}
Let $\pi\in S_d$ be a permutation.
Then there exists a unique path graph \(\boldsymbol{\gamma}\) whose vertices satisfy $
    \boldsymbol{\gamma}_{a,d} = (\pi^{-1}(a),d) \, ,$
or, equivalently, whose edges satisfy
\begin{equation}
\text{$\graphpath{a}{b}$ is of type } \setminus\text{ if }\ \pi^{-1}(a)<\pi^{-1}(b),
\qquad
\text{$\graphpath{a}{b}$ is of type } /\text{ if }\pi^{-1}(a)>\pi^{-1}(b).
\end{equation}
\end{lemma}

\begin{proof}
We proceed inductively. For \(d=1\), the result is trivial, as there exists only a single path graph. For \(d>1\), one checks that the requirement
$\boldsymbol{\gamma}_{a,d} = (\pi^{-1}(a),d)$
is equivalent to
\begin{equation}
    \text{$\graphpath{a}{d}$ is of type } \setminus\text{ if }\ \pi^{-1}(a)<\pi^{-1}(d),
    \qquad
    \text{$\graphpath{a}{d}$ is of type } /\text{ if }\pi^{-1}(a)>\pi^{-1}(d).
\end{equation}
We can first fix these edges. Now let
\begin{equation}
    \pi^{-1} = \sigma_{(\pi^{-1}(d)\cdots d)} \circ \tilde{\pi}^{-1} \, ,
\end{equation}
where \(\sigma_{(a\cdots b)}\in S_{d}\) is the cyclic permutation on the entries \(a,\ldots,b\), and \(\tilde{\pi} \in S_{d-1}\). It follows that with our previously fixed edges from level \(d\) to \(d-1\), we have
\begin{equation}
    \boldsymbol{\gamma}_{a,d-1} = (\tilde{\pi}^{-1}(a),d-1) \, .
\end{equation}
According to the induction hypothesis, there is a unique path graph \(\widetilde{\boldsymbol{\gamma}}\) corresponding to \(\tilde{\pi}\) on the subset of nodes
\begin{equation}
    \{(i,j) \mid 1 \le i \le j \le d-1\} \subseteq \{(i,j) \mid 1 \le i \le j \le d\} \, .
\end{equation}
We define \(\boldsymbol{\gamma}\) by gluing \(\widetilde{\boldsymbol{\gamma}}\) to the connections from level \(d-1\) to level \(d\) that we defined previously. For \(1\leq a < b \leq d-1\) we have
\begin{equation}
    \pi^{-1}(a)<\pi^{-1}(b) \quad \Leftrightarrow \quad \tilde{\pi}^{-1}(a)<\tilde{\pi}^{-1}(b) \, ,
\end{equation}
Together with
\begin{equation}
    \text{$\graphpath{a}{b}$ is of type } \setminus\text{ if }\ \tilde{\pi}^{-1}(a)<\tilde{\pi}^{-1}(b),
    \qquad
    \text{$\graphpath{a}{b}$ is of type } /\text{ if }\tilde{\pi}^{-1}(a)>\tilde{\pi}^{-1}(b)
\end{equation}
we obtain the lemma.
\end{proof}

We illustrate this correspondence in the following figure.

\begin{figure}[H]
    \centering
    \includegraphics[width=0.6\textwidth]{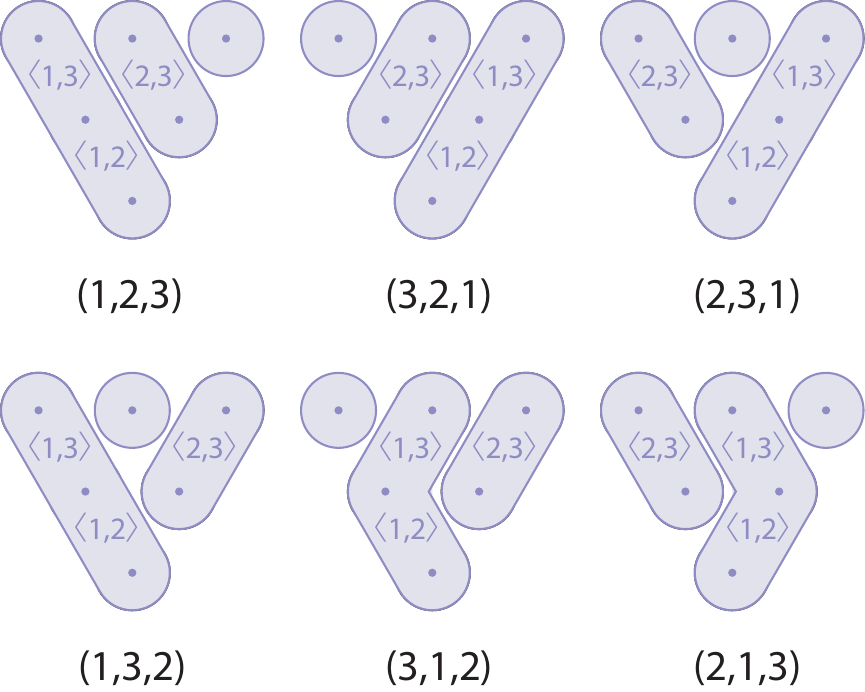}
    \caption{
        Correspondence between path graphs and permutations for $d=3$. The permutations are denoted by their cycles.
        }
    \label{fig:gt_path_param_permutations}
\end{figure}

We translate the path graph for permutations to arbitrarily ordered spectra.

\begin{corollary}[Path-graphs induced by spectrum ordering]
\label{cor:ordering_path_graph}
Fix $d\ge 1$ and a spectrum $\boldsymbol{q}\dot{=}(q_1,\ldots,q_d)$ with all entries distinct. There exists a unique tiling path graph $\boldsymbol{\gamma}$ such that the vertices at level \(d\) fulfill
\begin{equation}
    (\gamma_{a,d})_{1} < (\gamma_{b,d})_{1} \quad \Leftrightarrow \quad q_{a} > q_{b} \, .
\end{equation}
In addition, for every edge labeled $\graphpath{a}{b}$,
\begin{equation}
\graphpath{a}{b}\text{ is of type } \setminus \;\;\text{if}\;\; q_a>q_b,
\qquad
\graphpath{a}{b}\text{ is of type } / \;\;\text{if}\;\; q_a<q_b .
\end{equation}
\end{corollary}

\begin{proof}
This is a direct consequence of~\cref{lem:ordering_path_graph}.
\end{proof}

For a spectrum \(\boldsymbol{q}\) as given in~\cref{cor:ordering_path_graph}, the path-graph parametrization allows us to write the occupancy in a simple fashion. Let \(\boldsymbol{\gamma}\) be the corresponding path graph from~\cref{cor:ordering_path_graph}. For each \(a\in\{1,\ldots,d\}\), the occupancy reads
\begin{equation}
\begin{aligned}
    \#_a\mleft(\wt{w}\mright)
    =&
    \sum_{i=1}^{a} w_{i,a}
    -
    \sum_{i=1}^{a-1} w_{i,a-1}\\
    =&w_{(\gamma_{a,a})_{1},a}+\sum_{b=1}^{a-1} w_{(\gamma_{b,a})_{1},a}-\sum_{b=1}^{a-1} w_{(\gamma_{b,a-1})_{1},a-1}\\
    =&\varsigma_{(\gamma_{a,d})_{1}}+
    \sum_{b=a+1}^{d} t^\prime_{a,b}
    \;-\;
    \sum_{b=1}^{a-1} t^\prime_{b,a},
\end{aligned}
\end{equation}
where in the second equality we rewrite the summands in terms of the path-graph parametrization, and in the third equality we invoke~\cref{equ:path_graph_w_value}. The normalized weight of a diagram \(\wt{w}\) is therefore:
\begin{equation}
\label{eq:monomial_weight}
\begin{aligned}
    &\frac{\prod_{a=1}^d q_a^{\#_a\mleft(\wt{w}\mright)}}{\sigma^{-1}(\boldsymbol{q})^{\#\mleft(\wt{\lw}\mright)}}=\frac{\prod_{a=1}^d q_a^{\#_a}}{\prod_{a=1}^d q_a^{\varsigma_{(\gamma_{a,d})_{1}}}}=
    \prod_{a=1}^d q_a^{\#_a-\varsigma_{(\gamma_{a,d})_{1}}}
    =\prod_{1\le a<b\le d}\mleft(\frac{q_a}{q_b}\mright)^{t^\prime_{a,b}} \, .
\end{aligned}
\end{equation}
We see the advantage of this parametrization, as each term in the product is \(< 1\). This is due to~\cref{def:path_graph_parametrization} and the second result of~\cref{cor:ordering_path_graph}. We therefore have
\begin{equation}
\begin{aligned}
    t^\prime_{a,b}\leq0 \quad \Leftrightarrow \quad \graphpath{a}{b}\text{ is of type } \setminus \quad \Leftrightarrow \quad q_{a} > q_{b} \, , \\
    t^\prime_{a,b}\geq0 \quad \Leftrightarrow \quad \graphpath{a}{b}\text{ is of type } / \quad \Leftrightarrow \quad q_{a} < q_{b} \, .
\end{aligned}
\end{equation}
We alternatively denote the monomial in~\cref{eq:monomial_weight} as
\begin{equation}
\label{equ:path_graph_monomial}
    \prod_{1\le a<b\le d}\mleft(\frac{q_a}{q_b}\mright)^{t^\prime_{a,b}}=\prod_{a,b}{r_{a,b}^{t_{a,b}}},
\end{equation}
where $t_{a,b} \coloneqq |t^\prime_{a,b}|$ and
\begin{equation}
    r_{a,b} \coloneqq \mleft\{\begin{aligned}
    &\frac{q_{a}}{q_{b}}, & q_{a}<q_{b}, \\
    &\frac{q_{b}}{q_{a}}, & q_{b}<q_{a}.
    \end{aligned}\mright.
\end{equation}
We determine the admissible domain of the nonnegative edge variables $\boldsymbol{t}$. For the path graph determined by the ordering of $\boldsymbol{q}$ via~\cref{cor:ordering_path_graph}, we denote this domain by $\mathcal{W}^{\yd{\varsigma}}$.

To enforce the GT betweenness condition, the allowed values of $\boldsymbol{t}$ are specified recursively by constraints that propagate along the two adjacent paths in the path graph. At each level $\ell$, the admissible range depends on the edge type.
For a $\setminus$-type edge the corresponding variable satisfies
\begin{equation}
\begin{aligned}
    t_{a,\ell}&\in\{0,\ldots,\, w_{(\gamma_{a,\ell})_{1},\,\ell}-w_{(\gamma_{a,\ell})_{1}+1,\,\ell}\},
\end{aligned}
\end{equation}
whereas for a $/$-type edge it satisfies
\begin{equation}
\begin{aligned}
    t_{a,\ell}&\in\{0,\ldots,\, w_{(\gamma_{a,\ell})_{1}-1,\,\ell}-w_{(\gamma_{a,\ell})_{1},\,\ell}\}
\end{aligned}
\end{equation}
These bounds propagate recursively: the admissible values of the variables
$t_{a,\ell-1}$ at level $\ell-1$ are obtained from the values of the variables $t_{a,\ell}$ at level $\ell$. We take the index $b(a,\ell)$ as the path index that corresponds to these restrictions. More precisely, we take $b(a,\ell)$ such that
\begin{subequations}
\begin{align}
    (\gamma_{b(a,\ell),\ell})_{1} &= (\gamma_{a,\ell})_{1} + 1, \text{ if }\graphpath{a}{\ell}\text{ is of type } \setminus, \\
    (\gamma_{b(a,\ell),\ell})_{1} &= (\gamma_{a,\ell})_{1} - 1, \text{ if }\graphpath{a}{\ell}\text{ is of type } /.
\end{align}
\end{subequations}
If we insert this definition into the expressions above, we can write $t_{a,\ell} \in \{0,\ldots,\Delta_{\mathrm{eff}}(a,\ell)\}$ with
\begin{equation}
    \label{equ:effective_constraint_a_l}
    \Delta_{\mathrm{eff}}(a,\ell) \coloneqq
    \mleft\{
    \begin{aligned}
    &\varsigma_{(\gamma_{a,d})_{1}}-\varsigma_{(\gamma_{b(a,\ell),d})_{1}}+\sum_{s=\ell+1}^{d}t_{a,s}-\sum_{s=\ell+1}^{d}t_{b(a,\ell),s}, & \text{if } \graphpath{a}{\ell}\text{ is of type } \setminus \, , \\
    &\varsigma_{(\gamma_{b(a,\ell),d})_{1}}-\varsigma_{(\gamma_{a,d})_{1}}+\sum_{s=\ell+1}^{d}t_{b(a,\ell),s}-\sum_{s=\ell+1}^{d}t_{a,s}, & \text{if } \graphpath{a}{\ell}\text{ is of type } / \, .
    \end{aligned}
    \mright.
\end{equation}
Since the paths are all disjoint, we have $(\gamma_{a,\ell})_{1}\lessgtr(\gamma_{b,\ell})_{1}$ if and only if $(\gamma_{a,d})_{1}\lessgtr(\gamma_{b,d})_{1}$, which ensures that
\begin{equation}
\begin{aligned}
    &\varsigma_{(\gamma_{a,d})_{1}}-\varsigma_{(\gamma_{b(a,\ell),d})_{1}} >0, \text{ if } \graphpath{a}{\ell}\text{ is of type } \setminus \, , \\
    &\varsigma_{(\gamma_{b(a,\ell),d})_{1}}-\varsigma_{(\gamma_{a,d})_{1}} > 0, \text{ if }\graphpath{a}{\ell}\text{ is of type } / \, .
\end{aligned}
\end{equation}
We further know that the asymptotic scaling is
\begin{equation}
    \varsigma_{i} - \varsigma_{j} = n(q_{\sigma(i)}-q_{\sigma(j)}) + \lilo{n} \, ,
\end{equation}
and therefore the constant contribution of $\Delta_{\mathrm{eff}}(a,\ell)$ scales as $\bigTheta{n}$. In particular, any given $\boldsymbol{t} \in \mathbb{Z}_{\geq0}^{d(d-1)/2}$ will become valid for sufficiently large $n$, and intuitively, we can let $0\leq t_{a,\ell} < \infty$ in the asymptotic limit. This allows us to evaluate summations similar to those developed in Ref.~\cite{LFIC24S}. We will see that the normalization factor of Weyl averages is given by the following geometric sum
\begin{equation}
    \Lambda(\mathbf{r}) \coloneqq \prod_{a,b}\frac{1}{1-r_{a,b}} =  \prod_{a,b} \sum_{t_{a,b}=0}^{\infty} {r_{a,b}^{t_{a,b}}} \, .
\end{equation}
This is precisely the normalization factor of the unnormalized geometric distribution $\mathrm{G}(\boldsymbol{r})$ for $\boldsymbol{t} \in \mathbb{Z}_{\geq0}^{d(d-1)/2}$, where $\boldsymbol{r}\coloneqq \{r_{a,b}\}_{a,b}$ denotes the tuple of edge ratios.

\begin{theorem}[Uniform geometric approximation of Weyl averages]
\label{thm:uniform_asymptotics_weyl_distribution}
\label{thm:asymptotics_weyl_distribution}
Let $\mathsf{R}_\epsilon$ be the gap-separated typical region defined in~\cref{eq:gap_separated_typical_region}.
Let $D=d(d-1)/2$ be the number of edge variables, and let $\mathcal{W}^{\yd{\varsigma}}\subseteq \mathbb{Z}_{\geq0}^{D}$ be the admissible domain of the variables $\boldsymbol{t}$ for the YD $\yd{\varsigma}$.
Let $f_n:\mathbb{Z}_{\geq0}^{D}\to\mathbb{R}$ be a function satisfying, for some fixed $s>0$,
\begin{equation}
f_n(\boldsymbol{t})
=
\bigO[d,\epsilon][n]{s^{\|\boldsymbol{t}\|_1}}.
\end{equation}
Let $r\coloneqq \max_{a,b}r_{a,b}$ and assume $sr<1$.
Then, uniformly for all $\yd{\varsigma}$ with $\overline{\yd{\varsigma}}\in\mathsf{R}_\epsilon$,
\begin{equation}
\weylavg{f_n(\boldsymbol{t})}{\boldsymbol{q}}{\yd{\varsigma}}
=
\left\langle
f_n(\boldsymbol{t})
\right\rangle_{\mathrm{G}(\boldsymbol{r})}
+
\bigO[d,\mathbf{r},\epsilon][n]{\exp(-n)}.
\end{equation}
\end{theorem}

\begin{proof}
On $\mathsf{R}_\epsilon$, all relevant normalized gaps satisfy $\overline\Delta_{i,j}\geq \epsilon$.
Set $L\coloneqq \floor{\epsilon n/(2D)}$. We claim that, for every $\yd{\varsigma}\vdash_d n$ with $\overline{\yd{\varsigma}}\in\mathsf{R}_\epsilon$, any $\boldsymbol{t}\in[0,L]^D$ is admissible, i.e., $[0,L]^D\subseteq\mathcal{W}^{\yd{\varsigma}}$.
Indeed, for each edge variable $t_{a,\ell}\in[0,L]$, the admissibility upper bound is $\Delta_{\mathrm{eff}}(a,\ell)$. Its row-difference part is at least $n\epsilon$, and the remaining terms involve at most $D$ other edge variables, each bounded by $L$. Thus $\Delta_{\mathrm{eff}}(a,\ell)\geq n\epsilon-DL\geq L$, so $\boldsymbol{t}\in\mathcal{W}^{\yd{\varsigma}}$.
We then have
\begin{equation}
\begin{aligned}
\left|
\sum_{\boldsymbol{t}\in\mathcal{W}^{\yd{\varsigma}}}
f_n(\boldsymbol{t})
\boldsymbol{r}^{\boldsymbol{t}}
-
\sum_{\boldsymbol{t}\in\mathbb{Z}_{\geq0}^{D}}
f_n(\boldsymbol{t})
\boldsymbol{r}^{\boldsymbol{t}}
\right|
&\leq
\sum_{\boldsymbol{t}\notin\mathcal{W}^{\yd{\varsigma}}}
|f_n(\boldsymbol{t})|
\boldsymbol{r}^{\boldsymbol{t}} \\
&\leq
\sum_{\boldsymbol{t}\notin[0,L]^D}
|f_n(\boldsymbol{t})|
\boldsymbol{r}^{\boldsymbol{t}} \\
&\leq
\sum_{\boldsymbol{t}\notin[0,L]^D}
C_{d,\epsilon} s^{\|\boldsymbol{t}\|_1}r^{\|\boldsymbol{t}\|_1}.
\end{aligned}
\end{equation}
Since $sr<1$,
\begin{equation}
\sum_{\boldsymbol{t}\notin[0,L]^D}
(sr)^{\|\boldsymbol{t}\|_1}
\leq
D
\left(\sum_{t>L}(sr)^t\right)
\left(\sum_{t\geq0}(sr)^t\right)^{D-1}
=
\bigO[d,\mathbf{r},\epsilon][n]{\exp(-n)}.
\end{equation}
The constants in this estimate have the dependence specified in the theorem statement and are uniform over $\yd{\varsigma}$ with $\overline{\yd{\varsigma}}\in\mathsf{R}_\epsilon$.
This proves the unnormalized estimate.

Applying the same estimate to $f_n\equiv 1$ gives
\begin{equation}
\widetilde{s^{\yd{\varsigma}}}
=
\Lambda(\boldsymbol{r}) + \bigO[d,\mathbf{r},\epsilon,s][n]{\exp(-n)}.
\end{equation}
Since $\Lambda(\boldsymbol{r})>0$, dividing the unnormalized estimate by the normalization gives the result.
\end{proof}

\begin{corollary}[Weyl average with a controlled remainder]
\label{cor:weyl_average_main_term_remainder}
Let $\boldsymbol{q}$, $\boldsymbol{r}$, and $\mathsf{R}_{\epsilon}$ be as in~\cref{thm:asymptotics_weyl_distribution}.
Let $f_n:\mathbb{Z}_{\geq0}^{D}\to\mathbb{R}$ admit a decomposition $f_n(\boldsymbol{t})=g_n(\boldsymbol{t})+R_n(\boldsymbol{t})$.
All estimates below are assumed uniform in $\boldsymbol{t}$ and $\overline{\yd{\varsigma}}\in\mathsf{R}_{\epsilon}$.
Assume that there is $s>1$ with $sr<1$ such that
\begin{equation}
|g_n(\boldsymbol{t})|
=
\bigO[d,\mathbf{r},\epsilon][n]{s^{\|\boldsymbol{t}\|_1}}.
\end{equation}
Assume moreover that, for some sequence $\alpha_n\to0$, $\alpha_n=\bigOm[][n]{\exp(-n)}$, and
\begin{equation}
|R_n(\boldsymbol{t})|
=
\bigO[d,\mathbf{r},\epsilon][n]{\alpha_n s^{\|\boldsymbol{t}\|_1}}.
\end{equation}
\begin{equation}
    \weylavg{f_n(\boldsymbol{t})}{\boldsymbol{q}}{\yd{\varsigma}}
    =
    \left\langle
    g_n(\boldsymbol{t})
    \right\rangle_{\mathrm{G}(\boldsymbol{r})}
    +
    \bigO[d,\mathbf{r},\epsilon][n]{\alpha_n}.
\end{equation}
\end{corollary}

\begin{proof}
For all sufficiently large $n$, the assumptions give $|f_n(\boldsymbol{t})|=\bigO[d,\mathbf{r},\epsilon][n]{s^{\|\boldsymbol{t}\|_1}}$.
Thus~\cref{thm:asymptotics_weyl_distribution}, linearity, the remainder bound, and $\alpha_n=\bigOm[][n]{\exp(-n)}$ give the claimed estimate directly.
\end{proof}

In order to make use of~\cref{thm:asymptotics_weyl_distribution,cor:weyl_average_main_term_remainder}, we need to reparametrize the expression for the fidelity given in~\cref{thm:utility_component}. To this end, we first discuss the path-graph structure specific to our problem, see~\cref{equ:permutation}. Thus,
\begin{equation}
\label{equ:spectrum_ordering_q}
    q_{1} > \ldots > q_{k-1} > q_{d} > q_{k} > \ldots > q_{d-1} \geq 0 \, ,
\end{equation}
An example for a path graph of such an ordering is given in~\cref{fig:path_graph_fidelity}. With this choice, the GT lattice naturally decomposes into two disjoint regions, corresponding to edges of type \(/\) and \(\setminus\).
\begin{figure}[H]
    \centering
    \centering\includegraphics[scale=0.6]{"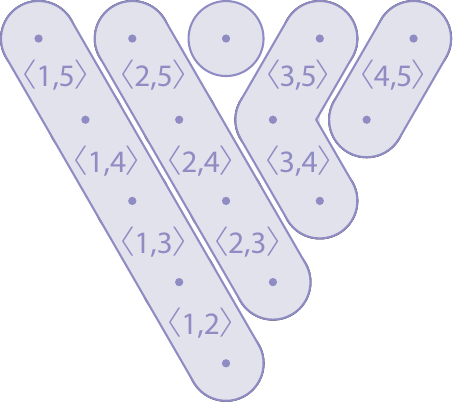"}
    \caption{Illustration of the path graph used to parametrize the summation variables for the case of $k=3$ and $d=5$.
    Edges of type \(/\) and \(\setminus\) define two summation regions that are evaluated separately.}
    \label{fig:path_graph_fidelity}
\end{figure}
According to this split, we separate the GT variables into the two regions depending on whether $i$ lies below or above $k$. We define the shifted row gaps $\widetilde{\Delta}_{i,j}\coloneqq\Delta_{i,j}+j-i$ and we can rewrite the expression from~\cref{thm:utility_component} as
{\scriptsize\normalsizeequationtag
\begin{equation}
\label{equ:nm_factorization}
\begin{aligned}
    &\res{f}{\yd{\mu}}{\yd{m}}{\yd{\varsigma}}_\mathrm{all}\mleft(\wt{w}\mright)\\
    =&\binom{m}{\mathbf m}\prod_{\ell=1}^d
    \frac{
      \prod_{i=1,i\neq\ell}^{k-1}
      \rpoch{\widetilde w_{i,d-1}-\widetilde w_{\ell,d}}{m_\ell}
      \;\;
      \prod_{i=k,i\neq\ell-1}^{d-1}
      \rpoch{\widetilde w_{i,d-1}-\widetilde w_{\ell,d}}{m_\ell}
    }{
      \prod_{i=1,i\neq \ell}^{k-1}
      \rpoch{\widetilde w_{i,d}-\widetilde w_{\ell,d}+1}{m_\ell}
      \prod_{i=k+1,i\neq \ell}^{d}
      \rpoch{\widetilde w_{i,d}-\widetilde w_{\ell,d}+1}{m_\ell}
    }\mleft\{\begin{aligned}
      &\frac{\rpoch{\widetilde w_{\ell,d-1}-\widetilde w_{\ell,d}}{m_\ell}}
      {\rpoch{\widetilde w_{k,d}-\widetilde w_{\ell,d}+1}{m_\ell}},
      & \ell<k,\\
      &\frac{\rpoch{\widetilde w_{\ell-1,d-1}-\widetilde w_{\ell,d}}{m_\ell}}
      {\rpoch{\widetilde w_{k,d}-\widetilde w_{\ell,d}+1}{m_\ell}},
      & \ell>k,\\
      &1, & \ell=k,
    \end{aligned}\mright.
    \\[10pt]
    =&
    \binom{m}{\mathbf m}\prod_{\ell=1}^d
      \prod_{i=1,i\neq\ell}^{k-1}
      \prod_{r=1}^{m_\ell}
      \mleft(1-\frac{t_{i,d}+1}{\widetilde{\Delta}_{i,\ell}+r}\mright)
      \prod_{i=k,i\neq \ell-1}^{d-1}
      \prod_{r=1}^{m_\ell}
      \mleft(1+\frac{t_{i,d}}{\widetilde{\Delta}_{i+1,\ell}+r}\mright)
      \mleft\{\begin{aligned}
      &\frac{\rpoch{-t_{\ell,d}}{m_\ell}}{\rpoch{\widetilde{\Delta}_{k,\ell}+1}{m_\ell}},
      & \ell<k,\\
      &\frac{\rpoch{t_{\ell-1,d}+1}{m_\ell}}{\rpoch{\widetilde{\Delta}_{k,\ell}+1\,}{m_\ell}},
      & \ell>k,\\
      &1, & \ell=k,
    \end{aligned}\mright..
\end{aligned}
\end{equation}
}
In the final step, we expanded the components of $\wt{w}$ using the path-graph parametrization, expressing the result in terms of the edge variables $t_{i,d}$ and the YD row gaps.

\subsection{Intensive outputs}
\label{subsubsec:intensive_supp}

We first consider intensive outputs where $m$ is constant. To clarify the asymptotic behavior, we begin with the case $m=1$, where a single box is removed and~\cref{equ:nm_factorization} simplifies according to the choice of the removed box.
We start with the removal of one box from the target row, i.e., the $k$-th row. In this case, we have

\begin{equation}
\label{equ:single_output_target_row_product}
\begin{aligned}
    \res{f}{\yd{\varsigma - \boldsymbol{e}_{k}}}{\yd{1}}{\yd{\varsigma}}_\mathrm{all}(\boldsymbol{t})
    &=
    \prod_{i=1}^{k-1}\mleft(1-\frac{t_{i,d}+1}{\widetilde{\Delta}_{i,k}+1}\mright)
    \prod_{i=k}^{d-1}\mleft(1-\frac{t_{i,d}}{\widetilde{\Delta}_{k,i+1}-1}\mright) \, .
\end{aligned}
\end{equation}

To study the leading large-$n$ behavior of~\cref{equ:single_output_target_row_product}, we first need the following denominator-expansion lemma.

\begin{lemma}[Uniform expansion of shifted gap denominators]
\label{lem:uniform_shifted_gap_expansion}
Fix $\epsilon>0$, $\eta\in\mathbb{Z}_{\geq0}$, and $c\in\mathbb{R}$.
For $n\geq 1$ and $\Delta\in\mathbb{N}$, assume that $\Delta\geq n\epsilon$. Then, for all sufficiently large $n$,
\begin{equation}
\frac{1}{\Delta+c}
=
\sum_{i=0}^{\eta}
\frac{(-c)^i}{\Delta^{i+1}}
+
\bigO[\epsilon,c,\eta][n]{n^{-(\eta+2)}}.
\end{equation}
\end{lemma}

\begin{proof}
For all sufficiently large $n$, we have $n\epsilon\geq2|c|$, and hence
$\Delta+c\geq\Delta-|c|\geq\Delta/2>0$.
We have $1/(\Delta+c)=\Delta^{-1}(1+c/\Delta)^{-1}$.
Using the finite geometric identity $(1+x)^{-1}=\sum_{i=0}^{\eta}(-x)^i+(-x)^{\eta+1}/(1+x)$ with $x=c/\Delta$, we get
\begin{equation}
\frac{1}{\Delta+c}
=
\sum_{i=0}^{\eta}
\frac{(-c)^i}{\Delta^{i+1}}
+
\frac{(-c)^{\eta+1}}{\Delta^{\eta+1}(\Delta+c)}.
\end{equation}
Define the remainder by $E\coloneqq(-c)^{\eta+1}/(\Delta^{\eta+1}(\Delta+c))$.
Using $\Delta+c\geq\Delta/2$ and $\Delta\geq n\epsilon$, we obtain
\begin{equation}
|E|
\leq
\frac{2|c|^{\eta+1}}{\Delta^{\eta+2}}
\leq
2|c|^{\eta+1}\epsilon^{-(\eta+2)}
n^{-(\eta+2)}.
\end{equation}
This proves the claim.
\end{proof}

\begin{theorem}[Uniform expansion bounds for finite products]
\label{thm:uniform_product_expansion_bounds}
Fix $s>1$, a polynomial degree $\mathfrak q\in\mathbb{Z}_{\geq0}$, and an expansion order $\iota\in\mathbb{N}$. For $\eta\in\mathbb{N}$ and $\epsilon>0$, take each $\Delta_i\in\mathbb{N}$ to be of the form $\Delta_i=n\overline\Delta_i$, with $\overline\Delta_i\geq\epsilon$.
Fix signs $\sigma_i\in\{+1,-1\}$, shifts $c_i=\bigO[d]{1}$, and polynomials $P_i$ satisfying
$P_i(\boldsymbol t)=\bigO{(1+\|\boldsymbol t\|_1)^{\mathfrak q}}$. Consider the product
\begin{equation}
    f_n(\boldsymbol{t})
    =
    \prod_{i=1}^{\eta}
    \mleft(
    1+\sigma_i
    \frac{P_{i}(\boldsymbol{t})}
    {\Delta_{i}+c_{i}}
    \mright),
\end{equation}
There are coefficient functions $f^{(j)}(\boldsymbol{t})$, $0\leq j\leq\iota$, such that
\begin{equation}
    f_n(\boldsymbol{t})
    =
    \sum_{j=0}^{\iota}
    \frac{1}{n^j}f^{(j)}(\boldsymbol{t})
    +
    \bigO[d,\epsilon,\eta][n]{n^{-(\iota+1)}s^{\|\boldsymbol{t}\|_1}}.
\end{equation}
Set $Q_{i,0}(\boldsymbol t)\coloneqq1$ and, for $r\geq1$,
\begin{equation}
\label{equ:uniform_product_factor_coefficients}
Q_{i,r}(\boldsymbol t)
=
\sigma_i\frac{(-c_i)^{r-1}}{\overline\Delta_i^{\,r}}P_i(\boldsymbol t).
\end{equation}
The coefficient functions are explicitly
\begin{equation}
\label{equ:uniform_product_coefficients}
f^{(j)}(\boldsymbol t)
=
\sum_{\substack{r_1,\ldots,r_\eta\geq0\\r_1+\cdots+r_\eta=j}}
\prod_{i=1}^{\eta}Q_{i,r_i}(\boldsymbol t),
\qquad 0\leq j\leq\iota.
\end{equation}
In particular, $f^{(0)}(\boldsymbol t)=1$ and
$f^{(1)}(\boldsymbol t)=\sum_{i=1}^{\eta}\sigma_i P_i(\boldsymbol t)/\overline\Delta_i$.
Moreover, $f^{(j)}(\boldsymbol t)=\bigO[d,\epsilon,\eta]{s^{\|\boldsymbol t\|_1}}$ for $0\leq j\leq\iota$.
These estimates are uniform over $\boldsymbol t$ and the normalized gaps $\overline\Delta_i$.
\end{theorem}

\begin{proof}
Since $|c_i|\leq C_d$ and $\Delta_i\geq n\epsilon$, the threshold in~\cref{lem:uniform_shifted_gap_expansion} can be chosen uniformly in $i$ and $\Delta_i$.
Applying~\cref{lem:uniform_shifted_gap_expansion} with $\eta=\iota-1$ gives
\begin{equation}
\frac{1}{\Delta_{i}+c_{i}}
=
\sum_{j=0}^{\iota-1}
\frac{(-c_{i})^j}{n^{j+1}\overline\Delta_{i}^{j+1}}
+
\bigO[d,\epsilon,\eta][n]{n^{-(\iota+1)}}.
\end{equation}
Multiplying by $\sigma_iP_{i}(\boldsymbol{t})$ gives
\begin{equation}
1+\sigma_i
\frac{P_{i}(\boldsymbol{t})}
{\Delta_{i}+c_{i}}
=
1+
\sum_{j=1}^{\iota}
\frac{Q_{i,j}(\boldsymbol{t})}{n^j}
 + \bigO[d,\epsilon,\eta][n]{n^{-(\iota+1)}(1+\|\boldsymbol{t}\|_1)^{\mathfrak{q}}},
\end{equation}
where $Q_{i,j}(\boldsymbol{t})$ is given by~\cref{equ:uniform_product_factor_coefficients} and satisfies
$Q_{i,j}(\boldsymbol{t})
    =
    \bigO[d,\epsilon,\eta]{(1+\|\boldsymbol{t}\|_1)^{\mathfrak{q}}}$.
Set $\mathfrak{p}\coloneqq \eta\mathfrak{q}$.
Since the number of factors is $\eta$, multiplying them gives only finitely many terms, whose number depends only on $\eta$ and $\iota$.
Collecting terms of total order $n^{-j}$ for $0\leq j\leq \iota$ gives~\cref{equ:uniform_product_coefficients}.
Each coefficient is a finite sum of products of at most $\eta$ polynomials $Q_{i,j}$, and therefore satisfies
\begin{equation}
    f^{(j)}(\boldsymbol{t})
    =
    \bigO[d,\epsilon,\eta]{(1+\|\boldsymbol{t}\|_1)^{\mathfrak{p}}}.
\end{equation}
The discarded terms have total order at least $n^{-(\iota+1)}$ and contribute
\begin{equation}
    \bigO[d,\epsilon,\eta][n]{n^{-(\iota+1)}(1+\|\boldsymbol{t}\|_1)^{\mathfrak{p}}}.
\end{equation}
Finally, exponential growth dominates polynomial growth, so there exists $C'_{\eta}>0$ such that
\begin{equation}
    (1+\|\boldsymbol{t}\|_1)^{\mathfrak{p}}
    \leq
    C'_{\eta}s^{\|\boldsymbol{t}\|_1}
\end{equation}
for all $\boldsymbol{t}\in\mathbb{Z}_{\geq0}^{D}$.
All constants used above are uniform over $\boldsymbol t$ and the normalized gaps, which proves the claim.
\end{proof}
\begin{theorem}[Single-output asymptotics]
\label{thm:single_output_intensive_all_site_asymptotics}
For \(m=1\), the all-site utility of the overhang-removal protocol satisfies
\begin{equation}
\label{equ:intensive_asymptotic_overall_fidelity}
     ^k\mathcal{F}_\mathrm{all}=1-\frac{1}{n}\sum_{i=1,i\neq k}^{d}\frac{p_i}{D_{k,i}^2}+\lilo[d,\mathbf{r}]{n^{-1}} \, .
\end{equation}
\end{theorem}

\begin{proof}
We apply~\cref{cor:weyl_average_main_term_remainder} to the fidelity component. Pick any $s>1$ such that $sr<1$.
Applying~\cref{thm:uniform_product_expansion_bounds} to~\cref{equ:single_output_target_row_product} with $\iota=1$ and $\eta=d-1$, and using its explicit formula for $f^{(1)}$, gives
\begin{equation}\begin{aligned}
    \res{f}{\yd{\varsigma - \boldsymbol{e}_{k}}}{\yd{1}}{\yd{\varsigma}}_\mathrm{all}(\boldsymbol{t})
    &=1
    -\sum_{i=1}^{k-1}\frac{t_{i,d}+1}{\Delta_{i,k}}
    -\sum_{i=k}^{d-1}\frac{t_{i,d}}{\Delta_{k,i+1}}
    + \bigO[d,\mathbf{r},\epsilon][n]{n^{-2}s^{\|\boldsymbol{t}\|_1}} \, .
\end{aligned}\end{equation}
The coefficients and tail term in this expansion satisfy the hypotheses of~\cref{cor:weyl_average_main_term_remainder}.
This controls the tail by $\bigO[d,\mathbf{r},\epsilon][n]{n^{-2}}$.
In addition, we have
\begin{equation}
    \mleft\langle t_{i,d} \mright\rangle_{\mathrm{G}(\boldsymbol{r})} = \frac{r_{i,d}}{1-r_{i,d}} \, .
\end{equation}
Using the ordering in~\cref{equ:spectrum_ordering_q}, we obtain
\begin{equation}
    ^k\res{f_\mathrm{all}}{\yd{\varsigma-\mathbf{e}_k}}{\yd{1}}{\yd{\varsigma}} = \weylavg{\res{f}{\yd{\varsigma - \boldsymbol{e}_{k}}}{\yd{1}}{\yd{\varsigma}}_\mathrm{all}\mleft(\wt{w}\mright)}{\boldsymbol{q}}{\yd{\varsigma}}\!\!\! = 1 - \sum_{i=1}^{k-1}\frac{1}{\Delta_{i,k}}\frac{q_{i}}{q_i-q_d} - \sum_{i=k}^{d-1}\frac{1}{\Delta_{k,i+1}}\frac{q_{i}}{q_d-q_{i}}
    + \bigO[d,\mathbf{r},\epsilon][n]{n^{-2}} \, .
\end{equation}
Substituting the sorted spectrum $\boldsymbol{p}$ from~\cref{equ:permutation}, we get
\begin{equation}
\label{equ:intensive_asymptotic_sector_wise_fidelity}
    ^k\res{f_\mathrm{all}}{\yd{\varsigma-\mathbf{e}_k}}{\yd{1}}{\yd{\varsigma}} = 1 - \sum_{i\neq k}\frac{1}{\Delta_{k,i}}\frac{p_{i}}{D_{k,i}}
    + \bigO[d,\mathbf{r},\epsilon][n]{n^{-2}} \, .
\end{equation}
Applying~\cref{lem:sw_averaging_uniform_sector_expansion} to~\cref{equ:intensive_asymptotic_sector_wise_fidelity} proves~\cref{equ:intensive_asymptotic_overall_fidelity}.
\end{proof}
\begin{theorem}[Single-output off-target asymptotics]
\label{thm:single_output_intensive_off_target_all_site_asymptotics}
For $m=1$ and an off-target removal row $\ell\neq k$, the fidelity component is
\begin{equation}
\label{equ:off_target_removal_asymptotic_one}
\begin{aligned}
    \res{f}{\yd{\varsigma - \boldsymbol{e}_{\ell}}}{\yd{1}}{\yd{\varsigma}}_\mathrm{all}(\boldsymbol{t})
    &=
      \prod_{i=1,i\neq\ell}^{k-1}
    \mleft(1-\frac{t_{i,d}+1}{\widetilde{\Delta}_{i,\ell}+1}\mright)
      \prod_{i=k,i\neq \ell-1}^{d-1}
    \mleft(1+\frac{t_{i,d}}{\widetilde{\Delta}_{i+1,\ell}+1}\mright) \mleft\{\begin{aligned}
      &\frac{-t_{\ell,d}}{\widetilde{\Delta}_{k,\ell}+1},
      & \ell<k,
      \\[4pt]
      &\frac{t_{\ell-1,d}+1}{\widetilde{\Delta}_{k,\ell}+1},
      & \ell>k,
    \end{aligned}\mright.\\
    &= 0+ \mleft\{
    \begin{aligned}
    &\frac{t_{\ell,d}}{\Delta_{\ell,k}},
    & \ell < k,\\\\[6pt]
    &\frac{t_{\ell-1,d}+1}{\Delta_{k,\ell}},
    & \ell > k.
    \end{aligned}
    \mright.
    + \bigO[d,\mathbf{r},\epsilon][n]{n^{-2}s^{\|\boldsymbol{t}\|_1}} \, .
\end{aligned}
\end{equation}
For the sorted spectrum $\boldsymbol{p}$, the corresponding sector-wise utility is
\begin{equation}
\label{equ:off_target_removal_asymptotic_sector_wise_fidelity}
    ^k\res{f_\mathrm{all}}{\yd{\varsigma-\mathbf{e}_\ell}}{\yd{1}}{\yd{\varsigma}}
    =
    \frac{1}{\Delta_{k,\ell}} \frac{p_k}{D_{k,\ell}}
    + \bigO[d,\mathbf{r},\epsilon][n]{n^{-2}} \, .
\end{equation}
\end{theorem}

\begin{proof}
The product formula follows by specializing~\cref{thm:utility_component} to removal from row $\ell$. Applying~\cref{thm:uniform_product_expansion_bounds} with $\iota=1$ gives~\cref{equ:off_target_removal_asymptotic_one}. Weyl averaging as in the proof of~\cref{thm:single_output_intensive_all_site_asymptotics} then gives~\cref{equ:off_target_removal_asymptotic_sector_wise_fidelity}.
\end{proof}
The sector-wise utilities we have calculated here give us the overlap of the output with the $k$-th eigenstate of the input. In addition, the output $\rho^\prime$ is diagonal in this eigenbasis; see Ref.~\cite[Lemma~S1.7]{CI26S}. Therefore, the leading-order structure of the output state is immediately fixed:
\begin{equation}
  \rho^\prime=\psi_k+\sum_{i\neq k}\frac{1}{\Delta_{k,i}}\frac{p_{i}}{D_{k,i}}\mleft(\psi_i-\psi_k\mright) + \bigO[d,\mathbf{r},\epsilon][n]{n^{-2}} \, .
\end{equation}
This expansion also allows other figures of merit to be evaluated directly. Following the discussion in the main text on robustness (see~\cref{sec:general_QPA}), we obtain e.g., the trace distance asymptotics
\begin{equation}
\frac{1}{2}\|\rho^\prime-\psi_k\|_1
 =
 \sum_{i\neq k}\frac{1}{\Delta_{k,i}}\frac{p_{i}}{D_{k,i}} + \bigO[d,\mathbf{r},\epsilon][n]{n^{-2}} \, .
\end{equation}

With this picture in place, we turn to the case of constant $m>1$, as characterized by the following theorem.
\begin{theorem}[Intensive-output all-site asymptotics]
\label{thm:constant_output_intensive_sector_all_site_expansion}
For all sufficiently large \(n\), the overhang removal irrep \(\yd{\mu}\) removes all boxes from row \(k\), and
\begin{equation}
    \res{^kf_\mathrm{all}}{\yd{\varsigma}}{\yd{m}}{\yd{\mu}} = 1 - \sum_{i\neq k}\frac{m}{\Delta_{k,i}}\frac{p_i}{D_{k,i}} + \bigO[d,m,\boldsymbol p,\epsilon][n]{n^{-2}} \, .
\end{equation}
Consequently, the overall all-site utility of the overhang-removal protocol satisfies
\begin{equation}
\label{equ:constant_output_intensive_overall_all_site_expansion}
{}^k\mathcal F_{\mathrm{all}}
=1-\frac{m}{n}\sum_{i\neq k}\frac{p_i}{D_{k,i}^2}
+\lilo[d,\boldsymbol p,m]{n^{-1}} \, .
\end{equation}
If $m\geq 1$ boxes are removed from an off-target row $\ell\neq k$, the contribution is of order $\bigO[d,m,\mathbf{r},\epsilon][n]{n^{-1}}$.
\end{theorem}

\begin{proof}
For all sufficiently large $n$, we have $m\leq \Delta_{k,k+1}$, and for $\yd{\mu}$ given by the removal rule all boxes are removed from row $k$. This gives
\begin{equation}\begin{aligned}
    \res{f}{\yd{\mu}}{\yd{m}}{\yd{\varsigma}}_\mathrm{all}(\boldsymbol{t})
    &=
    \prod_{i=1}^{k-1}\prod_{r=1}^m\mleft(1-\frac{t_{i,d}+1}{\widetilde{\Delta}_{i,k}+r}\mright)
    \prod_{i=k}^{d-1}\prod_{r=1}^m\mleft(1-\frac{t_{i,d}}{\widetilde{\Delta}_{k,i+1}-r}\mright) \\
    &=1
    -m\mleft(\sum_{i=1}^{k-1}\frac{t_{i,d}+1}{\Delta_{i,k}}
    +\sum_{i=k}^{d-1}\frac{t_{i,d}}{\Delta_{k,i+1}}\mright) + \bigO[d,m,\mathbf{r},\epsilon][n]{n^{-2}s^{\|\boldsymbol{t}\|_1}} \, .
\end{aligned}\end{equation}
Applying~\cref{thm:uniform_product_expansion_bounds} with the same expansion order $\iota=1$ and $\eta=m(d-1)$, followed by~\cref{cor:weyl_average_main_term_remainder}, proves the overhang removal expansion in the theorem statement.
Applying~\cref{lem:sw_averaging_uniform_sector_expansion} to this sector-wise expansion gives~\cref{equ:constant_output_intensive_overall_all_site_expansion}.
When $m$ boxes are removed from $\ell\neq k$, there will be higher-order intensive terms with an intensive numerator and extensive denominator, resulting in a contribution of order $\bigO[d,m,\mathbf{r},\epsilon][n]{n^{-1}}$.

Otherwise, the resulting contribution is a combination of the cases above and is therefore of order $\bigO[d,m,\mathbf{r},\epsilon][n]{n^{-1}}$.
\end{proof}

\subsection[Asymptotic analysis of the all-site utility]{Asymptotic analysis of the all-site utility}
\label{subsubsec:extensive_supp}
We now study extensive removals for arbitrary $m$. We first derive an asymptotic expression for the utility of the overhang removal rule. Then, for the case of linear extensive output, we compare the overhang-removal rule with other removal rules described by distinct removal rays. For a competing removal rule with $m_n=\Theta(n)$, we define its removal ray $\{R_\ell^\infty\}_{\ell=1}^d$ by $m_{\ell,n}/n\to R_\ell^\infty$ whenever these limits exist for all $\ell=1,\ldots,d$. We show that the overhang rule attains the limiting value in~\cref{equ:extensive_asymptote}, whereas distinct-ray competitors give exponentially suppressed contributions.

\subsubsection{Overhang Removal Ray}
\label{subsubsec:extensive_overhang_ray_supp}

For the overhang removal rule, only rows $k,\ldots,i^\ast$ are used.
\Cref{equ:nm_factorization} reduces to
{\scriptsize\normalsizeequationtag
\begin{equation}
\label{equ:sector_wise_utility_incorporated_conditions}
\res{f}{\yd{\mu}}{\yd{m}}{\yd{\varsigma}}_\mathrm{all}(\boldsymbol{t})=\!\binom{m}{\mathbf m}\!
  \prod_{\ell=k+1}^{{i^\ast}}
  \frac{\rpoch{t_{\ell-1,d}+1}{m_\ell}}{\rpoch{\widetilde{\Delta}_{k,\ell}+1}{m_\ell}}\!
    \prod_{\substack{1\leq i\le k-1\\k\le \ell \le {i^\ast}}}
    \prod_{r=1}^{m_\ell}
    \mleft(1-\frac{t_{i,d}+1}{\widetilde{\Delta}_{i,\ell}+r}\mright)\!
\prod_{\Big\{\substack{k\le \ell \le {i^\ast}-1\\ k\le i\le d-1\\ i\neq \ell-1,\ell}\Big\}\!\bigcup\!\Big\{\substack{\ell={i^\ast}\\ k\le i\le d-1\\ i\neq {i^\ast}-1}\Big\}\!\bigcup\!\Big\{\substack{k\le \ell \le {i^\ast}-1\\ i=\ell}\Big\}}
\!\prod_{r=1}^{m_\ell}
\mleft(1-\frac{t_{i,d}}{\widetilde{\Delta}_{\ell,i+1}-r}\mright).
\end{equation}
}
We split the analysis into three parts: extensive products, delta terms, and intensive numerators.\\

\paragraph{Extensive products}
We first find an asymptotic approximation for the extensive products.

\begin{lemma}[Geometric-type approximation for extensive products]
\label{cor:signed_endpoint_products}
\label{lem:consecutive_endpoint_product}
Fix $s>1$. Let $\epsilon>0$, $c>0$, $R\geq0$, $\delta\geq\epsilon$, and $t\in\mathbb{Z}_{\geq0}$, and assume $m=nR\in\mathbb{Z}_{>0}$, $n\delta\in\mathbb Z$, and $c\in\mathbb Z$. We define the increasing $(+)$ and decreasing $(-)$ extensive products as follows.
For the $(-)$ case, assume further $0\leq R\leq\delta$. Then
\begin{equation}
\label{equ:extensive_products_definition}
P_n^\pm(t)
\coloneqq
\mleft\{
\begin{aligned}
&\prod_{r=1}^{m}
\mleft(
1-\frac{t}{n\delta\pm r+c}
\mright), & 0\le t<A_n^\pm,\\[8pt]
&0, & t\ge A_n^\pm,
\end{aligned}
\mright.
\end{equation}
where we have defined the positive endpoint integers by
\begin{subequations}
\begin{align}
 A_n^+&\coloneqq n\delta+c+1,
&
 B_n^+&\coloneqq n(\delta+R)+c+1,\\
 A_n^-&\coloneqq n(\delta-R)+c,
&
 B_n^-&\coloneqq n\delta+c.
\end{align}
\end{subequations}
Then
\begin{equation}
P_n^\pm(t)
=
\mleft(\frac{A_n^\pm}{B_n^\pm}\mright)^t
+
\bigO[c,\epsilon][n]{n^{-1}s^t},
\end{equation}
uniformly over the corresponding allowed range of $R$ and $\delta$, and uniformly in $t$.
Equivalently,
\begin{subequations}
\begin{align}
P_n^+(t)
&=
\mleft(
\frac{n\delta+c+1}{n(\delta+R)+c+1}
\mright)^t
+
\bigO[c,\epsilon][n]{n^{-1}s^t},\\
P_n^-(t)
&=
\mleft(
\frac{n(\delta-R)+c}{n\delta+c}
\mright)^t
+
\bigO[c,\epsilon][n]{n^{-1}s^t}.
\end{align}
\end{subequations}
\end{lemma}

\begin{proof}
If $A_n^\pm=B_n^\pm$, then $P_n^\pm(t)=1=(A_n^\pm/B_n^\pm)^t$, so the claim is exact.
Assume $A_n^\pm<B_n^\pm$.
For $t=0$, both sides are equal to $1$.
For $1\leq t<A_n^\pm$, we have
\begin{equation}
P_n^\pm(t)
=
\prod_{q=A_n^\pm}^{B_n^\pm-1}\frac{q-t}{q}
=
\prod_{i=1}^{t}\frac{A_n^\pm-i}{B_n^\pm-i}.
\end{equation}
Set $a_i\coloneqq (A_n^\pm-i)/(B_n^\pm-i)$ and $b\coloneqq A_n^\pm/B_n^\pm$.
Then $0\leq a_i,b\leq1$ and
\begin{equation}
a_i-b
=
\frac{i(A_n^\pm-B_n^\pm)}{B_n^\pm(B_n^\pm-i)}.
\end{equation}
Because $i<A_n^\pm$, we have $B_n^\pm-i\geq B_n^\pm-A_n^\pm$, and hence $|a_i-b|\leq i/B_n^\pm$.
Since $\delta\geq\epsilon$ and $c>0$, we have $B_n^\pm=\bigOm[\epsilon][n]{n}$.
Using the product-difference identity \begin{equation}
\prod_{i=1}^{t}a_i-b^t
=
\sum_{j=1}^{t}
\left(
\prod_{i=1}^{j-1}a_i
\right)
(a_j-b)b^{t-j}
\end{equation}
and the fact that \(0\le a_i,b\le 1\), we obtain
\begin{equation}
\mleft|P_n^\pm(t)-\mleft(\frac{A_n^\pm}{B_n^\pm}\mright)^t\mright|
\leq
\bigO[c,\epsilon][n]{n^{-1}t^2}
\leq
\bigO[c,\epsilon][n]{n^{-1}s^t}.
\end{equation}
This gives the required bound for $t<A_n^\pm$.
For $t\geq A_n^\pm$, $P_n^\pm(t)=0$ by definition, and therefore
\begin{equation}
\mleft|P_n^\pm(t)-\mleft(\frac{A_n^\pm}{B_n^\pm}\mright)^t\mright|
=
\mleft(\frac{A_n^\pm}{B_n^\pm}\mright)^t
\leq
s^t\mleft(\frac{A_n^\pm}{sB_n^\pm}\mright)^{A_n^\pm}.
\end{equation}
The bound
\begin{equation}
\sup_{1\leq a\leq B}\mleft(\frac{a}{sB}\mright)^a
\leq
\bigO[][B]{B^{-1}}
\end{equation}
follows because $a\ln(a/(sB))$ is convex on $[1,B]$, so the maximum occurs at an endpoint, and exponential decay is faster than polynomial decay.
Applying this with $a=A_n^\pm$ and $B=B_n^\pm$, and using $B_n^\pm=\bigOm[\epsilon][n]{n}$, gives the desired bound for $t\geq A_n^\pm$.
This proves the claim.
\end{proof}

\begin{corollary}[Normalized approximation for extensive products]
\label{cor:normalized_endpoint_product_bases}
In the setting of~\cref{cor:signed_endpoint_products}, uniformly over the corresponding allowed range of $R$ and $\delta$, and uniformly in $t$,
\begin{subequations}
\begin{align}
P_n^+(t)
&=
\mleft(
\frac{\delta}{\delta+R}
\mright)^t
+
\bigO[c,\epsilon][n]{n^{-1}s^t},\\
P_n^-(t)
&=
\mleft(
\frac{\delta-R}{\delta}
\mright)^t
+
\bigO[c,\epsilon][n]{n^{-1}s^t}.
\end{align}
\end{subequations}
Here we use the convention $0^0=1$, so that if the normalized base vanishes, its $t$-th power is interpreted as $\delta_t$.
\end{corollary}

\begin{proof}
For the increasing product, set
\begin{equation}
G_n^+\coloneqq \frac{n\delta+c+1}{n(\delta+R)+c+1},
\qquad
G^+\coloneqq \frac{\delta}{\delta+R}.
\end{equation}
Since $\delta\geq\epsilon$ and $c>0$, we have $|G_n^+ - G^+|=\bigO[c,\epsilon][n]{n^{-1}}$, and $0\leq G_n^+,G^+\leq1$.
Therefore
\begin{equation}
\mleft|(G_n^+)^t-(G^+)^t\mright|
\leq
t|G_n^+ - G^+|
\leq
\bigO[c,\epsilon][n]{n^{-1}s^t}.
\end{equation}
Combining this with~\cref{cor:signed_endpoint_products} gives the first claim.

For the decreasing product, set
\begin{equation}
G_n^-\coloneqq \frac{n(\delta-R)+c}{n\delta+c},
\qquad
G^-\coloneqq \frac{\delta-R}{\delta}.
\end{equation}
Since $\delta\geq\epsilon$ and $c>0$, we have $|G_n^- - G^-|=\bigO[c,\epsilon][n]{n^{-1}}$.
Moreover, $0\leq G_n^-,G^-\leq1$ on the allowed range $0\leq R\leq\delta$.
Hence
\begin{equation}
\mleft|(G_n^-)^t-(G^-)^t\mright|
\leq
t|G_n^- - G^-|
\leq
\bigO[c,\epsilon][n]{n^{-1}s^t}.
\end{equation}
Combining this with~\cref{cor:signed_endpoint_products} proves the second claim.
\end{proof}

\paragraph{Delta terms}
When $c=1$ in the $P_n^-(t)$ expressions above, we obtain the following exact identity:
\begin{corollary}[Vanishing product for full removal]
\label{lem:full_removal_product}
Let $\Delta\in\mathbb{Z}_{\ge 1}$ and let $t\in\mathbb{Z}$ with $0\le t\le \Delta$. Then
\begin{equation}
\prod_{r=1}^{\Delta}\mleft(1-\frac{t}{\Delta-r+1}\mright)
=
\prod_{q=1}^{\Delta}\mleft(1-\frac{t}{q}\mright)
=
\frac{\rpoch{-t}{\Delta}}{\Delta!}
=
\delta_t,
\end{equation}
\end{corollary}

By~\cref{lem:full_removal_product}, the factors with the constraint $k\le \ell \le {i^\ast}-1,\, i=\ell$ reduce to $\delta_{t_{\ell,d}}$. This allows us to set all edge variables $t_{\ell,d}$ with $k\le\ell\le {i^\ast}-1$ to zero and restrict attention to the remaining variables. Consequently, the remaining extensive constraints reduce to $k\le \ell\le {i^\ast}$ and ${i^\ast}\le i \le d-1$.

We apply~\cref{cor:signed_endpoint_products} to the remaining extensive products.
Fix \(s>1\) such that \(sr_\ast<1\), where \(r_\ast=\max_{a,b}r_{a,b}\).
For the increasing product, we obtain, uniformly over the allowed gaps and removal rates,
\begin{equation}
\begin{aligned}
\prod_{r=1}^{m_\ell}
\left(
1-\frac{t_{i,d}+1}{\widetilde{\Delta}_{i,\ell}+r}
\right)
&=
\left(
\frac{\widetilde{\Delta}_{i,\ell}+1}
{\widetilde{\Delta}_{i,\ell}+m_\ell+1}
\right)^{t_{i,d}+1}
\mkern-2mu+
\bigO[d,\epsilon][n]{n^{-1}s^{t_{i,d}+1}} \\
&=
\left(
\frac{\overline{\Delta}_{i,\ell}}
{\overline{\Delta}_{i,\ell}+R_\ell}
\right)^{t_{i,d}+1}
\mkern-2mu+
\bigO[d,\epsilon][n]{n^{-1}s^{t_{i,d}+1}} .
\end{aligned}
\end{equation}
For the decreasing product, a factor vanishes whenever $t_{i,d}\geq A_n^-$. Thus it agrees with the cutoff convention in~\cref{equ:extensive_products_definition}, and the endpoint estimate applies uniformly in $t_{i,d}$.
\begin{equation}
\begin{aligned}
\prod_{r=1}^{m_\ell}
\left(
1-\frac{t_{i,d}}{\widetilde{\Delta}_{\ell,i+1}-r}
\right)
&=
\left(
\frac{\widetilde{\Delta}_{\ell,i+1}-m_\ell}
{\widetilde{\Delta}_{\ell,i+1}}
\right)^{t_{i,d}}
\mkern-2mu+
\bigO[d,\epsilon][n]{n^{-1}s^{t_{i,d}}} \\
&=
\left(
\frac{\overline{\Delta}_{\ell,i+1}-R_\ell}
{\overline{\Delta}_{\ell,i+1}}
\right)^{t_{i,d}}
\mkern-2mu+
\bigO[d,\epsilon][n]{n^{-1}s^{t_{i,d}}}.
\end{aligned}
\end{equation}
Here the second equality in each equation uses that all additive shifts are \(O_d(1)\) and all normalized gaps are bounded below by \(\epsilon\).

\paragraph{Intensive numerators}
We now focus on the factors with intensive numerators, denoted by $I_n(\boldsymbol{t})$, namely
\begin{equation}
I_n(\boldsymbol t)=\binom{m}{\mathbf m}\;
  \prod_{\ell=k+1}^{{i^\ast}}
  \frac{\rpoch{t_{\ell-1,d}+1}{m_\ell}}{\rpoch{\widetilde{\Delta}_{k,\ell}+1}{m_\ell}}.
\end{equation}
By part (b), the delta constraints set $t_{j,d}=0$ for $j\in\{k,\ldots,i^\ast-1\}$.
On their support, each intensive numerator satisfies $\rpoch{t_{\ell-1,d}+1}{m_\ell}=m_\ell!$, so these factors cancel the corresponding factors in the multinomial denominator.
Thus
\begin{equation}
I_n(\boldsymbol{t})\prod_{j=k}^{i^\ast-1}\delta_{t_{j,d}}
=
C_n\prod_{j=k}^{i^\ast-1}\delta_{t_{j,d}},
\end{equation}
where
\begin{equation}
C_n
=
\frac{m!}{m_k!}
\prod_{\ell=k+1}^{i^\ast}
\frac{1}{\rpoch{\widetilde{\Delta}_{k,\ell}+1}{m_\ell}}.
\end{equation}
\begin{lemma}[Intensive prefactor under overhang removal]
\label{lem:intensive_numerators_overhang}
Let $d$, $k$, $m$, $i^\ast$, and $\mathbf m$ be as defined above.
Assume $\overline{\yd{\varsigma}}\in\mathsf{R}_\epsilon$ and define $R\coloneqq m/n$ and $\overline{\Delta}_{k,\ell}\coloneqq \Delta_{k,\ell}/n$.
Then
\begin{equation}
\label{eq:intensive_numerator_cn_asymptotic}
C_n
=
\prod_{\ell=k+1}^{i^\ast}\frac{\overline{\Delta}_{k,\ell}}{R}
+
\bigO[d,\epsilon][n]{n^{-1}}.
\end{equation}
\end{lemma}

\begin{proof}
Writing the denominator products as ratios of factorials gives
\begin{equation}
C_n
=
\frac{m!}{m_k!}
\prod_{\ell=k+1}^{i^\ast}
\frac{\widetilde{\Delta}_{k,\ell}!}
{(\widetilde{\Delta}_{k,\ell}+m_\ell)!}.
\end{equation}
If $m_k=m$, then $i^\ast=k$ and the product is empty, therefore, $C_n=1$.
Thus the claim is immediate.

Assume now that $m>m_k$; then $i^\ast>k$. By the overhang removal rule, $m_k=\Delta_{k,k+1}$, $m_\ell=\Delta_{\ell,\ell+1}$ for $k+1\leq \ell<i^\ast$, and $m_{i^\ast}=m-\Delta_{k,i^\ast}$. Hence, $\Delta_{k,\ell}+m_\ell=\Delta_{k,\ell+1}$ for $k+1\leq \ell<i^\ast$, while the terminal row gives $\Delta_{k,i^\ast}+m_{i^\ast}=m$.
Therefore
\begin{equation}
\begin{aligned}
C_n
&=
\frac{m!}{\Delta_{k,k+1}!}
\frac{\widetilde{\Delta}_{k,k+1}!}{(\widetilde{\Delta}_{k,k+2}-1)!}
\frac{\widetilde{\Delta}_{k,k+2}!}{(\widetilde{\Delta}_{k,k+3}-1)!}
\cdots
\frac{\widetilde{\Delta}_{k,i^\ast}!}{(m+i^\ast-k)!} \\
&=
\frac{
(\widetilde{\Delta}_{k,k+1})
(\widetilde{\Delta}_{k,k+2})
\cdots
(\widetilde{\Delta}_{k,i^\ast})
}{
(m+1)(m+2)\cdots(m+i^\ast-k)
}\\
&=
\prod_{\ell=k+1}^{i^\ast}
\frac{\widetilde{\Delta}_{k,\ell}}
{m-k+\ell}
\end{aligned}
\end{equation}
Since $\overline{\yd{\varsigma}}\in\mathsf{R}_\epsilon$ gives $\overline{\Delta}_{k,\ell}\geq\epsilon$, while $i^\ast>k$ gives $R=m/n\geq m_k/n=\overline{\Delta}_{k,k+1}\geq\epsilon$, the first-order expansion is uniform, proving~\cref{eq:intensive_numerator_cn_asymptotic}.
\end{proof}

Assembling the pieces, we obtain the final expression for the utility component:
\begin{equation}
\label{equ:sector_wise_utility_assembled}
\begin{aligned}
    \res{f}{\yd{\mu}}{\yd{m}}{\yd{\varsigma}}_\mathrm{all}(\boldsymbol{t})
    =&\prod_{i=k}^{i^\ast-1}\frac{\overline\Delta_{k,i+1}\delta_{t_{i,d}}}{R}
    \prod_{\substack{k\le \ell \le {i^\ast}\\
    1\leq i\le k-1}}
    \Bigl(\frac{\overline{\Delta}_{i,\ell}}{\overline{\Delta}_{i,\ell}+R_\ell}\Bigr)^{t_{i,d}+1}
    \prod_{\substack{k\le \ell\le {i^\ast}\\i^\ast\le i \le d-1}}
    \Bigl(\frac{\overline{\Delta}_{\ell,i+1}-R_\ell}{\overline{\Delta}_{\ell,i+1}}\Bigr)^{t_{i,d}} + \bigO[d,\epsilon,s][n]{n^{-1}s^{\|\boldsymbol{t}\|_1}}.
\end{aligned}
\end{equation}

The pointwise remainders have the envelope \(\bigO[d,\epsilon][n]{n^{-1}s^{\|\boldsymbol{t}\|_1}}\).
\noindent We remark here that
\begin{equation}
    \mleft\langle c^{t_{i,d}} \mright\rangle_{\mathrm{G}(\boldsymbol{r})} = \frac{1-r_{i,d}}{1-cr_{i,d}} \, , \quad \mleft\langle \delta_{t_{i,d}} \mright\rangle_{\mathrm{G}(\boldsymbol{r})} = 1-r_{i,d}
\end{equation}
for $c$ such that $|cr_{i,d}|<1$, and that the individual terms in the product above are independent under the distribution $\mathrm{G}(\boldsymbol{r})$. Thus,
\begin{equation}
\begin{aligned}
    &\mleft\langle \prod_{i=k}^{i^\ast-1}\delta_{t_{i,d}} \prod_{\substack{k\le \ell \le {i^\ast}\\
    1\leq i\le k-1}}
    \Bigl(\frac{\overline{\Delta}_{i,\ell}}{\overline{\Delta}_{i,\ell}+R_\ell}\Bigr)^{t_{i,d}+1}
    \prod_{\substack{k\le \ell\le {i^\ast}\\i^\ast\le i \le d-1}}
    \Bigl(\frac{\overline{\Delta}_{\ell,i+1}-R_\ell}{\overline{\Delta}_{\ell,i+1}}\Bigr)^{t_{i,d}} \mright\rangle_{\mathrm{G}(\boldsymbol{r})} \\
    =& \prod_{i=k}^{i^\ast-1}(1-r_{i,d}) \prod_{i=1}^{k-1}
  \mleft\langle\Bigl(\prod_{\ell = k}^{i^\ast}\frac{\overline{\Delta}_{i,\ell}}{\overline{\Delta}_{i,\ell}+R_\ell}\Bigr)^{t_{i,d}+1}\mright\rangle_{\mathrm{G}(\boldsymbol{r})}
    \prod_{i=i^\ast}^{d-1}
    \mleft\langle\Bigl(\prod_{\ell=k}^{i^\ast} \frac{\overline{\Delta}_{\ell,i+1}-R_\ell}{\overline{\Delta}_{\ell,i+1}}\Bigr)^{t_{i,d}}\mright\rangle_{\mathrm{G}(\boldsymbol{r})} \\
    = &\prod_{i=k}^{i^\ast-1}(1-r_{i,d}) \prod_{i=1}^{k-1}
    \frac{(1-r_{i,d})\prod_{\ell = k}^{i^\ast}\frac{\overline{\Delta}_{i,\ell}}{\overline{\Delta}_{i,\ell}+R_\ell}}{
    1-r_{i,d}\prod_{\ell = k}^{i^\ast}\frac{\overline{\Delta}_{i,\ell}}{\overline{\Delta}_{i,\ell}+R_\ell}}
    \prod_{i={i^\ast}}^{d-1}
    \frac{1-r_{i,d}}{
    1-r_{i,d}\prod_{\ell=k}^{i^\ast} \frac{\overline{\Delta}_{\ell,i+1}-R_\ell}{\overline{\Delta}_{\ell,i+1}}}\, .
\end{aligned}
\end{equation}

The main term in~\cref{equ:sector_wise_utility_assembled} is bounded by \(\bigO[d,\epsilon,s]{s^{\|\boldsymbol{t}\|_1}}\).
Together with the remainder bound in~\cref{equ:sector_wise_utility_assembled}, this verifies the hypotheses of~\cref{cor:weyl_average_main_term_remainder} with \(\alpha_n=n^{-1}\).
Therefore,
\begin{equation}
\begin{aligned}
    ^k\res{f_\mathrm{all}}{\yd{\varsigma}}{\yd{m}}{\yd{\mu}}
    =\prod_{i=k}^{i^\ast-1}\frac{\overline\Delta_{k,i+1}(1-r_{i,d})}{R}
    \prod_{i=1}^{k-1}
    \frac{1-r_{i,d}}{
    \prod_{\ell=k}^{i^\ast}\mleft(1+\frac{R_\ell}{\overline{\Delta}_{i,\ell}}\mright)-r_{i,d}}
    \prod_{i={i^\ast+1}}^{d}
    \frac{1-r_{i-1,d}}{
    1-\prod_{\ell=k}^{{i^\ast}}\mleft(1-\frac{R_\ell}{\overline{\Delta}_{\ell,i}}\mright)\,r_{i-1,d}}+\bigO[d,\boldsymbol{r},\epsilon][n]{n^{-1}}.
\end{aligned}
\end{equation}
To express this in terms of the sorted spectrum $\boldsymbol{p}$, we remark that
\begin{equation}
    r_{i,d} =
    \mleft\{\begin{aligned}
        &\frac{p_{k}}{p_{i}}, & 1\leq i\leq k-1, \\
        &\frac{p_{i+1}}{p_{k}}, & k\leq i \leq d-1.
    \end{aligned}\mright.
\end{equation}
Thus,
\begin{equation}
    ^k\res{f_\mathrm{all}}{\yd{\varsigma}}{\yd{m}}{\yd{\mu}} =
    \prod_{i=k}^{i^\ast-1}
    \frac{\overline\Delta_{k,i+1}D_{k,i+1}}{p_{k}R}
    \prod_{i\notin \{k,\ldots,i^\ast\}} \frac{D_{k,i}}{p_{k} - \prod_{\ell=k}^{i^\ast}\mleft(1-\frac{R_\ell}{\overline{\Delta}_{\ell,i}}\mright)p_{i}}
    +\bigO[d,\boldsymbol{p},\epsilon][n]{n^{-1}}.
\end{equation}
We insert the overhang removal rules from~\cref{def:overhang_removal_rule}, which gives
\begin{equation}
    R_{\ell} =
    \mleft\{\begin{aligned}
        &0, & \ell<k, \\
        &\overline{\Delta}_{\ell,\ell+1}, & k\leq \ell <i^\ast, \\
        &R-\overline{\Delta}_{k,i^\ast}, & \ell = i^\ast, \\
        &0, & \ell>i^\ast.
    \end{aligned}\mright.
\end{equation}
Using the telescopic product in
\begin{equation}
    \prod_{\ell=k}^{i^\ast}\mleft(1-\frac{R_\ell}{\overline{\Delta}_{\ell,i}}\mright) = \frac{\overline{\Delta}_{i^\ast,i} - R + \overline{\Delta}_{k,i^\ast}}{\overline{\Delta}_{i^\ast,i}} \prod_{\ell=k}^{i^\ast-1}\mleft(\frac{\overline{\Delta}_{\ell,i} - \overline{\Delta}_{\ell,\ell+1}}{\overline{\Delta}_{\ell,i}}\mright) = \frac{\overline{\Delta}_{k,i} - R}{\overline{\Delta}_{i^\ast,i}} \prod_{\ell=k}^{i^\ast-1}\mleft(\frac{\overline{\Delta}_{\ell+1,i}}{\overline{\Delta}_{\ell,i}}\mright) = \frac{\overline{\Delta}_{k,i} - R}{\overline{\Delta}_{k,i}} \, .
\end{equation}
This gives
\begin{equation}
\label{eq:overhang_leading_branch}
^k\res{f_\mathrm{all}}{\yd{\varsigma}}{\yd{m}}{\yd{\mu}} =
    \prod_{i=k+1}^{i^\ast}
    \frac{\overline\Delta_{k,i}D_{k,i}}{p_{k}R}
    \prod_{i\notin \{k,\ldots,i^\ast\}} \frac{D_{k,i}}{D_{k,i}+\frac{p_iR}{\overline\Delta_{k,i}}}
    +\bigO[d,\boldsymbol{p},\epsilon][n]{n^{-1}} \, ,
\end{equation}

\begin{lemma}[Continuity of the overhang main term]
\label{lem:overhang_leading_continuity}
For fixed $R\geq0$, let $C(\overline{\yd{\varsigma}})$ be the main term in~\cref{eq:overhang_leading_branch}. Then $C(\overline{\yd{\varsigma}})$ is continuous in $\overline{\yd{\varsigma}}\in\mathsf R_\epsilon$.
\end{lemma}

\begin{proof}
For fixed terminal index $i^\ast$, the main term is a rational function of the normalized gaps $\overline\Delta_{k,i}=\overline\varsigma_k-\overline\varsigma_i$. On $\mathsf R_\epsilon$, all denominators are nonzero, so $C(\overline{\yd{\varsigma}})$ is continuous within each branch.

It remains to check continuity when the terminal index changes. A branch boundary occurs at $\overline\Delta_{k,i}=R$.
The new prefactor at the boundary evaluates to
\begin{equation}
\label{eq:overhang_main_term_boundary_prefactor}
\frac{\overline\Delta_{k,i}D_{k,i}}{p_kR}
=
\frac{D_{k,i}}{p_k}.
\end{equation}
On the adjacent branch with terminal index $i-1$, the corresponding factor satisfies
\begin{equation}
\label{eq:overhang_main_term_adjacent_branch_limit}
\frac{D_{k,i}}
{D_{k,i}+\dfrac{p_iR}{\overline\Delta_{k,i}}}
\longrightarrow
\frac{D_{k,i}}{D_{k,i}+p_i}
=
\frac{D_{k,i}}{p_k}
\qquad
\text{as }\overline\Delta_{k,i}\to R.
\end{equation}
Thus the adjacent branch converges to the value of the branch containing the boundary point. Since all remaining factors are continuous and unchanged across the boundary, the branch formulas glue continuously. Hence $C(\overline{\yd{\varsigma}})$ is continuous on $\mathsf R_\epsilon$.
\end{proof}

\begin{lemma}[Uniform Lipschitz bound for the overhang main term]
\label{lem:overhang_leading_uniform_lipschitz}
For $R\geq0$, write $C_R(\overline{\yd{\varsigma}})$ for the main term in~\cref{eq:overhang_leading_branch}. Then $C_R(\overline{\yd{\varsigma}})$ is uniformly Lipschitz in $\overline{\yd{\varsigma}}\in\mathsf R_\epsilon$, uniformly over $R\geq0$: there exists $L_{d,\boldsymbol p,\epsilon}<\infty$ such that, for every $R\geq0$ and $\overline{\yd{\varsigma}},\overline{\yd{\varsigma^\prime}}\in\mathsf R_\epsilon$,
\begin{equation}
\label{eq:overhang_main_term_uniform_lipschitz}
\left|
C_R\mleft(\overline{\yd{\varsigma}}\mright)
-
C_R\mleft(\overline{\yd{\varsigma^\prime}}\mright)
\right|
\leq
L_{d,\boldsymbol p,\epsilon}
\left\|
\overline{\yd{\varsigma}}
-
\overline{\yd{\varsigma^\prime}}
\right\|_2.
\end{equation}
\end{lemma}

\begin{proof}
Fix $R\geq0$ and a branch with terminal index $i^\ast$. The two types of factors in $C_R$ are
\begin{equation}
\label{eq:overhang_main_term_factor_types}
a_i(\Delta,R)
\coloneqq
\frac{\Delta D_{k,i}}{p_kR},
\qquad
b_i(\Delta,R)
\coloneqq
\frac{D_{k,i}\Delta}{D_{k,i}\Delta+p_iR}.
\end{equation}
A factor $a_i$ occurs only if $i^\ast>k$, in which case the branch conditions give $R\geq\overline\Delta_{k,k+1}\geq\epsilon$. Hence
\begin{equation}
\label{eq:overhang_main_term_prefactor_derivative}
\left|
\frac{\partial a_i}{\partial\Delta}(\Delta,R)
\right|
=
\frac{|D_{k,i}|}{p_kR}
\leq
\frac{|D_{k,i}|}{p_k\epsilon}.
\end{equation}
For a factor $b_i$, the relevant indices satisfy $D_{k,i}\Delta>0$. Setting $x\coloneqq|D_{k,i}\Delta|$ and $y\coloneqq p_iR$, we obtain
\begin{equation}
\label{eq:overhang_main_term_remaining_factor_derivative}
\left|
\frac{\partial b_i}{\partial\Delta}(\Delta,R)
\right|
=
\frac{|D_{k,i}|p_iR}{\mleft(D_{k,i}\Delta+p_iR\mright)^2}
=
\frac{1}{|\Delta|}\frac{xy}{(x+y)^2}
\leq
\frac{1}{4|\Delta|}
\leq
\frac{1}{4\epsilon}.
\end{equation}
Both derivative bounds are independent of $R$. Since all factors in~\cref{eq:overhang_leading_branch} lie in $[0,1]$ and there are at most $d-1$ factors, we have, on every branch,
\begin{equation}
\label{eq:overhang_main_term_branch_gradient_bound}
\left\|
\nabla_{\overline{\yd{\varsigma}}}
C_R\mleft(\overline{\yd{\varsigma}}\mright)
\right\|_2
\leq
L_{d,\boldsymbol p,\epsilon},
\end{equation}
with the same constant for every $R\geq0$ and every branch. By~\cref{lem:overhang_leading_continuity}, neighboring branch formulas agree at their common boundary. Applying the branchwise gradient bound along the line segment joining $\overline{\yd{\varsigma}}$ and $\overline{\yd{\varsigma^\prime}}$, split at the finitely many branch boundaries it crosses, proves~\cref{eq:overhang_main_term_uniform_lipschitz}.
\end{proof}

\begin{theorem}[Extensive overhang all-site utility]
\label{thm:extensive_overhang_all_site_fidelity}
Let $R\coloneqq m/n$ and $I^\ast\coloneqq\min\{i\mid D_{k,i+1}>R\}$.
For the overhang removal rule,
\begin{equation}
\label{eq:sectorwise_fidelity_plug_defs}
{}^k\mathcal{F}_\mathrm{all}
=
\prod_{i=k+1}^{I^\ast}\frac{D_{k,i}^2}{p_kR}
\prod_{i\notin\{k,\ldots,I^\ast\}}
\frac{D_{k,i}^2}{D_{k,i}^2+p_iR}
+\lilo[d,\boldsymbol p][n]{1}.
\end{equation}
\end{theorem}

\begin{proof}
Set $R=m/n$ in~\cref{eq:overhang_leading_branch}. Then
\begin{equation}
^k\res{f_\mathrm{all}}{\yd{\varsigma}}{\yd{m}}{\yd{\mu}}
=
\prod_{i=k+1}^{i^\ast}
\frac{\overline\Delta_{k,i}D_{k,i}}{p_kR}
\prod_{i\notin \{k,\ldots,i^\ast\}}
\frac{D_{k,i}}{D_{k,i}+\frac{p_iR}{\overline\Delta_{k,i}}}
+
\lilo[d,\boldsymbol{p},\epsilon][n]{1}.
\end{equation}
Here $i^\ast\coloneqq\min\{i\mid\overline\Delta_{k,i+1}>R\}$ is the sector-wise terminal index. Under Schur--Weyl concentration, this becomes the macroscopic terminal index $I^\ast\coloneqq\min\{i\mid D_{k,i+1}>R\}$.
By~\cref{lem:overhang_leading_uniform_lipschitz}, the main-term expression is uniformly Lipschitz in $\overline{\yd{\varsigma}}$ on $\mathsf R_\epsilon$, so after averaging we may evaluate it at the spectrum $\boldsymbol p$.
Applying the concentration argument in~\cref{subsec:concentration_supp} gives~\cref{eq:sectorwise_fidelity_plug_defs}.
\end{proof}

\subsubsection{Distinct-ray competitors}
\label{subsubsec:extensive_distinct_ray_supp}

We study removal rules with removal rays different from the overhang removal ray, and show that they result in sector-wise utility negligible in the large-$n$ limit. First, consider any removal rule that removes \(\bigTheta{n}\) boxes from some row \(\ell<k\), so that \(R^\infty_\ell>0\). Then, the utility component contains the factor \(\rpoch{-t_{\ell,d}}{m_\ell}\), which vanishes for every fixed \(t_{\ell,d}\) once \(m_\ell>t_{\ell,d}\). Since \(m_\ell\to\infty\), the terms $\res{f}{\yd{\mu}}{\yd{m}}{\yd{\varsigma}}_\mathrm{all}(\boldsymbol{t})$ tend pointwise to zero. Since the utility component is bounded by \(1\), DCT shows that the sector-wise utility also tends to zero.

We henceforth consider removal rules with $R^\infty_\ell=0$ for every $\ell<k$; these rows may still have $m_\ell=\lilo[][n]{n}$. Define the cumulative limiting removals $S^\infty_\ell\coloneqq \sum_{j=k}^{\ell}R^\infty_j$.
\begin{lemma}[Polynomial bound for signed extensive products]
\label{lem:polynomial_bound_signed_extensive_products}
Fix \(C>0\), let \(t,c\in\mathbb Z\), and let the positive integer \(m\) satisfy \(m\leq Cn\).
Assume that \(n\delta+c\in\mathbb Z\), and define
\begin{equation}
P_n
:=
\prod_{r=1}^{m}
\left(
1-\frac{t}{n\delta+c+r}
\right).
\end{equation}
If \(n\delta+c+r\neq0\) for all \(1\le r\le m\), then $|P_n|
=
\bigO[t,C][n]{n^{2|t|}}$, uniformly over all such \(m\).
\end{lemma}

\begin{proof}
If one factor is $0$, the result is immediate.
Otherwise set \(q_r=n\delta+c+r\).
Since \(n\delta+c\in\mathbb Z\), the numbers \(q_1,\ldots,q_m\) are consecutive nonzero integers.
Let \(B=2|t|+1\).
There are at most \(2B+1=\bigO[t][n]{1}\) indices \(r\) for which \(|q_r|\le B\).
For these indices, since \(q_r\neq0\), we have $\left|
1-\frac{t}{q_r}
\right|
\le
1+|t|,$
so all small-denominator factors contribute at most \(\bigO[t][n]{1}\).

For the remaining factors, \(|q_r|>B\), and therefore
\begin{equation}
\ln\left|
1-\frac{t}{q_r}
\right|
\le
\ln\left(
1+\frac{|t|}{|q_r|}
\right)
\le
\frac{|t|}{|q_r|}.
\end{equation}
Because the \(q_r\)'s are consecutive integers, the harmonic sum is bounded by
\begin{equation}
\sum_{r:\,|q_r|>B}\frac{1}{|q_r|}
\le
2\sum_{s=1}^{m+B}\frac1s
\le
2\ln n+\bigO[t,C][n]{1}.
\end{equation}
It follows that $\ln |P_n|\le 2|t|\ln n+\bigO[t,C][n]{1}$. Exponentiating gives \(|P_n|=\bigO[t,C][n]{n^{2|t|}}\), uniformly over all \(m\leq Cn\).
\end{proof}

\begin{lemma}[Global logarithmic Stirling bound]
\label{lem:global_log_stirling}
Let $g\colon\mathbb R_{\geq0}\to\mathbb R$ be given by $g(0):=0$ and $g(x):=x\ln x$ for $x>0$. Let $y=y(n)\in\mathbb Z_{\geq0}$. Then
\begin{equation}
\label{eq:global_log_stirling}
\ln(y!)
=
g(y)-y
+
\bigO[][n]{\ln(y+1)},
\end{equation}
The implicit constant is absolute and independent of $y$.
\end{lemma}

\begin{proof}
For all sufficiently large $y$, Stirling's formula gives
\begin{equation}
\label{eq:global_log_stirling_large_y}
\ln(y!)
=
g(y)-y
+
\bigO[][n]{\ln y},
\end{equation}
Since $y\in\mathbb Z_{\geq0}$, enlarging the implicit constant covers the finitely many positive integers below this threshold. The case $y=0$ holds exactly.
\end{proof}

\begin{lemma}[Exponential rate of the intensive numerator]
\label{lem:multinomial_block_exponential_rate}
Fix \(\boldsymbol{R}^\infty\dot{=}(R^\infty_1,\ldots,R^\infty_d)\) with \(R^\infty_\ell\geq0\), and assume that positive-rate rows are confined to \(\ell\geq k\), i.e., $R^\infty_\ell=0$ for $\ell<k$. Let \(m_\ell\) be nonnegative integers satisfying
\begin{equation}
m_\ell=nR^\infty_\ell+\lilo[][n]{n}.
\end{equation}
Set \(R^\infty\coloneqq\sum_{\ell=1}^{d}R^\infty_\ell\) and \(m\coloneqq\sum_{\ell=1}^{d}m_\ell\).
For \(\ell>k\), fix \(\Delta_{k,\ell}\) and \(\delta^\infty_{k,\ell}>0\) satisfying
\begin{equation}
\Delta_{k,\ell}=n\delta^\infty_{k,\ell}+\lilo[][n]{n}.
\end{equation}
Then
\begin{equation}
I_n(\boldsymbol t)=\binom{m}{\mathbf m}
\prod_{\ell=1}^{k-1}
\frac{\rpoch{-t_{\ell,d}}{m_\ell}}
{\rpoch{\widetilde\Delta_{k,\ell}+1}{m_\ell}}
\prod_{\ell=k+1}^{d}
\frac{\rpoch{t_{\ell-1,d}+1}{m_\ell}}
{\rpoch{\widetilde\Delta_{k,\ell}+1}{m_\ell}}\leq e^{n\Phi+\lilo[\boldsymbol t,d][n]{n}},
\end{equation}
where \(\Phi\) is given by
\begin{equation}
\label{eq:Phi_general_eta}
\Phi
=
g(R^\infty)-g(R^\infty_k)
+
\sum_{\ell=k+1}^{d}
\mleft(
g(\delta^\infty_{k,\ell})-g(\delta^\infty_{k,\ell}+R^\infty_\ell)
\mright),
\end{equation}
\end{lemma}

\begin{proof}
Set $R_\ell:=m_\ell/n$ and $R:=m/n=\sum_{\ell=1}^{d}R_\ell$.
Since \(m=\bigO[][n]{n}\), \cref{lem:global_log_stirling} gives, for every \(\ell\),
\begin{equation}
\ln(m_\ell!)
=
m_\ell\ln n
+
n g(R_\ell)
-
m_\ell
+
\bigO[][n]{\ln n},
\end{equation}
and similarly for \(m\). These bounds hold without distinguishing zero, bounded, sublinear, or extensive \(m_\ell\). Using \(m=\sum_{\ell=1}^{d}m_\ell\), the logarithmic and linear terms cancel exactly, giving
\begin{equation}
\ln\binom{m}{\mathbf m}
=
n g(R)
-
n\sum_{\ell=1}^{d}g(R_\ell)
+
\bigO[d][n]{\ln n}.
\end{equation}
Since \(R_\ell\to R^\infty_\ell\), \(R^\infty_\ell=0\) for \(1\leq\ell<k\), \(R\to R^\infty\), and \(g\) is continuous,
\begin{equation}
\ln\binom{m}{\mathbf m}
=
n\mleft(
g(R^\infty)
-
\sum_{\ell=k}^{d}g(R^\infty_\ell)
\mright)
+
\lilo[d][n]{n}.
\end{equation}
For each \(\ell>k\), set \(\delta_{k,\ell}\coloneqq\Delta_{k,\ell}/n\). Then \cref{lem:global_log_stirling} gives
\begin{equation}
\ln\rpoch{t_{\ell-1,d}+1}{m_\ell}
=
m_\ell\ln n
+
n g(R_\ell)
-
m_\ell
+
\bigO[\boldsymbol t,d][n]{\ln n}.
\end{equation}
Similarly,
\begin{equation}
\ln\rpoch{\widetilde\Delta_{k,\ell}+1}{m_\ell}
=
m_\ell\ln n
+
n\mleft(
g(\delta_{k,\ell}+R_\ell)
-
g(\delta_{k,\ell})
\mright)
-
m_\ell
+
\bigO[d][n]{\ln n}.
\end{equation}
Therefore,
\begin{equation}
\ln\frac{\rpoch{t_{\ell-1,d}+1}{m_\ell}}
{\rpoch{\widetilde\Delta_{k,\ell}+1}{m_\ell}}
=
n\mleft(
g(R_\ell)
+
g(\delta_{k,\ell})
-
g(\delta_{k,\ell}+R_\ell)
\mright)
+
\bigO[\boldsymbol t,d][n]{\ln n}.
\end{equation}
Passing to \(R_\ell\to R^\infty_\ell\) and \(\delta_{k,\ell}\to\delta^\infty_{k,\ell}\) gives
\begin{equation}
\ln\frac{\rpoch{t_{\ell-1,d}+1}{m_\ell}}
{\rpoch{\widetilde\Delta_{k,\ell}+1}{m_\ell}}
=
n\mleft(
g(R^\infty_\ell)
+
g(\delta^\infty_{k,\ell})
-
g(\delta^\infty_{k,\ell}+R^\infty_\ell)
\mright)
+
\lilo[\boldsymbol t,d][n]{n}.
\end{equation}
For each \(\ell<k\), the factor \(\rpoch{-t_{\ell,d}}{m_\ell}/\rpoch{\widetilde{\Delta}_{k,\ell}+1}{m_\ell}\) from~\cref{equ:nm_factorization} vanishes if \(m_\ell>t_{\ell,d}\), in which case the desired bound is immediate. Otherwise \(m_\ell\leq t_{\ell,d}\), so its numerator is \(\bigO[\boldsymbol t][n]{1}\) and the factor contributes at most a polynomial prefactor in \(n\), hence only \(\lilo[\boldsymbol t,d][n]{n}\) to the logarithm. Combining the multinomial estimate with all \(\ell>k\) factors gives
\begin{equation}
\begin{aligned}
\ln I_n(\boldsymbol t)
&\leq
n\mleft(
g(R^\infty)
-
g(R_k^\infty)
+
\sum_{\ell=k+1}^{d}
\mleft(
g(\delta^\infty_{k,\ell})
-
g(\delta^\infty_{k,\ell}+R^\infty_\ell)
\mright)
\mright)
+
\lilo[\boldsymbol t,d][n]{n}\\
&=
n\Phi+\lilo[\boldsymbol t,d][n]{n},
\end{aligned}
\end{equation}
where the last equality follows from the definition of \(\Phi\). Exponentiating proves the claim.
\end{proof}

\begin{corollary}[Extensive deviations from overhang removal]
\label{cor:extensive_deviation_overhang_suppressed}
For all admissible extensive removal rates
\begin{equation}
\label{eq:extensive_admissible_cumulative}
0\leq R^\infty_{\ell}\le \delta^\infty_{\ell,\ell+1},
\qquad
\ell=k,\ldots,d,
\end{equation}
with the convention $\delta^\infty_{d,d+1}=\infty$
, the intensive-numerator term $I_n(\boldsymbol t)$ in~\cref{lem:multinomial_block_exponential_rate} satisfies $I_n(\boldsymbol t)\leq e^{n\Phi+\lilo[\boldsymbol t,d][n]{n}}$, and the exponential rate \(\Phi\) satisfies \(\Phi\le0\).
Moreover, \(\Phi=0\) if and only if the limiting rates are those of the saturated asymptotic overhang removal rule from~\cref{def:overhang_removal_rule}: for the terminal row \(i^\ast\),
\begin{equation}
R^\infty_\ell
=
\mleft\{\begin{aligned}
&0, & \ell<k,\\
&\delta^\infty_{\ell,\ell+1}, & k\le \ell<i^\ast,\\
&R^\infty-\delta^\infty_{k,i^\ast}, & \ell=i^\ast,\\
&0, & \ell>i^\ast.
\end{aligned}\mright.
\end{equation}
Consequently, any extensive deviation from overhang removal has \(\Phi<0\) and is exponentially suppressed.
\end{corollary}

\begin{proof}
The exponential rate in~\cref{lem:multinomial_block_exponential_rate} can be rewritten as
\begin{equation}
\label{eq:Phi_cumulative_form}
\Phi
=
\sum_{\ell=k+1}^{d}\Phi_\ell
=
\sum_{\ell=k+1}^{d}
\mleft(
g(S^\infty_\ell)-g(S^\infty_{\ell-1})
-
g(\delta^\infty_{k,\ell}+R^\infty_\ell)
+
g(\delta^\infty_{k,\ell})
\mright).
\end{equation}
For each \(\ell\), write \(S^\infty_\ell=S^\infty_{\ell-1}+R^\infty_\ell\). Then
\begin{equation}
\Phi_\ell
\coloneqq
\mleft(g(S^\infty_{\ell-1}+R^\infty_\ell)-g(S^\infty_{\ell-1})\mright)
-
\mleft(g(\delta^\infty_{k,\ell}+R^\infty_\ell)-g(\delta^\infty_{k,\ell})\mright).
\end{equation}
For fixed \(r>0\), the increment \(g(x+r)-g(x)\) is strictly increasing for \(x\ge0\), since for \(x>0\)
\begin{equation}
\frac{d}{dx}\mleft(g(x+r)-g(x)\mright)
=
\ln\mleft(1+\frac{r}{x}\mright)>0.
\end{equation}
Moreover, since
\begin{equation}
\begin{aligned}
S_{\ell-1}^\infty=
\sum_{j=k}^{\ell-1}R_j^\infty\leq
\sum_{j=k}^{\ell-1}\delta_{j,j+1}^\infty=
\delta_{k,\ell}^\infty,
\end{aligned}
\end{equation}
\cref{eq:extensive_admissible_cumulative} gives \(\Phi_\ell\le0\) for every \(\ell\). If \(R^\infty_\ell=0\), then \(\Phi_\ell=0\). Strict monotonicity gives \(\Phi_\ell=0\) if and only if \(R^\infty_\ell=0\) or \(S^\infty_{\ell-1}=\delta^\infty_{k,\ell}\).
Therefore, $\Phi=0$ forces the active rows to be contiguous, which is precisely the saturated asymptotic overhang removal rule.
\end{proof}
Thus, if the removal rule deviates extensively from overhang removal, then \(\Phi<0\). By~\cref{lem:multinomial_block_exponential_rate}, the corresponding multinomial block is bounded above by $e^{n\Phi+\lilo[\mathbf t,d,\epsilon][n]{n}}$, which is exponentially small. All other extensive products have zero exponential rate by~\cref{lem:polynomial_bound_signed_extensive_products} and likewise contribute only polynomial prefactors. Since the utility component is bounded by \(1\), DCT shows that the utility contribution from these extensively deviating rules vanishes.

\subsection{Properties of the $F$-symbols}
\label{subsec:props_F_symbols}

To analyze the optimality of the one-site case better, we need a deeper understanding of the $F$-symbols.

\begin{lemma}[Ordering of $F$-symbols]
    \label{lem:ordering_of_F_symbols}

    Let \(\yd{\lambda}\) be a valid environment irrep for the input sector \(\yd{\varsigma}\), and fix a target index \(1 \leq k \leq d\). Consider the corresponding \(F\)-symbol $\resSix{F}{\yd{\lambda}}{\yd{m-1}}{\yd{1}}{\yd{\varsigma-\mathbf{e}_k}}{\yd{m}}{\yd{\varsigma}}.$
    We compare it with the $F$-symbol obtained from a one-step deformation of the environment YD,
    \begin{equation}
        \yd{\lambda^{\prime}}=\yd{\lambda-\mathbf{e}_s+\mathbf{e}_t},
    \end{equation}
    where \(1\le s,t\le d\) and \(s\neq t\). To ensure that \(\yd{\lambda^{\prime}}\) is again a well-defined environment YD compatible with \(\yd{\varsigma}\), we impose
    \begin{equation}
    \label{equ:s_t_conditions}
    \begin{aligned}
        &\lambda_s > \varsigma_{s+1} \quad \text{if} \quad s\le d-1 \, ,\\
        &\lambda_t < \varsigma_t \, .
    \end{aligned}
    \end{equation}
    Denote $F \coloneqq \resSix{F}{\yd{\lambda}}{\yd{m-1}}{\yd{1}}{\yd{\varsigma-\mathbf{e}_k}}{\yd{m}}{\yd{\varsigma}}$ and $F' \coloneqq \resSix{F}{\yd{\lambda^{\prime}}}{\yd{m-1}}{\yd{1}}{\yd{\varsigma-\mathbf{e}_k}}{\yd{m}}{\yd{\varsigma}}$.
    The possibilities for \(s\neq t\) are exhausted by
    \begin{equation}
    \begin{aligned}
    & s \ge k > t,\qquad
    k \le s < t,\qquad
    s < t < k,\\
    & s < k \le t,\qquad
    k > s > t,\qquad
    s > t \ge k.
    \end{aligned}
    \end{equation}
    Among these, one has
    \begin{equation}
    \begin{aligned}
        F' \ge F \quad \text{if } s \ge k > t,\ \ k \le s < t,\ \ \text{or } s < t < k, \\
        F' \le F \quad \text{if } s < k \le t,\ \ k > s > t,\ \ \text{or } s > t \ge k.
    \end{aligned}
    \end{equation}
    Moreover, the inequalities are strict, apart from the case where
    \begin{equation}
    \begin{aligned}
        F = F' = 0 \, .
    \end{aligned}
    \end{equation}
\end{lemma}

\begin{proof}
    The proof follows from the explicit formula for the $F$-symbol (see~\cref{lem:one_site_channel_decomposition}).
    First, note that the numerator terms $\{\varsigma_j-\lambda_k+k-j\}_{j=1}^d$ contain $0$ if and only if $\varsigma_k=\lambda_k$. This is due to the fact that for $\varsigma_k\neq\lambda_k$, we have $\varsigma_j>\lambda_k$ for $j\leq k$ and $\varsigma_j\leq \lambda_k$ for $j>k$, so the numerator never evaluates to $0$.

    We prove the relations in two cases.
    (I) Consider the case $\varsigma_k\neq\lambda_k$. We define $A=\varsigma_k-\lambda_s+s-k$ and $B=\varsigma_k-\lambda_t+t-k$. Since $\{\varsigma_j-\lambda_k+k-j\}_{j=1}^d$ contains no zero, both $A$ and $B$ are nonzero. We further write
    \begin{equation}\begin{aligned}
        \label{equ:proof_optimality_of_F_symbol_definition_a}
        a_{-s}^{+t} \coloneqq \frac{F^{\prime2}}{F^2} = \mleft(1 + \frac{1}{A}\mright)\mleft(1 - \frac{1}{B}\mright) \, .
    \end{aligned}\end{equation}
We go on to consider two subcases, (i) $AB>0$ and (ii) $AB<0$.\newline

\noindent(i) Consider the case $AB>0$, which corresponds to $s,t\geq k$ or $s,t<k$.
It follows that
    \begin{equation}\begin{aligned}
        \label{equ:proof_optimality_of_F_symbol_inequality_product}
        a_{-s}^{+t} =\mleft(1+\frac{1}{A}\mright)\mleft(1-\frac{1}{B}\mright) \gtrless  1  \;\Leftrightarrow\;  \frac{1}{A} \gtrless  \frac{1}{B} + \frac{1}{AB} \;\Leftrightarrow\;  B \gtrless A+1.
    \end{aligned}\end{equation}
    For $s < t$, since $\yd{\lambda}$ and $\yd{\varsigma}$ are valid YDs, we have
    \begin{equation}\begin{aligned}
        \label{equ:proof_optimality_of_F_symbol_inequality_first_case_geq}
         B = \varsigma_{k} - \lambda_{t} + t - k  \ge \varsigma_{k} - \lambda_{s} + s - k + 1 = A +1,
    \end{aligned}\end{equation}
    which by~\cref{equ:proof_optimality_of_F_symbol_inequality_product} implies that $a_{-s}^{+t} \geq 1$. However, examining further, we have equality only in the case where
    \begin{equation}\begin{aligned}
        t = s+1 \, , \quad \lambda_{t} = \lambda_{s} \, .
    \end{aligned}\end{equation}
This implies that \(s\leq d-1\) and that \(\lambda_{s} = \varsigma_{s+1} = \lambda_{s+1}\), which is not compatible with the restrictions we impose. Altogether, this implies $a_{-s}^{+t} > 1$.
    For \(s>t\), we similarly get
    \begin{equation}\begin{aligned}
         B = \varsigma_{k} - \lambda_{t} + t - k  < \varsigma_{k} - \lambda_{s} + s - k + 1 = A +1,
    \end{aligned}\end{equation}
Thus \(a_{-s}^{+t} < 1\).

    \noindent(ii) Consider the case $AB<0$, which occurs when $s<k \leq t$ or $s \geq k>t$. In this case the inequality signs in~\cref{equ:proof_optimality_of_F_symbol_inequality_product} are flipped:
    \begin{equation}\begin{aligned}
        \label{equ:proof_optimality_of_F_symbol_inequality_product_flipped}
        a_{-s}^{+t} =\mleft(1+\frac{1}{A}\mright)\mleft(1-\frac{1}{B}\mright) \gtrless  1  \;\Leftrightarrow\;  \frac{1}{A} \gtrless  \frac{1}{B} + \frac{1}{AB} \;\Leftrightarrow\;  A+1 \gtrless B.
    \end{aligned}\end{equation}
    Therefore, for $s > t$, we have $a_{-s}^{+t} > 1$; for $s < t$, we have $a_{-s}^{+t} < 1$. The strict inequality for the second case is due to the fact that \(s<k\leq t\) and therefore also \(\lambda_{s}>\varsigma_{k}> \lambda_{k}\).

    (II) We then consider the case where \(\yd{\lambda}\) has no removal from the $k$-th row, i.e., $\varsigma_k=\lambda_k$. Condition~\cref{equ:s_t_conditions} then implies $t\neq k$. For the cases where $k\neq s, t$, we have $F^\prime=F=0$. For the case where $k=s$, we always have $F^\prime>F$.
\end{proof}

\begin{lemma}[\rev{Domination of $F$-symbols}]
    \label{lem:optimality_of_F_symbol}
    Let \(\yd{\lambda}\) be a valid environment irrep for the input sector \(\yd{\varsigma}\), that is, $\res{c}{\yd{m}}{\yd{\lambda}}{\yd{\varsigma}} \geq 1$, and fix a target index \(1 \leq k \leq d\). Further, let $\Delta_{k,k+1} \geq 1$, and let $\yd{\mu}$ be defined by the overhang removal rule of~\cref{def:overhang_removal_rule}. Then we have for $\yd{\lambda}\neq\yd{\mu}$:
    \begin{equation}\begin{aligned}
        \label{equ:lem_optimality_of_F_symbol}
        \mleft\{\resSix{F}{\yd{\lambda}}{\yd{m-1}}{\yd{1}}{\yd{\varsigma-\mathbf{e}_i}}{\yd{m}}{\yd{\varsigma}}^2\mright\}_{i=k}^{d} \prec \mleft\{ \resSix{F}{\yd{\mu}}{\yd{m-1}}{\yd{1}}{\yd{\varsigma-\mathbf{e}_i}}{\yd{m}}{\yd{\varsigma}}^2\mright\}_{i=k}^{d} \, ,
    \end{aligned}\end{equation}
    Moreover, at the index \(i=k\), the \rev{domination relation} holds with strict inequality.
\end{lemma}

\begin{proof}
    We argue that for every $\yd{\lambda} \neq \yd{\mu}$, there exist $s,t$ with corresponding $\yd{\lambda^{\prime}}$ as in~\cref{lem:ordering_of_F_symbols} such that for all $k\leq j \leq d$,
        \begin{equation}\begin{aligned}
            \sum_{i=k}^{j} \resSix{F}{\yd{\lambda}}{\yd{m-1}}{\yd{1}}{\yd{\varsigma-\mathbf{e}_i}}{\yd{m}}{\yd{\varsigma}}^2 \leq \sum_{i=k}^{j} \resSix{F}{\yd{\lambda^{\prime}}}{\yd{m-1}}{\yd{1}}{\yd{\varsigma-\mathbf{e}_i}}{\yd{m}}{\yd{\varsigma}}^2 \, .
        \end{aligned}\end{equation}
        We proceed by considering two cases.

        \noindent(I) Suppose that for some $l<k$ we have $\lambda_l<\varsigma_l$, meaning boxes are removed from rows above the target row $k$. Choosing $t=l$ and $s=d$, we obtain from~\cref{lem:ordering_of_F_symbols} that
        \begin{equation}\begin{aligned}
            \resSix{F}{\yd{\lambda^{\prime}}}{\yd{m-1}}{\yd{1}}{\yd{\varsigma-\mathbf{e}_i}}{\yd{m}}{\yd{\varsigma}} \geq \resSix{F}{\yd{\lambda}}{\yd{m-1}}{\yd{1}}{\yd{\varsigma-\mathbf{e}_i}}{\yd{m}}{\yd{\varsigma}}
        \end{aligned}\end{equation}
        for all $k\leq i \leq d$, with equality only if both sides vanish. In particular, for all $k\leq j \leq d$,
        \begin{equation}\begin{aligned}
            \label{equ:proof_optimality_of_F_symbol_inequality_first_case}
            \sum_{i=k}^{j} \resSix{F}{\yd{\lambda^{\prime}}}{\yd{m-1}}{\yd{1}}{\yd{\varsigma-\mathbf{e}_i}}{\yd{m}}{\yd{\varsigma}}^2 \geq \sum_{i=k}^{j} \resSix{F}{\yd{\lambda}}{\yd{m-1}}{\yd{1}}{\yd{\varsigma-\mathbf{e}_i}}{\yd{m}}{\yd{\varsigma}}^2 \, ,
        \end{aligned}\end{equation}
        with equality only when both sides are zero.

        \noindent(II) Now suppose that
        \begin{equation}\begin{aligned}
            \lambda_{i} = \varsigma_{i} \quad \text{for all} \quad i<k \, ,
        \end{aligned}\end{equation}
        which implies
        \begin{equation}\begin{aligned}
            \label{equ:proof_optimality_of_F_symbol_F_symbol_zero}
            \resSix{F}{\yd{\lambda}}{\yd{m-1}}{\yd{1}}{\yd{\varsigma-\mathbf{e}_i}}{\yd{m}}{\yd{\varsigma}} = 0 \quad \text{for all} \quad i<k \, .
        \end{aligned}\end{equation}
        Let $s\geq k$ be the first row where $\lambda_{s} > \varsigma_{s+1}$ (using the convention $\lambda_{d}>\varsigma_{d+1}$), and let $t \geq k$ be the last row such that $\lambda_{t} < \varsigma_{t}$. If $s\geq t$, then $\yd{\lambda} = \yd{\mu}$ and there is nothing to prove. We therefore assume $s<t$.

        For $i>t$ we have $\lambda_{i} = \varsigma_{i}$, so
        \begin{equation}\begin{aligned}
            \label{equ:proof_optimality_of_F_symbol_F_symbol_zero_2}
            \resSix{F}{\yd{\lambda^{\prime}}}{\yd{m-1}}{\yd{1}}{\yd{\varsigma-\mathbf{e}_i}}{\yd{m}}{\yd{\varsigma}} = \resSix{F}{\yd{\lambda}}{\yd{m-1}}{\yd{1}}{\yd{\varsigma-\mathbf{e}_i}}{\yd{m}}{\yd{\varsigma}} = 0 \, .
        \end{aligned}\end{equation}
        For $s < i \leq t$,~\cref{lem:ordering_of_F_symbols} gives
        \begin{equation}\begin{aligned}
            \resSix{F}{\yd{\lambda^{\prime}}}{\yd{m-1}}{\yd{1}}{\yd{\varsigma-\mathbf{e}_i}}{\yd{m}}{\yd{\varsigma}} \leq \resSix{F}{\yd{\lambda}}{\yd{m-1}}{\yd{1}}{\yd{\varsigma-\mathbf{e}_i}}{\yd{m}}{\yd{\varsigma}} \, ,
        \end{aligned}\end{equation}
        with equality only if both sides vanish. Using~\cref{equ:proof_optimality_of_F_symbol_F_symbol_zero} and~\cref{equ:proof_optimality_of_F_symbol_F_symbol_zero_2} together with the fact that the $F$-symbols sum to $1$, we find for $s < j \leq d$ that
        \begin{equation}\begin{aligned}
            \label{equ:proof_optimality_of_F_symbol_inequality_second_case}
            \sum_{i=k}^{j} \resSix{F}{\yd{\lambda^{\prime}}}{\yd{m-1}}{\yd{1}}{\yd{\varsigma-\mathbf{e}_i}}{\yd{m}}{\yd{\varsigma}}^2 = 1 - \sum_{i=j+1}^{t} \resSix{F}{\yd{\lambda^{\prime}}}{\yd{m-1}}{\yd{1}}{\yd{\varsigma-\mathbf{e}_i}}{\yd{m}}{\yd{\varsigma}}^2 \\
            \geq 1 - \sum_{i=j+1}^{t} \resSix{F}{\yd{\lambda}}{\yd{m-1}}{\yd{1}}{\yd{\varsigma-\mathbf{e}_i}}{\yd{m}}{\yd{\varsigma}}^2 = \sum_{i=k}^{j} \resSix{F}{\yd{\lambda}}{\yd{m-1}}{\yd{1}}{\yd{\varsigma-\mathbf{e}_i}}{\yd{m}}{\yd{\varsigma}}^2 \, ,
        \end{aligned}\end{equation}
        with equality only if both sides equal $1$. For $j \leq s$, a symmetric argument yields
        \begin{equation}\begin{aligned}
            \label{equ:proof_optimality_of_F_symbol_inequality_third_case}
            \sum_{i=k}^{j} \resSix{F}{\yd{\lambda^{\prime}}}{\yd{m-1}}{\yd{1}}{\yd{\varsigma-\mathbf{e}_i}}{\yd{m}}{\yd{\varsigma}}^2 \geq \sum_{i=k}^{j} \resSix{F}{\yd{\lambda}}{\yd{m-1}}{\yd{1}}{\yd{\varsigma-\mathbf{e}_i}}{\yd{m}}{\yd{\varsigma}}^2 \, ,
        \end{aligned}\end{equation}
        with equality only if both sides vanish.

        Since the deformation rules are unidirectional, iterating from any $\yd{\lambda}$ eventually reaches the fixed point $\yd{\mu}$. This establishes that for all $\yd{\lambda}$ and for all $k\leq j \leq d$:
        \begin{equation}\begin{aligned}
            \label{equ:lem_optimality_of_F_symbol_partial_sums}
            \sum_{i=k}^{j} \resSix{F}{\yd{\lambda}}{\yd{m-1}}{\yd{1}}{\yd{\varsigma-\mathbf{e}_i}}{\yd{m}}{\yd{\varsigma}}^2 \leq \sum_{i=k}^{j} \resSix{F}{\yd{\mu}}{\yd{m-1}}{\yd{1}}{\yd{\varsigma-\mathbf{e}_i}}{\yd{m}}{\yd{\varsigma}}^2 .
        \end{aligned}\end{equation}

It remains to show that equality holds only for $\yd{\lambda} = \yd{\mu}$. Equality for all $j$ requires that each of the inequalities~\cref{equ:proof_optimality_of_F_symbol_inequality_first_case},~\cref{equ:proof_optimality_of_F_symbol_inequality_second_case}, and~\cref{equ:proof_optimality_of_F_symbol_inequality_third_case} is an equality, which forces both sides to be either $0$ or $1$ throughout the iterated deformations.

Consider the case $j=k$. Since $\Delta_{k,k+1}\geq 1$, we have
        \begin{equation}
            \resSix{F}{\yd{\mu}}{\yd{m-1}}{\yd{1}}{\yd{\varsigma-\mathbf{e}_k}}{\yd{m}}{\yd{\varsigma}} > 0 \, .
        \end{equation}
Therefore, the only way all partial sums can be equal is if they equal $1$, which forces $i^\ast = k$ and all removals to be concentrated at the target row for both $\yd{\lambda}$ and $\yd{\mu}$. This implies $\yd{\lambda} = \yd{\mu}$, completing the proof.
\end{proof}

\begin{remark}
    \label{rmk:F_symbol_row_removal}
    When $\yd{\mu}$ is given by the overhang removal rule, the $F$-symbol evaluates to
    \begin{equation}\begin{aligned}
        \resSix{F}{\yd{\mu}}{\yd{m-1}}{\yd{1}}{\yd{\varsigma-\mathbf{e}_i}}{\yd{m}}{\yd{\varsigma}}^2 = \mleft(\prod_{j=k+1, j \neq i}^{i^\ast} \mleft(1-\frac{1}{\widetilde{\Delta}_{i,j}}\mright)\mright)\mleft(\frac{m - \widetilde{\Delta}_{k,i} + (i^\ast-k)}{m}\mright)C_{i} \, ,
    \end{aligned}\end{equation}
    where
    \begin{equation}\begin{aligned}
        C_{i} \coloneqq \mleft\{\begin{aligned}
        &1, & i=k, \\
        &\dfrac{1}{\widetilde{\Delta}_{k,i}}, & k<i \leq i^\ast, \\
        &0, & \text{otherwise}.
        \end{aligned}\mright.
    \end{aligned}\end{equation}
\end{remark}

\begin{lemma}[Leading-order $F$-symbol expansion for overhang removal]
\label{lem:leading_order_F_symbol_overhang_removal}
For a normalized row vector \(\overline{\yd{\varsigma}}\dot{=}(\varsigma_{1}/n,\ldots,\varsigma_{d}/n)\) and an environment irrep $\yd{\lambda}$, set the removal rates $\boldsymbol R^{\yd{\lambda}}\dot{=}(R_1^{\yd{\lambda}},\ldots,R_d^{\yd{\lambda}})\coloneqq\rev{\overline{\yd{\varsigma-\lambda}}}$ and $R\coloneqq\sum_{i=1}^{d}R_i^{\yd{\lambda}}$. With the convention \(\rev{\overline\Delta_{k,d+1}\coloneqq\infty}\), \rev{the validity of $\yd{\lambda}$ is equivalent to} \(0\leq R_j^{\yd{\lambda}} \leq \overline{\Delta}_{j,j+1}\) for $1\leq j\leq d$. For the expansion below, assume additionally that $R\neq0$ and $\overline\Delta_{i,j}\neq0$ for every $j\neq i$. For normalized gaps $\overline{\boldsymbol\Delta}$ and removal rates $\boldsymbol R^{\yd{\lambda}}$, the squared $F$-symbol has the leading-order expansion
\begin{equation}
\label{equ:leading_order_F_symbol_expansion}
    \resSix{F}{\yd{\lambda}}{\yd{m-1}}{\yd{1}}{\yd{\varsigma-\mathbf{e}_i}}{\yd{m}}{\yd{\varsigma}}^2
    =A_i\mleft(\overline{\yd{\varsigma}},\boldsymbol R^{\yd{\lambda}}\mright)+\bigO[d,\overline{\boldsymbol\Delta},\boldsymbol R^{\yd{\lambda}}]{n^{-1}},
\end{equation}
where
\begin{equation}
\label{equ:leading_order_F_symbol_coefficient}
    A_i\mleft(\overline{\yd{\varsigma}},\boldsymbol R^{\yd{\lambda}}\mright)
    \coloneqq
    \frac{1}{R}
    \frac{\prod_{j=1}^{d}\mleft(\overline\Delta_{i,j}+R_j^{\yd{\lambda}}\mright)}
    {\prod_{j\neq i}\overline\Delta_{i,j}}.
\end{equation}
For the overhang removal rule, we get the rates \(\boldsymbol R^{\yd{\mu}}\). \rev{Let \(i^\ast \coloneqq \min\{i\mid \overline{\Delta}_{k,i+1}>R\}\)}. Then
\begin{equation}
    \label{equ:asymptotic_overhang_removal_rates}
    R_i^{\yd{\mu}} \coloneqq
    \mleft\{\begin{aligned}
    &0, & 1 \leq i < k,\\[4pt]
    &\overline{\Delta}_{i,i+1}, & k \leq i < i^\ast,\\[4pt]
    &R-\overline{\Delta}_{k,i^\ast}, & i = i^\ast,\\[4pt]
    &0, & i^\ast < i \leq d.
    \end{aligned}\mright.
\end{equation}
Then
\begin{equation}
\label{equ:leading_order_overhang_F_symbol}
    A_i\mleft(\overline{\yd{\varsigma}},\boldsymbol R^{\yd{\mu}}\mright)=\delta_{i,k},
    \qquad 1\leq i\leq d.
\end{equation}
\end{lemma}

\begin{proof}
With $t=1/n$, the exact $F$-symbol formula gives
\begin{equation}
\resSix{F}{\yd{\lambda}}{\yd{m-1}}{\yd{1}}{\yd{\varsigma-\mathbf{e}_i}}{\yd{m}}{\yd{\varsigma}}^2
=\frac{1}{R}
\frac{\prod_{j=1}^{d}\mleft(\overline\Delta_{i,j}+R_j^{\yd{\lambda}}+(j-i)t\mright)}
{\prod_{j\neq i}\mleft(\overline\Delta_{i,j}+(j-i)t\mright)}.
\end{equation}
The assumptions $R\neq0$ and $\overline\Delta_{i,j}\neq0$ for every $j\neq i$ make the right-hand side nonsingular at $t=0$, so it is analytic in a neighborhood of $t=0$. Evaluating it at $t=0$ gives~\cref{equ:leading_order_F_symbol_coefficient}, and Taylor expansion at $t=0$ proves~\cref{equ:leading_order_F_symbol_expansion}. Substituting~\cref{equ:asymptotic_overhang_removal_rates} into~\cref{equ:leading_order_F_symbol_coefficient} gives~\cref{equ:leading_order_overhang_F_symbol}.
\end{proof}

The following lemma shows that the \(O(1)\) expansion of the \(F\)-symbols given by \(A_i\mleft(\overline{\yd{\varsigma}},\boldsymbol R^{\yd{\lambda}}\mright)\) is uniform over all \(\overline{\yd{\varsigma}}\in\mathsf{R}_{\epsilon}\) and all environment irreps \(\yd{\lambda}\), and tells us more about the optimality of the overhang removal rule.

\begin{lemma}[Asymptotic behavior of $F$-symbols]
    \label{lem:asymptotic_behavior_of_F_symbol}
    For all \(\overline{\yd{\varsigma}}\in \mathsf{R}_{\epsilon}\), $i\in\{1,\ldots,d\}$, and all environment irreps \(\yd{\lambda}\), set $\boldsymbol R^{\yd{\lambda}}\coloneqq\rev{\overline{\yd{\varsigma-\lambda}}}$. Then
    \begin{equation}
        \label{equ:asymptotic_F_symbol_statement_bound}
        \left|\resSix{F}{\yd{\lambda}}{\yd{m-1}}{\yd{1}}{\yd{\varsigma-\mathbf{e}_i}}{\yd{m}}{\yd{\varsigma}}^2 - A_i\mleft(\overline{\yd{\varsigma}},\boldsymbol R^{\yd{\lambda}}\mright)\right| \leq \frac{4d^3}{n\epsilon} \, .
    \end{equation}
    Here, \(\overline{\yd{\varsigma}}\dot{=}(\varsigma_{1}/n,\ldots,\varsigma_{d}/n)\) denotes the normalized row-length vector and \(\yd{\varsigma}\vdash_{d}n\). In addition, let \(\boldsymbol R^{\yd{\lambda}},\boldsymbol R^{\yd{\lambda^\prime}}\) be valid removal rates for \(\overline{\yd{\varsigma}}\) with
    \begin{equation}
        \sum_{i=1}^{d}R_i^{\yd{\lambda}} = R = \sum_{i=1}^{d}R_i^{\yd{\lambda^\prime}} \, .
    \end{equation}
    Then we have
    \begin{equation}
        \left|A_i\mleft(\overline{\yd{\varsigma}},\boldsymbol R^{\yd{\lambda}}\mright) - A_i\mleft(\overline{\yd{\varsigma}},\boldsymbol R^{\yd{\lambda^\prime}}\mright)\right| \leq \rev{\frac{R+\epsilon}{R\epsilon}\left\|\boldsymbol R^{\yd{\lambda}} - \boldsymbol R^{\yd{\lambda^\prime}}\right\|_{1}}  \, .
    \end{equation}
\end{lemma}

\begin{proof}
    (I) First, both functions are bounded by \(1\), so the bound is vacuous when \(4d^3>n\epsilon\). We may therefore assume that \(4d^3\leq n\epsilon\), which in particular implies \(n>2d/\epsilon\). Set $\boldsymbol\delta\dot{=}\rev{(1,\ldots,d)}$. We have
    \begin{equation}
    \begin{aligned}
        \resSix{F}{\yd{\lambda}}{\yd{m-1}}{\yd{1}}{\yd{\varsigma-\mathbf{e}_i}}{\yd{m}}{\yd{\varsigma}}^2 = \dfrac{n}{m}\dfrac{\prod_{j=1}^{d}(\overline{\varsigma}_{i} - \overline{\lambda}_{j} + (j-i)/n)}{\prod_{j\neq i}(\overline{\varsigma}_{i} - \overline{\varsigma}_{j} + (j-i)/n)} = A_i\mleft(\rev{\overline{\yd{\varsigma-\boldsymbol{\delta}}}},\boldsymbol R^{\yd{\lambda}}\mright) \, ,
    \end{aligned}
    \end{equation}
    \rev{where we define \(\widetilde{\overline{\yd{\varsigma}}}(x)\) componentwise by \(\widetilde{\overline{\varsigma}}_i(x) \coloneqq \overline{\varsigma}_{i}-i\,x.\) and \(\widetilde{\overline{\Delta}}_{i,j}(x)\coloneqq\widetilde{\overline{\varsigma}}_{i}(x)-\widetilde{\overline{\varsigma}}_{j}(x)\).} For \(n>2d/\epsilon\), we have \(|\widetilde{\overline{\Delta}}_{j,\ell}(x)| > \epsilon/2\) for \(j\neq\ell\). Writing the difference as an integral gives
    \begin{equation}
        \label{equ:proof_asymptotic_optimality_of_F_symbol_derivative}
        \left|\resSix{F}{\yd{\lambda}}{\yd{m-1}}{\yd{1}}{\yd{\varsigma-\mathbf{e}_i}}{\yd{m}}{\yd{\varsigma}}^2 - A_i\mleft(\overline{\yd{\varsigma}},\boldsymbol R^{\yd{\lambda}}\mright)\right| \leq \frac{d^2}{n}\max_{0\leq x \leq 1}\left\|\nabla_{\widetilde{\overline{\yd{\varsigma}}}}A_i\mleft(\widetilde{\overline{\yd{\varsigma}}}(x),\boldsymbol R^{\yd{\lambda}}\mright)\right\|_\infty \, .
    \end{equation}
    \noindent\rev{(i)} For the $\ell$-th component of the gradient, with $\ell\neq i$, we have
    \begin{equation}
        \frac{\partial}{\partial \widetilde{\overline{\varsigma}}_{\ell}} A_i\mleft(\widetilde{\overline{\yd{\varsigma}}}(x),\boldsymbol R^{\yd{\lambda}}\mright) = \frac{A_i\mleft(\widetilde{\overline{\yd{\varsigma}}}(x),\boldsymbol R^{\yd{\lambda}}\mright)}{\widetilde{\overline{\Delta}}_{i,\ell}(x)} - \dfrac{1}{R}\dfrac{\prod_{j\neq \ell}(\widetilde{\overline{\Delta}}_{i,j}(x) + R_j^{\yd{\lambda}})}{\prod_{j\neq i}\widetilde{\overline{\Delta}}_{i,j}(x)} \, .
    \end{equation}
    For the first term, we have
    \begin{equation}
    \label{equ:proof_asymptotic_optimality_of_F_symbol_first_term}
        \left|\frac{A_i\mleft(\widetilde{\overline{\yd{\varsigma}}}(x),\boldsymbol R^{\yd{\lambda}}\mright)}{\widetilde{\overline{\Delta}}_{i,\ell}(x)}\right| \leq \frac{2}{\epsilon} \, ,
    \end{equation}
    since the numerator is bounded by \(1\) and the denominator by \(\epsilon/2\), as discussed before. For the second term, define \(\rev{\boldsymbol R^{(\ell)}}\) by setting the \(\ell\)-th component of \(\boldsymbol R^{\yd{\lambda}}\) to zero. If \(\rev{R^{(\ell)}\coloneqq\sum_jR_j^{(\ell)}}>0\), then
    \begin{equation}
        \label{equ:proof_asymptotic_optimality_of_F_symbol_second_term}
         \left|\dfrac{1}{R}\dfrac{\prod_{j\neq \ell}(\widetilde{\overline{\Delta}}_{i,j}(x) + R_j^{\yd{\lambda}})}{\prod_{j\neq i}\widetilde{\overline{\Delta}}_{i,j}(x)}\right| = \frac{1}{|\widetilde{\overline{\Delta}}_{i,\ell}(x)|} \frac{\rev{R^{(\ell)}}}{R} A_i\mleft(\widetilde{\overline{\yd{\varsigma}}}(x),\rev{\boldsymbol R^{(\ell)}}\mright) \leq \frac{2}{\epsilon} \, .
    \end{equation}
    Explicitly,
    \begin{equation}
        \label{equ:proof_asymptotic_optimality_of_F_symbol_removed_component}
        \rev{R_j^{(\ell)}} =
        \mleft\{\begin{aligned}
            &R_j^{\yd{\lambda}}, & j\neq \ell, \\[4pt]
            &0, & j = \ell.
        \end{aligned}\mright.
    \end{equation}
    If \(\rev{R^{(\ell)}}=0\), then $R=R_\ell^{\yd{\lambda}}$; since $\ell\neq i$, we have $R_i^{\yd{\lambda}}=0$, so the $j=i$ factor $\widetilde{\overline\Delta}_{i,i}(x)+R_i^{\yd{\lambda}}$ in the numerator vanishes and the second term is zero.

    \noindent\rev{(ii)} For \(\ell=i\), we have
    \begin{equation}
        \frac{\partial}{\partial \widetilde{\overline{\varsigma}}_{i}} A_i\mleft(\widetilde{\overline{\yd{\varsigma}}}(x),\boldsymbol R^{\yd{\lambda}}\mright) = -\sum_{j\neq i}\frac{\partial}{\partial \widetilde{\overline{\varsigma}}_{j}} A_i\mleft(\widetilde{\overline{\yd{\varsigma}}}(x),\boldsymbol R^{\yd{\lambda}}\mright) \, .
    \end{equation}
    Inserting this together with~\cref{equ:proof_asymptotic_optimality_of_F_symbol_first_term,equ:proof_asymptotic_optimality_of_F_symbol_second_term} into~\cref{equ:proof_asymptotic_optimality_of_F_symbol_derivative} gives exactly~\cref{equ:asymptotic_F_symbol_statement_bound}.

    \rev{(II)} For the second statement, we set
    \begin{equation}
        \rev{\boldsymbol R(x)} \coloneqq \rev{x\boldsymbol R^{\yd{\lambda^\prime}}+(1-x)\boldsymbol R^{\yd{\lambda}}} \, ,
    \end{equation}
    and we find
    \begin{equation}
        \label{equ:proof_asymptotic_optimality_of_F_symbol_interpolation}
        \left|A_i\mleft(\overline{\yd{\varsigma}},\boldsymbol R^{\yd{\lambda}}\mright) - A_i\mleft(\overline{\yd{\varsigma}},\boldsymbol R^{\yd{\lambda^\prime}}\mright)\right| \leq  \max_{0\leq x \leq 1}\left|\frac{d}{d x} A_i\mleft(\overline{\yd{\varsigma}},\rev{\boldsymbol R(x)}\mright)\right|.
    \end{equation}
    We have
    \begin{equation}
        \frac{d}{d x} A_i\mleft(\overline{\yd{\varsigma}},\rev{\boldsymbol R(x)}\mright) = \sum_{\ell=1}^{d}\mleft(R_\ell^{\yd{\lambda^\prime}}-R_\ell^{\yd{\lambda}}\mright) \dfrac{1}{R}\dfrac{\prod_{j\neq \ell}\mleft(\overline{\Delta}_{i,j} + \rev{R_j(x)}\mright)}{\prod_{j\neq i}\overline{\Delta}_{i,j}} \, .
    \end{equation}
    \noindent\rev{(i)} \rev{For \(\ell\neq i\), let \(\boldsymbol R^{(\ell)}(x)\) be obtained from \(\boldsymbol R(x)\) by setting its \(\ell\)-th component to zero, and write \(R^{(\ell)}(x)\) for its total mass. When \(R^{(\ell)}(x)>0\), the same argument as in~\cref{equ:proof_asymptotic_optimality_of_F_symbol_second_term,equ:proof_asymptotic_optimality_of_F_symbol_removed_component} gives}
    \begin{equation}
        \rev{\left|\dfrac{1}{R}\dfrac{\prod_{j\neq \ell}\mleft(\overline{\Delta}_{i,j} + R_j(x)\mright)}{\prod_{j\neq i}\overline{\Delta}_{i,j}}\right|
        =\frac{1}{|\overline\Delta_{i,\ell}|}
        \frac{R^{(\ell)}(x)}{R}
        A_i\mleft(\overline{\yd{\varsigma}},\boldsymbol R^{(\ell)}(x)\mright)
        \leq \frac{1}{\epsilon}} \, .
    \end{equation}
    \rev{Here, the total-mass ratio and \(A_i\) are at most \(1\), while \(|\overline\Delta_{i,\ell}|\geq\epsilon\). If \(R^{(\ell)}(x)=0\), the \(j=i\) factor in the numerator again vanishes, so the same bound holds.}
    \noindent\rev{(ii)} For \(\ell=i<d\), we can similarly argue that
    \begin{equation}
        \dfrac{1}{R}\dfrac{\prod_{j\neq i}\mleft(\overline{\Delta}_{i,j} + \rev{R_j(x)}\mright)}{\prod_{j\neq i}\overline{\Delta}_{i,j}} = \frac{\rev{R^{(i)}(x)}}{R\overline{\Delta}_{i,i+1}} A_i\mleft(\overline{\yd{\varsigma}},\rev{\boldsymbol R^{(i)}(x)}\mright) \leq \frac{R + \overline{\Delta}_{i,i+1}}{R\overline{\Delta}_{i,i+1}} \rev{\leq \frac{1}{R}+\frac{1}{\epsilon}} \, ,
    \end{equation}
    where \(\rev{\boldsymbol R^{(i)}(x)}\) is given by
    \begin{equation}
        \rev{R_j^{(i)}(x)} =
        \mleft\{\begin{aligned}
            &\rev{R_j(x)}, & j\neq i, \\[4pt]
            &\overline{\Delta}_{i,i+1}, & j = i.
        \end{aligned}\mright.
    \end{equation}
    \rev{In particular, \(R^{(i)}(x)\geq\overline\Delta_{i,i+1}\geq\epsilon>0\), so the coefficient \(A_i\mleft(\overline{\yd{\varsigma}},\boldsymbol R^{(i)}(x)\mright)\) is well-defined.}

    \noindent\rev{(iii) For \(\ell=i=d\), we have \(0\leq R_j(x)\leq\overline\Delta_{j,j+1}\leq\overline\Delta_{j,d}\) for every \(j<d\), so every factor lies in \([0,1]\):}
    \begin{equation}
        \rev{0\leq
        \dfrac{1}{R}\dfrac{\prod_{j<d}\mleft(\overline\Delta_{d,j}+R_j(x)\mright)}{\prod_{j<d}\overline\Delta_{d,j}}
        =\frac{1}{R}\prod_{j<d}\left(1-\frac{R_j(x)}{\overline\Delta_{j,d}}\right)
        \leq\frac{1}{R}} \, .
    \end{equation}
    \rev{Altogether, the right-hand side of~\cref{equ:proof_asymptotic_optimality_of_F_symbol_interpolation} is bounded by \(\frac{R+\epsilon}{R\epsilon}\left\|\boldsymbol R^{\yd{\lambda}}-\boldsymbol R^{\yd{\lambda^\prime}}\right\|_1\), which proves the second statement.}
\end{proof}

\subsection{Asymptotic analysis of the one-site utility}
\label{subsec:asymptotic_one_site_supp}

\begin{theorem}
    \label{cor:sector_wise_one_site_fidelity}
    Let $\yd{\varsigma}\dot{=}[\varsigma_{1},\ldots,\varsigma_{d}]$, write $R\coloneqq m/n$, and let $\yd{\mu}$ be given by the overhang removal rule, with $i^\ast$ the terminal index. Then we have the asymptotic expansion
    \begin{equation}\begin{aligned}
        \label{equ:sector_wise_one_site_fidelity_expansion}
        \res{{^kf}}{\yd{\varsigma}}{\yd{m}}{\yd{\mu}}_{\mathrm{one}}
        = 1 - \frac{1}{n}\mleft(\sum_{i\neq k}\frac{p_{i}}{\overline{\Delta}_{k,i}D_{k,i}} + \sum_{i=k+1}^{i^\ast} \mleft(\frac{1}{\overline{\Delta}_{k,i}} - \frac{1}{R}\mright)\mright) + \lilo[d,\mathbf{r},\epsilon]{n^{-1}} \, ,
    \end{aligned}\end{equation}
    uniformly in $R$ and $\overline{\yd{\varsigma}}\in\mathsf R_\epsilon$.
\end{theorem}

\begin{proof}
    From~\cref{lem:one_site_channel_decomposition}, we know that
    \begin{equation}\begin{aligned} \Tr_{2,\ldots,m}\mleft(\extChannel{\mathcal{T}}{\yd{\varsigma}}{\yt{1\cdots m}}{\yd{\mu}}(\cdot)\mright) = \sum_{i=1}^{d}\resSix{F}{\yd{\mu}}{\yd{m-1}}{\yd{1}}{\yd{\varsigma-\mathbf{e}_i}}{\yt{1\cdots m}}{\yd{\varsigma}}^2 \extChannel{\mathcal{T}}{\yd{\varsigma}}{\yd{1}}{\yd{\varsigma-\mathbf{e}_i}}(\cdot) \, ,
    \end{aligned}\end{equation}
By~\cref{thm:single_output_intensive_off_target_all_site_asymptotics}, the off-target fidelity components for \(i\neq k\) are \(\lilo[d,\mathbf r,\epsilon][n]{1}\), while \cref{rmk:F_symbol_row_removal,lem:leading_order_F_symbol_overhang_removal,lem:asymptotic_behavior_of_F_symbol} show that the corresponding \(F\)-symbol factors are \(\bigO[d,\epsilon]{n^{-1}}\), uniformly in \(\overline\Delta_{i,j}\) and \(R\); hence each off-target component contributes \(\lilo[d,\mathbf r,\epsilon]{n^{-1}}\). Therefore, only the on-target contribution remains:
    \begin{equation}
    \begin{aligned}
        \resSix{F}{\yd{\mu}}{\yd{m-1}}{\yd{1}}{\yd{\varsigma-\mathbf{e}_k}}{\yd{m}}{\yd{\varsigma}}^2 = 1 + \frac{1}{n}\mleft(\frac{i^\ast-k}{R} - \sum_{i=k+1}^{i^{*}}\frac{1}{\overline{\Delta}_{k,i}}\mright) + \lilo[d,\epsilon]{n^{-1}} \, .
    \end{aligned}
    \end{equation}
    Combining this with the asymptotics for the overhang-removal protocol proves the claim.

\end{proof}

\begin{lemma}[Uniform Lipschitz continuity of the one-site coefficient]
\label{lem:one_site_coefficient_lipschitz}
The first-order coefficient in~\cref{equ:sector_wise_one_site_fidelity_expansion} is Lipschitz in \(\overline{\yd{\varsigma}}\in\mathsf R_\epsilon\), uniformly over \(R>0\).
\end{lemma}

\begin{proof}
Write \(x_+\coloneqq\max(x,0)\). For \(\overline{\yd{\varsigma}}\in\mathsf R_\epsilon\), let \(i^\ast\) be its terminal index and set
\begin{equation}
\label{eq:one_site_overhang_coefficient}
c\mleft(\overline{\yd{\varsigma}}\mright)
\coloneqq
\sum_{i=k+1}^{i^\ast}
\mleft(
\frac{1}{\overline\Delta_{k,i}}-\frac{1}{R}
\mright)
=
\sum_{i=k+1}^{d}
\mleft(
\frac{1}{\overline\Delta_{k,i}}-\frac{1}{R}
\mright)_+.
\end{equation}
Since \(\lvert x_+-y_+\rvert\leq\lvert x-y\rvert\), for any \(\Delta,\Delta^\prime\geq\epsilon\),
\begin{equation}
\label{eq:one_site_threshold_lipschitz}
\begin{aligned}
\left|
\mleft(\frac{1}{\Delta}-\frac{1}{R}\mright)_+
-
\mleft(\frac{1}{\Delta^\prime}-\frac{1}{R}\mright)_+
\right|
&\leq
\left|
\mleft(\frac{1}{\Delta}-\frac{1}{R}\mright)
-
\mleft(\frac{1}{\Delta^\prime}-\frac{1}{R}\mright)
\right|\\
&=
\left|
\frac{1}{\Delta}-\frac{1}{\Delta^\prime}
\right|
\leq
\frac{1}{\epsilon^2}|\Delta-\Delta^\prime|.
\end{aligned}
\end{equation}
Thus the potentially singular \(R^{-1}\) terms cancel exactly. Since
\begin{equation}
\label{eq:one_site_gap_lipschitz}
\left|
\overline\Delta_{k,i}
-
\overline\Delta_{k,i}^{\prime}
\right|
\leq
2
\left\|
\overline{\yd{\varsigma}}
-
\overline{\yd{\varsigma^\prime}}
\right\|_2,
\end{equation}
Combining~\cref{eq:one_site_threshold_lipschitz,eq:one_site_gap_lipschitz} and summing over \(i=k+1,\ldots,d\) gives
\begin{equation}
\label{eq:one_site_overhang_coefficient_lipschitz}
\left|
c\mleft(\overline{\yd{\varsigma}}\mright)
-
c\mleft(\overline{\yd{\varsigma^\prime}}\mright)
\right|
\leq
\frac{2(d-k)}{\epsilon^2}
\left\|
\overline{\yd{\varsigma}}
-
\overline{\yd{\varsigma^\prime}}
\right\|_2,
\end{equation}
uniformly over \(R>0\). The full first-order coefficient is
\begin{equation}
\label{eq:one_site_coefficient_function}
c^\prime\mleft(\overline{\yd{\varsigma}}\mright)
\coloneqq
\sum_{i\neq k}
\frac{p_i}{\overline\Delta_{k,i}D_{k,i}}
+
c\mleft(\overline{\yd{\varsigma}}\mright).
\end{equation}
The same reciprocal-gap estimate applied to the first sum in~\cref{eq:one_site_coefficient_function}, together with~\cref{eq:one_site_overhang_coefficient_lipschitz}, gives
\begin{equation}
\label{eq:one_site_full_coefficient_lipschitz}
\left|
c^\prime\mleft(\overline{\yd{\varsigma}}\mright)
-
c^\prime\mleft(\overline{\yd{\varsigma^\prime}}\mright)
\right|
\leq
\frac{2}{\epsilon^2}
\mleft(
d-k+
\sum_{i\neq k}\frac{p_i}{|D_{k,i}|}
\mright)
\left\|
\overline{\yd{\varsigma}}
-
\overline{\yd{\varsigma^\prime}}
\right\|_2,
\end{equation}
uniformly over \(R>0\), proving the claim.
\end{proof}

\begin{theorem}[Overall one-site asymptotics]
\label{thm:asymptotics_one_fidelity}
Write \(R\coloneqq m/n\), and let the macroscopic terminal index of the overhang removal rule be \(I^\ast=\min\{i\mid D_{k,i+1}>R\}\). Then the one-site risk of the overhang-removal protocol satisfies
\begin{equation}
\label{eq:overall_one_site_asymptotics_supp}
{}^k\mathcal{L}_{\mathrm{one}}
=
\frac{1}{n}
\mleft(
\sum_{i\neq k}\frac{p_i}{D_{k,i}^2}
+
\sum_{i=k+1}^{I^\ast}
\mleft(D_{k,i}^{-1}-R^{-1}\mright)
\mright)
+
\lilo[d,\boldsymbol{p}]{n^{-1}} \, .
\end{equation}
\end{theorem}

\begin{proof}
By~\cref{cor:sector_wise_one_site_fidelity}, the sector-wise utility has a \(\overline{\yd{\varsigma}}\)-uniform expansion on \(\mathsf{R}_\epsilon\). By~\cref{lem:one_site_coefficient_lipschitz}, its first-order coefficient is Lipschitz in \(\overline{\yd{\varsigma}}\), uniformly in \(R\), so after SW averaging we may replace \(\overline{\yd{\varsigma}}\) by the spectrum \(\boldsymbol{p}\) and the microscopic terminal index \(i^\ast\) by the macroscopic terminal index \(I^\ast\). Applying~\cref{lem:sw_averaging_uniform_sector_expansion} to this uniform sector-wise expansion gives
\begin{equation}
{}^k\mathcal{F}_{\mathrm{one}}
=
1-
\frac{1}{n}
\mleft(
\sum_{i\neq k}\frac{p_i}{D_{k,i}^2}
+
\sum_{i=k+1}^{I^\ast}
\mleft(D_{k,i}^{-1}-R^{-1}\mright)
\mright)
+
\lilo[d,\boldsymbol{p}]{n^{-1}} \, .
\end{equation}
Passing from utility to risk proves~\cref{eq:overall_one_site_asymptotics_supp}.
\end{proof}

\section[Nonasymptotic analysis of the sector-wise utilities]{\NoCaseChange{Nonasymptotic analysis of the sector-wise utilities}}
\label{sec:qpa_irrep_level_all_copy}

\subsection{gYDs with constraints}
\label{subsec:constrained_yd_supp}
\subsubsection{Definitions and monotonicity of averages under constraints}
\label{subsubsec:definitions_monotonicity_averages}

For the lower bounds on CGCs, we need better control over averages of functions under the Weyl distribution.
In particular, we want to understand how these averages change when boxes are added or removed, especially for functions that depend only on a subset of the boxes in the YD. To this end, we introduce \emph{gYDs} and \emph{constraints} on them.

\begin{definition}
    Let $\yd{\lambda}_{\mathrm{g}} = \{(i_{1},j_{1}),\ldots,(i_{n},j_{n})\}$  be a finite subset of  $ \mathbb{Z}^2$. We call $\yd{\lambda}_{\mathrm{g}}$ a \emph{gYD}, and we refer to its elements $(i_{k},j_{k})$ as its boxes. A map $\wt{w}_{\mathrm{g}}:\yd{\lambda}_{\mathrm{g}} \rightarrow \{1,\ldots,d\}$, interpreted as a filling of the boxes of \(\yd{\lambda}_{\mathrm{g}}\), is called a \emph{generalized WT of $\yd{\lambda}_{\mathrm{g}}$} (gWT) if it is nondecreasing to the right and increasing downwards, i.e.,
    \begin{equation}\begin{aligned}
        \wt{w}_{\mathrm{g}}(i,j) \leq \wt{w}_{\mathrm{g}}(i,j+1) \quad \text{for} \quad (i,j),(i,j+1) \in \yd{\lambda}_{\mathrm{g}} \, , \\
        \wt{w}_{\mathrm{g}}(i,j) < \wt{w}_{\mathrm{g}}(i+1,j) \quad \text{for} \quad (i,j),(i+1,j) \in \yd{\lambda}_{\mathrm{g}} \, .
    \end{aligned}\end{equation}
    We denote this by $\wt{w}_{\mathrm{g}} \vdash_{d} \yd{\lambda}_{\mathrm{g}}$, and denote by \(\mathcal{W}_{\yd{\lambda}_{\mathrm{g}}}^{d}\) the set of all such gWTs.
\end{definition}

\begin{remark}
    If $\yd{\lambda}_{\mathrm{g}}$ corresponds to a (skew) YD, then $\wt{w}_{\mathrm{g}} \vdash_{d} \yd{\lambda}_{\mathrm{g}}$ just denotes a (skew) SSYT. Therefore, the definition above is a generalization of the usual definition. In this case, we use the following interchangeable descriptions
    \begin{equation}
\yd{\lambda}\dot{=}\yd{\lambda_{1},\ldots,\lambda_{d}} \quad \sim \quad \yd{\lambda}_{\mathrm{g}} = \left\{(i,j) \, \middle| \, 1\leq i \leq d \land 1\leq j \leq \lambda_{i}\right\} \, .
    \end{equation}
\end{remark}

\begin{remark}
    The definitions in this section are invariant under the translations $(i,j)\rightarrow (i+k,j+l)$ for $k,l\in \mathbb{Z}^2$, as there is a natural bijection between the respective gWTs.
\end{remark}

\begin{example}
    Panel (a) of~\cref{fig:gYT} shows the gYD
    \begin{equation}\begin{aligned}
        \yd{\lambda}_{\mathrm{g}} \dot{=} &\{(0,4),(0,5),(1,1),(1,2),(1,5),(2,2),(2,4),(4,3),(4,4)\}\subset \mathbb{Z}^{2}.
    \end{aligned}\end{equation}
    Translating every box by the same vector gives an equivalent gYD. For example,
    \begin{equation}
        \yd{\lambda^{\prime}}_{\mathrm{g}}\dot{=}\{(i+2,j-1)\mid(i,j)\in \yd{\lambda}_{\mathrm{g}}\}
    \end{equation}
    defines the same gYD.
\end{example}

We will also need a version of the concepts above with constraints.

\begin{definition}
    Let \(\yd{\lambda}_{\mathrm{g}}\) be a gYD. We define \(\boldsymbol{x}\) to be a \emph{constraint of \(\yd{\lambda}_{\mathrm{g}}\)} if it is a map that assigns two values to each box in \(\yd{\lambda}_{\mathrm{g}}\). Equivalently, we have
    \begin{equation}
        \boldsymbol{x}: \yd{\lambda}_{\mathrm{g}} \mapsto \{1,\ldots,d\}^2 \, , \quad \boldsymbol x(i,j) = (x_{l}(i,j),x_{u}(i,j)) \, .
    \end{equation}
    We further say that \(\wt{w}_{\mathrm{g}}\vdash_{d} \yd{\lambda}_{\mathrm{g}}|\boldsymbol{x}\) is a \emph{gWT with constraints} if
    \begin{equation}
        x_{l}(i,j) \leq \wt{w}_{\mathrm{g}}(i,j) \leq x_{u}(i,j) \quad \text{for all} \quad (i,j)\in \yd{\lambda}_{\mathrm{g}} \, ,
    \end{equation}
    and we write  \(\mathcal{W}_{\yd{\lambda}_{\mathrm{g}}|\boldsymbol{x}}^{d}\) for the set of all such gWTs.
\end{definition}
\begin{example}
    For the gYD in panel (b) of~\cref{fig:gYT}, the displayed constraint is
    \begin{equation}\begin{aligned}
        \boldsymbol{x}(0,4)&=(1,4), & \boldsymbol{x}(1,1)&=(1,4), & \boldsymbol{x}(1,2)&=(1,5), &
        \boldsymbol{x}(1,4)&=(2,4), \\ \boldsymbol{x}(1,5)&=(1,4), & \boldsymbol{x}(2,2)&=(1,4),&
        \boldsymbol{x}(2,4)&=(2,4), & \boldsymbol{x}(2,5)&=(1,4), \\
        \boldsymbol{x}(4,3)&=(1,2), & \boldsymbol{x}(4,4)&=(2,3).
    \end{aligned}\end{equation}
    One compatible gYT is the map
    \begin{equation}\begin{aligned}
        \wt{w}_{\mathrm{g}}(0,4)&=1,
        &\wt{w}_{\mathrm{g}}(1,1)&=3,
        &\wt{w}_{\mathrm{g}}(1,2)&=3,
        &\wt{w}_{\mathrm{g}}(1,4)&=2, \\
        \wt{w}_{\mathrm{g}}(1,5)&=2,
        &\wt{w}_{\mathrm{g}}(2,2)&=4,
        &\wt{w}_{\mathrm{g}}(2,4)&=3,
        &\wt{w}_{\mathrm{g}}(2,5)&=4, \\
        \wt{w}_{\mathrm{g}}(4,3)&=1,
        &\wt{w}_{\mathrm{g}}(4,4)&=2.
    \end{aligned}\end{equation}
    Thus \(\wt{w}_{\mathrm{g}}\vdash_{d}\yd{\lambda}_{\mathrm{g}}|\boldsymbol{x}\) for \(d\geq5\). By contrast, the filling \(\wt{w'}_{\mathrm{g}}\) obtained from \(\wt{w}_{\mathrm{g}}\) by replacing only \(\wt{w}_{\mathrm{g}}(2,2)=4\) with \(\wt{w'}_{\mathrm{g}}(2,2)=5\) is infeasible, since \(5>x_u(2,2)=4\).
\end{example}

\begin{remark}
    The trivial constraint given by \(x_{l}(i,j)=1\) and \(x_{u}(i,j)=d\) is equivalent to having no constraint on the box \((i,j)\).
\end{remark}

We define the weights of a gWT, and generalized (weighted) Schur polynomials.

\begin{definition}
    Let $\wt{w}_{\mathrm{g}}$ be a gWT. Then we define the \emph{row weight $\#_{k,\ell}\mleft(\wt{w}_{\mathrm{g}}\mright)$} as
    \begin{equation}\begin{aligned}
        \#_{k,\ell}\mleft(\wt{w}_{\mathrm{g}}\mright) \coloneqq \mleft|\{(\ell,j)\mid\wt{w}_{\mathrm{g}}(\ell,j) = k\}\mright| \, .
    \end{aligned}\end{equation}
    We interpret $\#_{kl}\mleft(\wt{w}_{\mathrm{g}}\mright)$ as the number of entries \(k\) in the \(\ell\)-th row. We further define the
    \emph{weight $\#_{k}\mleft(\wt{w}_{\mathrm{g}}\mright)$} as
    \begin{equation}\begin{aligned}
        \#_{k}\mleft(\wt{w}_{\mathrm{g}}\mright) \coloneqq \sum_{\ell} \#_{k,\ell}\mleft(\wt{w}_{\mathrm{g}}\mright) = \mleft|\{(i,j)\mid\wt{w}_{\mathrm{g}}(i,j) = k\}\mright| \, .
    \end{aligned}\end{equation}
    We interpret it as the total number of entries $k$. In addition, for some
    \begin{equation}
        \boldsymbol{p}\dot{=}(p_{1},\ldots,p_{d}) \in \mathbb{R}_{\geq 0}^{d} \, ,
    \end{equation}
    we write
    \begin{equation}
        \boldsymbol{p}^{\#\mleft(\wt{w}_{\mathrm{g}}\mright)} \coloneqq \prod_{i=1}^{d}p_{i}^{\#_{i}\mleft(\wt{w}_{\mathrm{g}}\mright)} \, .
    \end{equation}
\end{definition}

\begin{definition}
    Let $\yd{\lambda}_{\mathrm{g}}$ be a gYD, and let $\boldsymbol{x}$ be a constraint of $\yd{\lambda}_{\mathrm{g}}$. Further, let $F:\mathcal{W}_{\yd{\lambda}_{\mathrm{g}}}^{d} \rightarrow \mathbb{R}$ be a real-valued function, and let $\boldsymbol{p}\dot{=}(p_{1},\ldots,p_{d})\in \mathbb{R}_{+}^{d}$. Then we define the \emph{generalized Schur polynomial of $F$ with constraints} as
    \begin{equation}\begin{aligned}
        \schur[F]{\yd{\lambda}_{\mathrm{g}} | \boldsymbol{x}}[\boldsymbol{p}] = \sum_{\wt{w}_{\mathrm{g}}\vdash_{d}\yd{\lambda}_{\mathrm{g}}|\boldsymbol{x}} F\mleft(\wt{w}_{\mathrm{g}}\mright) \boldsymbol{p}^{\#\mleft(\wt{w}_{\mathrm{g}}\mright)} \, .
    \end{aligned}\end{equation}
    For $F\equiv 1$, we call $\schur[1]{\yd{\lambda}_{\mathrm{g}} | \boldsymbol{x}}[\boldsymbol{p}]$ the \emph{generalized Schur polynomial with constraints}, and we write
    \begin{equation}\begin{aligned}
        \schur{\yd{\lambda}_{\mathrm{g}} | \boldsymbol{x}}[\boldsymbol{p}] = \schur[1]{\yd{\lambda}_{\mathrm{g}} | \boldsymbol{x}}[\boldsymbol{p}] \, .
    \end{aligned}\end{equation}
    We can further define the \emph{Weyl average of $F$ on $\yd{\lambda}_{\mathrm{g}}| \boldsymbol{x}$ with constraints}. If $\schur{\yd{\lambda}_{\mathrm{g}} | \boldsymbol{x}}[\boldsymbol{p}]>0$, we set
    \begin{equation}
        \weylavg{F}{\boldsymbol{p}}{\yd{\lambda}_{\mathrm{g}} | \boldsymbol{x}} = \frac{\schur[F]{\yd{\lambda}_{\mathrm{g}} | \boldsymbol{x}}[\boldsymbol{p}]}{\schur{\yd{\lambda}_{\mathrm{g}} | \boldsymbol{x}}[\boldsymbol{p}]} \, .
\end{equation}
     For the case $\schur{\yd{\lambda}_{\mathrm{g}} | \boldsymbol{x}}[\boldsymbol{p}]=0$, we set $\weylavg{F}{\boldsymbol{p}}{\yd{\lambda}_{\mathrm{g}} | \boldsymbol{x}}=0$.
    In the case of a trivial constraint that is always fulfilled, we omit the \(\boldsymbol{x}\).
\end{definition}

\begin{remark}
    For a YD \(\yd{\lambda}\dot{=}\yd{\lambda_{1},\ldots,\lambda_{d}}\), we have
    \begin{equation}
        \schur[F]{\yd{\lambda}_{\mathrm{g}}}[\boldsymbol{p}] = \schur[F]{\yd{\lambda}}[\boldsymbol{p}] \, , \quad \weylavg{F}{\boldsymbol{p}}{\yd{\lambda}_{\mathrm{g}}} = \weylavg{F}{\boldsymbol{p}}{\yd{\lambda}} \, .
\end{equation}
Thus, these notions are also generalizations of the respective objects defined for YDs.
\end{remark}

Before we can state the main theorem of this section, we need two preparatory lemmas:~\cref{lem:constraints_incremental_sequence,lem:differences_inequalities}.

\begin{lemma}[Incremental constraint connections]
    \label{lem:constraints_incremental_sequence}
    Let $\boldsymbol{x},\boldsymbol{x}'$ be two constraints such that
    \begin{equation}\begin{aligned}
        x_{l}(i,j) \leq x'_{l}(i,j) \, , \quad x_{u}(i,j) \leq x'_{u}(i,j) \quad \text{for} \quad (i,j)\in \yd{\lambda}_{\mathrm{g}} \, ,
    \end{aligned}\end{equation}
    and let \(\schur{\yd{\lambda}_{\mathrm{g}} | \boldsymbol{x}}[\boldsymbol{p}]\,,\, \schur{\yd{\lambda}_{\mathrm{g}} | \boldsymbol{x}'}[\boldsymbol{p}] > 0\). Then there exists a sequence of $\boldsymbol{y}_{k}$ for $0\leq k \leq N$ such that
    \begin{equation}\begin{aligned}
        \boldsymbol{y}_{0} = \boldsymbol{x} \, , \quad \boldsymbol{y}_{N} = \boldsymbol{x}' \, ,
    \end{aligned}\end{equation}
    and there exist some $(i_{k},j_{k})$ and $t_{k}\in\{l,u\}$ such that
    \begin{equation}\begin{aligned}
        (y_{k})_{t+1}(i,j) = (y_{k})_{t}(i,j) + \delta_{i,i_{k}}\delta_{j,j_{k}}\delta_{t,t_{k}} \, ,
    \end{aligned}\end{equation}
    with the property that
    \begin{equation}\begin{aligned}
        \schur{\yd{\lambda}_{\mathrm{g}} | \boldsymbol{y}_{k}}[\boldsymbol{p}]>0 \quad \text{for all} \quad 0\leq k \leq N \, .
    \end{aligned}\end{equation}
\end{lemma}

\begin{proof}
    The fact that
    \begin{equation}
        \schur{\yd{\lambda}_{\mathrm{g}} | \boldsymbol{x}}[\boldsymbol{p}]\,,\, \schur{\yd{\lambda}_{\mathrm{g}} | \boldsymbol{x}'}[\boldsymbol{p}] > 0
    \end{equation}
    means that there exist
    \begin{equation}
        \wt{w}_{\mathrm{g}}\vdash_{d}\yd{\lambda}_{\mathrm{g}}|\boldsymbol{x} \, , \quad \wt{w'}_{\mathrm{g}}\vdash_{d}\yd{\lambda}_{\mathrm{g}}|\boldsymbol{x}'
    \end{equation}
    so that
    \begin{equation}
        \boldsymbol{p}^{\#\mleft(\wt{w}_{\mathrm{g}}\mright)}, \boldsymbol{p}^{\#\mleft(\wt{w'}_{\mathrm{g}}\mright)} > 0 \, .
    \end{equation}
    Increasing \(x_{u}(i,j)\) for any \((i,j)\in \yd{\lambda}_{\mathrm{g}}\) increases the number of allowed gWTs, so we can stepwise increase the upper bounds until we arrive at \(\boldsymbol{y}_{k^\ast}\) with
    \begin{equation}
    \begin{aligned}
        (y_{k^\ast})_{l}(i,j) = x_{l}(i,j) \, , \quad (y_{k^\ast})_{u}(i,j) = x'_{u}(i,j) \quad \text{for} \quad (i,j)\in \yd{\lambda}_{\mathrm{g}} \, ,
    \end{aligned}
    \end{equation}
    and the argument above shows
    \begin{equation}
        \wt{w}_{\mathrm{g}}\vdash_{d}\yd{\lambda}_{\mathrm{g}}|\boldsymbol{y}_{k} \quad \text{for} \quad 0\leq k \leq k^\ast
    \end{equation}
    and therefore
    \begin{equation}
    \begin{aligned}
        \schur{\yd{\lambda}_{\mathrm{g}} | \boldsymbol{y}_{k}}[\boldsymbol{p}]>0 \quad \text{for} \quad 0\leq k \leq k^\ast \, .
    \end{aligned}
    \end{equation}
    In addition, lowering \(x'_{l}(i,j)\) for any \((i,j)\in \yd{\lambda}_{\mathrm{g}}\) increases the number of allowed gWTs as well. We then increase the lower bounds until reaching \(\boldsymbol{y}_{N} = \boldsymbol{x}'\). By the argument about lowering \(x'_{l}(i,j)\), we again have
    \begin{equation}
        \wt{w'}_{\mathrm{g}}\vdash_{d}\yd{\lambda}_{\mathrm{g}}|\boldsymbol{y}_{k} \quad \text{for} \quad k^\ast\leq k \leq N
    \end{equation}
    and therefore
    \begin{equation}
    \begin{aligned}
        \schur{\yd{\lambda}_{\mathrm{g}} | \boldsymbol{y}_{k}}[\boldsymbol{p}]>0 \quad \text{for} \quad k^\ast\leq k \leq N \, .
    \end{aligned}
    \end{equation}
\end{proof}

\begin{lemma}[Relative difference inequalities]
\label{lem:differences_inequalities}
Let $0<A<B$, and let $a,b \in \mathbb{R}$. Then the following statements are equivalent:
\begin{equation}\begin{aligned}
    \frac{a}{A} \leq \frac{b}{B} \quad \Leftrightarrow \quad \frac{a}{A} \leq \frac{b-a}{B-A} \quad \Leftrightarrow \quad \frac{b}{B} \leq \frac{b-a}{B-A} \, .
\end{aligned}\end{equation}
The reversed inequalities are equivalent as well.
\end{lemma}

\begin{proof}
    Multiplying out, we have
    \begin{equation}\begin{aligned}
    aB \leq bA \quad \Leftrightarrow \quad a(B-A) \leq (b-a)A \quad \Leftrightarrow \quad b(B-A) \leq (b-a)B \, ,
\end{aligned}\end{equation}
which is easily verified.
\end{proof}

The following theorem is the main workhorse for gYDs with constraints, describing how the constraints modify the corresponding Weyl averages.

\begin{theorem}[Monotonicity of Weyl averages with constraints]
    \label{thm:monotonicity_constrained_generalized_averages}
    Let $\boldsymbol{x},\boldsymbol{x}'$ be two constraints such that
    \begin{equation}\begin{aligned}
        x_{l}(i,j) \leq x'_{l}(i,j) \, , \quad x_{u}(i,j) \leq x'_{u}(i,j) \quad \text{for} \quad (i,j)\in \yd{\lambda}_{\mathrm{g}} \, ,
    \end{aligned}\end{equation}
    and let
    \begin{equation}\begin{aligned}
        \schur{\yd{\lambda}_{\mathrm{g}} | \boldsymbol{x}}[\boldsymbol{p}]\,,\, \schur{\yd{\lambda}_{\mathrm{g}} | \boldsymbol{x}'}[\boldsymbol{p}] \, > 0.
    \end{aligned}\end{equation}
    Further, let $F:\mathcal{W}_{\yd{\lambda}_{\mathrm{g}}}^{d} \rightarrow \mathbb{R}$ be nondecreasing in every argument, that is, for \(\wt{w}_{\mathrm{g}}\leq \wt{w'}_{\mathrm{g}}\) pointwise we have
    \begin{equation}
        F\mleft(\wt{w}_{\mathrm{g}}\mright) \leq F\mleft(\wt{w'}_{\mathrm{g}}\mright) \, .
    \end{equation}
    Then we get
    \begin{equation}
        \weylavg{F}{\boldsymbol{p}}{\yd{\lambda}_{\mathrm{g}} | \boldsymbol{x}} \leq \weylavg{F}{\boldsymbol{p}}{\yd{\lambda}_{\mathrm{g}} | \boldsymbol{x}'} \, .
\end{equation}
\end{theorem}

\begin{remark}
    The content is that larger lower constraints \(x_{l}\) force the gWTs to have higher entries, while larger upper constraints allow the gWTs to have higher entries.
\end{remark}

\begin{proof}
    We use induction on the number of boxes $n$. Assume the statement holds for all gYDs $\yd{\lambda}_{\mathrm{g}}$ with $n$ boxes and for all constraints. We consider $\mleft|\yd{\lambda}_{\mathrm{g}}\mright|=n+1$. By~\cref{lem:constraints_incremental_sequence}, we can assume that \(\boldsymbol{x},\boldsymbol{x}'\) differ only on one site \((i^\ast,j^\ast)\), and we have
    \(x'_{t^\ast}(i^\ast,j^\ast) = x_{t^\ast}(i^\ast,j^\ast) + 1\)
    for \(t^\ast \in \{l,u\}\).

    (I) We first treat the case where \(t^\ast = l\). Let
    \begin{equation}
        \yd{\nu}_{\mathrm{g}} \coloneqq \yd{\lambda}_{\mathrm{g}} \setminus (i^\ast,j^\ast) \, .
    \end{equation}
    For \(m\in\{1,\ldots,d\}\), we define the constraints \(\boldsymbol{z}_{m}\) on \(\yd{\nu}_{\mathrm{g}}\) as
    \begin{subequations}
    \begin{align}
        (z_{m})_{l}(i,j)
        &\coloneqq
        \mleft\{\begin{alignedat}{2}
            &\max(m+1,x_{l}(i,j)), & (i,j) = (i^\ast,j^\ast + 1), \\
            &\max(m,x_{l}(i,j)), & (i,j) = (i^\ast + 1, j^\ast), \\
            &x_{l}(i,j), & \text{otherwise},
        \end{alignedat}\mright. \\
        (z_{m})_{u}(i,j)
        &\coloneqq
        \mleft\{\begin{alignedat}{2}
            &\min (m-1,x_{u}(i,j)), & (i,j) = (i^\ast,j^\ast - 1), \\
            &\min (m,x_{u}(i,j)), & (i,j) = (i^\ast - 1, j^\ast), \\
            &x_{u}(i,j), & \text{otherwise}.
        \end{alignedat}\mright.
    \end{align}
    \end{subequations}
    A direct calculation gives
    \begin{equation}
        \schur{\yd{\lambda}_{\mathrm{g}} | \boldsymbol{x}}[\boldsymbol{p}] = \sum_{m=x_{l}(i^\ast,j^\ast)}^{x_{u}(i^\ast,j^\ast)} p_{m}\schur{\yd{\nu}_{\mathrm{g}} | \boldsymbol{z}_{m}}[\boldsymbol{p}] \, ,
    \end{equation}
    where we can interpret the term \(p_{m}\schur{\yd{\nu}_{\mathrm{g}} | \boldsymbol{z}_{m}}[\boldsymbol{p}]\) as the weighted sum over all gWTs of \(\yd{\lambda}_{\mathrm{g}}| \boldsymbol{x}\) with entry \(m\) in the box \((i^\ast,j^\ast)\). More generally, let \(F_m^\ast:\mathcal{W}_{\yd{\nu}_{\mathrm{g}}}^d\to\mathbb{R}\) be defined as \(F\) but with the entry on \((i^\ast,j^\ast)\) replaced by \(m\). Then we have
    \begin{equation}
        \label{equ:proof_monotonicity_constrained_generalized_averages_sum_decomposition_lower}
        \schur[F]{\yd{\lambda}_{\mathrm{g}} | \boldsymbol{x}}[\boldsymbol{p}] = \sum_{m=x_{l}(i^\ast,j^\ast)}^{x_{u}(i^\ast,j^\ast)} p_{m}\schur[F_{m}^\ast]{\yd{\nu}_{\mathrm{g}} | \boldsymbol{z}_{m}}[\boldsymbol{p}] \, .
    \end{equation}
    Now let \(x_{l}(i^\ast,j^\ast) = L\). By the equation above, we get
    \begin{equation}\begin{aligned}
        \label{equ:proof_monotonicity_constrained_generalized_averages_decomposition_lower}
        \schur[F]{\yd{\lambda}_{\mathrm{g}} | \boldsymbol{x}}[\boldsymbol{p}] - \schur[F]{\yd{\lambda}_{\mathrm{g}} | \boldsymbol{x}'}[\boldsymbol{p}] = p_{L}\schur[F_{L}^\ast]{\yd{\nu}_{\mathrm{g}} | \boldsymbol{z}_{L}}[\boldsymbol{p}].
    \end{aligned}\end{equation}
    We first consider the case where \(p_{L}\schur{\yd{\nu}_{\mathrm{g}} | \boldsymbol{z}_{L}}[\boldsymbol{p}] = 0\), so either \(p_L=0\) or \(\schur{\yd{\nu}_{\mathrm{g}} | \boldsymbol{z}_{L}}[\boldsymbol{p}]=0\). In either case, \(p_{L}\schur[F_{L}^\ast]{\yd{\nu}_{\mathrm{g}} | \boldsymbol{z}_{L}}[\boldsymbol{p}] = 0\), and therefore, by~\cref{equ:proof_monotonicity_constrained_generalized_averages_decomposition_lower}, we get
    \begin{equation}
            \weylavg{F}{\boldsymbol{p}}{\yd{\lambda}_{\mathrm{g}} | \boldsymbol{x}} = \weylavg{F}{\boldsymbol{p}}{\yd{\lambda}_{\mathrm{g}} | \boldsymbol{x}'} \, .
\end{equation}
    It remains to consider \(p_{L}\schur{\yd{\nu}_{\mathrm{g}} | \boldsymbol{z}_{L}}[\boldsymbol{p}] > 0\), equivalently \(p_L>0\) and \(\schur{\yd{\nu}_{\mathrm{g}} | \boldsymbol{z}_{L}}[\boldsymbol{p}]>0\). We have by~\cref{equ:proof_monotonicity_constrained_generalized_averages_sum_decomposition_lower} that
    \begin{equation}
    \begin{aligned}
    \label{equ:proof_monotonicity_constrained_generalized_averages_decomposition_inequality_chain_lower}
        &\schur[F]{\yd{\lambda}_{\mathrm{g}} | \boldsymbol{x}'}[\boldsymbol{p}] = \sum_{m=L+1}^{x_{u}(i^\ast,j^\ast)} p_{m} \schur[F_{m}^\ast]{\yd{\nu}_{\mathrm{g}} | \boldsymbol{z}_{m}}[\boldsymbol{p}]\\ =& \sum_{m=L+1}^{x_{u}(i^\ast,j^\ast)} p_{m} \schur{\yd{\nu}_{\mathrm{g}} | \boldsymbol{z}_{m}}[\boldsymbol{p}] \weylavg{F_{m}^\ast}{\boldsymbol{p}}{\yd{\nu}_{\mathrm{g}} | \boldsymbol{z}_{m}}
        \geq \sum_{m=L+1}^{x_{u}(i^\ast,j^\ast)} p_{m} \schur{\yd{\nu}_{\mathrm{g}} | \boldsymbol{z}_{m}}[\boldsymbol{p}] \weylavg{F_{L}^\ast}{\boldsymbol{p}}{\yd{\nu}_{\mathrm{g}} | \boldsymbol{z}_{m}}\\ \geq &\weylavg{F_{L}^\ast}{\boldsymbol{p}}{\yd{\nu}_{\mathrm{g}} | \boldsymbol{z}_{L}} \sum_{m=L+1}^{x_{u}(i^\ast,j^\ast)} p_{m} \schur{\yd{\nu}_{\mathrm{g}} | \boldsymbol{z}_{m}}[\boldsymbol{p}]
        = \weylavg{F_{L}^\ast}{\boldsymbol{p}}{\yd{\nu}_{\mathrm{g}} | \boldsymbol{z}_{L}} \schur{\yd{\lambda}_{\mathrm{g}} | \boldsymbol{x}'}[\boldsymbol{p}] \, ,
    \end{aligned}
\end{equation}
    where we used in the first inequality that $F_{m}^\ast \geq F_{L}^\ast$ pointwise for $m> L$, since $F$ is nondecreasing in every argument. For the cases where $\schur{\yd{\nu}_{\mathrm{g}} | \boldsymbol{z}_{m}}[\boldsymbol{p}]=0$, the second inequality is trivial. For the other cases, we used the induction hypothesis. Since $\schur{\yd{\lambda}_{\mathrm{g}} | \boldsymbol{x}'}[\boldsymbol{p}]>0$ by assumption, we have
    \begin{equation}
            \frac{\schur[F]{\yd{\nu}_{\mathrm{g}} | \boldsymbol{z}_{L}}[\boldsymbol{p}]}{\schur{\yd{\nu}_{\mathrm{g}} | \boldsymbol{z}_{L}}[\boldsymbol{p}]} = \weylavg{F_{L}^\ast}{\boldsymbol{p}}{\yd{\nu}_{\mathrm{g}} | \boldsymbol{z}_{L}} \leq \frac{\schur[F]{\yd{\lambda}_{\mathrm{g}} | \boldsymbol{x}'}[\boldsymbol{p}]}{\schur{\yd{\lambda}_{\mathrm{g}} | \boldsymbol{x}'}[\boldsymbol{p}]} = \weylavg{F}{\boldsymbol{p}}{\yd{\lambda}_{\mathrm{g}} | \boldsymbol{x}'} \, .
\end{equation}
Here, we used that \(\schur{\yd{\nu}_{\mathrm{g}} | \boldsymbol{z}_{L}}[\boldsymbol{p}] > 0\). From~\cref{equ:proof_monotonicity_constrained_generalized_averages_decomposition_lower}, setting \(F\equiv 1\) and \((p_{L}\schur{\yd{\nu}_{\mathrm{g}} | \boldsymbol{z}_{L}}[\boldsymbol{p}] > 0)\), we obtain $\schur{\yd{\lambda}_{\mathrm{g}} | \boldsymbol{x}}[\boldsymbol{p}] - \schur{\yd{\lambda}_{\mathrm{g}} | \boldsymbol{x}'}[\boldsymbol{p}] > 0$, and hence
    \begin{equation}
    \label{equ:proof_monotonicity_constrained_generalized_averages_decomposition_difference_lower}
            \frac{\schur[F]{\yd{\lambda}_{\mathrm{g}} | \boldsymbol{x}}[\boldsymbol{p}] - \schur[F]{\yd{\lambda}_{\mathrm{g}} | \boldsymbol{x}'}[\boldsymbol{p}]}{\schur{\yd{\lambda}_{\mathrm{g}} | \boldsymbol{x}}[\boldsymbol{p}] - \schur{\yd{\lambda}_{\mathrm{g}} | \boldsymbol{x}'}[\boldsymbol{p}]} = \frac{p_{L}\schur[F_{L}^\ast]{\yd{\nu}_{\mathrm{g}} | \boldsymbol{z}_{L}}[\boldsymbol{p}]}{p_{L}\schur{\yd{\nu}_{\mathrm{g}} | \boldsymbol{z}_{L}}[\boldsymbol{p}]} = \weylavg{F_{L}^\ast}{\boldsymbol{p}}{\yd{\nu}_{\mathrm{g}} | \boldsymbol{z}_{L}} \leq \frac{\schur[F]{\yd{\lambda}_{\mathrm{g}} | \boldsymbol{x}'}[\boldsymbol{p}]}{\schur{\yd{\lambda}_{\mathrm{g}} | \boldsymbol{x}'}[\boldsymbol{p}]} = \weylavg{F}{\boldsymbol{p}}{\yd{\lambda}_{\mathrm{g}} | \boldsymbol{x}'} \, .
\end{equation}
~\cref{lem:differences_inequalities} then implies
    \begin{equation}
            \weylavg{F}{\boldsymbol{p}}{\yd{\lambda}_{\mathrm{g}} | \boldsymbol{x}} = \frac{\schur[F]{\yd{\lambda}_{\mathrm{g}} | \boldsymbol{x}}[\boldsymbol{p}]}{\schur{\yd{\lambda}_{\mathrm{g}} | \boldsymbol{x}}[\boldsymbol{p}]} \leq \frac{\schur[F]{\yd{\lambda}_{\mathrm{g}} | \boldsymbol{x}'}[\boldsymbol{p}]}{\schur{\yd{\lambda}_{\mathrm{g}} | \boldsymbol{x}'}[\boldsymbol{p}]} = \weylavg{F}{\boldsymbol{p}}{\yd{\lambda}_{\mathrm{g}} | \boldsymbol{x}'} \, .
\end{equation}
    \noindent (II) We now consider the case $t^{*}=u$. We take \(U = x'_{u}(i^\ast,j^\ast)\), and define \(\yd{\nu}_{\mathrm{g}}, \boldsymbol{z}_{m}\) and \(F_{m}^\ast\) as before. Similarly, we get
    \begin{equation}\begin{aligned}
        \label{equ:proof_monotonicity_constrained_generalized_averages_decomposition_upper}
        \schur[F]{\yd{\lambda}_{\mathrm{g}} | \boldsymbol{x}'}[\boldsymbol{p}] - \schur[F]{\yd{\lambda}_{\mathrm{g}} | \boldsymbol{x}}[\boldsymbol{p}] = p_{U}\schur[F_{U}^\ast]{\yd{\nu}_{\mathrm{g}} | \boldsymbol{z}_{U}}[\boldsymbol{p}] \, .
    \end{aligned}\end{equation}
    If \(p_{U}\schur{\yd{\nu}_{\mathrm{g}} | \boldsymbol{z}_{U}}[\boldsymbol{p}] = 0\), we again obtain
    \begin{equation}
            \weylavg{F}{\boldsymbol{p}}{\yd{\lambda}_{\mathrm{g}} | \boldsymbol{x}} = \weylavg{F}{\boldsymbol{p}}{\yd{\lambda}_{\mathrm{g}} | \boldsymbol{x}'} \, .
\end{equation}
For the case \(p_{U}\schur{\yd{\nu}_{\mathrm{g}} | \boldsymbol{z}_{U}}[\boldsymbol{p}] > 0\) we have
    \begin{equation}
    \begin{aligned}
        \schur[F]{\yd{\lambda}_{\mathrm{g}} | \boldsymbol{x}}[\boldsymbol{p}] = &\sum_{m=x_{l}(i^\ast,j^\ast)}^{U-1} p_{m} \schur[F_{m}^\ast]{\yd{\nu}_{\mathrm{g}} | \boldsymbol{z}_{m}}[\boldsymbol{p}] = \sum_{m=x_{l}(i^\ast,j^\ast)}^{U-1} p_{m} \schur{\yd{\nu}_{\mathrm{g}} | \boldsymbol{z}_{m}}[\boldsymbol{p}] \weylavg{F_{m}^\ast}{\boldsymbol{p}}{\yd{\nu}_{\mathrm{g}} | \boldsymbol{z}_{m}} \\
        \leq& \sum_{m=x_{l}(i^\ast,j^\ast)}^{U-1} p_{m} \schur{\yd{\nu}_{\mathrm{g}} | \boldsymbol{z}_{m}}[\boldsymbol{p}] \weylavg{F_{U}^\ast}{\boldsymbol{p}}{\yd{\nu}_{\mathrm{g}} | \boldsymbol{z}_{m}} \leq \weylavg{F_{U}^\ast}{\boldsymbol{p}}{\yd{\nu}_{\mathrm{g}} | \boldsymbol{z}_{U}} \sum_{m=x_{l}(i^\ast,j^\ast)}^{U-1} p_{m} \schur{\yd{\nu}_{\mathrm{g}} | \boldsymbol{z}_{m}}[\boldsymbol{p}] \\
        = &\weylavg{F_{U}^\ast}{\boldsymbol{p}}{\yd{\nu}_{\mathrm{g}} | \boldsymbol{z}_{U}} \schur{\yd{\lambda}_{\mathrm{g}} | \boldsymbol{x}}[\boldsymbol{p}] \, ,
    \end{aligned}
\end{equation}
    where we proceeded as in~\cref{equ:proof_monotonicity_constrained_generalized_averages_decomposition_inequality_chain_lower}. Since $\schur{\yd{\lambda}_{\mathrm{g}} | \boldsymbol{x}}[\boldsymbol{p}]>0$ by assumption, we have
    \begin{equation}
            \frac{\schur[F]{\yd{\nu}_{\mathrm{g}} | \boldsymbol{z}_{U}}[\boldsymbol{p}]}{\schur{\yd{\nu}_{\mathrm{g}} | \boldsymbol{z}_{U}}[\boldsymbol{p}]} = \weylavg{F_{U}^\ast}{\boldsymbol{p}}{\yd{\nu}_{\mathrm{g}} | \boldsymbol{z}_{U}} \geq \frac{\schur[F]{\yd{\lambda}_{\mathrm{g}} | \boldsymbol{x}}[\boldsymbol{p}]}{\schur{\yd{\lambda}_{\mathrm{g}} | \boldsymbol{x}}[\boldsymbol{p}]} = \weylavg{F}{\boldsymbol{p}}{\yd{\lambda}_{\mathrm{g}} | \boldsymbol{x}} \, .
\end{equation}
    Here we used \(\schur{\yd{\nu}_{\mathrm{g}} | \boldsymbol{z}_{U}}[\boldsymbol{p}] > 0\). As before, we see from~\cref{equ:proof_monotonicity_constrained_generalized_averages_decomposition_upper}, setting \(F\equiv 1\) and \(p_{U}\schur{\yd{\nu}_{\mathrm{g}} | \boldsymbol{z}_{U}}[\boldsymbol{p}] > 0\) that $\schur{\yd{\lambda}_{\mathrm{g}} | \boldsymbol{x}'}[\boldsymbol{p}] - \schur{\yd{\lambda}_{\mathrm{g}} | \boldsymbol{x}}[\boldsymbol{p}] > 0$, so we get
    \begin{equation}
            \frac{\schur[F]{\yd{\lambda}_{\mathrm{g}} | \boldsymbol{x}'}[\boldsymbol{p}] - \schur[F]{\yd{\lambda}_{\mathrm{g}} | \boldsymbol{x}}[\boldsymbol{p}]}{\schur{\yd{\lambda}_{\mathrm{g}} | \boldsymbol{x}'}[\boldsymbol{p}] - \schur{\yd{\lambda}_{\mathrm{g}} | \boldsymbol{x}}[\boldsymbol{p}]} = \frac{p_{U}\schur[F_{U}^\ast]{\yd{\nu}_{\mathrm{g}} | \boldsymbol{z}_{U}}[\boldsymbol{p}]}{p_{U}\schur{\yd{\nu}_{\mathrm{g}} | \boldsymbol{z}_{U}}[\boldsymbol{p}]} = \weylavg{F_{U}^\ast}{\boldsymbol{p}}{\yd{\nu}_{\mathrm{g}} | \boldsymbol{z}_{U}} \geq \frac{\schur[F]{\yd{\lambda}_{\mathrm{g}} | \boldsymbol{x}}[\boldsymbol{p}]}{\schur{\yd{\lambda}_{\mathrm{g}} | \boldsymbol{x}}[\boldsymbol{p}]} = \weylavg{F}{\boldsymbol{p}}{\yd{\lambda}_{\mathrm{g}} | \boldsymbol{x}} \, .
\end{equation}
Again, it follows from~\cref{lem:differences_inequalities} that
    \begin{equation}
            \weylavg{F}{\boldsymbol{p}}{\yd{\lambda}_{\mathrm{g}} | \boldsymbol{x}} = \frac{\schur[F]{\yd{\lambda}_{\mathrm{g}} | \boldsymbol{x}}[\boldsymbol{p}]}{\schur{\yd{\lambda}_{\mathrm{g}} | \boldsymbol{x}}[\boldsymbol{p}]} \leq \frac{\schur[F]{\yd{\lambda}_{\mathrm{g}} | \boldsymbol{x}'}[\boldsymbol{p}]}{\schur{\yd{\lambda}_{\mathrm{g}} | \boldsymbol{x}'}[\boldsymbol{p}]} = \weylavg{F}{\boldsymbol{p}}{\yd{\lambda}_{\mathrm{g}} | \boldsymbol{x}'} \, .
\end{equation}
\end{proof}

\subsubsection{Restrictions and Schur monotonicity}
\label{subsubsec:log_convexity_Schur_polynomials}

We investigate the implications that~\cref{thm:main_part_monotonicity_constrained_generalized_averages} has on the relations between Schur polynomials. We first introduce some notation.

\begin{definition}
    Let \(\yd{\nu}_{\mathrm{g}} \subseteq \yd{\lambda}_{\mathrm{g}}\), and let \(\yd{\mu}_{\mathrm{g}} = \yd{\lambda}_{\mathrm{g}} \setminus \yd{\nu}_{\mathrm{g}}\). For \(\wt{w}_{\mathrm{g}}\vdash_{d} \yd{\lambda}_{\mathrm{g}}\), we define \(\wt{w|_{\yd{\nu}_{\mathrm{g}}}}_{\mathrm{g}}\) as the \emph{restriction of \(\wt{w}_{\mathrm{g}}\) onto \(\yd{\nu}_{\mathrm{g}}\)}, with
    \begin{equation}
        \wt{w|_{\yd{\nu}_{\mathrm{g}}}}_{\mathrm{g}}(i,j) = \wt{w}_{\mathrm{g}}(i,j) \quad \text{for} \quad (i,j)\in \yd{\nu}_{\mathrm{g}} \, .
    \end{equation}
    For \(\wt{w'}_{\mathrm{g}}\vdash_{d} \yd{\nu}_{\mathrm{g}}\) and \(\wt{w''}_{\mathrm{g}}\vdash_{d} \yd{\mu}_{\mathrm{g}}\), we define their \emph{union} as
    \begin{equation}
        \wt{w'}_{\mathrm{g}} \cup \wt{w''}_{\mathrm{g}} = \wt{w} \Leftrightarrow \wt{w|_{\yd{\nu}_{\mathrm{g}}}}_{\mathrm{g}} = \wt{w'}_{\mathrm{g}} \quad \text{and} \quad \wt{w|_{\yd{\mu}_{\mathrm{g}}}}_{\mathrm{g}} = \wt{w''}_{\mathrm{g}} \, .
    \end{equation}
    For \(\wt{w''}_{\mathrm{g}}\vdash_{d} \yd{\mu}_{\mathrm{g}}\), we define a \emph{corresponding restriction} \(\boldsymbol{x}\mleft(\wt{w''}_{\mathrm{g}}\mright)\) as
    \begin{equation}
    \begin{aligned}
        \label{equ:definition_corresponding_restrictions}
        \mleft(x\mleft(\wt{w''}_{\mathrm{g}}\mright)\mright)_{l}(i,j) \coloneqq \mleft\{\begin{aligned}
                        &\wt{w''}_{\mathrm{g}}(i,j-1), & (i,j-1)\in \yd{\mu}_{\mathrm{g}}, \\
                        &\wt{w''}_{\mathrm{g}}(i-1,j) + 1, & (i-1,j)\in \yd{\mu}_{\mathrm{g}}, \\
                        &\max\mleft(\wt{w''}_{\mathrm{g}}(i,j-1),\wt{w''}_{\mathrm{g}}(i-1,j) + 1\mright), & (i,j-1),(i-1,j)\in \yd{\mu}_{\mathrm{g}},\\
                        &1, & \text{otherwise},
                    \end{aligned}\mright.
                    \\
                    \mleft(x\mleft(\wt{w''}_{\mathrm{g}}\mright)\mright)_{u}(i,j) \coloneqq
                    \mleft\{\begin{aligned}
                        &\wt{w''}_{\mathrm{g}}(i,j+1), & (i,j+1)\in \yd{\mu}_{\mathrm{g}}, \\
                        &\wt{w''}_{\mathrm{g}}(i+1,j) - 1, & (i+1,j)\in \yd{\mu}_{\mathrm{g}}, \\
                        &\min\mleft(\wt{w''}_{\mathrm{g}}(i,j+1),\wt{w''}_{\mathrm{g}}(i+1,j)-1\mright), & (i,j+1),(i+1,j)\in \yd{\mu}_{\mathrm{g}},\\
                        &d, & \text{otherwise}.
                    \end{aligned}\mright.
            \end{aligned}
        \end{equation}
    The entry in every box is constrained between \(1\) and \(d\), i.e., it is left unconstrained, unless a box directly to the left, above, to the right or below belongs to \(\yd{\mu}_{\mathrm{g}}\).
In those cases, we take the filling rules for gWTs -- nondecreasing to the right and increasing toward the bottom -- to obtain constraints between the neighboring boxes.
\end{definition}

\begin{remark}
    One checks that
    \begin{equation}
        \wt{w'} \vdash_{d} \yd{\nu}_{\mathrm{g}}|\boldsymbol{x}\mleft(\wt{w''}_{\mathrm{g}}\mright) \Leftrightarrow \wt{w'} \cup \wt{w''} \vdash_{d} \yd{\lambda}_{\mathrm{g}} \, ,
    \end{equation}
    That is, the corresponding constraint limits the gWTs of \(\yd{\nu}_{\mathrm{g}}\) to those that form a larger gWT of \(\yd{\lambda}_{\mathrm{g}}\) when combined with \(\wt{w''}_{\mathrm{g}}\).
\end{remark}

\begin{example}
    Panel (f) of~\cref{fig:gYT} illustrates the corresponding restriction for the same gYD as in panels (b)--(d). Let
    \begin{equation}
    \begin{aligned}
        &\yd{\lambda}_{\mathrm{g}}\dot{=}\{(0,4),(1,1),(1,2),(1,4),(1,5),(2,2),(2,4),(2,5),(4,3),(4,4)\}, \\
        &\yd{\nu}_{\mathrm{g}}\dot{=}\{(0,4),(1,1),(1,4),(1,5),(4,4)\},
        \quad
        \yd{\mu}_{\mathrm{g}}\dot{=}\{(1,2),(2,2),(2,4),(2,5),(4,3)\}.
    \end{aligned}
    \end{equation}
    For \(d\geq3\), fix the filling \(\wt{u}_{\mathrm{g}}\vdash_{d}\yd{\mu}_{\mathrm{g}}\) by
    \begin{equation}\begin{aligned}
        \wt{u}_{\mathrm{g}}(0,4)&=1,
        &\wt{u}_{\mathrm{g}}(1,1)&=1,
        &\wt{u}_{\mathrm{g}}(1,4)&=2,\\
        \wt{u}_{\mathrm{g}}(1,5)&=1,
        &\wt{u}_{\mathrm{g}}(4,4)&=2.
    \end{aligned}\end{equation}
    Then the corresponding restriction on \(\yd{\nu}_{\mathrm{g}}\) is
    \begin{equation}\begin{aligned}
        \boldsymbol{x}\mleft(\wt{u}_{\mathrm{g}}\mright)(1,2)&=(1,d),
        &\boldsymbol{x}\mleft(\wt{u}_{\mathrm{g}}\mright)(2,2)&=(1,d),\\
        \boldsymbol{x}\mleft(\wt{u}_{\mathrm{g}}\mright)(2,4)&=(3,d),
        &\boldsymbol{x}\mleft(\wt{u}_{\mathrm{g}}\mright)(2,5)&=(3,d),\\
        \boldsymbol{x}\mleft(\wt{u}_{\mathrm{g}}\mright)(4,3)&=(1,1).
    \end{aligned}\end{equation}
    For instance, the filling \(\wt{v}_{\mathrm{g}}\) of \(\yd{\nu}_{\mathrm{g}}\) defined by
    \begin{equation}\begin{aligned}
        \wt{v}_{\mathrm{g}}(1,1)&=1,
        &\wt{v}_{\mathrm{g}}(0,4)&=1,
        &\wt{v}_{\mathrm{g}}(1,4)&=2,\\
        \wt{v}_{\mathrm{g}}(1,5)&=2,
        &\wt{v}_{\mathrm{g}}(4,4)&=1
    \end{aligned}\end{equation}
    satisfies \(\wt{v}_{\mathrm{g}}\vdash_{d}\yd{\nu}_{\mathrm{g}}|\boldsymbol{x}\mleft(\wt{u}_{\mathrm{g}}\mright)\) for \(d\geq3\), and \(\wt{v}_{\mathrm{g}}\cup\wt{u}_{\mathrm{g}}\vdash_{d}\yd{\lambda}_{\mathrm{g}}\).
\end{example}

The following theorem was already proven in Ref.~\cite[Theorem~5]{LPP05S}. We provide an alternative proof, highlighting the utility of~\cref{thm:monotonicity_constrained_generalized_averages} while keeping the argument more self-contained.

\begin{theorem}[log-convexity of Schur polynomials]
\label{thm:inequality_Schur_polynomials_products}
    Let \(\yd{\lambda_{1}},\yd{\lambda_{2}},\yd{\nu_{1}},\yd{\nu_{2}}\) be YDs with
    \begin{equation}
        (\nu_{1})_{i} \leq (\lambda_{1})_{i} \, ,\quad (\nu_{2})_{i} \leq (\lambda_{2})_{i}
    \end{equation}
    for all \(1\leq i \leq d\). Further, let \(\yd{\lambda_{\min}},\yd{\lambda_{\max}}\) be given by
    \begin{equation}
    \begin{aligned}
        &(\lambda_{\min})_{i} \coloneqq \min\mleft((\lambda_{1})_{i},(\lambda_{2})_{i}\mright) \, ,\quad &(\nu_{\min})_{i} \coloneqq \min\mleft((\nu_{1})_{i},(\nu_{2})_{i}\mright) \, , \\
        &(\lambda_{\max})_{i} \coloneqq \max\mleft((\lambda_{1})_{i},(\lambda_{2})_{i}\mright) \, ,\quad &(\nu_{\max})_{i} \coloneqq \max\mleft((\nu_{1})_{i},(\nu_{2})_{i}\mright) \, ,
    \end{aligned}
    \end{equation}
    for \(1\leq i \leq d\). Then we have
    \begin{equation}
        \schur{\yd{\lambda_{\max}}\setminus \yd{\nu_{\max}}}[\boldsymbol{p}] \schur{\yd{\lambda_{\min}}\setminus \yd{\nu_{\min}}}[\boldsymbol{p}] \geq \schur{\yd{\lambda_{1}}\setminus \yd{\nu_{1}}}[\boldsymbol{p}] \schur{\yd{\lambda_{2}}\setminus \yd{\nu_{2}}}[\boldsymbol{p}] \, .
    \end{equation}
    Here we interpret \(\yd{\lambda}\setminus\yd{\nu}\) as the skew YD given by removing the smaller diagram from the larger one.
\end{theorem}

\begin{proof}
    We define the following gYDs
    \begin{equation}
    \begin{aligned}
        \yd{\nu_{1\Delta2}}_{\mathrm{g}} \coloneqq \yd{\nu_{1}}_{\mathrm{g}} \setminus \yd{\nu_{2}}_{\mathrm{g}} \, ,\quad \yd{\nu_{2\Delta1}}_{\mathrm{g}} \coloneqq \yd{\nu_{2}}_{\mathrm{g}} \setminus \yd{\nu_{1}}_{\mathrm{g}} \, ,\quad \yd{\nu_{1\cap2}}_{\mathrm{g}} \coloneqq \yd{\nu_{1}}_{\mathrm{g}} \cap \yd{\nu_{2}}_{\mathrm{g}} \, , \\
        \yd{\lambda_{1\Delta2}}_{\mathrm{g}} \coloneqq \yd{\lambda_{1}}_{\mathrm{g}} \setminus \yd{\lambda_{2}}_{\mathrm{g}} \, ,\quad \yd{\lambda_{2\Delta1}}_{\mathrm{g}} \coloneqq \yd{\lambda_{2}}_{\mathrm{g}} \setminus \yd{\lambda_{1}}_{\mathrm{g}} \, ,\quad \yd{\lambda_{1\cap2}}_{\mathrm{g}} \coloneqq \yd{\lambda_{1}}_{\mathrm{g}} \cap \yd{\lambda_{2}}_{\mathrm{g}} \, .
    \end{aligned}
    \end{equation}
    Instead of the alphabet \(\{1,\ldots,d\}\), we work with the extended alphabet \(\{0,\ldots,d+1\}\). We further define
    \begin{equation}
        \yd{\gamma}_{\mathrm{g}} \coloneqq \yd{\lambda_{\max}}_{\mathrm{g}} \setminus \yd{\nu_{2}}_{\mathrm{g}} \, ,
    \end{equation}
    and we remark here that
    \begin{align}
        \mleft(\yd{\lambda_{\max}}_{\mathrm{g}} \setminus \yd{\nu_{\max}}_{\mathrm{g}} \mright) \cup \yd{\nu_{1\Delta2}}_{\mathrm{g}} = \yd{\gamma}_{\mathrm{g}} = \mleft(\yd{\lambda_{2}}_{\mathrm{g}} \setminus \yd{\nu_{2}}_{\mathrm{g}} \mright) \cup \yd{\lambda_{1\Delta2}}_{\mathrm{g}} \, .
    \end{align}
    We define the constraints \(\boldsymbol{x},\boldsymbol{x}'\) on \(\yd{\gamma}_{\mathrm{g}}\). The constraint \(\boldsymbol{x}\) is trivial on \(\yd{\lambda_{2}}_{\mathrm{g}} \setminus \yd{\nu_{2}}_{\mathrm{g}}\) and fixes the boxes of \(\yd{\lambda_{1\Delta2}}_{\mathrm{g}}\) to \(d+1\), that is,
    \begin{equation}
        x_{l}(i,j)=x_{u}(i,j)\coloneqq d+1 \quad \text{for} \quad (i,j)\in\yd{\lambda_{1\Delta2}}_{\mathrm{g}} \, .
    \end{equation}
    Similarly, \(\boldsymbol{x}'\) is trivial on \(\yd{\lambda_{\max}}_{\mathrm{g}} \setminus \yd{\nu_{\max}}_{\mathrm{g}}\) and fixes the boxes of \(\yd{\nu_{1\Delta2}}_{\mathrm{g}}\) to \(0\), that is,
    \begin{equation}
        x'_{l}(i,j)=x'_{u}(i,j)\coloneqq 0 \quad \text{for} \quad (i,j)\in\yd{\nu_{1\Delta2}}_{\mathrm{g}} \, .
    \end{equation}
    It follows that \(\boldsymbol{x}' \leq \boldsymbol{x}\) entrywise. In addition, we have
    \begin{equation}
        \wt{w}_{\mathrm{g}}\vdash_{0}^{d+1}\yd{\gamma}_{\mathrm{g}}|\boldsymbol{x} \quad \Leftrightarrow \quad \wt{w|_{\yd{\lambda_{2}}_{\mathrm{g}} \setminus \yd{\nu_{2}}_{\mathrm{g}}}}_{\mathrm{g}} \vdash_{d} \mleft(\yd{\lambda_{2}}_{\mathrm{g}} \setminus \yd{\nu_{2}}_{\mathrm{g}}\mright) \, ,
    \end{equation}
    where \(\vdash_{0}^{d+1}\) indicates the extended alphabet \(\{0,\ldots,d+1\}\). Similarly, we have
    \begin{equation}
        \wt{w}_{\mathrm{g}}\vdash_{0}^{d+1}\yd{\gamma}_{\mathrm{g}}|\boldsymbol{x}' \quad \Leftrightarrow \quad \wt{w|_{\yd{\lambda_{\max}}_{\mathrm{g}} \setminus \yd{\nu_{\max}}_{\mathrm{g}}}}_{\mathrm{g}} \vdash_{d} \mleft(\yd{\lambda_{\max}}_{\mathrm{g}} \setminus \yd{\nu_{\max}}_{\mathrm{g}}\mright) \, .
    \end{equation}
    The reasoning behind this is that \(\boldsymbol{x}\)  fixes the values on \(\yd{\lambda_{1\Delta2}}_{\mathrm{g}}\) to be \(d+1\), and \(\boldsymbol{x}'\) fixes the values on \(\yd{\nu_{1\Delta2}}_{\mathrm{g}}\) to be \(0\), while the rest of the gWT can have values in \(1,\ldots,d\). The above equivalences tell us in particular that, with \(n_{1}=|\yd{\lambda_{1\Delta2}}_{\mathrm{g}}|\) and \(n_{2}=|\yd{\nu_{1\Delta2}}_{\mathrm{g}}|\),
    \begin{equation}
    \begin{aligned}
        \label{equ:proof_inequality_Schur_polynomials_products_Schur_polynomials}
        \schur{\yd{\gamma}_{\mathrm{g}}|\boldsymbol{x}}[\boldsymbol{p}] = (p_{d+1})^{n_{1}} \, \schur{\yd{\lambda_{2}}_{\mathrm{g}}\setminus\yd{\nu_{2}}_{\mathrm{g}}}[\boldsymbol{p}|_{1}^{d}] \, , \\
        \schur{\yd{\gamma}_{\mathrm{g}}|\boldsymbol{x}'}[\boldsymbol{p}] = (p_{0})^{n_{2}} \, \schur{\yd{\lambda_{\max}}_{\mathrm{g}}\setminus\yd{\nu_{\max}}_{\mathrm{g}}}[\boldsymbol{p}|_{1}^{d}] \, .
    \end{aligned}
    \end{equation}
    where we take \(\boldsymbol{p}\dot{=}(p_{0},p_{1},\ldots,p_{d},p_{d+1})\) and \(\boldsymbol{p}|_{1}^{d}\dot{=}(p_{1},\ldots,p_{d})\). We introduce a slightly modified version of the corresponding restrictions. Recall that
    \begin{equation}
        \mleft(\yd{\lambda_{\max}}_{\mathrm{g}} \setminus \yd{\lambda_{\min}}_{\mathrm{g}}\mright) \setminus \yd{\nu_{2\Delta1}}_{\mathrm{g}} = \yd{\gamma}_{\mathrm{g}} \, ,
    \end{equation}
    so for \(\wt{w}_{\mathrm{g}}\vdash_{0}^{d+1}\yd{\gamma}_{\mathrm{g}}\), we can talk about the corresponding restriction \(\boldsymbol{x}\mleft(\wt{w}_{\mathrm{g}}\mright)\) on \(\yd{\nu_{2\Delta1}}_{\mathrm{g}}\). We modify this restriction as follows
    \begin{equation}
    \begin{aligned}
        \mleft(\boldsymbol{x}_{1}^{d}\mleft(\wt{w}_{\mathrm{g}}\mright)\mright)_{l}(i,j) \coloneqq \max\mleft(\mleft(\boldsymbol{x}\mleft(\wt{w}_{\mathrm{g}}\mright)\mright)_{l}(i,j) \, , \, 1\mright) \, , \\
        \mleft(\boldsymbol{x}_{1}^{d}\mleft(\wt{w}_{\mathrm{g}}\mright)\mright)_{u}(i,j) \coloneqq \min\mleft(\mleft(\boldsymbol{x}\mleft(\wt{w}_{\mathrm{g}}\mright)\mright)_{u}(i,j) \, , \, d\mright) \, .
    \end{aligned}
    \end{equation}
    The construction enforces the restriction \(\boldsymbol{x}\mleft(\wt{w}_{\mathrm{g}}\mright)\), and additionally we require that every entry is in \(\{1,\ldots,d\}\). We finally define the function
    \begin{equation}
    \begin{aligned}
        &F:\mathcal{W}_{\yd{\gamma}_{\mathrm{g}}}^{d+2} \rightarrow \mathbb{R}_{\geq 0} \, , \\
        &F\mleft(\wt{w}_{\mathrm{g}}\mright) \coloneqq
        \mleft\{\begin{aligned}
            &0, & \wt{w}_{\mathrm{g}}(i,j) \neq d+1 \text{ for some } (i,j) \in \yd{\lambda_{2\Delta1}}_{\mathrm{g}} \cap \yd{\gamma}_{\mathrm{g}}, \\
            &\schur{\yd{\nu_{2\Delta1}}_{\mathrm{g}}|\boldsymbol{x}_{1}^{d}\mleft(\wt{w}_{\mathrm{g}}\mright)}[\boldsymbol{p}], & \text{otherwise}.
        \end{aligned}\mright.
    \end{aligned}
    \end{equation}
    Thus, \(F\) is nonzero only if \(\yd{\lambda_{2\Delta1}}_{\mathrm{g}} \cap \yd{\gamma}_{\mathrm{g}}\) is filled with \(d+1\)'s. In this case, it evaluates to the sum over all possible extensions of \(\wt{w}_{\mathrm{g}}\) with entries strictly in \(\{1,\ldots,d\}\). We notice two facts about \(F\). First, it is nondecreasing in every entry. On the one hand, this is because it takes nonzero values only when a certain subset of boxes is filled with \(d+1\). On the other hand, the gYD \(\yd{\nu_{2\Delta1}}_{\mathrm{g}}\) is only to the left and above \(\yd{\gamma}_{\mathrm{g}}\), so any increase in the entries of \(\wt{w}_{\mathrm{g}}\) only leads to more potential extensions
    \begin{equation}
        \wt{w'}_{\mathrm{g}} \vdash_{d} \yd{\nu_{2\Delta1}}_{\mathrm{g}}|\boldsymbol{x}_{1}^{d}\mleft(\wt{w}_{\mathrm{g}}\mright) \, ,
    \end{equation}
    and therefore an increase in the number of terms of the sum \(\schur{\yd{\nu_{2\Delta1}}_{\mathrm{g}}|\boldsymbol{x}_{1}^{d}\mleft(\wt{w}_{\mathrm{g}}\mright)}[\boldsymbol{p}]\), which are all positive. The second fact we can verify about \(F\) is the following, where \(n_{1}=|\yd{\lambda_{1\Delta2}}_{\mathrm{g}}|\), \(n_{2}=|\yd{\nu_{1\Delta2}}_{\mathrm{g}}|\), and \(n_{3}=|\yd{\lambda_{2\Delta1}}_{\mathrm{g}}\cap \yd{\gamma}_{\mathrm{g}}|\):
    \begin{equation}
    \begin{aligned}
        \label{equ:proof_inequality_Schur_polynomials_products_weighted_Schur_polynomials}
        \schur[F]{\yd{\gamma}_{\mathrm{g}}|\boldsymbol{x}}[\boldsymbol{p}] = (p_{d+1})^{n_{1} + n_{3}} \, \schur{\yd{\lambda_{\min}}_{\mathrm{g}}\setminus\yd{\nu_{\min}}_{\mathrm{g}}}[\boldsymbol{p}|_{1}^{d}] \, , \\
        \schur[F]{\yd{\gamma}_{\mathrm{g}}|\boldsymbol{x}'}[\boldsymbol{p}] = (p_{0})^{n_{2}} (p_{d+1})^{n_{3}} \, \schur{\yd{\lambda_{1}}_{\mathrm{g}}\setminus\yd{\nu_{1}}_{\mathrm{g}}}[\boldsymbol{p}|_{1}^{d}] \, .
    \end{aligned}
    \end{equation}
    To see the first equality, we remark that
    \begin{equation}
        \yd{\lambda_{\min}}_{\mathrm{g}} \setminus \yd{\nu_{2}} = \yd{\gamma}_{\mathrm{g}} \setminus \mleft(\yd{\lambda_{1\Delta2}}_{\mathrm{g}} \cup \yd{\lambda_{2\Delta1}}_{\mathrm{g}}\mright) \, ,
    \end{equation}
    and we define
    \begin{equation}
    \begin{aligned}
        &F|_{\yd{\lambda_{\min}}_{\mathrm{g}} \setminus \yd{\nu_{2}}} :\mathcal{W}_{\yd{\lambda_{\min}}_{\mathrm{g}} \setminus \yd{\nu_{2}}}^{d+2} \rightarrow \mathbb{R}_{\geq 0} \, , \\
        &F|_{\yd{\lambda_{\min}}_{\mathrm{g}} \setminus \yd{\nu_{2}}}\mleft(\wt{w}_{\mathrm{g}}\mright) \coloneqq
        \schur{\yd{\nu_{2\Delta1}}_{\mathrm{g}}|\boldsymbol{x}_{1}^{d}\mleft(\wt{w}_{\mathrm{g}}\mright)}[\boldsymbol{p}] \, .
    \end{aligned}
    \end{equation}
    This can be interpreted as a version of \(F\) where we already fixed the entries on \(\yd{\lambda_{1\Delta2}}_{\mathrm{g}}\) to be equal to \(d+1\) due to \(\boldsymbol{x}\), and the entries on \(\yd{\lambda_{2\Delta1}}_{\mathrm{g}} \cap \yd{\gamma}_{\mathrm{g}}\) to be equal to \(d+1\) so that \(F\) is non-zero. This gives
    \begin{equation}
    \begin{aligned}
        \schur[F]{\yd{\gamma}_{\mathrm{g}}|\boldsymbol{x}}[\boldsymbol{p}] &= (p_{d+1})^{n_{1} + n_{3}} \sum_{\wt{w}_{\mathrm{g}}\vdash_{0}^{d+1} \yd{\lambda_{\min}}_{\mathrm{g}}\setminus\yd{\nu_{2}}_{\mathrm{g}}|\boldsymbol{x}} F|_{\yd{\lambda_{\min}}_{\mathrm{g}} \setminus \yd{\nu_{2}}}\mleft(\wt{w}_{\mathrm{g}}\mright) \\
        &= (p_{d+1})^{n_{1} + n_{3}} \sum_{\wt{w}_{\mathrm{g}}\vdash_{d} \yd{\lambda_{\min}}_{\mathrm{g}}\setminus\yd{\nu_{2}}_{\mathrm{g}}} \schur{\yd{\nu_{2\Delta1}}_{\mathrm{g}}|\boldsymbol{x}_{1}^{d}\mleft(\wt{w}_{\mathrm{g}}\mright)}[\boldsymbol{p}|_{1}^{d}] \\
        &= (p_{d+1})^{n_{1} + n_{3}} \, \schur{\yd{\lambda_{\min}}_{\mathrm{g}}\setminus\yd{\nu_{\min}}_{\mathrm{g}}}[\boldsymbol{p}|_{1}^{d}] \, .
    \end{aligned}
    \end{equation}
    Here we used in the second row that \(\boldsymbol{x}\) restricts the gWTs to entries in \(\{1,\ldots,d\}\) on \(\yd{\lambda_{\min}}_{\mathrm{g}}\setminus\yd{\nu_{2}}_{\mathrm{g}}\). For the second equality of~\cref{equ:proof_inequality_Schur_polynomials_products_weighted_Schur_polynomials}, we similarly remark that
    \begin{equation}
        \yd{\lambda_{1}}_{\mathrm{g}} \setminus \yd{\nu_{\max}}_{\mathrm{g}} = \yd{\gamma}_{\mathrm{g}} \setminus \mleft(\yd{\nu_{1\Delta2}}_{\mathrm{g}} \cup \yd{\lambda_{2\Delta1}}_{\mathrm{g}}\mright) \, ,
    \end{equation}
    and we define in the same way
    \begin{equation}
    \begin{aligned}
        &F|_{\yd{\lambda_{1}}_{\mathrm{g}} \setminus \yd{\nu_{\max}}_{\mathrm{g}}} :\mathcal{W}_{\yd{\lambda_{1}}_{\mathrm{g}} \setminus \yd{\nu_{\max}}_{\mathrm{g}}}^{d+2} \rightarrow \mathbb{R}_{\geq 0} \, , \\
        &F|_{\yd{\lambda_{1}}_{\mathrm{g}} \setminus \yd{\nu_{\max}}_{\mathrm{g}}}\mleft(\wt{w}_{\mathrm{g}}\mright) \coloneqq
        \schur{\yd{\nu_{2\Delta1}}_{\mathrm{g}}|\boldsymbol{x}_{1}^{d}\mleft(\wt{w}_{\mathrm{g}}\mright)}[\boldsymbol{p}] \, .
    \end{aligned}
    \end{equation}
    The same mechanism applies, with \(\boldsymbol{x}'\) forcing the gWTs to have value \(0\) on \(\yd{\nu_{1\Delta2}}_{\mathrm{g}}\) instead of \(d+1\) on \(\yd{\lambda_{1\Delta2}}_{\mathrm{g}}\). Again, it follows
    \begin{equation}
    \begin{aligned}
        \schur[F]{\yd{\gamma}_{\mathrm{g}}|\boldsymbol{x}'}[\boldsymbol{p}] &= (p_{0})^{n_{2}} (p_{d+1})^{n_{3}} \sum_{\wt{w}_{\mathrm{g}}\vdash_{0}^{d+1} \yd{\lambda_{1}}_{\mathrm{g}} \setminus \yd{\nu_{\max}}_{\mathrm{g}}|\boldsymbol{x}'} F|_{\yd{\lambda_{1}}_{\mathrm{g}} \setminus \yd{\nu_{\max}}_{\mathrm{g}}}\mleft(\wt{w}_{\mathrm{g}}\mright) \\
        &= (p_{0})^{n_{2}} (p_{d+1})^{n_{3}} \sum_{\wt{w}_{\mathrm{g}}\vdash_{d} \yd{\lambda_{1}}_{\mathrm{g}} \setminus \yd{\nu_{\max}}_{\mathrm{g}}} \schur{\yd{\nu_{2\Delta1}}_{\mathrm{g}}|\boldsymbol{x}_{1}^{d}\mleft(\wt{w}_{\mathrm{g}}\mright)}[\boldsymbol{p}|_{1}^{d}] \\
        &= (p_{0})^{n_{2}} (p_{d+1})^{n_{3}} \, \schur{\yd{\lambda_{1}}_{\mathrm{g}} \setminus \yd{\nu_{1}}_{\mathrm{g}}}[\boldsymbol{p}|_{1}^{d}] \, .
    \end{aligned}
    \end{equation}
    Using~\cref{thm:monotonicity_constrained_generalized_averages,equ:proof_inequality_Schur_polynomials_products_Schur_polynomials,equ:proof_inequality_Schur_polynomials_products_weighted_Schur_polynomials}, we see
    \begin{equation}
    \begin{aligned}
        &(p_{d+1})^{n_{1}} \frac{\schur{\yd{\lambda_{\min}}_{\mathrm{g}}\setminus\yd{\nu_{\min}}_{\mathrm{g}}}[\boldsymbol{p}|_{1}^{d}]}{\schur{\yd{\lambda_{2}}_{\mathrm{g}}\setminus\yd{\nu_{2}}_{\mathrm{g}}}[\boldsymbol{p}|_{1}^{d}]}
        = \frac{\schur[F]{\yd{\gamma}_{\mathrm{g}}|\boldsymbol{x}}[\boldsymbol{p}]}{\schur{\yd{\gamma}_{\mathrm{g}}|\boldsymbol{x}}[\boldsymbol{p}]}
        = \weylavg{F\mleft(\wt{w}_{\mathrm{g}}\mright)}{\boldsymbol{p}}{\yd{\gamma}_{\mathrm{g}}|\boldsymbol{x}}\\
        \geq &\weylavg{F\mleft(\wt{w}_{\mathrm{g}}\mright)}{\boldsymbol{p}}{\yd{\gamma}_{\mathrm{g}}|\boldsymbol{x}'}
        =\frac{\schur[F]{\yd{\gamma}_{\mathrm{g}}|\boldsymbol{x}'}[\boldsymbol{p}]}{\schur{\yd{\gamma}_{\mathrm{g}}|\boldsymbol{x}'}[\boldsymbol{p}]}
        = (p_{d+1})^{n_{1}} \frac{\schur{\yd{\lambda_{1}}_{\mathrm{g}}\setminus\yd{\nu_{1}}_{\mathrm{g}}}[\boldsymbol{p}|_{1}^{d}]}{\schur{\yd{\lambda_{\max}}_{\mathrm{g}}\setminus\yd{\nu_{\max}}_{\mathrm{g}}}[\boldsymbol{p}|_{1}^{d}]} \, .
    \end{aligned}
    \end{equation}
    Multiplying by the denominator on both sides and choosing \(p_{d+1}=1\) gives the result.
\end{proof}

Following the approach in Ref.~\cite[Proposition~8.1]{KT21S}, we derive the next result.

\begin{corollary}[Monotonicity of Occupancy Observables]
\label{cor:monotonicity_functions_young_diagrams}
Let $\yd{\varsigma}$ and $\yd{\varrho}$ be two YDs such that $\yd{\varrho}$ is obtained from $\yd{\varsigma}$ by adding boxes, i.e., $\yd{\varrho} \geq \yd{\varsigma}$. Further, let $F:\mathbb{Z}_{\geq0}\mapsto \mathbb{R}$ be a nondecreasing function. Then, for any entry $k \in \{1, \ldots, d\}$ and spectrum $\boldsymbol{p} \in \mathbb{R}_{>0}^d$:
\begin{equation}
    \weylavg{F(\#_k)}{\boldsymbol{p}}{\yd{\varsigma}} \le \weylavg{F(\#_k)}{\boldsymbol{p}}{\yd{\varrho}} \, .
\end{equation}
\end{corollary}

\begin{proof}
    Because of the symmetry of the distribution \(\mathrm{W}(\boldsymbol{p},\yd{\varrho})\) it suffices to consider the case \(k=1\). For a YD \(\yd{\lambda}\) and \(i\leq \lambda_{1}\), let \(\yd{\lambda}\setminus\yd{i}\) denote the skew YD with the first \(i\) boxes removed from the first row. We note that
    \begin{subequations}
    \begin{align}
        \schur{\yd{\lambda}}[\boldsymbol{p}] &= \sum_{i=0}^{\lambda_{1}} p_{1}^{i}\schur{\yd{\lambda}\setminus\yd{i}}[\boldsymbol{p}|_{2}^{d}] \, , \\
        \weylavg{F(\#_1)}{\boldsymbol{p}}{\yd{\lambda}} &= \frac{\sum_{i=0}^{\lambda_{1}} F(i)p_{1}^{i}\schur{\yd{\lambda}\setminus\yd{i}}[\boldsymbol{p}|_{2}^{d}]}{\schur{\yd{\lambda}}[\boldsymbol{p}]} \, .
    \end{align}
    \end{subequations}

    This follows by ordering the WTs according to the number of \(1\)'s (all of which are at the beginning of the first row). We proceed as in Ref.~\cite[Proposition~8.1]{KT21S}, using the inequality from Ref.~\cite[Theorem~5]{LPP05S}, which gives

    \begin{equation}
        \label{equ:proof_monotonicity_functions_young_diagrams_Schur_polynomials_positive}
        \schur{\yd{\varrho}\setminus\yd{i}}[\boldsymbol{p}|_{2}^{d}]\schur{\yd{\varsigma}\setminus\yd{j}}[\boldsymbol{p}|_{2}^{d}] - \schur{\yd{\varrho}\setminus\yd{j}}[\boldsymbol{p}|_{2}^{d}]\schur{\yd{\varsigma}\setminus\yd{i}}[\boldsymbol{p}|_{2}^{d}] \geq 0 \, ,
    \end{equation}
    for \(\varrho \geq \varsigma\) and \(i\geq j\). We have
    \begin{equation}
    \begin{aligned}
        \label{equ:proof_monotonicity_functions_young_diagrams_difference}
        \schur{\yd{\varrho}}[\boldsymbol{p}]\schur{\yd{\varsigma}}[\boldsymbol{p}]\mleft(\weylavg{F(\#_k)}{\boldsymbol{p}}{\yd{\varrho}} - \weylavg{F(\#_k)}{\boldsymbol{p}}{\yd{\varsigma}}\mright) \\
        = \sum_{i=0}^{\varrho_{1}}\sum_{j=0}^{\varsigma_{1}} p_{1}^{i+j}(F(i)-F(j))\schur{\yd{\varrho}\setminus\yd{i}}[\boldsymbol{p}|_{2}^{d}]\schur{\yd{\varsigma}\setminus\yd{j}}[\boldsymbol{p}|_{2}^{d}] \, .
    \end{aligned}
    \end{equation}
    The diagonal part of the right-hand side of~\cref{equ:proof_monotonicity_functions_young_diagrams_difference}, with \(i=j\), cancels out, and by virtue of \(F\) being nondecreasing and \(\schur{\yd{\varrho}}[\boldsymbol{p}]\schur{\yd{\varsigma}}[\boldsymbol{p}] > 0\), the part for \(i>\varsigma_{1}\) is \(\geq 0\). Omitting these parts, we can sum pairs of similar indices to get
    \begin{equation}
    \begin{aligned}
        \sum_{0\leq j<i \leq \varsigma_{1}} p_{1}^{i+j}(F(i)-F(j))\mleft(\schur{\yd{\varrho}\setminus\yd{i}}[\boldsymbol{p}|_{2}^{d}]\schur{\yd{\varsigma}\setminus\yd{j}}[\boldsymbol{p}|_{2}^{d}] - \schur{\yd{\varrho}\setminus\yd{j}}[\boldsymbol{p}|_{2}^{d}]\schur{\yd{\varsigma}\setminus\yd{i}}[\boldsymbol{p}|_{2}^{d}]\mright) \geq 0 \, .
    \end{aligned}
    \end{equation}
    Here, we used~\cref{equ:proof_monotonicity_functions_young_diagrams_Schur_polynomials_positive} and the fact that \(F\) is nondecreasing. Since \(\schur{\yd{\varrho}}[\boldsymbol{p}]\schur{\yd{\varsigma}}[\boldsymbol{p}] > 0\), we can divide by this term and we obtain the theorem.
\end{proof}

\subsubsection{Splitting averages}
\label{subsubsec:splitting_averages}

Finally, we want to understand how we can split Weyl averages of products of functions, which only depend on the WT values of complementary subsets of the gYD. To this end, we introduce the following notion.

\begin{definition}
    For a gYD \(\yd{\lambda}_{\mathrm{g}}\) we define the \emph{highest weight gWT \(\wt{\hw}_{\mathrm{g}}\)} and the \emph{lowest weight gWT \(\wt{\lw}_{\mathrm{g}}\)} as the gWTs such that for all \(\wt{w}_{\mathrm{g}} \vdash_{d} \yd{\lambda}_{\mathrm{g}}\) we have
    \begin{equation}
        \label{equ:definition_lowest_highest_weight_gWTs}
        \wt{\lw}_{\mathrm{g}}(i,j) \leq \wt{w}_{\mathrm{g}}(i,j) \leq \wt{\hw}_{\mathrm{g}}(i,j) \quad \text{for} \quad (i,j)\in \yd{\lambda}_{\mathrm{g}} \, .
    \end{equation}
\end{definition}

\begin{remark}
    For the case where \(\yd{\lambda}_{\mathrm{g}}\) is just given by \(\yd{\lambda}\dot{=}\yd{\lambda_{1},\ldots,\lambda_{d}}\), the lowest and highest weight gWTs just correspond to the lowest and highest weight SSYTs with weight ordering \((1,\ldots,d)\). Equivalently, the lowest weight is given by
    \begin{equation}
        \wt{\lw}_{\mathrm{g}}(i,j) = i \quad \text{for} \quad 1\leq i \leq d \, , \quad 1\leq j \leq \lambda_{i} \, ,
    \end{equation}
    and the highest weight is given by
    \begin{equation}
        \wt{\hw}_{\mathrm{g}}(i,j) = d - \ell + 1 \quad \text{for} \quad 1\leq i \leq d \, , \quad \max(1,\lambda_{\ell+1}) \leq j \leq \lambda_{\ell} \, , \quad i \leq \ell \leq d \, ,
    \end{equation}
    with the convention that \(\lambda_{d+1}=-\infty\).
\end{remark}

\begin{example}
    For the gYD in panels (c,d) of~\cref{fig:gYT}, take \(d=5\).
    The lowest weight gWT is
    \begin{equation}\begin{aligned}
        \wt{\lw}_{\mathrm{g}}(0,4)&=1,
        &\wt{\lw}_{\mathrm{g}}(1,1)&=1,
        &\wt{\lw}_{\mathrm{g}}(1,2)&=1,
        &\wt{\lw}_{\mathrm{g}}(1,4)&=2, \\
        \wt{\lw}_{\mathrm{g}}(1,5)&=2,
        &\wt{\lw}_{\mathrm{g}}(2,2)&=2,
        &\wt{\lw}_{\mathrm{g}}(2,4)&=3,
        &\wt{\lw}_{\mathrm{g}}(2,5)&=3, \\
        \wt{\lw}_{\mathrm{g}}(4,3)&=1,
        &\wt{\lw}_{\mathrm{g}}(4,4)&=1.
    \end{aligned}\end{equation}
    The highest weight gWT is
    \begin{equation}\begin{aligned}
        \wt{\hw}_{\mathrm{g}}(0,4)&=3,
        &\wt{\hw}_{\mathrm{g}}(1,1)&=4,
        &\wt{\hw}_{\mathrm{g}}(1,2)&=4,
        &\wt{\hw}_{\mathrm{g}}(1,4)&=4, \\
        \wt{\hw}_{\mathrm{g}}(1,5)&=4,
        &\wt{\hw}_{\mathrm{g}}(2,2)&=5,
        &\wt{\hw}_{\mathrm{g}}(2,4)&=5,
        &\wt{\hw}_{\mathrm{g}}(2,5)&=5, \\
        \wt{\hw}_{\mathrm{g}}(4,3)&=5,
        &\wt{\hw}_{\mathrm{g}}(4,4)&=5.
    \end{aligned}\end{equation}
    Thus every \(\wt{w}_{\mathrm{g}}\vdash_{5}\yd{\lambda}_{\mathrm{g}}\) satisfies \(\wt{\lw}_{\mathrm{g}}(i,j)\leq \wt{w}_{\mathrm{g}}(i,j)\leq \wt{\hw}_{\mathrm{g}}(i,j)\) for all \((i,j)\in\yd{\lambda}_{\mathrm{g}}\).
\end{example}

We can show that there always exist lowest and highest weights.

\begin{lemma}
    \label{lem:existence_lowest_highest_weights}
    If there exists any \(\wt{w}\vdash_{d}\yd{\lambda}_{\mathrm{g}}\), then there exist (not necessarily distinct) lowest and highest weight gWTs.
\end{lemma}

\begin{proof}
    Let \(R_{\pi}\) be the rotation of \(\mathbb{Z}^2\) by \(180\) degrees. We define
    \begin{equation}
        \yd{\lambda_{\mathrm{rot}}}_{\mathrm{g}} \coloneqq R_{\pi}\mleft(\yd{\lambda}_{\mathrm{g}}\mright) \, ,
    \end{equation}
    that is,
    \begin{equation}
        (i,j) \in \yd{\lambda}_{\mathrm{g}} \quad \Leftrightarrow \quad (-i,-j) \in \yd{\lambda_{\mathrm{rot}}}_{\mathrm{g}} \, .
    \end{equation}
    Further, let \(\iota\) be the inverse labeling of \(\{1,\ldots,d\}\) given by
    \begin{align}
        \iota(i) \coloneqq d - i + 1\, .
    \end{align}
    Finally, let \(I:\mathcal{W}_{\yd{\lambda}_{\mathrm{g}}}^{d} \rightarrow \mathcal{W}_{\yd{\lambda_{\mathrm{rot}}}_{\mathrm{g}}}^{d}\) be the map
    \begin{equation}
        I\mleft(\wt{w}\mright) \coloneqq \iota \circ \wt{w} \circ R_{\pi}^{-1} \, .
    \end{equation}
    The map \(I\) takes a gWT, rotates it by \(180\) degrees, and then applies the relabeling \(\iota\). One verifies that this again defines a valid gWT of \(\yd{\lambda_{\mathrm{rot}}}_{\mathrm{g}}\) and that it is a bijection. In addition, \(I\) inverts entrywise ordering, that is, for \(\wt{w}_{\mathrm{g}},\wt{w'}_{\mathrm{g}}\vdash_{d}\yd{\lambda}_{\mathrm{g}}\), and for any \((i,j)\in \yd{\lambda}_{\mathrm{g}}\) we have
    \begin{equation}
        \wt{w}_{\mathrm{g}}(i,j) \leq \wt{w'}_{\mathrm{g}}(i,j) \quad \Leftrightarrow \quad I\mleft(\wt{w}\mright) \geq I\mleft(\wt{w'}\mright) \, .
    \end{equation}
    In particular,
    \begin{equation}
        I^{-1}\mleft(\wt{w_{\mathrm{lw}}^{\yd{\lambda_{\mathrm{rot}}}_{\mathrm{g}}}}_{\mathrm{g}}\mright) = \wt{\hw}_{\mathrm{g}} \, .
    \end{equation}
    Therefore it suffices to show that the existence of some \(\wt{w}\vdash_{d}\yd{\lambda}_{\mathrm{g}}\) implies the existence of the lowest weight gWT.

    We now prove the existence of the lowest weight gWT. For any \(\yd{\nu}_{\mathrm{g}}\), let
    \begin{equation}
    \begin{aligned}
        \mathrm{Top}\mleft(\yd{\nu}_{\mathrm{g}}\mright)
        \coloneqq \left\{(i,j)\in \yd{\nu}_{\mathrm{g}} \, \middle|   \left((i,j-1) \notin \yd{\nu}_{\mathrm{g}}
        \lor(i,j-1)\in \mathrm{Top}\mleft(\yd{\nu}_{\mathrm{g}}\mright)\right) \land (i-1,j) \notin \yd{\nu}_{\mathrm{g}}\right\} \, .
    \end{aligned}
    \end{equation}
    The construction first takes all boxes with no boxes to their left or above, and then iteratively takes all boxes with a box already chosen to the left and no box above. This produces the top-left border of the diagram. We iteratively remove this top-left border, and for the step \(k\) we define the upper leftover part \(\yd{\mu_{k}}_{\mathrm{g}}\), the removed row \(\yd{\tau_{k}}_{\mathrm{g}}\) and the lower leftover part \(\yd{\nu_{k}}_{\mathrm{g}}\) as
    \begin{equation}
    \begin{aligned}
        &\yd{\mu_{0}}_{\mathrm{g}} \coloneqq \varnothing \, , \quad &\yd{\mu_{k+1}}_{\mathrm{g}} \coloneqq \yd{\mu_{k}}_{\mathrm{g}} \cup \mathrm{Top}\mleft(\yd{\nu_{k}}_{\mathrm{g}}\mright) \, , \\
        &\yd{\tau_{0}}_{\mathrm{g}} \coloneqq \varnothing \, , \quad &\yd{\tau_{k+1}}_{\mathrm{g}} \coloneqq \mathrm{Top}\mleft(\yd{\nu_{k}}_{\mathrm{g}}\mright) \, , \\
        &\yd{\nu_{0}}_{\mathrm{g}} \coloneqq \yd{\lambda}_{\mathrm{g}} \, , \quad &\yd{\nu_{k+1}}_{\mathrm{g}} \coloneqq \yd{\nu_{k}}_{\mathrm{g}}\setminus{\mathrm{Top}\mleft(\yd{\nu_{k}}_{\mathrm{g}}\mright)} \, .
    \end{aligned}
    \end{equation}
    We remark here that for every \(0\leq k\) we have
    \begin{equation}
        \yd{\mu_{k}}_{\mathrm{g}} \cup \yd{\tau_{k}}_{\mathrm{g}} \cup \yd{\nu_{k}}_{\mathrm{g}} = \yd{\lambda}_{\mathrm{g}} \, .
    \end{equation}
    Now let \(\wt{w}_{\mathrm{g}}\vdash_{d}\yd{\lambda}_{\mathrm{g}}\). We further define the gWTs
    \begin{equation}
    \begin{aligned}
        \label{equ:proof_existence_lowest_highest_weights_definition_v}
        &\wt{v_{0}}_{\mathrm{g}} \coloneqq \wt{w}_{\mathrm{g}} \, , \\
        &\wt{v_{k+1}}_{\mathrm{g}}(i,j) \coloneqq
        \mleft\{\begin{aligned}
            &\wt{v_{k}}_{\mathrm{g}}(i,j), & (i,j)\in \yd{\mu_{k+1}}_{\mathrm{g}}, \\
            &k + 1, & (i,j)\in \yd{\tau_{k+1}}_{\mathrm{g}}, \\
            &\wt{w}_{\mathrm{g}}(i,j), & (i,j)\in \yd{\nu_{k+1}}_{\mathrm{g}}.
        \end{aligned}\mright.
    \end{aligned}
    \end{equation}
    We prove by induction that \(\wt{v_{k}}_{\mathrm{g}}\vdash_{d}\yd{\lambda}_{\mathrm{g}}\) and that
    \begin{equation}
    \begin{aligned}
        \label{equ:proof_existence_lowest_highest_weights_inequality_chain}
        &\wt{v_{k}}_{\mathrm{g}}(i,j)
        \mleft\{\begin{aligned}
            &\leq k - 1, & (i,j) \in \yd{\mu_{k}}_{\mathrm{g}}, \\
            &= k, & (i,j) \in \yd{\tau_{k}}_{\mathrm{g}}, \\
            &\geq k + 1, & (i,j) \in \yd{\nu_{k}}_{\mathrm{g}}.
        \end{aligned}\mright.
    \end{aligned}
    \end{equation}
    These conditions follow by construction for \(k=0\). Assume that this holds for \(k-1\). We have
    \begin{equation}
        \wt{v_{k-1}}_{\mathrm{g}}(i,j) \leq k-1 \quad \text{for} \quad (i,j) \in \yd{\mu_{k-1}}_{\mathrm{g}} \cup \yd{\tau_{k-1}}_{\mathrm{g}} = \yd{\mu_{k}}_{\mathrm{g}} \, .
    \end{equation}
    Therefore, the first line of~\cref{equ:proof_existence_lowest_highest_weights_inequality_chain} immediately follows, and the second line follows directly from the definition. By induction hypothesis, we also know that
    \begin{equation}
        \wt{w}_{\mathrm{g}}(i,j) = \wt{v_{k-1}}_{\mathrm{g}}(i,j) \geq k \quad \text{for} \quad (i,j) \in \yd{\nu_{k-1}}_{\mathrm{g}} \, .
    \end{equation}
    Now it follows that equality can only hold for
    \begin{equation}
        (i,j) \in \yd{\tau_{k}}_{\mathrm{g}} = \mathrm{Top}\mleft(\yd{\nu_{k-1}}_{\mathrm{g}}\mright) \, .
    \end{equation}
    All other boxes have either a box above them, or are to the right of a box with another one on top of it. The condition that gWTs are weakly increasing to the right and strongly increasing downwards now requires their values to be \(\geq k+1\).

    It remains to show that \(\wt{v_{k}}_{\mathrm{g}}\vdash_{d}\yd{\lambda}_{\mathrm{g}}\), which means that the row-wise weakly increasing and column-wise strongly increasing conditions are fulfilled. By~\cref{equ:proof_existence_lowest_highest_weights_definition_v} and the definition of \(\mathrm{Top}\mleft(.\mright)\) we know that this is true on the sets
    \begin{equation}
        \yd{\mu_{k}}_{\mathrm{g}} \, , \quad \yd{\tau_{k}}_{\mathrm{g}} \, , \quad \yd{\nu_{k}}_{\mathrm{g}}
    \end{equation}
    individually. We prove that the conditions also hold between these sets.\newline

    \noindent\(\yd{\mu_{k}}_{\mathrm{g}} \leftrightarrow \yd{\nu_{k}}_{\mathrm{g}}:\) We have
    \begin{equation}
        (i,j) \in \yd{\mu_{k}}_{\mathrm{g}} = \yd{\mu_{k-1}}_{\mathrm{g}} \cup \yd{\tau_{k-1}}_{\mathrm{g}} \, , \quad (i^\ast,j^\ast) \in \yd{\nu_{k}}_{\mathrm{g}} \subseteq \yd{\nu_{k-1}}_{\mathrm{g}} \, ,
    \end{equation}
    and by induction hypothesis the conditions between
    \begin{equation}
        \wt{v_{k}}_{\mathrm{g}}(i,j) = \wt{v_{k-1}}_{\mathrm{g}}(i,j) \quad \text{and} \quad \wt{v_{k}}_{\mathrm{g}}(i^\ast,j^\ast) = \wt{v_{k-1}}_{\mathrm{g}}(i^\ast,j^\ast)
    \end{equation}
    are fulfilled.\newline

    \noindent\(\yd{\mu_{k}}_{\mathrm{g}} \leftrightarrow \yd{\tau_{k}}_{\mathrm{g}}:\) We argue by the induction hypothesis that, for
    \begin{equation}
        (i,j)\in\yd{\mu_{k}}_{\mathrm{g}} = \yd{\mu_{k-1}}_{\mathrm{g}} \cup \yd{\tau_{k-1}}_{\mathrm{g}} \, , \quad (i^\ast,j^\ast)\in\yd{\tau_{k}}_{\mathrm{g}}\subseteq \yd{\nu_{k-1}}_{\mathrm{g}} \, ,
    \end{equation}
    we have
    \begin{equation}
        \wt{v_{k-1}}_{\mathrm{g}}(i,j) \leq k-1 < k \leq \wt{v_{k-1}}_{\mathrm{g}}(i^\ast,j^\ast) \, .
    \end{equation}
    Therefore \((i^\ast,j^\ast)\) cannot be to the left or above \((i,j)\). We additionally know from the previous calculations that
    \begin{equation}
        \wt{v_{k}}_{\mathrm{g}}(i,j) \leq k-1 < k = \wt{v_{k}}_{\mathrm{g}}(i^\ast,j^\ast) \, ,
    \end{equation}
    and so the conditions are fulfilled.\newline

    \noindent\(\yd{\tau_{k}}_{\mathrm{g}} \leftrightarrow \yd{\mu_{k}}_{\mathrm{g}}:\) By the definition of \(\mathrm{Top}\mleft(.\mright)\), we know that for
    \begin{equation}
        (i,j) \in \yd{\tau_{k}}_{\mathrm{g}} = \mathrm{Top}\mleft(\yd{\nu_{k-1}}_{\mathrm{g}}\mright) \, , \quad (i^\ast,j^\ast) \in \yd{\nu_{k}}_{\mathrm{g}} = \yd{\nu_{k-1}}_{\mathrm{g}} \setminus{\mathrm{Top}\mleft(\yd{\nu_{k-1}}_{\mathrm{g}}\mright)} \, ,
    \end{equation}
    the box \((i,j)\) cannot be to the right or below \((i^\ast,j^\ast)\). In addition, we know that
    \begin{equation}
        \wt{v_{k}}_{\mathrm{g}}(i,j) = k < k+1 \leq \wt{v_{k}}_{\mathrm{g}}(i^\ast,j^\ast) \, ,
    \end{equation}
    which again means that the conditions are fulfilled.\newline

    This concludes the induction. It follows that \(\yd{\nu_{d}}_{\mathrm{g}}=\varnothing\), since we know that \(\wt{w}_{\mathrm{g}}\vdash_{d}\yd{\lambda}_{\mathrm{g}}\), but also from~\cref{equ:proof_existence_lowest_highest_weights_definition_v} and~\cref{equ:proof_existence_lowest_highest_weights_inequality_chain} that
    \begin{equation}
        d \geq \wt{w}_{\mathrm{g}}(i,j) = \wt{v_{d}}_{\mathrm{g}}(i,j) \geq d + 1 \quad \text{for} \quad (i,j) \in \yd{\nu_{d}}_{\mathrm{g}} \, .
    \end{equation}
    Thus, \(\bigcup_{k=1}^{d}\yd{\tau_{k}}_{\mathrm{g}}=\yd{\lambda}_{\mathrm{g}}\). Since we have \(\wt{v_{d}}_{\mathrm{g}}(i,j)=k\) for \((i,j)\in\yd{\tau_{k}}_{\mathrm{g}}\), it follows that \(\wt{v_{d}}_{\mathrm{g}}\) is independent of the specific choice of \(\wt{w}_{\mathrm{g}}\vdash_{d}\yd{\lambda}_{\mathrm{g}}\), and therefore unique. In addition, from the definition in~\cref{equ:proof_existence_lowest_highest_weights_definition_v} and~\cref{equ:proof_existence_lowest_highest_weights_inequality_chain}, we have that
    \begin{equation}
        \wt{w}_{\mathrm{g}}(i,j) = \wt{v_{0}}_{\mathrm{g}}(i,j) \geq \wt{v_{1}}_{\mathrm{g}}(i,j) \geq \ldots \geq \wt{v_{d}}_{\mathrm{g}}(i,j) \quad .
    \end{equation}
    We conclude the proof by setting \(\wt{\lw}_{\mathrm{g}} \coloneqq \wt{v_{d}}_{\mathrm{g}}\).
\end{proof}

We quickly state the following ordering induced by the lowest and highest weight.

\begin{lemma}
    \label{lem:monotonicity_constraints_lowest_highest_weight}
    Let \(\yd{\nu}_{\mathrm{g}} \subseteq \yd{\lambda}_{\mathrm{g}}\), and let \(\yd{\mu}_{\mathrm{g}}= \yd{\lambda}_{\mathrm{g}} \setminus \yd{\nu}_{\mathrm{g}}\). Then we have for any \(\wt{w}_{\mathrm{g}}\vdash_{d} \yd{\lambda}_{\mathrm{g}}\) that entrywise
    \begin{equation}
            \boldsymbol{x}\mleft(\wt{\lw|_{\yd{\mu}_{\mathrm{g}}}}_{\mathrm{g}}\mright) \leq \boldsymbol{x}\mleft(\wt{w|_{\yd{\mu}_{\mathrm{g}}}}_{\mathrm{g}}\mright) \leq \boldsymbol{x}\mleft(\wt{\hw|_{\yd{\mu}_{\mathrm{g}}}}_{\mathrm{g}}\mright)\, .
        \end{equation}
\end{lemma}

\begin{proof}
    This follows directly from the definition of the lowest / highest weight gWTs in~\cref{equ:definition_lowest_highest_weight_gWTs} and the corresponding restrictions in~\cref{equ:definition_corresponding_restrictions}.
\end{proof}

\begin{theorem}[Splitting averages for gYDs with constraints]
    \label{thm:restriction_averages_generalized_YD_monotone_functions}
    Let \(\yd{\nu}_{\mathrm{g}} \subseteq \yd{\lambda}_{\mathrm{g}}\), and let \(\yd{\mu}_{\mathrm{g}}= \yd{\lambda}_{\mathrm{g}} \setminus \yd{\nu}_{\mathrm{g}}\). Further, let the functions
    \begin{equation}
        F:\mathcal{W}_{\yd{\lambda}_{\mathrm{g}}}^{d}\rightarrow \mathbb{R}_{\geq 0} \, , \quad F|_{\yd{\nu}_{\mathrm{g}}}:\mathcal{W}_{\yd{\nu}_{\mathrm{g}}}^{d}\rightarrow \mathbb{R}_{\geq 0} \, , \quad F|_{\yd{\mu}_{\mathrm{g}}}:\mathcal{W}_{\yd{\mu}_{\mathrm{g}}}^{d}\rightarrow \mathbb{R}_{\geq 0}
    \end{equation}
    be such that
    \begin{equation}
        F\mleft(\wt{w}_{\mathrm{g}}\mright) = F|_{\yd{\nu}_{\mathrm{g}}}\mleft(\wt{w|_{\yd{\nu}_{\mathrm{g}}}}_{\mathrm{g}}\mright) F|_{\yd{\mu}_{\mathrm{g}}}\mleft(\wt{w|_{\yd{\mu}_{\mathrm{g}}}}_{\mathrm{g}}\mright) \, .
    \end{equation}
    Finally, let
    \begin{equation}
    \begin{aligned}
        \boldsymbol{x}_{\min} \coloneqq \boldsymbol{x}\mleft(\wt{\lw|_{\yd{\mu}_{\mathrm{g}}}}_{\mathrm{g}}\mright) \, , \quad \boldsymbol{x}_{\max} \coloneqq \boldsymbol{x}\mleft(\wt{\hw|_{\yd{\mu}_{\mathrm{g}}}}_{\mathrm{g}}\mright) \, .
    \end{aligned}
    \end{equation}
    If \(F|_{\yd{\nu}_{\mathrm{g}}}\) is nondecreasing in every entry, then we have
    \begin{equation}
    \begin{aligned}
        \weylavg{F\mleft(\wt{w}_{\mathrm{g}}\mright)}{\boldsymbol{p}}{\yd{\lambda}_{\mathrm{g}}} \geq \weylavg{F|_{\yd{\nu}_{\mathrm{g}}}\mleft(\wt{w'}_{\mathrm{g}}\mright)}{\boldsymbol{p}}{\yd{\nu}_{\mathrm{g}}|\boldsymbol{x}_{\min}} \weylavg{F|_{\yd{\mu}_{\mathrm{g}}}\mleft(\wt{w|_{\yd{\mu}_{\mathrm{g}}}}_{\mathrm{g}}\mright)}{\boldsymbol{p}}{\yd{\lambda}_{\mathrm{g}}} \, , \\
        \weylavg{F\mleft(\wt{w}_{\mathrm{g}}\mright)}{\boldsymbol{p}}{\yd{\lambda}_{\mathrm{g}}} \leq \weylavg{F|_{\yd{\nu}_{\mathrm{g}}}\mleft(\wt{w'}_{\mathrm{g}}\mright)}{\boldsymbol{p}}{\yd{\nu}_{\mathrm{g}}|\boldsymbol{x}_{\max}} \weylavg{F|_{\yd{\mu}_{\mathrm{g}}}\mleft(\wt{w|_{\yd{\mu}_{\mathrm{g}}}}_{\mathrm{g}}\mright)}{\boldsymbol{p}}{\yd{\lambda}_{\mathrm{g}}} \, .
    \end{aligned}
    \end{equation}
    On the other hand, if \(F|_{\yd{\nu}_{\mathrm{g}}}\) is nonincreasing in every entry, then the inequality signs are reversed.
\end{theorem}

\begin{remark}
    This theorem means that we can split averages over monotone functions on disjoint smaller parts of a gYD into decoupled averages.
\end{remark}

\begin{proof}
    We have
    \begin{equation}
    \begin{aligned}
        \label{equ:proof_restriction_averages_generalized_YD_monotone_functions_decomposition_Schur_polynomial}
        \schur[F]{\yd{\lambda}_{\mathrm{g}}}[\boldsymbol{p}] &= \sum_{\wt{w''}_{\mathrm{g}}\vdash_{d}\yd{\mu}_{\mathrm{g}}} F|_{\yd{\mu}_{\mathrm{g}}}\mleft(\wt{w''}_{\mathrm{g}}\mright) \boldsymbol{p}^{\#\mleft(\wt{w''}_{\mathrm{g}}\mright)} \sum_{\wt{w'}_{\mathrm{g}}\vdash_{d}\yd{\nu}_{\mathrm{g}}|\boldsymbol{x}\mleft(\wt{w''}\mright)} F|_{\yd{\nu}_{\mathrm{g}}}\mleft(\wt{w'}_{\mathrm{g}}\mright) \boldsymbol{p}^{\#\mleft(\wt{w'}_{\mathrm{g}}\mright)} \\
        &= \sum_{\wt{w''}_{\mathrm{g}}\vdash_{d}\yd{\mu}_{\mathrm{g}}} F|_{\yd{\mu}_{\mathrm{g}}}\mleft(\wt{w''}_{\mathrm{g}}\mright) \boldsymbol{p}^{\#\mleft(\wt{w''}_{\mathrm{g}}\mright)} \schur[F|_{\yd{\nu}_{\mathrm{g}}}]{\yd{\nu}_{\mathrm{g}}|\boldsymbol{x}\mleft(\wt{w''}\mright)}[\boldsymbol{p}] \, .
    \end{aligned}
    \end{equation}
    Here we used the fact that we can decompose every \(\wt{w}_{\mathrm{g}}\vdash_{d} \yd{\lambda}_{\mathrm{g}}\) as
    \begin{equation}
        \wt{w}_{\mathrm{g}} = \wt{w|_{\yd{\nu}_{\mathrm{g}}}}_{\mathrm{g}} \cup \wt{w|_{\yd{\mu}_{\mathrm{g}}}}_{\mathrm{g}} \, ,
    \end{equation}
    and instead of summing over all \(\wt{w}_{\mathrm{g}}\) we first sum over the \(\wt{w|_{\yd{\nu}_{\mathrm{g}}}}_{\mathrm{g}}\) and then over all compatible \(\wt{w|_{\yd{\mu}_{\mathrm{g}}}}_{\mathrm{g}}\). In the second line, we used the definition of the generalized Schur polynomials with constraints. Recall that
    \begin{equation}
        \label{equ:proof_restriction_averages_generalized_YD_monotone_functions_average_ratio}
        \schur{\yd{\nu}_{\mathrm{g}}|\boldsymbol{x}\mleft(\wt{w''}\mright)}[\boldsymbol{p}] \weylavg{F|_{\yd{\nu}_{\mathrm{g}}}\mleft(\wt{w'}_{\mathrm{g}}\mright)}{\boldsymbol{p}}{\yd{\nu}_{\mathrm{g}}|\boldsymbol{x}\mleft(\wt{w''}\mright)} = \schur[F|_{\yd{\nu}_{\mathrm{g}}}]{\yd{\nu}_{\mathrm{g}}|\boldsymbol{x}\mleft(\wt{w''}\mright)}[\boldsymbol{p}] \, ,
    \end{equation}
    so we can multiply by the Schur polynomial and rewrite~\cref{equ:proof_restriction_averages_generalized_YD_monotone_functions_decomposition_Schur_polynomial} as
    \begin{equation}
    \begin{aligned}
        \schur[F]{\yd{\lambda}_{\mathrm{g}}}[\boldsymbol{p}] &= \sum_{\wt{w''}_{\mathrm{g}}\vdash_{d}\yd{\mu}_{\mathrm{g}}} F|_{\yd{\mu}_{\mathrm{g}}}\mleft(\wt{w''}_{\mathrm{g}}\mright) \boldsymbol{p}^{\#\mleft(\wt{w''}_{\mathrm{g}}\mright)} \schur{\yd{\nu}_{\mathrm{g}}|\boldsymbol{x}\mleft(\wt{w''}\mright)}[\boldsymbol{p}] \weylavg{F|_{\yd{\nu}_{\mathrm{g}}}\mleft(\wt{w'}_{\mathrm{g}}\mright)}{\boldsymbol{p}}{\yd{\nu}_{\mathrm{g}}|\boldsymbol{x}\mleft(\wt{w''}\mright)} \, .
        \end{aligned}
        \end{equation}
        For the case that
        \begin{equation}
            \schur{\yd{\nu}_{\mathrm{g}}|\boldsymbol{x}\mleft(\wt{w''}\mright)}[\boldsymbol{p}] > 0 \, ,
        \end{equation}
        we can find some \(\wt{w'}_{\mathrm{g}}\vdash_{d} \yd{\nu}_{\mathrm{g}}\) so that
        \begin{equation}
            \wt{w'}_{\mathrm{g}} \cup \wt{w''}_{\mathrm{g}} = \wt{w}_{\mathrm{g}} \vdash_{d} \yd{\lambda}_{\mathrm{g}} \, .
        \end{equation}
        Because of~\cref{lem:monotonicity_constraints_lowest_highest_weight} we have
        \begin{equation}
            \label{equ:proof_restriction_averages_generalized_YD_monotone_functions_lower_bound_x_min}
            \boldsymbol{x}_{\min} \leq \boldsymbol{x}\mleft(\wt{w|_{\yd{\mu}_{\mathrm{g}}}}_{\mathrm{g}}\mright) \leq \boldsymbol{x}_{\max} \, , \quad \boldsymbol{x}\mleft(\wt{w|_{\yd{\mu}_{\mathrm{g}}}}_{\mathrm{g}}\mright) = \boldsymbol{x}\mleft(\wt{w''}_{\mathrm{g}}\mright) \, .
        \end{equation}
        If we assume that \(F|_{\yd{\nu}_{\mathrm{g}}}\) is nondecreasing, we can use~\cref{thm:monotonicity_constrained_generalized_averages} to find that
        \begin{equation}
\label{equ:proof_restriction_averages_generalized_YD_monotone_functions_monotone_average}
            \schur{\yd{\nu}_{\mathrm{g}}|\boldsymbol{x}\mleft(\wt{w''}\mright)}[\boldsymbol{p}] \weylavg{F|_{\yd{\nu}_{\mathrm{g}}}\mleft(\wt{w'}_{\mathrm{g}}\mright)}{\boldsymbol{p}}{\yd{\nu}_{\mathrm{g}}|\boldsymbol{x}\mleft(\wt{w''}\mright)} \rev{\geq} \schur{\yd{\nu}_{\mathrm{g}}|\boldsymbol{x}\mleft(\wt{w''}\mright)}[\boldsymbol{p}] \weylavg{F|_{\yd{\nu}_{\mathrm{g}}}\mleft(\wt{w'}_{\mathrm{g}}\mright)}{\boldsymbol{p}}{\yd{\nu}_{\mathrm{g}}|\boldsymbol{x}_{\min}} \, .
        \end{equation}
        Using the reverse decomposition of \(\wt{w}_{\mathrm{g}}\vdash_{d} \yd{\lambda}_{\mathrm{g}}\) as in~\cref{equ:proof_restriction_averages_generalized_YD_monotone_functions_decomposition_Schur_polynomial} together with a similar equality as in~\cref{equ:proof_restriction_averages_generalized_YD_monotone_functions_average_ratio}, we finally obtain
        \begin{equation}
        \begin{aligned}
            \schur[F]{\yd{\lambda}_{\mathrm{g}}}[\boldsymbol{p}] &\geq \weylavg{F|_{\yd{\nu}_{\mathrm{g}}}\mleft(\wt{w'}_{\mathrm{g}}\mright)}{\boldsymbol{p}}{\yd{\nu}_{\mathrm{g}}|\boldsymbol{x}_{\min}} \sum_{\wt{w''}_{\mathrm{g}}\vdash_{d}\yd{\mu}_{\mathrm{g}}} F|_{\yd{\mu}_{\mathrm{g}}}\mleft(\wt{w''}_{\mathrm{g}}\mright) \boldsymbol{p}^{\#\mleft(\wt{w''}_{\mathrm{g}}\mright)} s^{\yd{\nu}_{\mathrm{g}}|\boldsymbol{x}\mleft(\wt{w''}\mright)}(\boldsymbol{p}) \\
            &= \weylavg{F|_{\yd{\nu}_{\mathrm{g}}}\mleft(\wt{w'}_{\mathrm{g}}\mright)}{\boldsymbol{p}}{\yd{\nu}_{\mathrm{g}}|\boldsymbol{x}_{\min}} \sum_{\wt{w}_{\mathrm{g}}\vdash_{d}\yd{\lambda}_{\mathrm{g}}} F|_{\yd{\mu}_{\mathrm{g}}}\mleft(\wt{w|_{\yd{\mu}_{\mathrm{g}}}}_{\mathrm{g}}\mright) \boldsymbol{p}^{\#\mleft(\wt{w}_{\mathrm{g}}\mright)} \\
            &= \weylavg{F|_{\yd{\nu}_{\mathrm{g}}}\mleft(\wt{w'}_{\mathrm{g}}\mright)}{\boldsymbol{p}}{\yd{\nu}_{\mathrm{g}}|\boldsymbol{x}_{\min}} \weylavg{F|_{\yd{\mu}_{\mathrm{g}}}\mleft(\wt{w|_{\yd{\mu}_{\mathrm{g}}}}_{\mathrm{g}}\mright)}{\boldsymbol{p}}{\yd{\lambda}_{\mathrm{g}}} \schur{\yd{\lambda}_{\mathrm{g}}}[\boldsymbol{p}] \, .
    \end{aligned}
    \end{equation}
    It is important that we use the fact that \(F|_{\yd{\mu}_{\mathrm{g}}}\mleft(\wt{w''}_{\mathrm{g}}\mright) \geq 0\) in the first line to obtain the right inequality sign. Dividing both sides of this inequality chain by \(\schur{\yd{\lambda}_{\mathrm{g}}}[\boldsymbol{p}]\) gives the first inequality. The second inequality follows in the same way, except that \(\boldsymbol{x}_{\max}\) gives an entrywise upper bound in~\cref{equ:proof_restriction_averages_generalized_YD_monotone_functions_lower_bound_x_min}, and so the inequality sign of~\cref{equ:proof_restriction_averages_generalized_YD_monotone_functions_monotone_average} is reversed. For the case where \(F|_{\yd{\nu}_{\mathrm{g}}}\) is nonincreasing, we replace \(F|_{\yd{\nu}_{\mathrm{g}}}\) by \(-F|_{\yd{\nu}_{\mathrm{g}}}\) in the proof above, and we note that everything works as well for negative \(F|_{\yd{\nu}_{\mathrm{g}}} \leq 0\).
\end{proof}

For the specific case of YDs and functions that only depend on the \(\#_{d,i}\mleft(\wt{w}\mright)\), we get the following corollary.

\begin{corollary}[Splitting averages for YDs]
    \label{cor:splitting_diagram_average_weight}
    Let $\yd{\varsigma}\dot{=}[\varsigma_{1},\ldots,\varsigma_{d}]$, let \(1\leq k \leq d\) and let
    \begin{equation}
        D|_{1}^{k-1}:\mathbb{Z}_{\geq0}^{k-1} \rightarrow \mathbb{R}_{\geq 0} \, , \quad I|_{k}^{d}:\mathbb{Z}_{\geq0}^{d-k+1} \rightarrow \mathbb{R}_{\geq 0}
    \end{equation}
    be two functions so that \(D|_{1}^{k-1}\) is nonincreasing and \(I|_{k}^{d}\) is nondecreasing in each argument. Then we have
    \begin{equation}
    \begin{aligned}
        &\weylavg{D|_{1}^{k-1}(\#_{d,1},\ldots,\#_{d,k-1}) I|_{k}^{d}(\#_{d,k},\ldots,\#_{d,d})}{\boldsymbol{p}}{\yd{\varsigma}} \\
        \geq \, &\weylavg{D|_{1}^{k-1}(\#_{k,1},\ldots,\#_{k,k-1})}{(\boldsymbol{p}|_{1}^{k-1},p_{d})}{\yd{\Delta|_{1}^{k}}}  \weylavg{I|_{k}^{d}(\#_{d-k+1,1},\ldots,\#_{d-k+1,d-k+1})}{\boldsymbol{p}|_{k}^{d}}{\yd{\varsigma|_{k}^{d}}} \, .
    \end{aligned}
\end{equation}
In turn, let
    \begin{equation}
        I|_{1}^{k-1}:\mathbb{Z}_{\geq0}^{k-1} \rightarrow \mathbb{R}_{\geq 0} \, , \quad D|_{k}^{d}:\mathbb{Z}_{\geq0}^{d-k+1} \rightarrow \mathbb{R}_{\geq 0}
    \end{equation}
    be two functions so that \(I|_{1}^{k-1}\) is nondecreasing and \(D|_{k}^{d}\) is nonincreasing in each argument. Then we have
    \begin{equation}
    \begin{aligned}
        &\weylavg{I|_{1}^{k-1}(\#_{d,1},\ldots,\#_{d,k-1}) D|_{k}^{d}(\#_{d,k},\ldots,\#_{d,d})}{\boldsymbol{p}}{\yd{\varsigma}} \\
        \leq \, &\weylavg{I|_{1}^{k-1}(\#_{k,1},\ldots,\#_{k,k-1})}{(\boldsymbol{p}|_{1}^{k-1},p_{d})}{\yd{\Delta|_{1}^{k}}} \weylavg{D|_{k}^{d}(\#_{d-k+1,1},\ldots,\#_{d-k+1,d-k+1})}{\boldsymbol{p}|_{k}^{d}}{\yd{\varsigma|_{k}^{d}}} \, .
    \end{aligned}
\end{equation}
Here \((\boldsymbol{p}|_{1}^{k-1},p_{d})\dot{=}(p_{1},\ldots,p_{k-1},p_{d})\), \(\yd{\Delta|_{1}^{k}}\dot{=}\yd{\Delta_{1,k},\ldots,\Delta_{k,k}}\), \(\boldsymbol{p}|_{k}^{d}\dot{=}(p_{k},\ldots,p_{d})\), and \(\yd{\varsigma|_{k}^{d}}\dot{=}\yd{\varsigma_{k},\ldots,\varsigma_{d}}\).
\end{corollary}

\begin{proof}
    We start with the first inequality, for which we set
    \begin{equation}
        \yd{\mu} \dot{=} \yd{\varsigma_{1},\ldots,\varsigma_{k-1}} \, , \quad \yd{\nu} \dot{=} \yd{\varsigma_{k},\ldots,\varsigma_{d}}
    \end{equation}
    in the context of~\cref{thm:restriction_averages_generalized_YD_monotone_functions}. This allows us to use the first inequality for nondecreasing functions there to get
    \begin{equation}
    \begin{aligned}
    \label{equ:proof_splitting_diagram_average_weight_nondecreasing_lower_bound}
        &\weylavg{D|_{1}^{k-1}(\#_{d,1},\ldots,\#_{d,k-1}) I|_{k}^{d}(\#_{d,k},\ldots,\#_{d,d})}{\boldsymbol{p}}{\yd{\varsigma}_{\mathrm{g}}} \\
        \geq \, &\weylavg{D|_{1}^{k-1}(\#_{d,1},\ldots,\#_{d,k-1})}{\boldsymbol{p}}{\yd{\varsigma}_{\mathrm{g}}} \weylavg{I|_{k}^{d}(\#_{d,k},\ldots,\#_{d,d})}{\boldsymbol{p}}{\yd{\nu}_{\mathrm{g}}|\boldsymbol{x}_{\min}} \, .
    \end{aligned}
\end{equation}
One checks that \(\boldsymbol{x}_{\min}\) restricts the gWTs precisely to those with entries in \(\{k,\ldots,d\}\), so we find
    \begin{equation}
    \label{equ:proof_splitting_diagram_average_weight_nondecreasing_restriction}
            \weylavg{I|_{k}^{d}(\#_{d,k},\ldots,\#_{d,d})}{\boldsymbol{p}}{\yd{\nu}_{\mathrm{g}}|\boldsymbol{x}_{\min}} = \weylavg{I|_{k}^{d}(\#_{d-k+1,1},\ldots,\#_{d-k+1,d-k+1})}{\boldsymbol{p}|_{k}^{d}}{\yd{\varsigma|_{k}^{d}}} \, .
\end{equation}
To lower bound with the other average
    \begin{equation}
        \mleft\langle D|_{1}^{k-1}(\#_{d,1},\ldots,\#_{d,k-1})\mright\rangle_{\mathrm{W}(\boldsymbol{p},\yd{\varsigma}_{\mathrm{g}})} \, ,
    \end{equation}
    we first define
    \begin{equation}
        \widetilde{\boldsymbol{p}}\dot{=}(p_{k},\ldots,p_{d-1},p_{1},\ldots,p_{k-1},p_{d}) \, .
    \end{equation}
    Since the expression is a symmetric polynomial in \(p_{1},\ldots,p_{d-1}\) (but not in \(p_{d}\)!), we have
    \begin{equation}
            \weylavg{D|_{1}^{k-1}(\#_{d,1},\ldots,\#_{d,k-1})}{\boldsymbol{p}}{\yd{\varsigma}_{\mathrm{g}}} = \weylavg{D|_{1}^{k-1}(\#_{d,1},\ldots,\#_{d,k-1})}{\widetilde{\boldsymbol{p}}}{\yd{\varsigma}_{\mathrm{g}}} \, .
\end{equation}
We seek to use~\cref{thm:restriction_averages_generalized_YD_monotone_functions} again and set
    \begin{equation}
        \yd{\mu} \dot{=} \yd{\varsigma_{k},\ldots,\varsigma_{k},\varsigma_{k+1},\ldots,\varsigma_{d}} \, , \quad \yd{\nu} \dot{=} \yd{\Delta_{1,k},\ldots,\Delta_{k-1,k}} \, .
    \end{equation}
    We can rewrite the function \(D|_{1}^{k-1}(\#_{d,1},\ldots,\#_{d,k-1})\) as
    \begin{equation}
        D|_{1}^{k-1}(\#_{d,1},\ldots,\#_{d,k-1}) = D|_{1}^{k-1}(\#_{d,1},\ldots,\#_{d,k-1}) 1(\#_{d,k},\ldots,\#_{d,d}) \, ,
    \end{equation}
    where \(1(\ldots)\) is the constant \(1\) function. From the case for a nonincreasing function in~\cref{thm:restriction_averages_generalized_YD_monotone_functions}, we get the lower bound
    \begin{equation}
    \begin{aligned}
    \label{equ:proof_splitting_diagram_average_weight_nonincreasing_lower_bound}
        &\weylavg{D|_{1}^{k-1}(\#_{d,1},\ldots,\#_{d,k-1}) \, 1(\#_{d,k},\ldots,\#_{d,d})}{\widetilde{\boldsymbol{p}}}{\yd{\varsigma}_{\mathrm{g}}} \\
        \geq \, &\weylavg{D|_{1}^{k-1}(\#_{d,1},\ldots,\#_{d,k-1})}{\widetilde{\boldsymbol{p}}}{\yd{\nu}_{\mathrm{g}}|\boldsymbol{x}_{\max}} \weylavg{1(\#_{d,k},\ldots,\#_{d,d})}{\widetilde{\boldsymbol{p}}}{\yd{\varsigma}_{\mathrm{g}}} \\
        = \, &\weylavg{D|_{1}^{k-1}(\#_{d,1},\ldots,\#_{d,k-1})}{\widetilde{\boldsymbol{p}}}{\yd{\nu}_{\mathrm{g}}|\boldsymbol{x}_{\max}} \, .
    \end{aligned}
\end{equation}
In this case, one checks that \(\boldsymbol{x}_{\max}\) restricts the gWTs precisely to those with entries in \(\{d-k+1,\ldots,d\}\), so we have
    \begin{equation}
\label{equ:proof_splitting_diagram_average_weight_nonincreasing_restriction}
            \weylavg{D|_{1}^{k-1}(\#_{d,1},\ldots,\#_{d,k-1})}{\widetilde{\boldsymbol{p}}}{\yd{\nu}_{\mathrm{g}}|\boldsymbol{x}_{\max}} = \weylavg{D|_{1}^{k-1}(\#_{k,1},\ldots,\#_{k,k-1})}{\widetilde{\boldsymbol{p}}|_{d-k+1}^{d}}{\yd{\Delta|_{1}^{k}}_{\mathrm{g}}} \, .
\end{equation}
We remark here that \(\yd{\nu}\) has \(k-1\) entries, while \(\widetilde{\boldsymbol{p}}|_{d-k}^{d}\) has \(k\) entries. Therefore we pad with a \(0\) and get
    \begin{equation}
        \yd{\Delta|_{1}^{k}} \dot{=} \yd{\Delta_{1,k},\ldots,\Delta_{k-1,k},0} \, ,
    \end{equation}
    which does not change the allowed gWTs. Combining~\cref{equ:proof_splitting_diagram_average_weight_nondecreasing_lower_bound,equ:proof_splitting_diagram_average_weight_nondecreasing_restriction,equ:proof_splitting_diagram_average_weight_nonincreasing_lower_bound,equ:proof_splitting_diagram_average_weight_nonincreasing_restriction} and substituting \(\widetilde{\boldsymbol{p}}|_{d-k}^{d}\) gives the first inequality. For the second inequality, we proceed in the same way, only we use the upper bounds of~\cref{thm:restriction_averages_generalized_YD_monotone_functions}.
\end{proof}

\begin{example}
    Panel (g) of~\cref{fig:gYT} depicts the specialization used in the corollary. The YD is split into the rows above and below the cut at \(k\), so the decreasing factor depends only on the upper-row occupancies and the increasing factor only on the lower-row occupancies.
\end{example}

\subsection{Nonasymptotic analysis of the all-site utility}
\label{subsec:nonasymptotic_allcopy_supp}

\subsubsection{Miscellaneous lemmas}
\label{subsec:misc_lemmas_supp}

\begin{lemma}[Bounds for almost-telescopic products]
\label{lem:bounds_telescopic_products}
    Let \(0= A_{0}<A_{1} \leq \ldots \leq A_{d} \) and let \(0 \leq B_{i} \leq A_{i} - A_{i-1}\) for all \(1\leq i \leq d\). Then we have
    \begin{equation}
    \begin{aligned}
        \frac{A_{1} - \sum_{i=1}^{d}B_{i}}{A_{1}} \leq \frac{A_{1} - B_{1}}{A_{1} + \sum_{i=2}^{d} B_{i}} \leq \prod_{i=1}^{d}\frac{A_{i}-B_{i}}{A_{i}} \leq \frac{A_{d} - \sum_{i=1}^{d}B_{i}}{A_{d}} \leq \frac{A_{d} - B_{1}}{A_{d} + \sum_{i=2}^{d} B_{i}} \, .
    \end{aligned}
    \end{equation}
\end{lemma}

\begin{proof}
    First, we note that for \(X>Y\geq 0\) we have
    \begin{equation}
        \frac{X-Y}{X} \leq \frac{X}{X+Y} \, .
    \end{equation}
    Thus, the two outer inequalities are trivial, and we only need to show
    \begin{equation}
        \frac{A_{1} - B_{1}}{A_{1} + \sum_{i=2}^{d} B_{i}} \leq \prod_{i=1}^{d}\frac{A_{i}-B_{i}}{A_{i}} \leq \frac{A_{d} - \sum_{i=1}^{d}B_{i}}{A_{d}} \, .
    \end{equation}
    Let
    \begin{equation}
    \begin{aligned}
        P(A_{1},\ldots,A_{d}) \coloneqq \prod_{i=1}^{d}\frac{A_{i}-B_{i}}{A_{i}} \, .
    \end{aligned}
    \end{equation}
    We find that
    \begin{equation}
    \begin{aligned}
        \label{equ:proof_bounds_telescopic_products_derivatives}
        \frac{\partial}{\partial A_{i}}P(A_{1},\ldots,B_{d}) = P(A_{1},\ldots,A_{d})\mleft(\frac{1}{A_{i} - B_{i}} - \frac{1}{A_{i}}\mright) \geq 0 \, .
    \end{aligned}
    \end{equation}
    We address the upper bound and define
    \begin{equation}
        \widetilde{A}_{i} \coloneqq A_{d} - \sum_{j=i+1}^{d} B_{j} \geq A_{i} \, .
    \end{equation}
    By~\cref{equ:proof_bounds_telescopic_products_derivatives}, it follows that
    \begin{equation}
        \frac{A_{d} - \sum_{i=1}^{d}B_{i}}{A_{d}} = P(\widetilde{A}_{1},\ldots,\widetilde{A}_{d}) \geq P(A_{1},\ldots,A_{d}) = \prod_{i=1}^{d}\frac{A_{i}-B_{i}}{A_{i}} \, .
    \end{equation}
    For the lower bound, we similarly define
    \begin{equation}
        \widetilde{A}_{i} \coloneqq A_{1} + \sum_{j=i+1}^{d} B_{j} \leq A_{i} \, ,
    \end{equation}
    and again we use~\cref{equ:proof_bounds_telescopic_products_derivatives} to see that
    \begin{equation}
        \frac{A_{1} - B_{1}}{A_{1} + \sum_{i=2}^{d} B_{i}} = P(\widetilde{A}_{1},\ldots,\widetilde{A}_{d}) \leq P(A_{1},\ldots,A_{d}) = \prod_{i=1}^{d}\frac{A_{i}-B_{i}}{A_{i}} \, .
    \end{equation}
    This concludes the proof.
\end{proof}

\begin{lemma}[Linearization of pochhammer symbols]
    \label{lem:taylor_pochhammer_symbol}
    For \(M\in\mathbb{Z}_{\geq0}\), \(A, B \geq 0\) and \(A\geq B+M\), we have
    \begin{equation}\begin{aligned}
        \fpoch{A-B}{M} \geq \fpoch{A}{M}\mleft(1 - \frac{MB}{A}\mright) \, , \quad \rpoch{A-B}{M} \geq \rpoch{A}{M}\mleft(1 - \frac{MB}{A+M-1}\mright) \, .
    \end{aligned}\end{equation}
\end{lemma}

\begin{proof}
    We start with the first inequality. We want to prove inductively that
    \begin{equation}
        \label{equ:proof_taylor_pochhammer_symbol_inequality}
        \fpoch{A-B}{M} \geq \fpoch{A}{M} - BM\fpoch{A-1}{M-1} \, .
    \end{equation}
    For \(M=1\), we have equality. Let us now proceed to the induction step for \(M+1\). We have
    \begin{equation}
    \begin{aligned}
        &\fpoch{A}{M} - \fpoch{A-B}{M} \\
        = &(A-(A-B))\fpoch{A-1}{M-1} + (A-B)\mleft(\fpoch{A-1}{M-1} - \fpoch{A-B-1}{M-1}\mright) \\
        \leq &B\fpoch{A-1}{M-1} + (A-B)B(M-1)\fpoch{A-2}{M-2} \leq BM\fpoch{A-1}{M-1} \, ,
    \end{aligned}
    \end{equation}
    where we used the induction hypothesis between the second and third row. Rearranging terms, we arrive at~\cref{equ:proof_taylor_pochhammer_symbol_inequality}. We remark that
    \begin{equation}\begin{aligned}
        \frac{\fpoch{A-1}{M-1}}{\fpoch{A}{M}} = \frac{1}{A} \, ,
    \end{aligned}\end{equation}
    which proves the first inequality. The second inequality follows from the fact that
    \begin{equation}\begin{aligned}
        \rpoch{K}{M} = \fpoch{K+M-1}{M} \, .
    \end{aligned}\end{equation}
\end{proof}

    \begin{lemma}[Averages of occupation numbers on the (dual) totally symmetric subspace]
\label{lem:upper_bounds_expectation_value_powers_of_weights_symmetric_subspace}
    Let $\boldsymbol{p}\dot{=}(p_1,\ldots,p_d)$, and let \(m\in \mathbb{N}\) and \(A>0\). For $\yd{\nu} \dot{=} [\mu_{1},\ldots,\mu_{1},0]$ and $p_{i} >p_{k}\geq 0$ for \(i\neq k\) we have
    \begin{equation}
        \weylavg{\mleft(\#_{k} + A\mright)^{m}}{\boldsymbol{p}}{\yd{\nu}} \leq m!\mleft(A+m-1+\sum_{i\neq k}\frac{p_{k}}{D_{i,k}}\mright)^{m} \, .
\end{equation}
Now let $p_{k} >p_{i}$ for all $i\neq k$, and let $\yd{\mu}\dot{=}[\mu_{1},0,\ldots,0]$. Then we have
    \begin{equation}
            \weylavg{\mleft(\mu_{1} - \#_{k} + A\mright)^{m}}{\boldsymbol{p}}{\yd{\mu}} \leq m!\mleft(A+m-1+\sum_{i\neq k}\frac{p_{i}}{D_{k,i}}\mright)^{m} \, .
\end{equation}
\end{lemma}

\begin{proof}
    First, the two inequalities are related by the fact that the irreps labeled by \(\yd{\nu}\) and \(\yd{\mu}\) are dual, that is, we have
    \begin{equation}
            \weylavg{\mleft(\#_{k} + A\mright)^{m}}{\boldsymbol{p}^{-1}}{\yd{\nu}} = \weylavg{\mleft(\mu_{1} - \#_{k} + A\mright)^{m}}{\boldsymbol{p}}{\yd{\mu}} \, ,
\end{equation}
    where \(\boldsymbol{p}^{-1}\dot{=}(p_{1}^{-1},\ldots,p_{d}^{-1})\). Therefore we can focus on the second inequality.

    Since all quantities are invariant under reordering of \(\boldsymbol{p}\), we can assume without loss of generality that \(k=d\). We will now prove that
    \begin{equation}
    \label{equ:proof_upper_bounds_expectation_value_powers_of_weights_symmetric_subspace_upper_bound_power}
            \weylavg{\mleft(\mu_{1} - \#_{d}\mright)^{m}}{\boldsymbol{p}}{\yd{\mu}} \leq m!\mleft(m-1+\sum_{i\neq d}\frac{p_{i}}{p_{d} - p_{i}}\mright)^{m} \, .
\end{equation}
We have
    \begin{equation}
            \weylavg{\mleft(\mu_{1} - \#_{d}\mright)^{m}}{\boldsymbol{p}}{\yd{\mu}} = \weylavg{\mleft(\sum_{i=1}^{d-1} \#_{i}\mright)^{m}}{\boldsymbol{p}}{\yd{\mu}} = \frac{\sum_{T=0}^{\mu_{1}}T^{m}h_{T}(r_{1},\ldots,r_{d-1})}{\sum_{T=0}^{\mu_{1}}h_{T}(r_{1},\ldots,r_{d-1})} \, ,
\end{equation}
    where \(h_{T}(r_{1},\ldots,r_{d-1})\) are the homogeneous symmetric polynomials of degree \(T\), that is,
    \begin{equation}
        h_{T}(r_{1},\ldots,r_{d-1}) \coloneqq s^{\yd{T}}(r_{1},\ldots,r_{d-1}) \, ,
    \end{equation}
    and \(r_{i} \coloneqq \frac{p_{i}}{p_{d}}\). We know that
    \begin{equation}
        \frac{\sum_{T=0}^{\mu_{1}}T^{m}h_{T}(r_{1},\ldots,r_{d-1})}{\sum_{T=0}^{\mu_{1}}h_{T}(r_{1},\ldots,r_{d-1})} \leq \mu_{1}^{m} < \frac{\sum_{T=\mu_{1}+1}^{\infty}T^{m}h_{T}(r_{1},\ldots,r_{d-1})}{\sum_{T=\mu_{1}+1}^{\infty}h_{T}(r_{1},\ldots,r_{d-1})} \, ,
    \end{equation}
    so we can use~\cref{lem:differences_inequalities} to get
    \begin{equation}
    \label{equ:proof_upper_bounds_expectation_value_powers_of_weights_symmetric_subspace_upper_bound_geometric}
            \weylavg{\mleft(\mu_{1} - \#_{d}\mright)^{m}}{\boldsymbol{p}}{\yd{\mu}} \leq \frac{\sum_{T=0}^{\infty}T^{m}h_{T}(r_{1},\ldots,r_{d-1})}{\sum_{T=0}^{\infty}h_{T}(r_{1},\ldots,r_{d-1})} = \mleft\langle \mleft(\sum_{i=1}^{d-1}X_{i} \mright)^{m}\mright\rangle_{\mathrm{G}(r_{1},\ldots,r_{d-1})} \, .
\end{equation}
Here, \(\mathrm{G}(r_{1},\ldots,r_{d-1})\) is the geometric distribution of \(d-1\) independent variables \(X_{i}\), that is,
    \begin{equation}
        P(X_{1},\ldots,X_{d-1}) = \prod_{i=1}^{d-1}(1-r_{i})r_{i}^{X_{i}} \, .
    \end{equation}
    Now let \(X\coloneqq \sum_{i=1}^{d-1}X_{i}\). We have
    \begin{equation}
        \mleft\langle S^{X_{i}} \mright\rangle_{\mathrm{G}(r_{1},\ldots,r_{d-1})} = \frac{1-r_{i}}{1-Sr_{i}} \, , \quad \mleft\langle S^{X} \mright\rangle_{\mathrm{G}(r_{1},\ldots,r_{d-1})} = \prod_{i=1}^{d-1}\frac{1-r_{i}}{1-Sr_{i}} \, ,
    \end{equation}
    and we can calculate
    \begin{equation}
        \mleft\langle \mleft(X_{i} \mright)^{\underline{m}} \mright\rangle_{\mathrm{G}(r_{1},\ldots,r_{d-1})} = \mleft.\mleft(\frac{d}{dS}\mright)^{m}\mright|_{S=1}\mleft( \mleft\langle S^{X_i} \mright\rangle_{\mathrm{G}(r_{1},\ldots,r_{d-1})} \mright) = m!\mleft(\frac{r_{i}}{1-r_{i}}\mright)^{m} \, .
    \end{equation}
    We set \(a_{i} \coloneqq \frac{r_{i}}{1-r_{i}}\), and as above, we see that
    \begin{equation}
    \begin{aligned}
        \mleft\langle \mleft(X \mright)^{\underline{m}} \mright\rangle_{\mathrm{G}(r_{1},\ldots,r_{d-1})} &= \sum_{m_{1},\ldots,m_{d-1}}\binom{m}{m_{1},\ldots,m_{d-1}}\prod_{i=1}^{d-1 }\mleft.\mleft(\frac{d}{dS}\mright)^{m_i}\mright|_{S=1}\mleft(\frac{1-r_{i}}{1-Sr_{i}}\mright) \\
        &= \sum_{m_{1},\ldots,m_{d-1}}\binom{m}{m_{1},\ldots,m_{d-1}}\prod_{i=1}^{d-1}(m_{i})!(a_{i})^{m_{i}}\\
        &= m! h_m(a_1,\ldots,a_{d-1})\\
    & \leq m!\mleft(\sum_{i=1}^{d-1}a_{i}\mright)^{m} \, .
    \end{aligned}
    \end{equation}
    We know that the falling factorials are a basis for the vector space of polynomials, so we have
    \begin{equation}
        \sum_{\ell=0}^{m}S(m,\ell)(x)^{\underline{\ell}} = x^{m} \, ,
    \end{equation}
    for
 coefficients \(S(m,\ell)\) (in this case the Stirling numbers of the second kind). Therefore, we get
    \begin{equation}
        \mleft\langle \mleft(X \mright)^{m} \mright\rangle_{\mathrm{G}(r_{1},\ldots,r_{d-1})} = \sum_{\ell=0}^{m}S(m,\ell)\mleft\langle \mleft(X \mright)^{\underline{\ell}} \mright\rangle_{\mathrm{G}(r_{1},\ldots,r_{d-1})} \leq m! \sum_{\ell=0}^{m}S(m,\ell) \mleft(\sum_{i=1}^{d-1}a_{i}\mright)^{\ell} \leq m!\mleft(m - 1 + \sum_{i=1}^{d-1}a_{i}\mright)^{m} \, .
    \end{equation}
    Together with~\cref{equ:proof_upper_bounds_expectation_value_powers_of_weights_symmetric_subspace_upper_bound_geometric} this proves~\cref{equ:proof_upper_bounds_expectation_value_powers_of_weights_symmetric_subspace_upper_bound_power}.

    We finally have
    \begin{equation}
            \weylavg{\mleft(\mu_{1} - \#_{d} + A\mright)^{m}}{\boldsymbol{p}}{\yd{\mu}} = \sum_{\ell=0}^{m}\binom{m}{\ell}\weylavg{\mleft(\mu_{1} - \#_{d}\mright)^{\ell}}{\boldsymbol{p}}{\yd{\mu}}A^{m-\ell} \leq m!\mleft(A+m-1+\sum_{i\neq d}a_{i}\mright)^{m} \, .
\end{equation}
    Substituting the definitions of \(a_{i}\) gives the result. From the aforementioned duality, we directly get the other inequality.
\end{proof}

\subsubsection{Sector-wise utility bounds}
\label{subsubsec:sector_utility_bounds_supp}
We give bounds on the sector-wise utility. This subsection will be dedicated to proving the following theorem:

\begin{theorem}[Nonasymptotic bounds on optimal sector-wise utility]
    \label{thm:sector_wise_utility}
    Consider
    \begin{equation}\begin{aligned}
        {\mathbf m}\dot{=}(m_{1},\ldots,m_{d}) \, , \quad 0\leq m_{i} \leq \Delta_{i,i+1} \, , \quad m = \sum_{i=1}^{d} m_{i} \, ,
    \end{aligned}\end{equation}
    with the convention that \(\Delta_{d,d+1} = \infty\). Finally, let
    \begin{equation}\begin{aligned}
        C_{<k} \coloneqq \sum_{i=1}^{k-1}\frac{q_i}{q_i-q_d} \, , \quad C_{\geq k} \coloneqq \sum_{i=k}^{d-1}\frac{q_i}{q_d-q_i} \, .
    \end{aligned}\end{equation}
    For the overhang-removal protocol, \(m_{k}=m\) and
    \begin{equation}\begin{aligned}
        \yd{\mu} \dot{=} [\varsigma_{1},\ldots,\varsigma_{k}-m,\ldots,\varsigma_{d}] \, .
    \end{aligned}\end{equation}
    At \(k=1\), the factors and conditions below involving \(\Delta_{k-1,k}\) are omitted; at \(k=d\), those involving \(\Delta_{k,k+1}\) are omitted.
    Then
    \begin{equation}\begin{aligned}
        \label{equ:thm_sector_wise_utility_lower_bound}
        \res{f}{\yd{\varsigma}}{\yd{m}}{\yd{\mu}}_{\mathrm{all}}\geq \frac{(\Delta_{k-1,k} + 2 - C_{<k})^{\overline{m}}}{(\Delta_{k-1,k} + 2)^{\overline{m}}} \frac{(\Delta_{k,k+1} - C_{\geq k})^{\underline{m}}}{(\Delta_{k,k+1})^{\underline{m}}} \, ,
    \end{aligned}\end{equation}
    as long as the expressions in the brackets are nonnegative. In addition, for any off-target contribution, i.e., any choice with \(m_{k}<m\), set
    \begin{equation}\begin{aligned}
        \yd{\lambda} \dot{=} [\varsigma_{1} - m_{1},\ldots,\varsigma_{d} - m_{d}] \, .
    \end{aligned}\end{equation}
    Then
    \begin{equation}\begin{aligned}
        \label{equ:thm_sector_wise_utility_upper_bound}
        \res{f}{\yd{\varsigma}}{\yd{m}}{\yd{\lambda}}_{\mathrm{all}} \leq \mleft(\frac{m\, m_{<k} (m_{<k}-1+C_{<k})}{\Delta_{k-1,k}}\mright)^{m_{<k}}\mleft(\frac{m\, m_{>k}(2m_{>k} - 1 + C_{\geq k})}{\Delta_{k,k+1}+m_{>k}+1}\mright)^{m_{>k}} \, ,
    \end{aligned}\end{equation}
    where
    \begin{equation}\begin{aligned}
        m_{<k} = \sum_{i=1}^{k-1}m_{i} \, , \quad m_{>k} = \sum_{i=k+1}^{d}m_{i} \, .
    \end{aligned}\end{equation}
\end{theorem}

\begin{proof}

    We first prove~\cref{equ:thm_sector_wise_utility_lower_bound}, so we start with the overhang removal rule. By~\cref{cor:utility_component_single_row}, or alternatively by choosing a path graph with the all-$\setminus$ convention, we can identify the edge variables $t_{i,d}$ with the number of boxes filled with the letter $d$ in the $i$-th row, also denoted by $\#_{d,i}$. We have
    \begin{equation}
        \label{equ:proof_fidelity_subcomponent}
        \res{f}{\yd{\varsigma}}{\yd{m}}{\yd{\mu}}_{\mathrm{all}}\mleft(\wt{w}\mright)=\prod_{i=1}^{d-1}f^{k,m}_{i}\mleft(\wt{w}\mright) \, ,
    \end{equation}
    where
    \begin{equation}\begin{aligned}
        f^{k,m}_{i}\mleft(\wt{w}\mright) \coloneqq \mleft\{\begin{aligned}
        &\dfrac{\mleft(\widetilde{\Delta}_{i,k} + 1 - (\#_{d,i} + 1)\mright)^{\overline{m}}}{\mleft(\widetilde{\Delta}_{i,k} + 1\mright)^{\overline{m}}}, & i < k, \\[10pt]
        &\dfrac{\mleft((\widetilde{\Delta}_{k,i+1} - 1) - \mleft(\Delta_{i,i+1} - \#_{d,i}\mright)\mright)^{\underline{m}}}{\mleft(\widetilde{\Delta}_{k,i+1} - 1\mright)^{\underline{m}}}, & i \geq k.
        \end{aligned}\mright.
    \end{aligned}\end{equation}
    We define
    $\#_d|_{s}^{t} \coloneqq \sum_{i=s}^{t}\#_{d,i}$.
    For \(1\leq i \leq k-1\), we set
    \begin{equation}\begin{aligned}
        A_{i} = \widetilde{\Delta}_{i,k}+1 \, , \quad B_{i} = \#_{d,i}+1 \, .
    \end{aligned}\end{equation}
    By applying the leftmost lower bound of~\cref{lem:bounds_telescopic_products} to each individual term in the rising factorial, we get
    \begin{equation}\begin{aligned}
        \label{equ:proof_fidelity_lower_bound_partial_product_less_than_k}
        \prod_{i=1}^{k-1}f^{k,m}_{i}\mleft(\wt{w}\mright) = \prod_{i=1}^{k-1} \frac{\mleft(A_{i} - B_{i}\mright)^{\overline{m}}}{\mleft(A_{i}\mright)^{\overline{m}}} \geq \frac{\mleft(A_{1} - \sum_{i=1}^{k-1}B_{i}\mright)^{\overline{m}}}{\mleft(A_{1}\mright)^{\overline{m}}} = \frac{\mleft(\widetilde{\Delta}_{k-1,k} + 1 - (\#_{d}|_{1}^{k-1} + k - 1)\mright)^{\overline{m}}}{\mleft(\widetilde{\Delta}_{k-1,k} + 1\mright)^{\overline{m}}} \, .
    \end{aligned}\end{equation}
    For $k\leq i\leq d$, we label
    \begin{equation}
        A_{i} = \Delta_{k,i+1} + i-k \, , \quad B_{i} = \Delta_{i,i+1} - \#_{d,i} \, ,
    \end{equation}
    and we get from~\cref{lem:bounds_telescopic_products} that
    \begin{equation}\begin{aligned}
        \label{equ:proof_fidelity_lower_bound_partial_product_bigger_than_k}
        \prod_{i=k}^{d-1}f^{k,m}_{i}\mleft(\wt{w}\mright) = \prod_{i=k}^{d-1} \frac{\mleft(A_{i} - B_{i}\mright)^{\underline{m}}}{\mleft(A_{i}\mright)^{\underline{m}}} \geq \frac{\mleft(A_{k} - \sum_{i=k}^{d-1}B_{i}\mright)^{\underline{m}}}{\mleft(A_{k}\mright)^{\underline{m}}} = \frac{\mleft(\#_{d}|_{k}^{d} -\varsigma_{k+1})\mright)^{\underline{m}}}{\mleft(\Delta_{k,k+1}\mright)^{\underline{m}}} \, .
    \end{aligned}\end{equation}
    For the last equality, we used that \(\#_{d,d}=\varsigma_{d}\). It follows that both functions on the right-hand side are convex in \(\#_{d}|_{1}^{k-1}\) and \(\#_{d}|_{k}^{d-1}\) respectively. We seek to use~\cref{cor:splitting_diagram_average_weight}, and we take
    \begin{equation}
        D|_{1}^{k-1}(\#_{d,1},\ldots,\#_{d,k-1}) = \mleft(\Delta_{k-1,k} + 2 - (\#_{d}|_{1}^{k-1} + k - 1)\mright)^{\overline{m}} \, , \quad I|_{k}^{d}(\#_{d,k},\ldots,\#_{d,d}) = \mleft(\#_{d}|_{k}^{d} -\varsigma_{k+1})\mright)^{\underline{m}} \, .
    \end{equation}
    Using the inequality from~\cref{cor:splitting_diagram_average_weight}, we find
    \begin{equation}
    \label{equ:proof_fidelity_on_target_split_bound}
            \weylavg{\res{f}{\yd{\varsigma}}{\yd{m}}{\yd{\mu}}_{\mathrm{all}}\mleft(\wt{w}\mright)}{\boldsymbol{q}}{\yd{\varsigma}} \geq \frac{\weylavg{D|_{1}^{k-1}(\#_{d,1},\ldots,\#_{d,k-1})}{(\boldsymbol{q}|_{1}^{k-1},q_{d})}{\yd{\Delta|_{1}^{k}}}}{\mleft(\Delta_{k-1,k} + 2\mright)^{\overline{m}}} \frac{\weylavg{I|_{k}^{d}(\#_{d,k},\ldots,\#_{d,d})}{\boldsymbol{q}|_{k}^{d}}{\yd{\varsigma|_{k}^{d}}}}{\mleft(\Delta_{k,k+1}\mright)^{\underline{m}}} \, .
\end{equation}
    Further applying Jensen's inequality allows us to pull the expectation value inside the brackets, and we get
    \begin{equation}
    \label{equ:proof_fidelity_on_target_lower_bound}
    \begin{aligned}
        &\weylavg{\res{f}{\yd{\varsigma}}{\yd{m}}{\yd{\mu}}_{\mathrm{all}}\mleft(\wt{w}\mright)}{\boldsymbol{q}}{\yd{\varsigma}}
\geq\frac{\mleft(\Delta_{k-1,k} + 2 - (\weylavg{\#_{k}}{(\boldsymbol{q}|_{1}^{k-1},q_{d})}{\yd{\Delta|_{1}^{k}}} + k - 1)\mright)^{\overline{m}}}{\mleft(\Delta_{k-1,k} + 2\mright)^{\overline{m}}} \frac{\mleft(\weylavg{\#_{d-k+1}}{\boldsymbol{q}|_{k}^{d}}{\yd{\varsigma|_{k}^{d}}} -\varsigma_{k+1}\mright)^{\underline{m}}}{\mleft(\Delta_{k,k+1}\mright)^{\underline{m}}} \, .
    \end{aligned}
\end{equation}
    By~\cref{cor:monotonicity_functions_young_diagrams} for the first and~\cref{lem:upper_bounds_expectation_value_powers_of_weights_symmetric_subspace} for the second inequality, we find that
   \begin{equation}
    \label{equ:proof_fidelity_1_to_k_upper_bound}
            \weylavg{\#_{k}}{(\boldsymbol{q}|_{1}^{k-1},q_{d})}{\yd{\Delta|_{1}^{k}}} \leq \weylavg{\#_{k}}{(\boldsymbol{q}|_{1}^{k-1},q_{d})}{[\Delta_{1,k},\ldots,\Delta_{1,k},0]} \leq \sum_{i=1}^{k-1}\frac{q_{d}}{q_{i} - q_{d}}\ \, .
\end{equation}
We remark here that
    \begin{equation}\begin{aligned}
        k - 1 + \sum_{i=1}^{k-1}\frac{q_{d}}{q_{i} - q_{d}} = \sum_{i=1}^{k-1} \frac{q_{i}}{q_{i} - q_{d}} \, .
    \end{aligned}\end{equation}
    In a similar fashion, we can use~\cref{cor:monotonicity_functions_young_diagrams} and~\cref{lem:upper_bounds_expectation_value_powers_of_weights_symmetric_subspace} to get
    \begin{equation}
    \label{equ:proof_fidelity_k_to_d_lower_bound}
            \weylavg{\#_{d-k+1}}{\boldsymbol{q}|_{k}^{d}}{\yd{\varsigma|_{k}^{d}}} \geq \weylavg{\#_{d-k+1}}{\boldsymbol{q}|_{k}^{d}}{[\varsigma_{k},0,\ldots,0]} \geq \varsigma_{k} - \sum_{i=k}^{d-1}\frac{q_{i}}{q_{d} - q_{i}} \, ,
\end{equation}
    where we used
    \begin{equation}
        \langle \varsigma_{k} - \#_{d-k+1}\rangle_{\mathrm{W}(\boldsymbol{q}|_{k}^{d}\rev{,}\yd{\varsigma_{k},0,\ldots,0})} \leq \sum_{i=k}^{d-1}\frac{q_{i}}{q_{d} - q_{i}} \, .
    \end{equation}
Inserting~\cref{equ:proof_fidelity_1_to_k_upper_bound} and~\cref{equ:proof_fidelity_k_to_d_lower_bound} into~\cref{equ:proof_fidelity_on_target_lower_bound} now gives~\cref{equ:thm_sector_wise_utility_lower_bound}.
    \begin{equation}
        \weylavg{\varsigma_{k} - \#_{d-k+1}}{\boldsymbol{q}|_{k}^{d}}{\yd{\varsigma_{k},0,\ldots,0}} \leq \sum_{i=k}^{d-1}\frac{q_{i}}{q_{d} - q_{i}} \, .
    \end{equation}

    We now consider a general off-target removal with environment irrep \(\yd{\lambda}\) and prove~\cref{equ:thm_sector_wise_utility_upper_bound}. As above, we use~\cref{thm:utility_component} to get
    \begin{equation}\begin{aligned}
        \res{f}{\yd{\varsigma}}{\yd{m}}{\yd{\lambda}}_{\mathrm{all}}\mleft(\wt{w}\mright)=\binom{m}{\mathbf m}\prod_{j=1}^{d}\prod_{i=1}^{d-1}f^{j,m_{j}}_i\mleft(\wt{w}\mright) \, .
    \end{aligned}\end{equation}
    Since each \(f^{j,m_{j}}_i\mleft(\wt{w}\mright)\leq 1\), we have
    \begin{equation}\begin{aligned}
        \label{equ:proof_fidelity_upper_bounds_batched}
        \res{f}{\yd{\varsigma}}{\yd{m}}{\yd{\lambda}}_{\mathrm{all}}\mleft(\wt{w}\mright) \leq \binom{m}{\mathbf m}\mleft(\prod_{j=1}^{k-1}\prod_{i=j}^{k-1}f^{j,m_{j}}_i\mleft(\wt{w}\mright)\mright)\mleft(\prod_{j=k+1}^{d}\prod_{i=k}^{j-1}f^{j,m_{j}}_i\mleft(\wt{w}\mright)\mright) \, .
    \end{aligned}\end{equation}
    We use the leftmost upper bound of~\cref{lem:bounds_telescopic_products} as in~\cref{equ:proof_fidelity_lower_bound_partial_product_bigger_than_k}, that is, we take
    \begin{equation}
        A_{i} = \widetilde{\Delta}_{j,i+1} - 1 \, , \quad B_{i} = \Delta_{i,i+1} - \#_{d,i} \, ,
    \end{equation}
    and we find that
    \begin{equation}
    \begin{aligned}
\label{equ:proof_fidelity_upper_bounds_index_wise_1}
        \prod_{i=j}^{k-1} f^{j,m_{j}}_i\mleft(\wt{w}\mright) &= \prod_{i=j}^{k-1} \frac{\mleft(A_{i} - B_{i}\mright)^{\underline{m}}}{\mleft(A_{i}\mright)^{\underline{m}}} \\
        &\leq \frac{\fpoch{A_{k-1} - \sum_{i=j}^{k-1}B_{i}}{m_{j}}}{\fpoch{A_{k-1}}{m_{j}}} \\
        &= \frac{(\#_{d}|_{j}^{k-1} + (k - 1 - j))^{\underline{m_{j}}}}{(\widetilde{\Delta}_{j,k}-1)^{\underline{m_{j}}}} \\
        &\leq \frac{(\#_{d}|_{1}^{k-1} + (k - 1))^{\underline{m_{j}}}}{(\Delta_{k-1,k})^{\underline{m_{j}}}} \, .
    \end{aligned}
    \end{equation}
    In the last inequality, we used trivial bounds on the numerator and denominator. For the complementary bound, following the approach of~\cref{equ:proof_fidelity_lower_bound_partial_product_less_than_k}, we can set
    \begin{equation}\begin{aligned}
        A_{i} = \widetilde{\Delta}_{i,j}+1 \, , \quad B_{i} = \#_{d,i}+1 \, ,
    \end{aligned}\end{equation}
    and we get
    \begin{equation}
    \label{equ:proof_fidelity_upper_bounds_index_wise_2}
    \begin{aligned}
        \prod_{i=k}^{j-1} f^{j,m_{j}}_i\mleft(\wt{w}\mright)
        &= \prod_{i=k}^{j-1} \frac{\mleft(A_{i} - B_{i}\mright)^{\overline{m}}}{\mleft(A_{i}\mright)^{\overline{m}}}
        \leq \frac{\rpoch{A_{k} - \sum_{i=k}^{j-1}B_{i}}{m_{j}}}{\rpoch{A_{k}}{m_{j}}} \\
        &= \frac{(\Delta_{k,j} - \#_{d}|_{k}^{j-1} + 1)^{\overline{m_{j}}}}{(\widetilde{\Delta}_{k,j}+1)^{\overline{m_{j}}}}
        \leq \frac{(\Delta_{k,d} - \#_{d}|_{k}^{d-1} + 1)^{\overline{m_{j}}}}{(\widetilde{\Delta}_{k,k+1} + 1)^{\overline{m_{j}}}} \, .
    \end{aligned}
    \end{equation}
    We use the fact that for any \(A>B>0\) we have
    \begin{equation}\begin{aligned}
        \frac{B^{\underline{s}}}{A^{\underline{s}}} \leq \mleft(\frac{B}{A}\mright)^{s} \, , \quad \frac{B^{\overline{s}}}{A^{\overline{s}}} \leq \mleft(\frac{B+s-1}{A+s-1}\mright)^{s}
    \end{aligned}\end{equation}
    together with~\cref{equ:proof_fidelity_upper_bounds_batched,equ:proof_fidelity_upper_bounds_index_wise_1,equ:proof_fidelity_upper_bounds_index_wise_2} to see that
    \begin{equation}\begin{aligned}
        \res{f}{\yd{\varsigma}}{\yd{m}}{\yd{\lambda}}_{\mathrm{all}}\mleft(\wt{w}\mright) \leq \binom{m}{\mathbf m} \mleft(\frac{\#_{d}|_{1}^{k-1} + (k - 1)}{\Delta_{k-1,k}}\mright)^{m_{<k}} \mleft(\frac{\Delta_{k,d} - \#_{d}|_{k}^{d-1} + m_{>k}}{\widetilde{\Delta}_{k,k+1} + m_{>k}}\mright)^{m_{>k}} \, .
    \end{aligned}\end{equation}
    We again want to use~\cref{cor:splitting_diagram_average_weight}, and we set
    \begin{equation}
        I|_{1}^{k-1}(\#_{d,1},\ldots,\#_{d,k-1}) = \mleft(\#_{d}|_{1}^{k-1} + (k - 1)\mright)^{m_{<k}} \, , \quad D|_{k}^{d}(\#_{d,k},\ldots,\#_{d,d}) = \mleft(\Delta_{k,d} - \#_{d}|_{k}^{d-1} + m_{>k}\mright)^{m_{>k}} \, .
    \end{equation}
    With the upper bound from~\cref{cor:splitting_diagram_average_weight}, we find
    \begin{equation}
    \begin{aligned}
        \weylavg{\res{f}{\yd{\varsigma}}{\yd{m}}{\yd{\lambda}}_{\mathrm{all}}\mleft(\wt{w}\mright)}{\boldsymbol{q}}{\yd{\varsigma}} \!\!\!\!\!\!\!\!\leq \binom{m}{\mathbf m} &\weylavg{\mleft(\frac{\#_{k} + (k - 1)}{\Delta_{k-1,k}}\mright)^{m_{<k}}}{(\boldsymbol{q}|_{1}^{k-1},q_{d})}{\yd{\Delta|_{1}^{k}}} \weylavg{\mleft(\frac{\Delta_{k,d} - \#_{d-k+1} + m_{>k}}{\widetilde{\Delta}_{k,k+1} + m_{>k}}\mright)^{m_{>k}}}{\boldsymbol{q}|_{k}^{d}}{\yd{\varsigma|_{k}^{d}}} \, .
    \end{aligned}
\end{equation}
By~\cref{cor:monotonicity_functions_young_diagrams} for the first and~\cref{lem:upper_bounds_expectation_value_powers_of_weights_symmetric_subspace} for the second inequality, we find that
    \begin{equation}
    \begin{aligned}
        &\weylavg{\mleft(\#_{k} + (k - 1)\mright)^{m_{<k}}}{(\boldsymbol{q}|_{1}^{k-1},q_{d})}{\yd{\Delta|_{1}^{k}}} \\ \leq &\weylavg{\mleft(\#_{k} + (k - 1)\mright)^{m_{<k}}}{(\boldsymbol{q}|_{1}^{k-1},q_{d})}{[\Delta_{1,k},\ldots,\Delta_{1,k},0]}\\
        \leq& (m_{<k})!\mleft(m_{<k}-1+k-1 + \sum_{i=1}^{k-1}\frac{q_{d}}{q_{i} - q_{d}}\mright)^{m_{<k}} \, .
    \end{aligned}
\end{equation}
To get the result, we again invoke that
    \begin{equation}
        k-1 + \sum_{i=1}^{k-1}\frac{q_{d}}{q_{i} - q_{d}} = \sum_{i=1}^{k-1}\frac{q_{i}}{q_{i} - q_{d}} \, .
    \end{equation}
    In a similar fashion, we can use~\cref{cor:monotonicity_functions_young_diagrams} and~\cref{lem:upper_bounds_expectation_value_powers_of_weights_symmetric_subspace} to get
    \begin{equation}
    \begin{aligned}
        &\weylavg{\mleft(\varsigma_{k} - \#_{d-k+1} + m_{>k}\mright)^{m_{>k}}}{\boldsymbol{q}|_{k}^{d}}{\yd{\varsigma|_{k}^{d}}} \\ \leq &\weylavg{\mleft(\varsigma_{k} - \#_{d-k+1} + m_{>k}\mright)^{m_{>k}}}{\boldsymbol{q}|_{k}^{d}}{[\varsigma_{k},0,\ldots,0]} \\
        \leq& (m_{>k})!\mleft(2m_{>k} - 1 + \sum_{i=k}^{d-1}\frac{q_{i}}{q_{d} - q_{i}} \mright)^{m_{>k}} \, .
    \end{aligned}
\end{equation}
We finally remark that
$(m_{<k})! \leq (m_{<k})^{m_{<k}}$, $(m_{>k})! \leq (m_{>k})^{m_{>k}} $, and $ \binom{m}{\mathbf m} \leq m^{m_{<k}+m_{>k}}$.
Putting everything together, we arrive at~\cref{equ:thm_sector_wise_utility_upper_bound}.
\end{proof}

This implies that whenever $\Delta_{k-1,k}$ and $\Delta_{k,k+1}$ are larger than a threshold $\Delta^\ast(\boldsymbol{p})$, the optimal channel for maximizing the sector-wise utility is given by the overhang removal rule. Moreover, the threshold is independent of $d$. The following corollary from~\cref{thm:sector_wise_utility} gives a simpler form of the lower and upper bounds.

\begin{corollary}
    \label{cor:sector_wise_utility_simplified}
    For $1 \leq m \leq \Delta_{k,k+1}$, we have
    \begin{equation}\begin{aligned}
        \res{^kf}{\yd{\varsigma}}{\yd{m}}{\yd{\varsigma-m\mathbf{e}_k}}_{\mathrm{all}}\geq 1 - \frac{m}{\widetilde{\Delta}_{k-1,k}+m}\sum_{i=1}^{k-1}\frac{p_i}{D_{i,k}} - \frac{m}{\Delta_{k,k+1}}\sum_{i=k+1}^{d}\frac{p_i}{D_{k,i}} \, .
    \end{aligned}\end{equation}
    In addition, for \(m=1\) and any off-target contribution \(\ell\neq k\), we have
        \begin{equation}
            \label{equ:cor_sector_wise_utility_off_target_upper_bound}
            \begin{aligned}
            \res{^kf}{\yd{\varsigma}}{\yd{1}}{\yd{\varsigma-\mathbf{e}_\ell}}_{\mathrm{all}}\leq \mleft\{\begin{aligned}
            &\dfrac{1}{\Delta_{\ell,k}}\displaystyle\sum_{i=1}^{k-1}\dfrac{p_i}{D_{i,k}}, & \ell<k, \\[10pt]
            &\dfrac{1}{\Delta_{k,\ell} + 2}\mleft(1+\displaystyle\sum_{i=k+1}^{d}\dfrac{p_i}{D_{k,i}}\mright), & \ell>k.
            \end{aligned}\mright.
            \end{aligned}
        \end{equation}
\end{corollary}

\begin{proof}
    For the first inequality, we apply~\cref{lem:taylor_pochhammer_symbol} to~\cref{thm:sector_wise_utility}. The second inequality follows directly from~\cref{thm:sector_wise_utility} after inserting \(m=1\). In both cases, we rewrite the reindexed spectrum using \(q_d=p_k\), \(q_i=p_i\) for \(i<k\), and \(q_i=p_{i+1}\) for \(k\leq i\leq d-1\).
\end{proof}

\subsection{Nonasymptotic analysis of the one-site utility}
\label{subsec:nonasymptotic_onecopy_supp}

\begin{theorem}
    \label{thm:bounds_one_utility}
    Let $\yd{\varsigma}\dot{=}[\varsigma_{1},\ldots,\varsigma_{d}]$, $m\geq 1$ and let $\yd{\mu}$ be defined by the overhang removal rule of~\cref{def:overhang_removal_rule}. Then the optimal sector-wise one-site utility satisfies the lower bound
    \begin{equation}\begin{aligned}
        \res{{^kf}}{\yd{\varsigma}}{\yd{m}}{\yd{\mu}}_{\mathrm{one}}
        \geq 1 - \frac{1}{\widetilde{\Delta}_{k-1,k}}\sum_{i=1}^{k-1}\frac{p_{i}}{D_{i,k}} - \frac{1}{\Delta_{k,k+1}}\sum_{i=k+1}^{d}\frac{p_{i}}{D_{k,i}} - (i^\ast-k)\frac{m-\Delta_{k,k+1}}{m(\Delta_{k,k+1} + (i^\ast-k))} \, .
    \end{aligned}\end{equation}
    Now let $\yd{\lambda}$ be a valid environment irrep. Then we have the following upper bound
    \begin{equation}\begin{aligned}
        \res{{^kf}}{\yd{\varsigma}}{\yd{m}}{\yd{\lambda}}_{\mathrm{one}}
        \leq 1 - \frac{(i^\ast-k)}{\widetilde{\Delta}_{k,i^\ast}}\frac{m-\Delta_{k,i^\ast}}{m}\mleft(1 - \frac{1}{\Delta_{k,k+1}}\mleft(1+\sum_{i=k+1}^{d}\frac{p_{i}}{D_{k,i}}\mright)\mright) \, .
    \end{aligned}\end{equation}
\end{theorem}

\begin{proof}
    We use~\cref{lem:one_site_channel_decomposition}, and we start by proving the lower and upper bounds. For any environment $\yd{\lambda}$, we have
    \begin{equation}\begin{aligned}
        &\resSix{F}{\yd{\lambda}}{\yd{m-1}}{\yd{1}}{\yd{\varsigma-\mathbf{e}_k}}{\yd{m}}{\yd{\varsigma}}^2 \res{{^kf}}{\yd{\varsigma}}{\yd{1}}{\yd{\varsigma-\mathbf{e}_k}}_{\mathrm{all}} \\
        \leq &\sum_{i=1}^{d}\resSix{F}{\yd{\lambda}}{\yd{m-1}}{\yd{1}}{\yd{\varsigma-\mathbf{e}_i}}{\yd{m}}{\yd{\varsigma}}^2 \res{{^kf}}{\yd{\varsigma}}{\yd{1}}{\yd{\varsigma-\mathbf{e}_i}}_{\mathrm{all}} = \res{{^kf}}{\yd{\varsigma}}{\yd{m}}{\yd{\lambda}}_{\mathrm{one}} \\
        \leq& \resSix{F}{\yd{\lambda}}{\yd{m-1}}{\yd{1}}{\yd{\varsigma-\mathbf{e}_k}}{\yd{m}}{\yd{\varsigma}}^2 + \mleft(1 - \resSix{F}{\yd{\lambda}}{\yd{m-1}}{\yd{1}}{\yd{\varsigma-\mathbf{e}_k}}{\yd{m}}{\yd{\varsigma}}^2\mright) \max_{\substack{\yd{\varsigma-\mathbf{e}_i}\vdash_d (n,1) \\ i \neq k}} \res{{^kf}}{\yd{\varsigma}}{\yd{1}}{\yd{\varsigma-\mathbf{e}_i}}_{\mathrm{all}} \, ,
        \label{equ:proof_bounds_one_fidelity_upper_bound}
    \end{aligned}\end{equation}
    where we used that any sector-wise utility is bounded above by $1$, and, for the second inequality, the $F$-symbols add up to $1$ (see~\cref{rem:F_symbols_isometry}). We set $i=k$, and recall~\cref{rmk:F_symbol_row_removal}, that is,
    \begin{equation}\begin{aligned}
        \resSix{F}{\yd{\mu}}{\yd{m-1}}{\yd{1}}{\yd{\varsigma-\mathbf{e}_k}}{\yd{m}}{\yd{\varsigma}}^2 = \mleft(\prod_{j=k+1}^{i^\ast} \frac{\widetilde{\Delta}_{k,j}-1}{\widetilde{\Delta}_{k,j}}\mright)\mleft(\frac{m + (i^\ast-k)}{m}\mright) \, ,
    \end{aligned}\end{equation}
    and we use~\cref{lem:bounds_telescopic_products} for
    \begin{equation}
        A_{j} = \widetilde{\Delta}_{k,j} \, , \quad B_{j} = 1 \, ,
    \end{equation}
    we find
    \begin{equation}
        \frac{\Delta_{k,k+1}}{\Delta_{k,k+1} + (i^{\ast}-k)} = \frac{A_{k+1} - B_{k+1}}{A_{k+1}+\sum_{j=k+2}^{i^{\ast}}B_{j}} \leq \prod_{j=k+1}^{i^\ast} \frac{\widetilde{\Delta}_{k,j}-1}{\widetilde{\Delta}_{k,j}} \leq \frac{A_{i^{\ast}} - B_{k+1}}{A_{i^{\ast}}+\sum_{j=k+2}^{i^{\ast}}B_{j}} = \frac{\Delta_{k,i^{\ast}}}{\widetilde{\Delta}_{k,i^{\ast}}} \, .
    \end{equation}
    Thus, we have
    \begin{equation}\begin{aligned}
        \frac{\Delta_{k,k+1}}{\Delta_{k,k+1}+(i^\ast-k)}\frac{m+(i^\ast-k)}{m} \leq \resSix{F}{\yd{\mu}}{\yd{m-1}}{\yd{1}}{\yd{\varsigma-\mathbf{e}_k}}{\yd{m}}{\yd{\varsigma}}^2 \leq \frac{\Delta_{k,i^\ast}}{\widetilde{\Delta}_{k,i^\ast}}\frac{m+(i^\ast-k)}{m} \, .
    \end{aligned}\end{equation}
    For the lower bound, it follows that
    \begin{equation}
        1 - (i^\ast-k)\frac{m-\Delta_{k,k+1}}{m(\Delta_{k,k+1} + (i^\ast-k))} \leq \frac{\Delta_{k,k+1}}{\Delta_{k,k+1}+(i^\ast-k)}\frac{m+(i^\ast-k)}{m} \, .
    \end{equation}
    For the upper bound, note that~\cref{equ:proof_bounds_one_fidelity_upper_bound} is monotone in $\resSix{F}{\yd{\lambda}}{\yd{m-1}}{\yd{1}}{\yd{\varsigma-\mathbf{e}_k}}{\yd{m}}{\yd{\varsigma}}^2$, since the maximum over the fidelities is always upper bounded by $1$. We can then apply~\cref{lem:optimality_of_F_symbol} together with~\cref{thm:sector_wise_utility} to obtain the result.
\end{proof}

\subsection{Nonasymptotic overall bounds}
\label{subsec:nonasymptotic_overall_bounds_supp}

\begin{lemma}[Polynomial tail bound]
\label{cor:poly_tail_bound}
Choose the parameter $\alpha = c\sqrt{\frac{\ln n}{n}}$ with $c = 12$. For $n \ge e$,
\begin{equation}
\Pr\mleft(\,|\overline{\Delta}_{k-1,k}-D_{k-1,k}|\ge \alpha\ \lor\ \ |\overline{\Delta}_{k,k+1}-D_{k,k+1}|\ge \alpha\,\mright)
\le
4n^{-3/2}.
\label{eq:tail_poly_bound_nminus2}
\end{equation}
For $k\!=\!1$ or $d$, only one row difference is needed and the factor $4$ becomes $2$.
\end{lemma}

\begin{proof}
\Cref{lem:upper_bound_probability_row_differences} gives
\begin{equation}
\begin{aligned}
\label{eq:tail_expanded_alpha_choice}
\Pr\mleft(\,|\overline{\Delta}_{k-1,k}-D_{k-1,k}|\ge \alpha\ \lor \ |\overline{\Delta}_{k,k+1}-D_{k,k+1}|\ge \alpha\,\mright)
&\le
4e^{-\frac{(c\sqrt{\ln n}-4)^2}{32}}\\
&=
4e^{-\frac{c^2}{32}\ln n+\frac{c}{4}\sqrt{\ln n}-\frac12}\\
&=
4e^{-1/2}e^{\frac{c}{4}\sqrt{\ln n}}\,n^{-c^2/32},
\end{aligned}
\end{equation}
For $n\ge e$ we have $\sqrt{\ln n}\le \ln n$, hence
\begin{equation}
e^{\frac{c}{4}\sqrt{\ln n}}\le e^{\frac{c}{4}\ln n}=n^{c/4},
\end{equation}
and therefore for all $n\ge e$,
\begin{equation}
\Pr\mleft(\,|\overline{\Delta}_{k-1,k}-D_{k-1,k}|\ge \alpha\ \lor \ |\overline{\Delta}_{k,k+1}-D_{k,k+1}|\ge \alpha\,\mright)
\le
4\,e^{-1/2}\,n^{-c^2/32+c/4}
\le
4\,n^{-c^2/32+c/4}.
\label{eq:tail_poly_bound_general_c}
\end{equation}
Choosing $c=12$ so that $-c^2/32+c/4=-3/2$ yields the estimate.
\end{proof}

For the nonasymptotic analysis, we derive a more precise bound for the worst-case utility. We show that:

\begin{lemma}[Overall lower bound]
\label{lem:overall_utility_bound}
Define
\begin{equation}
\label{eq:adjacent_gap_index_set}
J\coloneqq
\left\{
\begin{array}{lr}
\{1\}, & k=1,\\
\{d-1\}, & k=d,\\
\{k-1,k\}, & 1<k<d.
\end{array}
\right.
\end{equation}
If the sector-wise utility has a lower bound of the form
\(f^{\yd{\varsigma}}(\mathcal T)\geq 1-\sum_{j\in J}\frac{c_{j}}{n}\frac{1}{\overline{\Delta}_{j,j+1}}\)
for \(n\geq m/\overline{\Delta}_{k,k+1}\) when \(k<d\), with no such condition when \(k=d\), and with all \(c_j\geq0\), then the overall utility has a lower bound
\begin{equation}
\label{eq:F_collect_two_gaps}
\mathcal F\mleft(\mathcal T\mright)
\;\ge\;
1-\sum_{j\in J}\frac{c_{j}}{nD_{j,j+1}}
-
R,
\end{equation}
where
\begin{equation}
\label{eq:R_collect_two_gaps}
R
\le
\frac{4\sqrt{\ln n}}{n^{3/2}}
\mleft(
1+6\frac{\sum_{j\in J}c_j}{D_{k,\mathrm{min}}^2}
\mright).
\end{equation}
Here $D_{k,\mathrm{min}}\coloneqq\min_{j\in J}D_{j,j+1}$ and $c=12$.
Define
\begin{equation}
N_0\coloneqq8c^2D_{k,\mathrm{min}}^{-2}\ln\!\mleft(4c^2D_{k,\mathrm{min}}^{-2}\mright),
\qquad
N\coloneqq
\mleft\{\begin{aligned}
&\max\mleft(N_0,\dfrac{2m}{D_{k,k+1}}\mright), & k<d,\\
&N_0, & k=d.
\end{aligned}\mright.
\end{equation}
Then the bound holds whenever \(n\geq N\).
\end{lemma}
\begin{proof}
As in~\cref{cor:poly_tail_bound}, set \(\alpha\coloneqq c\sqrt{\ln n/n}\). We start by upper bounding the factor $\frac{1}{\overline{\Delta}_{j,j+1}}$ for $j\in J$. When the condition $\alpha\le\ \frac{D_{j,j+1}}{2}$ holds,
\begin{equation}
\frac{1}{D_{j,j+1}-\alpha}
=
\frac{1}{D_{j,j+1}}+\frac{\alpha}{D_{j,j+1}(D_{j,j+1}-\alpha)}
\le
\frac{1}{D_{j,j+1}}+\frac{2\alpha}{D_{j,j+1}^2}
\le
\frac{1}{D_{j,j+1}}
+\frac{2c}{D_{j,j+1}^2}\sqrt{\frac{\ln n}{n}}.
\end{equation}
To ensure \(\alpha\le D_{j,j+1}/2\) for all \(j\in J\), we first impose the sufficient condition
\(n\geq8c^2D_{k,\mathrm{min}}^{-2}\ln\!\mleft(4c^2D_{k,\mathrm{min}}^{-2}\mright)\).
Since \(x\mapsto x^{-2}\ln(4c^2x^{-2})\) is decreasing and \(D_{k,\mathrm{min}}\leq D_{j,j+1}\), this condition implies
\begin{equation}
\label{eq:gapwise_log_threshold}
n \;\ge\; \frac{8c^2}{D_{j,j+1}^2}\,\ln\!\Big(\frac{4c^2}{D_{j,j+1}^2}\Big),
\qquad j\in J.
\end{equation}
To see that~\cref{eq:gapwise_log_threshold} gives the desired inequality, for each \(j\in J\), set \(A\coloneqq4c^2/D_{j,j+1}^2\) and \(n_0\coloneqq2A\ln A\), and consider \(n\geq n_0\). Since \(c=12\) and \(0<D_{j,j+1}\leq1\), we have \(A>1\) and hence \(2\ln A\leq A\). It follows that \(\ln(2\ln A)\leq\ln A\), and therefore \(\ln(2A\ln A)\leq2\ln A\). Thus,
\(
\frac{\ln n_0}{n_0}
\leq
\frac{2\ln A}{2A\ln A}
=
\frac{1}{A}
=
\frac{D_{j,j+1}^2}{4c^2}.
\)
Since \(n_0>e\) and \(\ln n/n\) is decreasing for \(n\geq e\), the same bound holds for every \(n\geq n_0\), proving that the displayed condition implies \(\alpha\leq D_{j,j+1}/2\).
For \(k<d\), on the typical set the assumptions imply
\(
\Delta_{k,k+1}\geq n(D_{k,k+1}-\alpha)\geq nD_{k,k+1}/2\geq m,
\)
so the sector-wise lower bound applies. For \(k=d\), this first-overhang condition is absent.
Collecting the pieces, we obtain
\begin{equation}
\label{eq:F_collect_two_gaps_proof}
\mathcal F\mleft(\mathcal T\mright)
\;\ge\;
\Bigl(1-\sum_{j\in J}\frac{c_{j}}{n}\frac{1}{D_{j,j+1}-\alpha}\Bigr)\Bigl(1-4n^{-3/2}\Bigr)
\;\ge\;
1-\sum_{j\in J}\frac{c_{j}}{nD_{j,j+1}}
-
R,
\end{equation}
where
\begin{equation}
\label{eq:R_collect_two_gaps_proof}
R
=
\frac{4}{n^{3/2}}
+\sum_{j\in J}\mleft(\frac{2\,c\,c_{j}}{D_{j,j+1}^2}\frac{\sqrt{\ln n}}{n^{3/2}}\mright)
\;\le\;
\frac{4}{n^{3/2}}
+\frac{24(\sum_{j\in J}c_j)}{D_{k,\mathrm{min}}^2}\frac{\sqrt{\ln n}}{n^{3/2}}
\;\le\;
\frac{4\sqrt{\ln n}}{n^{3/2}}
\mleft(
1+6\frac{\sum_{j\in J}c_j}{D_{k,\mathrm{min}}^2}
\mright),
\end{equation}
Here, the last inequality uses $1\le \sqrt{\ln n}$, which holds for $n\ge e$.
\end{proof}

\begin{theorem}[Nonasymptotic all-site utility bound]
\label{thm:nonasymptotic_all_site_fidelity_bound}
    Let \(D_{k,\mathrm{min}}\), \(c\), and \(N\) be as in~\cref{lem:overall_utility_bound}.
    Assume \(n\geq N\). Then
\begin{equation}
    {^k}\mathcal{F}_\mathrm{all}\geq 1 - \frac{m}{nD_{k,\mathrm{min}}}\sum_{i\neq k}\frac{p_i}{|D_{k,i}|}-R_{\mathrm{all}}.
\end{equation}
where
\begin{equation}
\label{eq:all_site_nonasymptotic_remainder}
R_{\mathrm{all}}
\le
\frac{4\sqrt{\ln n}}{n^{3/2}}
\mleft(
1+
6\frac{m}{D_{k,\mathrm{min}}^2}
\sum_{i\ne k}\frac{p_i}{|D_{k,i}|}
\mright).
\end{equation}
In particular, this proves~\cref{equ:utility_bound_non_asymptotic}.
\end{theorem}
\begin{proof}
    The lower bound in~\cref{cor:sector_wise_utility_simplified} implies the hypothesis of~\cref{lem:overall_utility_bound} with the corresponding coefficients, since $\Delta_{k-1,k}+m+1\geq \Delta_{k-1,k}$ and $\Delta_{j,j+1}=n\overline{\Delta}_{j,j+1}$.
    Applying~\cref{lem:overall_utility_bound} gives a lower bound involving \(\sum_{j\in J}c_j/D_{j,j+1}\). Moreover,
    \begin{equation}
    \label{eq:nonasymptotic_all_site_coefficient_bound}
    \begin{aligned}
    \sum_{j\in J}\frac{c_j}{D_{j,j+1}}
    &=
    \frac{m}{D_{k-1,k}}
    \sum_{i<k}\frac{p_i}{D_{i,k}}
    +
    \frac{m}{D_{k,k+1}}
    \sum_{i>k}\frac{p_i}{D_{k,i}}
    \\
    &\le
    \frac{m}{D_{k,\mathrm{min}}}
    \sum_{i\ne k}\frac{p_i}{|D_{k,i}|}.
    \end{aligned}
    \end{equation}
    At the boundaries, the absent term is omitted. Substituting this estimate into the bound from~\cref{lem:overall_utility_bound} proves the claim.
\end{proof}

\begin{theorem}[Nonasymptotic one-site utility bound]
\label{thm:nonasymptotic_one_site_fidelity_bound}
Let \(D_{k,\mathrm{min}}\), \(c\), and \(N\) be as in~\cref{lem:overall_utility_bound}.
Assume \(n\geq N\).
Then the one-site utility satisfies
\begin{equation}
\label{equ:one_site_non_asymptotic_fidelity_bound}
{^k}\mathcal{F}_{\mathrm{one}}
\ge
1
-
\frac{1}{nD_{k,\mathrm{min}}}
\sum_{i\ne k}\frac{p_i}{|D_{k,i}|}
-
R_{\mathrm{one}}.
\end{equation}
where
\begin{equation}
\label{eq:one_site_nonasymptotic_remainder}
R_{\mathrm{one}}
\le
\frac{4\sqrt{\ln n}}{n^{3/2}}
\left(
1+
6\frac{1}{D_{k,\mathrm{min}}^2}
\sum_{i\neq k}\frac{p_i}{|D_{k,i}|}
\right).
\end{equation}
\end{theorem}

\begin{proof}
The correction term being subtracted in~\cref{thm:bounds_one_utility} is nonpositive, so dropping it only weakens the lower bound. The remaining sector-wise lower bound satisfies the hypotheses of~\cref{lem:overall_utility_bound} with the corresponding coefficients. Applying that lemma and the coefficient comparison in~\cref{eq:nonasymptotic_all_site_coefficient_bound} with \(m=1\) gives the stated utility bound.
\end{proof}

\section{Optimality of overhang removal}
\label{sec:optimality_supp}

\subsection{Asymptotic optimality}
\label{subsec:asymptotic_optimality_supp}
We prove that the overhang-removal protocol achieves the optimal asymptotic scaling. We assume the spectrum is nondegenerate and $\yd{\varsigma}$ is typical, i.e., $\overline{\yd{\varsigma}}\in\mathsf{R}_{\epsilon}$.

\begin{lemma}[\rev{Above-target deviations from overhang removal}] \label{lem:asymptotic_deviations_from_overhang_removal}
    Let \(\overline{\yd{\varsigma}}\in \mathsf{R}_{\epsilon}\), and let \(\boldsymbol R^{\yd{\mu}}\) be given by the overhang removal rule detailed in~\cref{lem:leading_order_F_symbol_overhang_removal} above. Further, let \(\boldsymbol{R}\) be a \rev{removal-rate vector} such that
    \begin{equation}
        \sum_{i=1}^{k-1}R_{i} = \rev{\delta \geq0} \, , \quad \sum_{i=1}^{d}R_{i} = \sum_{i=1}^{d}R_i^{\yd{\mu}} = R \, .
    \end{equation}
    Then we have
    \begin{equation}
        A_k\mleft(\overline{\yd{\varsigma}},\boldsymbol R^{\yd{\mu}}\mright) - A_k\mleft(\overline{\yd{\varsigma}},\boldsymbol{R}\mright) \geq \rev{\frac{1}{8}\min\left(\delta,\frac{\epsilon R}{R+\epsilon}\right)}\rev{\left(1 + \frac{1}{R}\right)} \, .
    \end{equation}
    On the other hand, for \(1\leq i \leq k-1\), we have
    \begin{equation}
        A_i\mleft(\overline{\yd{\varsigma}},\boldsymbol{R}\mright) \leq \rev{2\delta\frac{R+\epsilon}{R\epsilon}}
    \end{equation}
\end{lemma}

\begin{proof}
    We start with the first statement. We observe that
    \begin{equation}
        \frac{\partial}{\partial R_{j}}A_k\mleft(\overline{\yd{\varsigma}},\boldsymbol{R}\mright) = A_k\mleft(\overline{\yd{\varsigma}},\boldsymbol{R}\mright)\left(\frac{1}{\overline{\Delta}_{k,j}+R_{j}} - \frac{1}{R}\right) \, .
    \end{equation}
    Thus, for \(x\) sufficiently small, we have
    \begin{equation}
    \begin{aligned}
    \label{equ:proof_asymptotic_deviations_from_overhang_removal_integral}
    &A_k\mleft(\overline{\yd{\varsigma}},\boldsymbol{R}+x(\boldsymbol{e}_{s}-\boldsymbol{e}_{t})\mright) - A_k\mleft(\overline{\yd{\varsigma}},\boldsymbol{R}\mright) \\=&
        \int_{0}^{x} A_k\mleft(\overline{\yd{\varsigma}},\boldsymbol{R}+\alpha(\boldsymbol{e}_{s}-\boldsymbol{e}_{t})\mright)\left(\frac{1}{\overline{\Delta}_{k,s}+R_{s}+\alpha} - \frac{1}{\overline{\Delta}_{k,t}+R_{t}-\alpha}\right)\,\mathrm{d}\alpha \, .
    \end{aligned}
    \end{equation}
    As in the proof of~\cref{lem:optimality_of_F_symbol}, we can check that the above quantity is always \(\geq 0\) if \(k \leq s < t\), \(t< k \leq s\), or \(s<t<k\). Since $A_k$ does not decrease under any such transfer, repeatedly transferring mass from a later row $t$ to the earliest unsaturated row $s$, with $k\leq s<t$, while leaving the components above $k$ unchanged, yields the following removal rates:
    \begin{equation}
        \widetilde R_j \coloneqq
        \mleft\{\begin{aligned}
            &R_j, & 1\leq j\leq k-1,\\
            &\overline\Delta_{j,j+1}, & k\leq j<j^\ast,\\
            &R-\delta-\overline\Delta_{k,j^\ast}, & j=j^\ast,\\
            &0, & j^\ast<j\leq d,
        \end{aligned}\mright.
    \end{equation}
    Here \(j^\ast\coloneqq\min\{j\mid\overline\Delta_{k,j+1}>R-\delta\}\), and \(\sum_{j=1}^{d}\widetilde R_j=R\). This implies by the argument above that
    \begin{equation}
        A_k\mleft(\overline{\yd{\varsigma}},\boldsymbol{R}\mright) \leq A_k\mleft(\overline{\yd{\varsigma}},\widetilde{\boldsymbol{R}}\mright) \leq A_k\mleft(\overline{\yd{\varsigma}},\boldsymbol R^{\yd{\mu}}\mright) \, .
    \end{equation}
    In particular, $\boldsymbol R^{\yd{\mu}}$ maximizes $A_k\mleft(\overline{\yd{\varsigma}},\,\cdot\,\mright)$ over all valid removal rates, and~\cref{equ:leading_order_overhang_F_symbol} gives $A_k\mleft(\overline{\yd{\varsigma}},\boldsymbol R^{\yd{\mu}}\mright)=1$.
    \rev{If \(\delta=0\), the right-hand side of the first bound vanishes, so the first bound is immediate. This includes the \(k=1\) case. We therefore assume \(\delta>0\) and \(k>1\) for the remainder of its proof.}

    To lower-bound the difference between $A_k\mleft(\overline{\yd{\varsigma}},\boldsymbol R^{\yd{\mu}}\mright)$ and $A_k\mleft(\overline{\yd{\varsigma}},\widetilde{\boldsymbol R}\mright)$, we set
    \begin{equation}
        \boldsymbol{R}(x)\coloneqq x\boldsymbol R^{\yd{\mu}} + (1-x)\widetilde{\boldsymbol{R}} \, ,
    \end{equation}
    \begin{equation}
    \label{equ:proof_asymptotic_deviations_from_overhang_removal_interpolation}
    \begin{aligned}
        A_k\mleft(\overline{\yd{\varsigma}},\boldsymbol R^{\yd{\mu}}\mright) - A_k\mleft(\overline{\yd{\varsigma}},\boldsymbol{R}(x)\mright) &= \int_{x}^{1} \sum_{j=1}^{d}(R_j^{\yd{\mu}} - \widetilde{R}_{j})A_k\mleft(\overline{\yd{\varsigma}},\boldsymbol{R}(\alpha)\mright)\left(\frac{1}{\overline{\Delta}_{k,j}+R_{j}(\alpha)} - \frac{1}{R}\right)\,\mathrm{d}\alpha \\
        &=\int_{x}^{1} \left(\sum_{j=1}^{k-1} \frac{\widetilde{R}_{j} A_k\mleft(\overline{\yd{\varsigma}},\boldsymbol{R}(\alpha)\mright)}{\overline{\Delta}_{j,k}-(1-\alpha)\widetilde{R}_{j}} + \sum_{j=k}^{d}\frac{(R_j^{\yd{\mu}} - \widetilde{R}_{j}) A_k\mleft(\overline{\yd{\varsigma}},\boldsymbol{R}(\alpha)\mright)}{\overline{\Delta}_{k,j}+(1-\alpha)\widetilde{R}_{j}+\alpha R_j^{\yd{\mu}}}\right)\,\mathrm{d}\alpha \, .
    \end{aligned}
    \end{equation}
    Every term in the final integrand is nonnegative; hence \(A_k\mleft(\overline{\yd{\varsigma}},\boldsymbol R(x)\mright)\) is monotonically increasing in \(x\).
    Moreover, for \(1\leq j \leq k-1\),
    \begin{equation}
        \overline{\Delta}_{j,k}-(1-\alpha)\widetilde{R}_{j} \leq \overline{\Delta}_{1,k} \, ,
    \end{equation}
    and therefore
    \begin{equation}
        \sum_{j=1}^{k-1} \frac{\widetilde{R}_{j}}{\overline{\Delta}_{j,k}-(1-\alpha)\widetilde{R}_{j}} \geq \frac{\delta}{\overline{\Delta}_{1,k}} \, .
    \end{equation}
    Further, for the relevant indices \(j^\ast\leq j\leq i^\ast\), we have
    \begin{equation}
        \overline{\Delta}_{k,j}+(1-\alpha)\widetilde{R}_{j}+\alpha R_j^{\yd{\mu}} \leq \overline{\Delta}_{k,j} + R_j^{\yd{\mu}} \leq R \, ,
    \end{equation}
    and so
    \begin{equation}
        \sum_{j=k}^{d}\frac{R_j^{\yd{\mu}} - \widetilde{R}_{j}}{\overline{\Delta}_{k,j}+(1-\alpha)\widetilde{R}_{j}+\alpha R_j^{\yd{\mu}}} \geq \frac{\delta}{R} \, .
    \end{equation}
    Finally, we define
    \begin{equation}
        \widetilde{\delta}\coloneqq \rev{\frac{1}{4\delta}\min\left(\delta,\frac{\epsilon R}{R+\epsilon}\right)} \, ,
    \end{equation}
    The components of \(\boldsymbol R^{\yd{\mu}}\) and \(\widetilde{\boldsymbol R}\) on rows \(j<k\) and \(j\geq k\) each differ by total mass \(\delta\), so \(\left\|\boldsymbol R^{\yd{\mu}}-\widetilde{\boldsymbol R}\right\|_1=2\delta\).
    Consequently, for \(1-\widetilde{\delta}\leq\alpha\leq1\),
    \begin{equation}
        \left\|\boldsymbol R^{\yd{\mu}}-\boldsymbol R(\alpha)\right\|_1
        =(1-\alpha)\left\|\boldsymbol R^{\yd{\mu}}-\widetilde{\boldsymbol R}\right\|_1
        \leq2\widetilde{\delta}\delta \, .
    \end{equation}
    Applying~\cref{lem:asymptotic_behavior_of_F_symbol} and using \(A_k\mleft(\overline{\yd{\varsigma}},\boldsymbol R^{\yd{\mu}}\mright)=1\), we obtain
    \begin{equation}
        \label{equ:proof_asymptotic_deviations_from_overhang_removal_integrand_lower_bound}
        A_k\mleft(\overline{\yd{\varsigma}},\boldsymbol{R}(\alpha)\mright)
        \geq 1-\mleft|A_k\mleft(\overline{\yd{\varsigma}},\boldsymbol R^{\yd{\mu}}\mright)
        -A_k\mleft(\overline{\yd{\varsigma}},\boldsymbol{R}(\alpha)\mright)\mright|
        \geq 1-\rev{2\widetilde{\delta}\delta\frac{R+\epsilon}{\epsilon R}}
        \geq \frac{1}{2}\, .
    \end{equation}
    \rev{Using~\cref{equ:proof_asymptotic_deviations_from_overhang_removal_integrand_lower_bound} to lower-bound \(A_k\mleft(\overline{\yd{\varsigma}},\boldsymbol R(\alpha)\mright)\) in the integrand of~\cref{equ:proof_asymptotic_deviations_from_overhang_removal_interpolation}, we find that}
    \begin{equation}
        A_k\mleft(\overline{\yd{\varsigma}},\boldsymbol R^{\yd{\mu}}\mright) - A_k\mleft(\overline{\yd{\varsigma}},\boldsymbol{R}\mright) \geq A_k\mleft(\overline{\yd{\varsigma}},\boldsymbol R^{\yd{\mu}}\mright) - A_k\mleft(\overline{\yd{\varsigma}},\boldsymbol{R}(1-\widetilde{\delta})\mright) \geq \rev{\frac{1}{8}\min\left(\delta,\frac{\epsilon R}{R+\epsilon}\right)\left(1+\frac{1}{R}\right)} \, .
    \end{equation}
    \rev{This proves the first statement.}

    \rev{We now prove the second statement for all \(\delta\geq0\).} Set \(\delta_i\coloneqq\sum_{j=i}^{k-1}R_j\), and define \(\boldsymbol R^\prime\) by compacting the remaining mass toward row \(k\):
    \begin{equation}
        R_j^\prime\coloneqq
        \mleft\{\begin{aligned}
            &0, & 1\leq j<i,\\
            &R_j, & i\leq j<k,\\
            &\overline\Delta_{j,j+1}, & k\leq j<j_i^\ast,\\
            &R-\delta_i-\overline\Delta_{k,j_i^\ast}, & j=j_i^\ast,\\
            &0, & j_i^\ast<j\leq d
        \end{aligned}\mright.\rev{\, .}
    \end{equation}
    Here \(j_i^\ast\coloneqq\min\{j\mid\overline\Delta_{k,j+1}>R-\delta_i\}\), and \(\sum_{j=1}^{d}R_j^\prime=R\). To obtain \(\boldsymbol R^\prime\) from \(\boldsymbol R\), we first transfer the mass on rows \(j<i\) to rows \(j\geq k\), and then compact all mass on rows \(j\geq k\) toward row \(k\). For \(A_i\), these transfers satisfy \(t<i\leq s\) and \(i\leq s<t\), respectively, so the same calculation shows that neither transfer decreases \(A_i\). Moreover, \(\boldsymbol R^\prime\) and \(\boldsymbol R^{\yd{\mu}}\) use the same compacted rule on rows \(j\geq k\), with respective total masses \(R-\delta_i\) and \(R\). Therefore, \(\left\|\boldsymbol R^{\yd{\mu}}-\boldsymbol R^\prime\right\|_1=2\delta_i\leq2\delta\), and
    \begin{equation}
        0= A_i\mleft(\overline{\yd{\varsigma}},\boldsymbol R^{\yd{\mu}}\mright) \leq A_i\mleft(\overline{\yd{\varsigma}},\boldsymbol{R}\mright) \leq A_i\mleft(\overline{\yd{\varsigma}},\boldsymbol{R}^\prime\mright) \, .
    \end{equation}
    Together with~\cref{lem:asymptotic_behavior_of_F_symbol} we see that
    \begin{equation}
        A_i\mleft(\overline{\yd{\varsigma}},\boldsymbol{R}\mright) \leq A_i\mleft(\overline{\yd{\varsigma}},\boldsymbol{R}^\prime\mright) \leq \rev{\frac{R+\epsilon}{R\epsilon}\left\|\boldsymbol R^{\yd{\mu}} - \boldsymbol{R}^\prime\right\|_{1}} \leq \rev{2\delta\frac{R+\epsilon}{R\epsilon}} \, .
    \end{equation}
\end{proof}

\begin{theorem}[Sector-wise asymptotic optimality]
\label{thm:asymptotic_optimality_irrep_level}
Let $\overline{\yd{\varsigma}}\in\mathsf{R}_{\epsilon}$, and let $\yd{\mu}$ be the environment irrep given by the overhang removal rule. For fixed $m$ and the all-site loss function, for all sufficiently large $n$,
\begin{equation}
    \max_{\yd{\lambda}\,\mathrm{valid}}
    \res{{^kf}}{\yd{\varsigma}}{\yd{m}}{\yd{\lambda}}_{\mathrm{all}}
    =
    \res{{^kf}}{\yd{\varsigma}}{\yd{m}}{\yd{\mu}}_{\mathrm{all}} \, ,
\end{equation}
uniformly in $\overline{\yd{\varsigma}}\in\mathsf R_\epsilon$.
For the one-site loss function, we have
\begin{equation}
    \res{{^kf}}{\yd{\varsigma}}{\yd{m}}{\yd{\lambda}}_{\mathrm{one}}
    \leq
    \res{{^kf}}{\yd{\varsigma}}{\yd{m}}{\yd{\mu}}_{\mathrm{one}}
    +\lilo[d,\boldsymbol{p},\epsilon]{n^{-1}} \, ,
\end{equation}
uniformly in $\overline{\yd{\varsigma}}\in\mathsf R_\epsilon$, all valid environment irreps $\yd{\lambda}$, and $m$. Equivalently,
\begin{equation}
\label{equ:one_site_sector_uniform_optimality_equivalent}
\sup_m\sup_{\overline{\yd{\varsigma}}\in\mathsf R_\epsilon}
\sup_{\yd{\lambda}\,\mathrm{valid}}
\mleft[
\res{{^kf}}{\yd{\varsigma}}{\yd{m}}{\yd{\lambda}}_{\mathrm{one}}
-\res{{^kf}}{\yd{\varsigma}}{\yd{m}}{\yd{\mu}}_{\mathrm{one}}
\mright]
=\lilo[d,\boldsymbol{p},\epsilon]{n^{-1}} \, .
\end{equation}
\end{theorem}

\begin{proof}
For $\overline{\yd{\varsigma}}\in\mathsf R_\epsilon$, the target-row gaps grow at least as $n\epsilon$. Hence~\cref{cor:non_asymptotic_optimality_irrep_level} implies, for all sufficiently large $n$, that every valid $\yd{\lambda}\neq\yd{\mu}$ has strictly smaller all-site utility than $\yd{\mu}$. The threshold for $n$ depends only on $d,\boldsymbol p,m,\epsilon$, so the conclusion is uniform in $\overline{\yd{\varsigma}}\in\mathsf R_\epsilon$.
We now turn to the one-site case. By~\cref{lem:one_site_channel_decomposition},
\begin{equation}\begin{aligned}
    \res{{^kf}}{\yd{\varsigma}}{\yd{m}}{\yd{\lambda}}_{\mathrm{one}}
    = \sum_{i:\yd{\varsigma-\mathbf{e}_i} \vdash_{d}(n,1)}
    \resSix{F}{\yd{\lambda}}{\yd{m-1}}{\yd{1}}{\yd{\varsigma-\mathbf{e}_i}}{\yd{m}}{\yd{\varsigma}}^2
    \res{{^kf}}{\yd{\varsigma}}{\yd{1}}{\yd{\varsigma-\mathbf{e}_i}}_{\mathrm{all}} \, .
\end{aligned}\end{equation}
\rev{For $1\leq i\leq d$, define}
\begin{equation}
\label{equ:proof_asymptotic_optimality_irrep_level_coefficients}
\rev{
C_i(n)\coloneqq
\left\{
\begin{array}{@{}l@{\qquad}l@{}}
\displaystyle
1-\frac{1}{n}\sum_{\ell\neq k}
\frac{p_\ell}{\overline{\Delta}_{k,\ell}D_{k,\ell}},
& i=k,\\[10pt]
\displaystyle
\frac{1}{n}\frac{p_k}{\overline{\Delta}_{k,i}D_{k,i}},
& i\neq k.
\end{array}
\right.}
\end{equation}
Therefore
{\scriptsize\normalsizeequationtag
\begin{equation}
\label{equ:proof_asymptotic_optimality_irrep_level_decomposition}
    \begin{aligned}
        &\res{{^kf}}{\yd{\varsigma}}{\yd{m}}{\yd{\lambda}}_{\mathrm{one}} - \res{{^kf}}{\yd{\varsigma}}{\yd{m}}{\yd{\mu}}_{\mathrm{one}} \\
        ={}&\mleft(
        \resSix{F}{\yd{\lambda}}{\yd{m-1}}{\yd{1}}{\yd{\varsigma-\mathbf{e}_k}}{\yd{m}}{\yd{\varsigma}}^2
        -\resSix{F}{\yd{\mu}}{\yd{m-1}}{\yd{1}}{\yd{\varsigma-\mathbf{e}_k}}{\yd{m}}{\yd{\varsigma}}^2
        \mright)C_k(n)
        +\sum_{j\neq k}\mleft(
        \resSix{F}{\yd{\lambda}}{\yd{m-1}}{\yd{1}}{\yd{\varsigma-\mathbf{e}_j}}{\yd{m}}{\yd{\varsigma}}^2
        -\resSix{F}{\yd{\mu}}{\yd{m-1}}{\yd{1}}{\yd{\varsigma-\mathbf{e}_j}}{\yd{m}}{\yd{\varsigma}}^2
        \mright)C_j(n) + \lilo[d,\boldsymbol{p},\epsilon]{n^{-1}} \, ,
    \end{aligned}
\end{equation}
}
Here we used~\cref{thm:constant_output_intensive_sector_all_site_expansion}. The term $\lilo[d,\boldsymbol{p},\epsilon]{n^{-1}}$ is uniform in $m$, since only the squared $F$-symbols depend on $m$, and they are uniformly bounded by $1$.

Now let \(\yd{\lambda^\ast}\) be the environment irrep maximizing $\res{{^kf}}{\yd{\varsigma}}{\yd{m}}{\yd{\lambda^\ast}}_{\mathrm{one}}$.
Such a \(\yd{\lambda^\ast}\) always exists, since there are only finitely many possibilities for fixed \(\yd{\varsigma}\) and \(m\).
In particular,
\begin{equation}
\label{equ:proof_asymptotic_optimality_irrep_level_maximizer_bound}
    \res{{^kf}}{\yd{\varsigma}}{\yd{m}}{\yd{\lambda^\ast}}_{\mathrm{one}}
    -
    \res{{^kf}}{\yd{\varsigma}}{\yd{m}}{\yd{\mu}}_{\mathrm{one}}
    \geq
    \max\mleft(
    \res{{^kf}}{\yd{\varsigma}}{\yd{m}}{\yd{\lambda}}_{\mathrm{one}}
    -
    \res{{^kf}}{\yd{\varsigma}}{\yd{m}}{\yd{\mu}}_{\mathrm{one}},
    0
    \mright).
\end{equation}
Let \(\boldsymbol R^{\yd{\lambda^\ast}}\coloneqq\overline{\yd{\varsigma}}-\overline{\yd{\lambda^\ast}}\) be the removal-rate vector of \(\yd{\lambda^\ast}\), and write \(R\coloneqq\sum_{i=1}^{d}R_i^{\yd{\lambda^\ast}}=m/n\).
Set \(\delta\coloneqq\sum_{i=1}^{k-1}R_i^{\yd{\lambda^\ast}}\). \rev{By~\cref{lem:asymptotic_behavior_of_F_symbol,lem:asymptotic_deviations_from_overhang_removal}, for \(1<k\leq d\),}
\begin{equation}
\label{equ:proof_asymptotic_optimality_irrep_level_target_F_symbol_bound}
\begin{aligned}
    \rev{\resSix{F}{\yd{\lambda^\ast}}{\yd{m-1}}{\yd{1}}{\yd{\varsigma-\mathbf{e}_k}}{\yd{m}}{\yd{\varsigma}}^2
    -\resSix{F}{\yd{\mu}}{\yd{m-1}}{\yd{1}}{\yd{\varsigma-\mathbf{e}_k}}{\yd{m}}{\yd{\varsigma}}^2}
    &\rev{\leq \frac{8d^3}{n\epsilon} + A_k\mleft(\overline{\yd{\varsigma}},\boldsymbol R^{\yd{\lambda^\ast}}\mright) - A_k\mleft(\overline{\yd{\varsigma}},\boldsymbol R^{\yd{\mu}}\mright)} \\
    &\rev{\leq \frac{8d^3}{n\epsilon}
    -\frac{1}{8}\min\left(\delta,\frac{\epsilon R}{R+\epsilon}\right)
    \left(1 + \frac{1}{R}\right).}
\end{aligned}
\end{equation}

\rev{We now consider \(1<k\leq d\). We insert the bound above and the second inequality encoded in~\cref{equ:proof_asymptotic_optimality_irrep_level_maximizer_bound} into~\cref{equ:proof_asymptotic_optimality_irrep_level_decomposition}, together with the general bound \(\resSix{F}{\yd{\lambda}}{\yd{m-1}}{\yd{1}}{\yd{\varsigma-\mathbf{e}_i}}{\yd{m}}{\yd{\varsigma}}^2-\resSix{F}{\yd{\mu}}{\yd{m-1}}{\yd{1}}{\yd{\varsigma-\mathbf{e}_i}}{\yd{m}}{\yd{\varsigma}}^2\leq1\), to obtain}
\begin{equation}
\begin{aligned}
    \label{equ:proof_asymptotic_optimality_irrep_level_lower_bound}
    0 &\leq \res{{^kf}}{\yd{\varsigma}}{\yd{m}}{\yd{\lambda^\ast}}_{\mathrm{one}} - \res{{^kf}}{\yd{\varsigma}}{\yd{m}}{\yd{\mu}}_{\mathrm{one}} \\
    &\leq -\rev{\frac{1}{8}\min\!\left(\delta,\frac{\epsilon R}{R+\epsilon}\right)}
    \rev{\!\left(1 + \frac{1}{R}\right)}
    \!\left(1 - \frac{1}{n}\sum_{i\neq k}\frac{p_i}{\rev{\overline{\Delta}_{k,i}D_{k,i}}}\right)
    + \frac{1}{n}\sum_{i\neq k}\frac{\rev{p_k}}{\rev{\overline{\Delta}_{k,i}D_{k,i}}}
    + \frac{8d^3}{n\epsilon}+ \lilo[d,\boldsymbol{p},\epsilon]{n^{-1}} \, .
\end{aligned}
\end{equation}
\rev{When \(k=1\), we have \(\delta=0\), and the off-target sum in~\cref{equ:proof_asymptotic_optimality_irrep_level_off_target_bound} below is empty, so the required bound is immediate.}
For \(1<k\leq d\), \cref{equ:proof_asymptotic_optimality_irrep_level_lower_bound} implies
\begin{equation}
    \label{equ:proof_asymptotic_optimality_irrep_level_delta_minimum}
    \rev{\frac{1}{8}\min\left(\delta,\frac{\epsilon R}{R+\epsilon}\right)\left(1 + \frac{1}{R}\right)} = \bigO[d,\boldsymbol{p},\epsilon]{n^{-1}}
\end{equation}
uniformly for all choices of \(m,\yd{\varsigma}\). \rev{If the minimum in~\cref{equ:proof_asymptotic_optimality_irrep_level_delta_minimum} were attained by its second argument, the left-hand side would be at least \(\epsilon/8\). Hence, for sufficiently large \(n\), with a threshold independent of \(R\), the minimum is attained by its first argument, \(\delta\), giving}
\begin{equation}
    \label{equ:proof_asymptotic_optimality_irrep_level_delta_bound}
    \delta
    =
    \rev{\bigO[d,\boldsymbol p,\epsilon]{\frac{R}{n(R+1)}}} \, .
\end{equation}
Together with the trivial \(k=1\) case, these estimates and the second statement of~\cref{lem:asymptotic_deviations_from_overhang_removal} give
{\scriptsize\normalsizeequationtag
\begin{equation}
    \label{equ:proof_asymptotic_optimality_irrep_level_off_target_bound}
    \frac{1}{n}\sum_{i=1}^{k-1}\mleft(
    \resSix{F}{\yd{\lambda^\ast}}{\yd{m-1}}{\yd{1}}{\yd{\varsigma-\mathbf{e}_i}}{\yd{m}}{\yd{\varsigma}}^2
    -\resSix{F}{\yd{\mu}}{\yd{m-1}}{\yd{1}}{\yd{\varsigma-\mathbf{e}_i}}{\yd{m}}{\yd{\varsigma}}^2
    \mright)\frac{\rev{p_k}}{\rev{\overline{\Delta}_{k,i}D_{k,i}}}
    \leq \rev{\frac{2\delta}{n}\frac{R+\epsilon}{R\epsilon}} \sum_{i=1}^{k-1}\frac{\rev{p_k}}{\rev{\overline{\Delta}_{k,i}D_{k,i}}}
    +\frac{4d^3}{n^2\epsilon}\sum_{i=1}^{k-1}\frac{\rev{p_k}}{\rev{\overline{\Delta}_{k,i}D_{k,i}}}
    = \lilo[d,\boldsymbol{p},\epsilon]{n^{-1}} \, ,
\end{equation}
}\rev{where we use~\cref{equ:proof_asymptotic_optimality_irrep_level_delta_bound} and \((R+\epsilon)/(R+1)\leq1\), so both terms are \(\bigO[d,\boldsymbol p,\epsilon]{n^{-2}}\).}
We can again insert into~\cref{equ:proof_asymptotic_optimality_irrep_level_decomposition} and get
\begin{equation}
\begin{aligned}
    &\res{{^kf}}{\yd{\varsigma}}{\yd{m}}{\yd{\lambda}}_{\mathrm{one}} - \res{{^kf}}{\yd{\varsigma}}{\yd{m}}{\yd{\mu}}_{\mathrm{one}} \leq \res{{^kf}}{\yd{\varsigma}}{\yd{m}}{\yd{\lambda^\ast}}_{\mathrm{one}} - \res{{^kf}}{\yd{\varsigma}}{\yd{m}}{\yd{\mu}}_{\mathrm{one}} \\
    &={}\sum_{i=k}^{d}\mleft(
    \resSix{F}{\yd{\lambda^\ast}}{\yd{m-1}}{\yd{1}}{\yd{\varsigma-\mathbf{e}_i}}{\yd{m}}{\yd{\varsigma}}^2
    -\resSix{F}{\yd{\mu}}{\yd{m-1}}{\yd{1}}{\yd{\varsigma-\mathbf{e}_i}}{\yd{m}}{\yd{\varsigma}}^2
    \mright) C_{i}(n) + \lilo[d,\boldsymbol{p},\epsilon]{n^{-1}}
    \leq \lilo[d,\boldsymbol{p},\epsilon]{n^{-1}} \, .
\end{aligned}
\label{equ:proof_asymptotic_optimality_irrep_level_vector_product}
\end{equation}
By~\cref{equ:off_target_removal_asymptotic_sector_wise_fidelity}, $C_k(n)\geq\cdots\geq C_d(n)\geq0$ for all sufficiently large $n$. Hence,~\cref{lem:optimality_of_F_symbol} implies
\begin{equation}
\label{equ:proof_asymptotic_optimality_weighted_F_symbol_ordering}
\sum_{i=k}^{d}
\resSix{F}{\yd{\lambda^\ast}}{\yd{m-1}}{\yd{1}}{\yd{\varsigma-\mathbf{e}_i}}{\yd{m}}{\yd{\varsigma}}^2 C_i(n)
\leq
\sum_{i=k}^{d}
\resSix{F}{\yd{\mu}}{\yd{m-1}}{\yd{1}}{\yd{\varsigma-\mathbf{e}_i}}{\yd{m}}{\yd{\varsigma}}^2 C_i(n) \, .
\end{equation}
\noindent Therefore the sum in~\cref{equ:proof_asymptotic_optimality_irrep_level_vector_product} is nonpositive. The term \(\lilo[d,\boldsymbol{p},\epsilon]{n^{-1}}\) only depends on \(d,\boldsymbol{p},\epsilon\) and is uniform for all choices of \(m,\yd{\varsigma}\) by our previous derivation.
\end{proof}

\begin{theorem}[Asymptotic optimality]
    \label{thm:asymptotic_optimality_QPA}
The overhang-removal protocol is asymptotically minimax optimal. For each input sector $\yd{\varsigma}$, let $\yd{\mu}$ be given by the overhang removal rule, and denote the corresponding all-site and one-site utilities by ${}^k\mathcal F_{\mathrm{all}}$ and ${}^k\mathcal F_{\mathrm{one}}$, respectively. For fixed $m$ under the all-site loss,
\begin{equation}
\label{equ:all_site_asymptotic_minimax_utility}
{}^k\mathcal F_{\mathrm{all}}^\ast
= {}^k\mathcal F_{\mathrm{all}}
+\lilo[d,\boldsymbol p,m]{n^{-1}} \, .
\end{equation}
For the one-site loss,
\begin{equation}
\label{equ:one_site_asymptotic_minimax_utility}
{}^k\mathcal F_{\mathrm{one}}^\ast
= {}^k\mathcal F_{\mathrm{one}}
+\lilo[d,\boldsymbol p]{n^{-1}} \, ,
\end{equation}
uniformly over all $m\in\mathbb N$. The corresponding asymptotic utility expressions are given separately by~\cref{thm:constant_output_intensive_sector_all_site_expansion,thm:asymptotics_one_fidelity}, respectively.
\end{theorem}

\begin{proof}
By the symmetry reduction in~\cref{sec:math_prelim_supp}, the fidelity of any QPA protocol $\mathcal{T}$ is upper bounded by that of a symmetric protocol. By~\cref{thm:SM_QPA_choi}, a symmetric protocol decomposes into independent input sectors $\yd{\varsigma}$, and~\cref{equ:the_fidelity_is,equ:the_one_copy_fidelity_is} express its all-site and one-site utilities as SW averages of the corresponding sector-wise utilities. On the sectors with $\overline{\yd{\varsigma}}\in\mathsf R_\epsilon$, we apply the uniform sector-wise comparison of~\cref{thm:asymptotic_optimality_irrep_level}. The remaining sectors have total SW probability $\bigO[][n]{\exp(-n)}$ by~\cref{lem:sw_averaging_uniform_sector_expansion}, and all fidelities are bounded by $1$. Averaging the sector-wise comparison therefore gives~\cref{equ:all_site_asymptotic_minimax_utility,equ:one_site_asymptotic_minimax_utility}, with the stated uniformity in $m$ for the one-site utility.
\end{proof}

\subsection{Nonasymptotic optimality}
\label{subsec:nonasymptotic_optimality_supp}
With more refined analysis, we obtain nonasymptotic optimality in $n$ for sublinear $m$ with $m_{n}=\lilo{n^{1/3}}$.
We first present sector-wise results showing that, when the YD $\yd{\varsigma}$ has sufficiently large target row gaps, the overhang removal rule is uniquely optimal.

\begin{corollary}[Sector-wise nonasymptotic optimality]
    \label{cor:non_asymptotic_optimality_irrep_level}
    Let \(\yd{\mu}\) be defined as in~\cref{def:overhang_removal_rule}, and let \(\yd{\lambda}\) be any valid environment irrep with \(\yd{\lambda}\neq\yd{\mu}\).
    At \(k=1\), the first condition in each pair below is omitted; at \(k=d\), the second condition is omitted. In what follows, \(\Delta_{k,\mathrm{min}}\) denotes the minimum of the existing adjacent gaps.
    For the all-site utility, we have \(\res{{^kf}}{\yd{\varsigma}}{\yd{m}}{\yd{\mu}}_{\mathrm{all}} > \res{{^kf}}{\yd{\varsigma}}{\yd{m}}{\yd{\lambda}}_{\mathrm{all}}\) whenever
    \begin{equation}
    \begin{aligned}
       \Delta_{k-1,k}> &2m^2\mleft(m-1+\sum_{i=1}^{k-1}\frac{p_{i}}{D_{i,k}}\mright) &\,, \\
       \Delta_{k,k+1}> &2m^2\mleft(2m - 1 + \sum_{i=k+1}^{d}\frac{p_{i}}{D_{k,i}}\mright) \, .
    \end{aligned}
    \end{equation}
    In particular, it is sufficient to require \(\Delta_{k,\mathrm{min}}>2m^2\mleft(2m-1+1/D_{k,\mathrm{min}}\mright)\).
    For the one-site utility, we have \(\res{{^kf}}{\yd{\varsigma}}{\yd{m}}{\yd{\mu}}_{\mathrm{one}} > \res{{^kf}}{\yd{\varsigma}}{\yd{m}}{\yd{\lambda}}_{\mathrm{one}}\) whenever
    \begin{equation}
    \begin{aligned}
    \Delta_{k-1,k} > &2\sum_{i=1}^{k-1}\frac{p_{i}}{D_{i,k}}\,, \\
    \Delta_{k,k+1} > &\max\mleft(m,2\sum_{i=k+1}^{d}\frac{p_{i}}{D_{k,i}}\mright) \, .
    \end{aligned}
    \end{equation}
    It is sufficient to require \(\Delta_{k,\mathrm{min}}>\max\mleft(m,2/D_{k,\mathrm{min}}\mright)\).
\end{corollary}

\begin{proof}
    The first statement is a direct consequence of~\cref{thm:sector_wise_utility} and~\cref{cor:sector_wise_utility_simplified}. For the second statement, we know from~\cref{lem:one_site_channel_decomposition} that
    \begin{equation}
        \res{{^kf}}{\yd{\varsigma}}{\yd{m}}{\yd{\mu}}_{\mathrm{one}} = \res{{^kf}}{\yd{\varsigma}}{\yd{1}}{\yd{\varsigma - \mathbf{e}_k}}_{\mathrm{all}} \, ,\quad \res{{^kf}}{\yd{\varsigma}}{\yd{m}}{\yd{\lambda}}_{\mathrm{one}} = \sum_{i=1}^{d}\resSix{F}{\yd{\lambda}}{\yd{m-1}}{\yd{1}}{\yd{\varsigma-\mathbf{e}_i}}{\yt{1\cdots m}}{\yd{\varsigma}}^2 \res{{^kf}}{\yd{\varsigma}}{\yd{1}}{\yd{\varsigma - \mathbf{e}_i}}_{\mathrm{all}}  \, .
    \end{equation}
    Here, the first equality is due to the fact that \(\resSix{F}{\yd{\mu}}{\yd{m-1}}{\yd{1}}{\yd{\varsigma-\mathbf{e}_k}}{\yt{1\cdots m}}{\yd{\varsigma}}^2 = 1\) for \(\Delta_{k,k+1}\geq m\). The squares of the \(F\)-symbols sum up to \(1\), and together with~\cref{cor:sector_wise_utility_simplified} and~\cref{lem:optimality_of_F_symbol} this gives us
    \begin{equation}
        \res{{^kf}}{\yd{\varsigma}}{\yd{m}}{\yd{\mu}}_{\mathrm{one}} - \res{{^kf}}{\yd{\varsigma}}{\yd{m}}{\yd{\lambda}}_{\mathrm{one}} \geq \mleft(1 - \resSix{F}{\yd{\lambda}}{\yd{m-1}}{\yd{1}}{\yd{\varsigma-\mathbf{e}_k}}{\yt{1\cdots m}}{\yd{\varsigma}}^2 \mright)\mleft(\res{{^kf}}{\yd{\varsigma}}{\yd{1}}{\yd{\varsigma - \mathbf{e}_k}}_{\mathrm{all}} - \max_{i\neq k} \res{{^kf}}{\yd{\varsigma}}{\yd{1}}{\yd{\varsigma - \mathbf{e}_i}}_{\mathrm{all}}\mright) > 0 \, .
    \end{equation}
\end{proof}

We finally show that the overhang removal rule is exactly optimal for depolarized inputs under \(^1\mathcal{L}_{\mathrm{one}}\).

\begin{corollary}[Optimal QPA for depolarizing noise and \({}^1\mathcal{L}_{\mathrm{one}}\)]
    \label{cor:optimality_QPA_depolarizing_F_one}
    Consider the depolarized inputs with spectrum $\boldsymbol{p}\dot{=}(1-\eta \frac{d-1}{d},\frac{\eta}{d},\ldots,\frac{\eta}{d})$ for \(0<\eta<1\). Let $\yd{\varsigma}\dot{=}[\varsigma_{1},\ldots,\varsigma_{d}]$, and let $m\geq 1$ and $\yd{\lambda}$ such that $\res{c}{\yd{m}}{\yd{\lambda}}{\yd{\varsigma}} \geq 1$. Further, let $\yd{\mu}$ be given by the overhang removal rule of~\cref{def:overhang_removal_rule}, with $k\!=\!1$.
    Then we have
    \begin{equation}\begin{aligned}
       \res{{^1f}}{\yd{\varsigma}}{\yd{m}}{\yd{\mu}}_{\mathrm{one}} \geq  \res{{^1f}}{\yd{\varsigma}}{\yd{m}}{\yd{\lambda}}_{\mathrm{one}} \, ,
    \end{aligned}\end{equation}
    with equality only if \(\yd{\lambda} = \yd{\mu}\).
\end{corollary}

\begin{proof}
    By~\cref{lem:one_site_channel_decomposition}, we have
    \begin{equation}\begin{aligned}
        \res{{^1f}}{\yd{\varsigma}}{\yd{m}}{\yd{\lambda}}_{\mathrm{one}}
        = \sum_{i:\yd{\varsigma-\mathbf{e}_i} \vdash_{d}(n,1)}\resSix{F}{\yd{\lambda}}{\yd{m-1}}{\yd{1}}{\yd{\varsigma-\mathbf{e}_i}}{\yd{m}}{\yd{\varsigma}}^2 \res{{^1f}}{\yd{\varsigma}}{\yd{1}}{\yd{\varsigma-\mathbf{e}_i}}_{\mathrm{all}} \, .
    \end{aligned}\end{equation}
   ~\cref{lem:optimality_of_F_symbol} implies $\yd{\mu}$ maximizes the expressions
    \begin{equation}\begin{aligned}
        \label{equ:proof_optimality_QPA_depolarizing_F_one_maximal}
        \sum_{i=1}^{j} \resSix{F}{\yd{\lambda}}{\yd{m-1}}{\yd{1}}{\yd{\varsigma-\mathbf{e}_i}}{\yd{m}}{\yd{\varsigma}}^2 \leq \sum_{i=1}^{j} \resSix{F}{\yd{\mu}}{\yd{m-1}}{\yd{1}}{\yd{\varsigma-\mathbf{e}_i}}{\yd{m}}{\yd{\varsigma}}^2
    \end{aligned}\end{equation}
    for all $1\leq j \leq d$, with equality only if both sides are \(1\) or \(\yd{\lambda} = \yd{\mu}\). From Ref.~\cite[Appendix~S2B]{LFIC24S}, we know that for $i< j$, we have
    \begin{equation}\begin{aligned}
        \label{equ:proof_optimality_QPA_depolarizing_F_one_maximal_OQPA_result}
        \res{{^1f}}{\yd{\varsigma}}{\yd{1}}{\yd{\varsigma-\mathbf{e}_j}}_{\mathrm{all}} < \res{{^1f}}{\yd{\varsigma}}{\yd{1}}{\yd{\varsigma-\mathbf{e}_i}}_{\mathrm{all}} \, .
    \end{aligned}\end{equation}
    If all partial sums in~\cref{equ:proof_optimality_QPA_depolarizing_F_one_maximal} are equal to \(1\), this implies that $i^\ast=1$ where we remove \(m\) boxes from the first row. This forces $\yd{\lambda}=\yd{\mu}$. In any other case, at least for \(j=1\),~\cref{equ:proof_optimality_QPA_depolarizing_F_one_maximal} becomes strict.
    Together with~\cref{equ:proof_optimality_QPA_depolarizing_F_one_maximal_OQPA_result}, this proves the claim.
\end{proof}

\putbib[qpa_sm.bib]
\end{bibunit}
\end{document}

%% file: notations.tex
\usepackage{mleftright}
\usepackage{xparse}

\ExplSyntaxOn
\cs_set_protected:Npn \coh_asymp_print:nnnn #1#2#3#4
  {
    #1
    \tl_if_blank:nF {#3} { \sp{(#3)} }
    \tl_if_blank:nF {#2} { \sb{#2} }
    \mleft( #4 \mright)
  }

\DeclareDocumentCommand{\bigO}{O{} O{} m}
  {
    \ensuremath{\coh_asymp_print:nnnn {O} {#1} {#2} {#3}}
  }
\DeclareDocumentCommand{\lilo}{O{} O{} m}
  {
    \ensuremath{\coh_asymp_print:nnnn {o} {#1} {#2} {#3}}
  }
\DeclareDocumentCommand{\bigOm}{O{} O{} m}
  {
    \ensuremath{\coh_asymp_print:nnnn {\Omega} {#1} {#2} {#3}}
  }
\DeclareDocumentCommand{\lilom}{O{} O{} m}
  {
    \ensuremath{\coh_asymp_print:nnnn {\omega} {#1} {#2} {#3}}
  }
\DeclareDocumentCommand{\bigTheta}{O{} O{} m}
  {
    \ensuremath{\coh_asymp_print:nnnn {\Theta} {#1} {#2} {#3}}
  }
\ExplSyntaxOff

\newcommand{\yd}[1]{{\mleft[#1\mright]}}
\newcommand{\yt}[1]{{\mleft(#1\mright)}}
\newcommand{\wt}[1]{{\mleft\{#1\mright\}}}

\newcommand{\graphpath}[2]{\langle #1,#2 \rangle}

\def\multiset#1#2{\ensuremath{\mleft(\kern-.3em\mleft(\genfrac{}{}{0pt}{}{#1}{#2}\mright)\kern-.3em\mright)}}

\DeclarePairedDelimiter\floor{\lfloor}{\rfloor}

\newcommand{\triple}[3]{%
  {\begin{array}{@{}c@{}c@{}} & #1\,#2\\ & #3 \end{array}}%
}

\newcommand{\res}[4]{%
  #1^{\scalebox{0.6}{$\triple{#2}{#3}{#4}$}}%
}

\newcommand{\tripleSix}[6]{%
  {\begin{array}{@{}c@{}c@{}}%
    & #1\quad #2\quad #3\\%
    & \quad #4\,#5\\%
    & \quad\quad #6%
  \end{array}}%
}

\newcommand{\intertwiner}[5]{#1^{#2\xrightarrow{#5}#3\,#4}}
\newcommand{\intertwineradj}[5]{#1^{#3\,#4\xrightarrow{#5}#2}}

\newcommand\ue{\mathrel{\bullet\mkern-3mu{-}\mkern-3mu\bullet}}

\NewDocumentCommand{\extChannel}{m m m m o}{%
  #1^{#2 \xrightarrow{\IfValueTF{#5}{#4,#5}{#4}} \,#3}%
}

\newcommand{\resSix}[7]{%
  #1\mleft[{\scalebox{0.6}{$\tripleSix{#2}{#3}{#4}{#5}{#6}{#7}$}}\mright]%
}

\newcommand{\rpoch}[2]{\mleft(#1\mright)^{\overline{#2}}}
\newcommand{\fpoch}[2]{\mleft(#1\mright)^{\underline{#2}}}
\newcommand{\lw}[0]{{\rm lw}}
\newcommand{\hw}[0]{{\rm hw}}

\NewDocumentCommand{\schur}{o m o}{%
  \ensuremath{s^{#2}\IfValueT{#1}{_{#1}}\IfValueT{#3}{\!\mleft(#3\mright)}}%
}

\NewDocumentCommand{\weylavg}{m m m}{%
  \ensuremath{\mleft\langle #1 \mright\rangle_{\mathrm{W}\mleft(#2,#3\mright)}}%
}
\NewDocumentCommand{\swavg}{m m m}{%
  \ensuremath{\mleft\langle #1 \mright\rangle_{\mathrm{SW}\mleft(#2,#3\mright)}}%
}

\NewDocumentCommand{\swprob}{m m m}{%
  \ensuremath{\Pr\nolimits_{\mathrm{SW}\mleft(#2,#3\mright)}\mleft[#1\mright]}%
}

\usepackage{xcolor}
\definecolor{lavender}{rgb}{0.64,0.47,0.87}
\definecolor{darklavender}{rgb}{0.48,0.35,0.65}
\definecolor{blue}{rgb}{0.12156862745098039,0.4666666666666667,0.7058823529411764}
\definecolor{paleorange}{rgb}{0.8431372549019608,0.5568627450980392,0.12549019607843137}
\definecolor{defgreen}{rgb}{0.28627450980392155,0.5490196078431373,0.2549019607843137}
\definecolor{jam}{rgb}{0.7686274509803921, 0.34901960784313724, 0.6313725490196078}